\documentclass[unsortedaddress,superscriptaddress,pra,notitlepage]{revtex4}

\usepackage{amsmath,amssymb}
\usepackage{braket}
\usepackage{theorem}
\usepackage{color}
\usepackage{amsfonts}
\usepackage[mathscr]{eucal}
\usepackage[hidelinks]{hyperref}
\usepackage{graphicx}
\usepackage{xcolor}
\usepackage{tikz}
\usepackage{scalerel}
\usepackage{bm}

\usetikzlibrary{calc,decorations.pathreplacing}
\usetikzlibrary{arrows,shapes}

\newtheorem{definition}{Definition}[section]

\newtheorem{thm}[definition]{Theorem}
\newtheorem{prop}[definition]{Proposition}
\newtheorem{lemma}[definition]{Lemma}

\newtheorem{rem}[definition]{Remark}

\newtheorem{cor}[definition]{Corollary}

\newtheorem{example}[definition]{Example}

\newenvironment{proof}[1][Proof]{\begin{trivlist}
\item[\hskip \labelsep {\bfseries #1}]}{\hfill$\Box$\end{trivlist}}

\newcommand{\nn}{\nonumber \\}

\def\theta{\vartheta}
\def\hil{{\mathcal H}}
\def\kil{{\mathcal K}}

\def\A{{\mathcal A}}

\def\B{{\mathcal B}}

\def\cD{{\mathcal D}}
\def\E{{\mathcal E}}
\def\F{{\mathcal F}}

\def\I{{\mathcal I}}
\def\J{{\mathcal J}}

\def\M{\mathcal{M}}

\def\P{{\mathcal P}}
\def\S{{\mathcal S}}

\def\X{{\mathcal X}}
\def\Y{{\mathcal Y}}

\def\half{\frac{1}{2}}
\def\imp{\implies}
\def\pmi{\impliedby}
\def\ep{\varepsilon}
\def\bN{\mathbb{N}}
\def\bC{\mathbb{C}}

\def\bR{\mathbb{R}}
\def\bZ{\mathbb{Z}}

\def\bP{\mathbb{P}}
\def\bU{\mathbb{U}}
\def\bz{\left(}
\def\jz{\right)}

\def\inv{^{-1}}

\def\egy{\mathbf 1}
\def\map{\Phi}

\def\W{W}

\def\sa{\mathrm{sa}}

\def\what{\widehat}
\def\oll{\overline}

\def\trans{^{\mathrm{T}}}

\def\rho{\varrho}
\def\povm{\mathrm{POVM}}

\def\ptp{\mathrm{PTP}}

\def\cl{\mathrm{cl}}

\def\nn{\nonumber}

\def\cptp{\mathrm{CPTP}}

\def\pplus{\mathrm{P}^+}

\def\nw{^{*}}

\def\meas{\mathrm{meas}}

\def\fdd{^{[1]}}

\def\bary{\mathrm{b}}
\def\bal{\mathrm{left}}

\def\p{_{\ge 0}}
\def\pne{_{\gneq 0}}
\def\pp{_{>0}}
\def\bs{\max}
\def\um{\mathrm{Um}}
\def\valt{\cdot}

\def\qv{\mathbf{q}}
\def\qn{q_0}
\def\qo{q_1}
\def\BS{\max}
\def\Um{\mathrm{Um}}

\def\ac{\mathrm{ac}}
\def\sD{\tilde D}
\def\comm{\mathrm{cm}}
\def\first{\rho}
\def\second{\sigma}
\def\ft{\first}
\def\sd{\second}
\def\sigmasupp{S}

\newcommand{\ki}[1]{\emph{#1}}

\newcommand{\s}{\mbox{ }}
\newcommand{\ds}{\mbox{ }\mbox{ }}

\newcommand{\norm}[1]{\left\| #1\right\|}

\newcommand{\inner}[2]{\left\langle #1 , #2\right\rangle}

\newcommand{\vecc}[1]{\underline{#1}}

\newcommand{\diad}[2]{\left|#1\right\rangle\!\left\langle #2\right|}

\newcommand{\pr}[1]{\diad{#1}{#1}}
\newcommand{\wtilde}[1]{\widetilde{#1}}

\newcommand{\verteq}{\rotatebox{90}{$=$}}

\newcommand{\ch}[1]{W_{#1}}
\newcommand{\Wx}[1]{W_{#1}}
\newcommand{\persp}[1]{{#1}^{pp}}
\newcommand{\QP}[1]{Q^{\bary,\qv}_P(#1)}
\newcommand{\psiP}[1]{\psi^{\bary,\qv}_P(#1)}

\newcommand{\potimes}{\mathbin{\hat{\otimes}}}
\newcommand{\signed}[1]{\wtilde{#1}}
\newcommand{\acc}[2]{A(#1\,|\,#2)}

\renewcommand{\theenumi}{(\roman{enumi})}

\makeatletter
\renewcommand{\p@enumii}{}
\makeatother

\DeclareMathOperator{\id}{id}
\DeclareMathOperator{\Tr}{Tr}

\DeclareMathOperator{\supp}{supp}

\DeclareMathOperator{\ran}{ran}

\DeclareMathOperator{\spann}{span}

\DeclareMathOperator{\spec}{spec}

\DeclareMathOperator{\ccdots}{\cdot\ldots\cdot}

\DeclareMathOperator{\logn}{\widehat\log}

\DeclareMathOperator{\divv}{\Delta}

\DeclareMathOperator{\D}{\mathit{D}}
\DeclareMathOperator{\DU}{\mathit{D}^{Um}}
\DeclareMathOperator{\Q}{\mathit{Q}}
\DeclareMathOperator{\FD}{D}
\DeclareMathOperator{\rt}{\Gamma}
\DeclareMathOperator{\nlog}{\what\log}
\DeclareMathOperator*{\medcap}{\scalerel*{\cap}{\textstyle\sum}}

\usepackage{subfiles}

\begin{document}

%

\title{Geometric relative entropies and barycentric R\'enyi divergences}

\author{Mil\'an Mosonyi}
\email{milan.mosonyi@gmail.com}

\affiliation{MTA-BME Lend\"ulet Quantum Information Theory
Research Group, Budapest University of Technology and
Economics, M\H uegyetem rkp.~3., H-1111 Budapest,
Hungary}

\affiliation{HUN-REN Alfr\'ed R\'enyi Institute of Mathematics, Re\'altanoda street 13-15, H-1053, Budapest}

\affiliation{Department of Analysis and Operations Research, Institute of Mathematics,
Budapest University of Technology and Economics,
M\H uegyetem rkp.~3., H-1111 Budapest, Hungary}

\author{Gergely Bunth}
\email{gbunthy@gmail.com}

\affiliation{MTA-BME Lend\"ulet Quantum Information Theory
Research Group, Budapest University of Technology and
Economics, M\H uegyetem rkp.~3., H-1111 Budapest,
Hungary}

\affiliation{HUN-REN Alfr\'ed R\'enyi Institute of Mathematics, Re\'altanoda street 13-15, H-1053, Budapest}

\author{P\'eter Vrana}
\email{vranap@math.bme.hu}

\affiliation{MTA-BME Lend\"ulet Quantum Information Theory
Research Group, Budapest University of Technology and
Economics, M\H uegyetem rkp.~3., H-1111 Budapest,
Hungary}

\affiliation{Department of Algebra and Geometry, Institute of Mathematics,
Budapest University of Technology and Economics,
M\H uegyetem rkp.~3., H-1111 Budapest, Hungary}

\begin{abstract}
\centerline{\textbf{Abstract}}
\vspace{.3cm}

We give systematic ways of defining monotone quantum relative entropies and (multi-variate)
quantum R\'enyi divergences starting from a set of 
monotone quantum relative entropies. 

Interestingly, despite its central importance in information theory, only 
two additive and monotone quantum extensions of the classical relative entropy have been 
known so far, the Umegaki and the Belavkin-Staszewski relative entropies, which are the minimal and the maximal ones, respectively, with these properties.
Using the Kubo-Ando weighted geometric means, we give a general procedure to construct monotone and additive quantum relative entropies from a given one with the same properties; in particular, when starting from the Umegaki relative entropy, this gives a new one-parameter family of 
monotone (even under positive trace-preserving (PTP) maps) and additive quantum relative entropies interpolating between the Umegaki and the Belavkin-Staszewski ones on full-rank states.

In a different direction, we use a generalization of a classical 
variational formula 
to define multi-variate quantum R\'enyi quantities 
corresponding to any finite set of quantum relative entropies
$(D^{q_x})_{x\in\X}$ and real weights $(P(x))_{x\in\X}$ summing to $1$, as 
\begin{align*}
Q_P^{\bary,\qv}((\rho_x)_{x\in\X}):=\sup_{\tau\ge 0}
\left\{\Tr\tau-\sum_xP(x)D^{q_x}(\tau\|\rho_x)\right\}.
\end{align*}
We analyze in detail the properties of the resulting quantity inherited from the generating set of quantum relative entropies; in particular, we show that monotone quantum relative entropies define monotone R\'enyi quantities whenever $P$ is a probability measure. 
With the proper normalization, the negative logarithm of the above quantity gives 
a quantum extension of the 
classical R\'enyi $\alpha$-divergence in the $2$-variable case
($\X=\{0,1\}$, $P(0)=\alpha$).
We show that 
if both $D^{\qn}$ and $D^{\qo}$ are lower semi-continuous, monotone, and additive quantum relative entropies,
and at least one of them is strictly larger than the Umegaki relative entropy then the resulting 
barycentric R\'enyi divergences are strictly between the log-Euclidean and the maximal
R\'enyi divergences, 
and hence they are different from any previously studied quantum R\'enyi divergence.
\vspace{.2cm}

\vspace{.2cm}


\end{abstract}

\maketitle

\tableofcontents

\section{Introduction}

Dissimilarity measures of states of a system (classical or quantum) play a fundamental role in information theory, statistical physics, computer science, and various other disciplines. 
Probably the most relevant such measures for information theory are the \ki{R\'enyi divergences},
defined for finitely supported probability distributions $\rho,\sigma$ on a set $\I$ as
\begin{align}\label{classical Renyi}
\D_{\alpha}(\rho\|\sigma):=
\frac{1}{\alpha-1}\log\sum_{i\in\I}\rho(i)^{\alpha}\sigma(i)^{1-\alpha}
\end{align}
where $\alpha\in[0,+\infty)\setminus\{1\}$ is a parameter. (For simplicity, here we assume all probability distributions and quantum states to have full support; we give the formulas for the general case in Section \ref{sec:qdiv}.)
See, for instance, \cite{Csiszar} for the role of the R\'enyi divergences 
and derived information measures (entropy, divergence radius, channel capacity)
in classical state discrimination, as well as source- and channel coding. The limit $\alpha\to 1$ yields the \ki{Kullback-Leibler divergence},
or \ki{relative entropy}
\begin{align}\label{KL}
\lim_{\alpha\to 1}D_{\alpha}(\rho\|\sigma)=
\D(\rho\|\sigma):=
\sum_{i\in\I}\left[\rho(i)\log\rho(i)-\rho(i)\log\sigma(i)\right].
\end{align}
Interestingly, the relative entropy in itself also determines the whole one-parameter family of R\'enyi divergences, as for every $\alpha\in[0,1)\cup(1,+\infty)$, 
\begin{align}\label{classical variational}
\D_{\alpha}(\rho\|\sigma)=\frac{1}{1-\alpha}\min_{\omega\in\P(\I)}\left\{\alpha D(\omega\|\rho)+(1-\alpha)D(\omega\|\sigma)\right\},
\end{align}
where the optimization is over all finitely supported probability distributions $\omega$ on 
$\I$ \cite{CsM2003}.

Due to the non-commutativity of quantum states, for any given 
$\alpha\in[0,+\infty)$, there are infinitely many 
quantum extensions of the classical R\'enyi $\alpha$-divergence for 
pairs of quantum states, e.g., the measured, the maximal \cite{Matsumoto_newfdiv}, 
or the R\'enyi $(\alpha,z)$-divergences \cite{AD}.
Of particular importance are the Petz-type \cite{P86} and the 
sandwiched \cite{Renyi_new,WWY} R\'enyi divergences, which appear as exact quantifiers of the trade-off relations between the operationally relevant quantities in a number of 
information-theoretic problems, including state discrimination, classical-quantum channel coding, entanglement manipulation, and more 
\cite{ANSzV,hayashi2002error,Hayashicq,jensen2019asymptotic,LiYao_Reliability_2021,
LiYao_Channel_simulation,LiYao2022,LiYao_EA_sc,LiYaoHayashi2023,Nagaoka,
MO-cqconv,MO-cqconv-cc}.
While so far only these two families of quantum R\'enyi divergences have found such explicit operational interpretations, 
it is nevertheless useful, for a number of different reasons, to consider other quantum 
extensions as well. Indeed, apart from their study being interesting from the purely 
mathematical point of view of matrix analysis, some of these quantities serve as useful tools 
in proofs to arrive at the operationally relevant R\'enyi information quantities in various 
problems;
see, e.g., the role played by the so-called log-Euclidean R\'enyi divergences
$D_{\alpha,+\infty}$ in determining the strong converse exponent in various problems
\cite{LiYao_Reliability_2021,LiYao2022,LiYao_EA_sc,LiYaoHayashi2023,MO-cqconv,MO-cqconv-cc}, or 
the family of R\'enyi divergences $D_{\alpha}^{\#}$ 
introduced in \cite{FawziFawzi2021}, 
where it was used to determine the strong converse exponent of binary channel discrimination. 

Of course, only quantum extensions with a number of good mathematical properties may be interesting for quantum information theory, the most important being monotonicity, 
i.e., that for any two states $\rho,\sigma$, and 
completely positive trace-preserving (CPTP) map 
$\map$, the data processing inequality (DPI)  $D_{\alpha}(\map(\rho)\|\map(\sigma))\le D_{\alpha}(\rho\|\sigma)$ holds; this in turn is strongly related to the convexity properties of these quantities. Another desirable property is additivity
$D_{\alpha}(\rho_1\otimes\rho_2\|\sigma_1\otimes\sigma_2)
=D_{\alpha}(\rho_1\|\sigma_1)+D_{\alpha}(\rho_2\|\sigma_2)$.
However, quantum R\'enyi divergences without these properties might still be useful; indeed, 
$D_{\alpha,+\infty}$ is additive, but not monotone for $\alpha>1$ (the range of $\alpha$ values for which it was used in \cite{MO-cqconv,MO-cqconv-cc}), and 
$D_{\alpha}^{\#}$ is monotone, but not additive. 
The study of the mathematical properties of various quantum R\'enyi divergences and related information quantities and trace functionals has been the subject of intensive research in matrix analysis, functional analysis, and operator algebras in the past several decades; see, e.g., 
\cite{Ando,Beigi,BST,CFL,CL-MII,FL,Hiai_concavity2001,
Hiai-convexity,Hiai_fdiv_standard,Hiai_fdiv_max,Hiai_fdiv_Springer,sc_vN,
JP,Jencova_rev16,Jencova_NCLp,Jencova_NCLpII,Kosaki1982,
Lieb-Ruskai,OP,
Petz_Properties1986,Petz_QE_vN,P86,Zhang2018} and references therein.

Quantum divergences with good mathematical properties also play an important role in the 
study of the problem of state convertibility, where the question is whether 
a set of states
$(\rho_x)_{x\in\X}$ can be mapped into another set of states
$(\rho_x')_{x\in\X}$ with a joint quantum operation. This problem can be studied in a 
large variety of settings; single-shot or asymptotic, exact or approximate, with or without catalysts, allowing arbitrary quantum operations or only those respecting some symmetry or being free operations of some resource theory, 
any combination of these, and more. 
Necessary conditions for convertibility can be obtained using 
multi-variate functions on quantum states with suitable mathematical properties; 
for instance, for exact single-shot convertibility,
$F((\rho_x)_{x\in\Y})\ge F((\rho_x')_{x\in\Y})$
has to hold for any function $F$ whose variables are indexed by a subset 
$\Y\subseteq\X$ and which is 
monotone non-increasing under the joint application of an allowed quantum operation
on its arguments; the same has to hold
also for multi-copy or catalytic single-shot convertibility, if $F$
is additionally additive on tensor products, and 
for asymptotic catalytic convertibility, if, moreover, 
$F$ is lower semi-continuous in its variables. 
(We refer to \cite{farooq2023asymptotic} for the precise definitions of the various 
versions of state convertibility.)
Many of the quantum R\'enyi divergences provide such functions on pairs of states
(i.e., $|\Y|=2$), 
including the maximal
R\'enyi $\alpha$-divergences \cite{Matsumoto_newfdiv} with $\alpha\in[0,2]$, 
and the R\'enyi $(\alpha,z)$-divergences \cite{AD} for certain values of 
$\alpha$ and $z$ \cite{Zhang2018}.
In the converse direction, sufficient conditions in terms of R\'enyi divergences have been given for the convertibility of pairs of commuting states
in \cite{Brandao_etal_secondlaws,jensen2019asymptoticmajorization,Klimesh2007,Mu_Econometrica2020,turgut2007catalytic}, and these have been extended very recently 
in \cite{farooq2023asymptotic} to a complete characterization 
of asymptotic as well as approximate catalytic convertibility 
between finite sets of commuting states 
in terms of the monotonicity of the multi-variate R\'enyi quantities
\begin{align}\label{multiQ classical}
Q_{\vecc{\alpha}}(\rho_1,\ldots,\rho_r):=\Tr(\rho_1^{\alpha_1}\cdot\ldots\cdot\rho_r^{\alpha_r}),
\end{align}
where $\alpha_1+\ldots+\alpha_r=1$ and either all of them are non-negative or exactly one of 
them is positive. No sufficient conditions, however, are known in the general noncommutative 
case.

Motivated by the above, in this paper we set out to give systematic ways 
to define quantum R\'enyi divergences with good mathematical properties
(in particular, monotonicity), for two and for more variables. 
Other approaches to define multi-variate quantum R\'enyi divergences have been proposed very recently in \cite{bunth2021equivariant,FuruyaLashkariOuseph2023}, which we briefly review in Section \ref{sec:nc gm}; here we only note that our approach is completely different from the previous ones.
Indeed, while multi-variate quantum R\'enyi divergences in \cite{bunth2021equivariant,FuruyaLashkariOuseph2023}
were defined by putting a non-commutative geometric mean
(an iterated Kubo-Ando weighted geometric mean in \cite{FuruyaLashkariOuseph2023})
of the arguments
into the second argument of a R\'enyi $(\alpha,z)$-divergence 
(sandwiched R\'enyi divergenc in \cite{bunth2021equivariant}),
here we start with a collection of quantum relative entropies, and 
define corresponding multi-variate quantum R\'enyi divergences via a variational expression.

The structure of the paper is as follows. In Section \ref{sec:prelim} we give the necessary mathematical preliminaries. In Section \ref{sec:multidiv} we discuss in detail the definition 
and various properties of general multi-variate quantum divergences. 
Sections \ref{sec:cl Renyi} and \ref{sec:qRenyi} contain brief reviews of the definitions and properties of the 
classical and the quantum R\'enyi divergences that we use in the paper.
Section \ref{sec:nc gm} gives a high-level overview of the various ways we propose to 
define new quantum R\'enyi divergences from given ones, of which we work out in detail 
two approaches in this paper. We also illustrate other possible approaches
in Figure \ref{fig:multiRenyi} in Appendix \ref{sec:multiRenyi}.

In Section \ref{sec:gamma relentr} we focus on quantum relative entropies, 
i.e., various quantum extensions of the classical relative entropy.
Interestingly, despite its central importance in information theory, only 
two additive and monotone quantum extensions of the classical relative entropy have been 
known so far, the Umegaki \cite{Umegaki} and the Belavkin-Staszewski \cite{BS} relative entropies, which are the minimal and the maximal with these properties, respectively
\cite{Matsumoto_newfdiv}. 
Here we give a general procedure to construct monotone and additive quantum relative entropies from a given one with the same properties; in particular, when starting from the Umegaki relative entropy, this gives a new one-parameter family of 
monotone (even under positive trace-preserving (PTP) maps) and additive quantum relative entropies interpolating between the Umegaki and the Belavkin-Staszewski ones on full-rank states.

In Section \ref{sec:barycentric} we use a generalization of the classical 
variational formula in \eqref{classical variational}
to define quantum extensions of the multi-variate R\'enyi quantities 
\eqref{multiQ classical} corresponding to any set of quantum relative entropies
$(D^{q_x})_{x\in\X}$ and real weights $(P(x))_{x\in\X}$ summing to $1$, as 
\begin{align}\label{multiQ Intro}
Q_P^{\bary,\qv}((\rho_x)_{x\in\X}):=\sup_{\tau\ge 0}
\left\{\Tr\tau-\sum_xP(x)D^{q_x}(\tau\|\rho_x)\right\}.
\end{align}
We analyze in detail the properties of the resulting quantity inherited from the generating set of quantum relative entropies; in particular, we show that monotone quantum relative entropies define monotone R\'enyi quantities whenever $P$ is a probability measure. 
We also show that the negative logarithm of $Q_P^{\bary,\qv}((\rho_x)_{x\in\X})$
is equal to the $P$-weighted (left) relative entropy radius of the 
$(\rho_x)_{x\in\X}$, i.e., 
\begin{align}\label{multipsi Intro}
-\log Q_P^{\bary,\qv}((\rho_x)_{x\in\X})=\inf_{\omega}
\sum_xP(x)D^{q_x}(\omega\|\rho_x),
\end{align}
where the infimum is taken over all states on the given Hilbert space, and therefore we 
call the quantities in \eqref{multiQ Intro} \ki{barycentric R\'enyi quantities}.
With the proper normalization, the quantities in \eqref{multipsi Intro} give quantum extensions of the 
classical R\'enyi divergences \eqref{classical Renyi} in the $2$-variable case
($\X=\{0,1\}$, $P(0)=\alpha$), which we denote by $D_{\alpha}^{\bary,\qv}$. 
In Section \ref{sec:ex} we study the relation of the resulting ($2$-variable) quantum R\'enyi 
divergences to the known ones. 
It has been shown in \cite{MO-cqconv} that if $D^{\qn}=D^{\qo}=\DU$ 
is the Umegaki relative entropy \cite{Umegaki} then $D_{\alpha}^{\bary,\qv}$
is equal to the log-Euclidean R\'enyi divergence $D_{\alpha,+\infty}$. We show that 
if both $D^{\qn}$ and $D^{\qo}$ are lower semi-continuous monotone and additive quantum relative entropies,
and at least one of them is strictly larger than $\DU$ then the resulting 
barycentric R\'enyi divergences are strictly between the log-Euclidean and the maximal
R\'enyi divergences, 
and hence they are different from any previously studied quantum R\'enyi divergences.
Figure \ref{fig:barycentric} in Appendix \ref{sec:ordering} illustrates how the 
new quantum relative entropies and ($2$-variable) R\'enyi divergences considered in this paper are related
to other such quantities considered in the literature before.

Appendices \ref{sec:jointconc} and \ref{sec:proof of lemma} contain proofs
relegated from the main body of the paper.

\section{Preliminaries}
\label{sec:prelim}

For a finite-dimensional Hilbert space $\hil$, let $\B(\hil)$ denote the set of all linear operators on $\hil$, 
and let $\B(\hil)_{\sa}$, $\B(\hil)_{\ge 0}$, $\B(\hil)_{\gneq 0}$, and 
$\B(\hil)_{>0}$ denote the set of 
self-adjoint, positive semi-definite (PSD), non-zero positive semi-definite, and positive definite operators, respectively. 
For an interval $J\subseteq\bR$, let 
$\B(\hil)_{\sa,J}:=\{A\in\B(\hil)_{\sa}:\spec(A)\subseteq J\}$, i.e., 
the set of self-adjoint operators on $\hil$ with all their eigenvalues in $J$.
Let $\S(\hil):=\{\rho\in\B(\hil)\p,\,\Tr\rho=1\}$ denote the set of \ki{density operators}, or \ki{states}. 
For an operator $X\in\B(\hil)$, 
\begin{align*}
\norm{X}_{\infty}:=\max\{\norm{X\psi}:\,\psi\in\hil,\,\norm{\psi}=1\}
\end{align*}
denotes the \ki{operator norm} of $X$ (i.e., its largest singular value). 

Similarly, for a finite set $\I$, we will use the notation $\F(\I):=\bC^{\I}$ for the set of 
complex-valued functions on $\I$, and 
$\F(\I)\p$, $\F(\I)\pne$, $\F(\I)\pp$ for the set of non-negative, non-negative and not constant zero, and strictly positive functions on $\I$. The set of probability density functions on $\I$ will be denoted by $\P(\I)$. When equipped with the maximum norm, 
$\F(\I)$ becomes a commutative $C^*$-algebra, which we denote by $\ell^{\infty}(\I)$. 
In the more general case when  $\I$ is an arbitrary non-empty set, we will also use the notations $\P_f(\I)$ for the set of 
finitely supported probability measures, and 
 $\P_f^{\pm}(\I)$ for the set of 
finitely supported signed probability measures on $\I$, i.e., 
\begin{align*}
\P_f^{\pm}(\I):=\left\{P\in\bR^{\I}:\,|\supp P|<+\infty,\,\sum_{i\in\I}P(i)=1\right\},\ds\ds
\supp P:=\{i\in\I:\,P(i)\ne 0\}.
\end{align*}
We also introduce the following subset of signed probability measures:
\begin{align*}
\P_{f,1}^{\pm}(\I):=
\left\{P\in\P_f^{\pm}(\I):\,
\exists\,i_+\in\I\text{ s.t. }P(i_+)>0\text{ and }P(i)\le 0,\,i\in\I\setminus\{i_+\} \right\},
\end{align*}
which plays an important role in the definition of multi-variate R\'enyi divergences; 
see Lemma \ref{lemma:cl Q uniqueness}. 

For any non-empty set $\X$, let 
\begin{align*}
\B(\X,\hil),\ds\ds
\B(\X,\hil)\p,\ds\ds
\B(\X,\hil)\pne,\ds\ds
\B(\X,\hil)\pp,\ds\ds
\S(\X,\hil),
\end{align*}
denote the set of functions mapping from $\X$ into 
$\B(\hil)$,
$\B(\hil)\p$, $\B(\hil)\pne$, $\B(\hil)\pp$, and
$\S(\hil)$, respectively. Elements of $\S(\X,\hil)$ are called \ki{classical-quantum channels}, or \ki{cq channels}, and we will use the terminology 
\ki{generalized classical-quantum channels}, or \ki{gcq channels}, 
for the elements of $\B(\X,\hil)\pne$.
We will normally use the notation $W=(W_x)_{x\in\X}$ to denote elements of 
$\B(\X,\hil)\p$.
We say that $W\in\B(\X,\hil)$ is \ki{classical} if 
there exists an orthonormal basis $(e_i)_{i\in\I}$ in $\hil$ such that 
$W_x=\sum_{i\in\I}\inner{e_i}{W_xe_i}\pr{e_i}$, $x\in\X$;
we call any such orthonormal basis a \ki{$W$-basis}. Equivalently, we may identify
$W$ with the collection of functions $((\wtilde W_x(i):=\inner{e_i}{W_xe_i})_{i\in\I})_{x\in\X}\in \F(\X,\I)$, where we use the notations
\begin{align*}
\F(\X,\I),\ds\ds
\F(\X,\I)\p,\ds\ds
\F(\X,\I)\pne,\ds\ds
\F(\X,\I)\pp,\ds\ds
\P(\X,\I),
\end{align*}
for the sets of functions mapping elements of $\X$ into functions $f_x\in\F(\I)$, 
$x\in\X$,
on the finite set $\I$, 
such that the $f_x$ are arbitrary/non-negative/non-negative and not constant zero/strictly positive/probability density functions on $\I$. 

Operations on elements of 
$\B(\X,\hil)$ are always meant pointwise; e.g., 
for any $\W,\W^{(1)},\W^{(2)}\in\B(\X,\hil)$, $V\in\B(\hil,\kil)$, and 
$\sigma\in\B(\kil)$, where $\kil$ is an arbitrary finite-dimensional Hilbert space,
\begin{align}
&V\W V^*:=(V\Wx{x}V^*)_{x\in\X},\ds\ds\\
&W^{(1)}\potimes W^{(2)}:=\bz \W^{(1)}_x\otimes\W^{(2)}_x\jz_{x\in\X},\ds\ds
\label{gqc tensor product soros}\\
&\W\potimes \sigma:=\W\otimes \sigma:=
(\Wx{x}\otimes \sigma)_{x\in\X}\,.\label{tensor with constant}
\end{align}
Note that here we only consider the (pointwise) tensor product of functions defined on the same set, and that this notion of tensor product is different from the one used to describe the 
parallel action of two cq channels, given by
\begin{align}\label{gqc tensor product}
W^{(1)}\otimes W^{(2)}:=\bz \W^{(1)}_{x_1}\otimes\W^{(2)}_{x_2}\jz_{(x_1,x_2)\in\X_1\times\X_2},
\end{align}
where $W^{(i)}\in\B(\X^{(i)},\hil^{(i)})$, $i=1,2$, and possibly 
$\X^{(1)}\ne \X^{(2)}$, $\hil^{(1)}\ne \hil^{(2)}$.
The tensor product in \eqref{tensor with constant} can be interpreted either in this setting, with $\X^{(1)}=\X$, $W^{(1)}=W$ and $\X^{(2)}=\{0\}$, $W^{(2)}_0=\sigma$, or as the pointwise 
tensor product between $W^{(1)}=W\in\B(\X,\hil)$ and the constant function
$W^{(2)}\in\B(\X,\hil)$, $W^{(2)}_x=\sigma$, $x\in\X$.

The set of projections on $\hil$ is denoted by 
$\bP(\hil):=\{P\in\B(\hil):\,P^2=P=P^*\}$.
For $P,Q\in\bP(\hil)$, the projection onto $(\ran P)\cap(\ran Q)$ is denoted by $P\wedge Q$.
For a sequence of projections $P_1,\ldots,P_r\in\B(\hil)$ summing to $I$, the corresponding 
\ki{pinching} operation is 
\begin{align*}
\B(\hil)\ni X\mapsto\sum_{i=1}^r P_iXP_i.
\end{align*}

For a self-adjoint operator $A$, let $P^A_a:=\egy_{\{a\}}(A)$ denote the spectral projection of $A$
corresponding to the singleton $\{a\}\subset\bR$. 
(Here and henceforth $\egy_H$ stands for the characteristic (or indicator) function
of a set $H$.)
The projection onto the support of $A$ is $\sum_{a\ne 0}P^A_a$; in particular, if $A$ is 
positive semi-definite, it is equal to $\lim_{\alpha\searrow 0}A^{\alpha}=:A^0$. In general, 
we follow the convention that real powers of a positive semi-definite operator $A$ are taken 
only on its support, i.e., for any $x\in\bR$, $A^x:=\sum_{a>0}a^x P^A_a$.
In particular, $A\inv:=\sum_{a>0}a\inv P^A_a$ stands for the generalized inverse of $A$,
and $A\inv A=AA\inv=A^0$. For $A\in\B(\hil)\p$ and a projection $P$ on $\hil$ we write
$A\in\B(P\hil)\p$ if $A^0\le P$. 

For two PSD operators $\rho$, $\sigma$, we write $\rho\perp\sigma$ if 
$\ran\rho\perp\ran\sigma$, which is equivalent to $\rho\sigma=0$, and further to 
$\inner{\rho}{\sigma}_{HS}:=\Tr\rho\sigma=0$, and to $\rho^0\sigma^0=0$.
In particular, it implies $\rho^0\wedge\sigma^0=0$, but not the other way around.

For two finite-dimensional Hilbert spaces $\hil,\kil$, we will use the notations 
$\ptp(\hil,\kil)$ and $\cptp(\hil,\kil)$ for the set of positive trace-preserving linear maps and 
the set of completely positive trace-preserving linear maps, respectively, from $\B(\hil)$ to $\B(\kil)$. We will also use the notation 
$\pplus(\hil,\kil)$ for the set of (positive) linear maps 
from $\B(\hil)$ to $\B(\kil)$ such that 
$\map(\rho)\in\B(\kil)\pne$ for all
$\rho\in\B(\hil)\pne$. We will also consider (completely) positive maps of the form
$\map:\,\B(\hil)\to\ell^{\infty}(\I)$ and 
$\map:\,\ell^{\infty}(\I)\to\B(\hil)$.

For a finite-dimensional Hilbert space $\hil$ and a natural number $n$, 
we denote by 
\begin{align*}
\povm(\hil,[n]):=\left\{M=(M_i)_{i=0}^{n-1}\in\B(\hil)\p^{[n]}:\,\sum_{i=0}^{n-1} M_i=I\right\}
\end{align*}
the set of \ki{$n$-outcome positive operator valued measures (POVMs)} on $\hil$, where
\begin{align*}
[n]:=\{0,\ldots,n-1\}.
\end{align*}
Any $M\in\povm(\hil,[n])$ determines a CPTP map $\M:\,\B(\hil)\to\ell^{\infty}([n])$ by 
\begin{align*}
\M(\valt):=\sum_{i=0}^{n-1}(\Tr M_i(\valt))\egy_{\{i\}}.
\end{align*}

For a differentiable function $f$ defined on an interval $J\subseteq\bR$, let 
$f\fdd:\,J\times J\to\bR$ be its \ki{first divided difference function}, defined as
\begin{align*}
f\fdd(a,b):=\begin{cases}
\frac{f(a)-f(b)}{a-b},&a\ne b,\\
f'(a),&a=b,
\end{cases}\ds\ds\ds a,b\in J.
\end{align*}
If $f$ is a continuously differentiable function on an open interval $J\subseteq\bR$ then for any finite-dimensional  Hilbert space $\hil$, $A\mapsto f(A)$ is Fr\'echet differentiable on $\B(\hil)_{\sa,J}$, and its 
Fr\'echet derivative $(Df)[A]$ at a point $A\in \B(\hil)_{\sa,J}$ is given by 
\begin{align}\label{opfunction derivative}
(Df)[A](Y)=\sum_{i,j=1}^rf\fdd(a_i,a_j)P_iYP_j,\ds\ds\ds Y\in\B(\hil)_{\sa},
\end{align}
for any $P_1,\ldots,P_r\in\bP(\hil)$ and $a_1,\ldots,a_r\in\bR$ such that 
$\sum_{i=1}^r P_i=I$, $\sum_{i=1}^r a_i P_i=A$.
See, e.g., \cite[Theorem V.3.3]{Bhatia} or \cite[Theorem 2.3.1]{Hiai_book}.

By $\log$ we denote the natural logarithm, and
we use two different extensions of it to $[0,+\infty]$, defined as
\begin{align*}
\log x:=\begin{cases}
-\infty,&x=0,\\
\log x,& x\in(0,+\infty),\\
+\infty,&x=+\infty,
\end{cases}
\ds\ds\ds
\nlog x:=\begin{cases}
0,&x=0,\\
\log x,& x\in(0,+\infty),\\
+\infty,&x=+\infty.
\end{cases}
\end{align*}
Throughout the paper we use the convention
\begin{align}\label{zero times infty}
0\cdot(\pm\infty):=0.
\end{align}

For a function $f:\,(0,+\infty)\to\bR$, the corresponding 
\ki{operator perspective function} \cite{Effros,ENG,HiaiMosonyi2017}
$\persp{f}$ is defined on pairs 
of positive definite operators $\rho,\sigma\in\B(\hil)\pp$ as 
\begin{align}\label{persp def}
\persp{f}(\rho,\sigma):=\sigma^{1/2}f\bz \sigma^{-1/2}\rho\sigma^{-1/2}\jz\sigma^{1/2},
\end{align}
and it is extended to pairs of positive semi-definite operators $\rho,\sigma$ as 
\begin{align}\label{persp extension}
\persp{f}(\rho,\sigma):=\lim_{\ep\searrow 0}\persp{f}(\rho+\ep I,\sigma+\ep I),
\end{align}
whenever the limit exists. 
It is easy to see that for the \ki{transpose function} $\tilde f(x):=xf(1/x)$, $x>0$, we have
\begin{align}\label{persp trans}
\persp{f}(\rho,\sigma)=\persp{\tilde f}(\sigma,\rho),
\end{align}
whenever both sides are well-defined. 
For any $\gamma\in(0,1)$,
the choice $f_{\gamma}:=\id_{[0,+\infty)}^{\gamma}$ gives the
\ki{Kubo-Ando $\gamma$-weighted geometric mean}, denoted by 
$\persp{(f_{\gamma})}(\rho,\sigma)=:\sigma\#_{\gamma}\rho$, $\rho,\sigma\in\B(\hil)\p$. 
We give a somewhat detailed exposition of these concepts for the purposes of the present paper in Appendix \ref{sec:oppersp}.
\medskip

The following is completely elementary; we state it explicitly because we use it 
multiple times in the paper.

\begin{lemma}\label{lemma:optimization}
For any $c\in\bR$, the supremum of the function 
$[0,+\infty)\ni t\mapsto t-t\log t-tc=:f_c(t)$ is $e^{-c}$,
attained uniquely at $t=e^{-c}$.
\end{lemma}
\begin{proof}
We have $f'(t)=-\log t-c=0\iff t=e^{-c}$,  $f''(t)=-1/t<0$, $t\in(0,+\infty)$,
from which the statement follows immediately.
\end{proof}

The following minimax theorem is from \cite[Corollary A.2]{MH}.
\begin{lemma}\label{lemma:minimax2}
Let $X$ be a compact topological space, $Y$ be an ordered set, and let $f:\,X\times Y\to \bR\cup\{-\infty,+\infty\}$ be a function. Assume that
\smallskip

\s(i) $f(.\,,\,y)$ is lower semi-continuous for every $y\in Y$ and
\smallskip

(ii) \begin{minipage}[t]{15cm}
$f(x,.)$ is monotonic increasing for every $x\in X$, or
$f(x,.)$ is monotonic decreasing for every\\
 $x\in X$.
\end{minipage}
Then 
\begin{align}\label{minimax statement}
\inf_{x\in X}\sup_{y\in Y}f(x,y)=
\sup_{y\in Y}\inf_{x\in X}f(x,y),
\end{align}
and the infima in \eqref{minimax statement} can be replaced by minima.
\end{lemma}
\medskip

The following might be known; however, we could not find a reference for it, so we provide a detailed proof. For simplicity, we include the Hausdorff property in the definition of a topological space, although the statement also holds without this.

\begin{lemma}\label{lemma:usc}
Let $X$ be a topological space, $Y$ be an arbitrary set, and 
$f:\,X\times Y\to\bR\cup\{\pm\infty\}$ be a function.
\begin{enumerate}
\item\label{usc1}
If $f(.,y)$ is lower semi-continuous for every $y\in Y$ then $\sup_{y\in Y}f(.,y)$ is lower semi-continuous.

\item\label{usc2}
If $Y$ is a compact topological space, and $f$ is lower semi-continuous on $X\times Y$ w.r.t.~the product topology, then 
$\inf_{y\in Y}f(.,y)$ is lower semi-continuous.

\item\label{usc3}
If $Y$ is a compact topological space, and $f$ is continuous on $X\times Y$ w.r.t.~the product topology, then 
$\inf_{y\in Y}f(.,y)$ and
$\sup_{y\in Y}f(.,y)$ are both continuous.
\end{enumerate}
\end{lemma}
\begin{proof}
The assertion in \ref{usc1} is trivial from the definition of lower semi-continuity.

To prove \ref{usc2}, let
$(x_i)_{i\in I}\subseteq X$ be a generalized sequence and $\oll x=\lim_{i\to\infty}x_i$.
Since $f(x_i,.)$ is  lower semi-continuous on $Y$, and $Y$ is compact, there exists a $y_i\in Y$ such that 
$\inf_{y\in Y}f(x_i,y)=f(x_i,y_i)$. 
Let 
$(f(x_{\alpha(j)},y_{\alpha(j)}))_{j\in J}$ be a subnet such that  
\begin{align*}
\liminf_{i\to\infty}f(x_i,y_i)=\lim_{j\to\infty}f(x_{\alpha(j)},y_{\alpha(j)}).
\end{align*}
Since $Y$ is compact, there exists a 
subnet $(y_{\alpha(\beta(k))})_{k\in K}\subseteq Y$ converging to some $\oll y\in Y$.
Then 
\begin{align*}
\inf_{y\in Y}f(\oll x,y)
&\le
f(\oll x,\oll y)
\le
\lim_{k\to\infty}f(x_{\alpha(\beta(k))},y_{\alpha(\beta(k))})
=
\lim_{j\to\infty}f(x_{\alpha(j)},y_{\alpha(j)})\\
&=
\liminf_{i\to\infty}f(x_i,y_i)
=
\liminf_{i\to\infty}\inf_{y\in Y}f(x_i,y),
\end{align*}
where the second inequality follows from the lower semi-continuity of $f$ on $X\times Y$. 

The assertion in \ref{usc3} follows immediately from \ref{usc1} and \ref{usc2}.
\end{proof}

We will also use Fekete's lemma, and a straightforward generalization of it, stated together below. 

\begin{lemma}\label{lemma:Fekete}
Let $a_n\in[-\infty,+\infty]$, $n\in\bN$, be a subadditive sequence, i.e., 
for every $n,m\in\bN$, $a_n+a_m$ is well-defined and 
$a_{n+m}\le a_n+a_m$. 
Then
\begin{align}\label{Fekete limit0}
\liminf_{n\to+\infty}\frac{1}{n}a_n=\inf_{n\in\bN}\frac{1}{n}a_n,
\end{align}
and the following are equivalent:
\begin{enumerate}
\item\label{Fekete1}
$a_n=+\infty$ for every $n\in\bN$, or there exists an $n_0\in\bN$ such that 
$a_n<+\infty$ for every $n\ge n_0$.
\item\label{Fekete2-1}
$a_n=+\infty$ for every $n\in\bN$, or there exists a $k\in\bN$ such that 
$a_k,a_{k+1}<+\infty$.
\item\label{Fekete2}
$a_n=+\infty$ for every $n\in\bN$, or there exist relative primes 
$k,m\in\bN$ such that 
$a_k,a_{m}<+\infty$. 
\item\label{Fekete3}
The sequence $(\frac{1}{n}a_n)_{n\in\bN}$ has a limit.
\end{enumerate}
In particular, \ref{Fekete1}--\ref{Fekete3} hold if the sequence is additionally monotone
non-increasing.
\end{lemma}
\begin{proof}
Obviously, $x:=\inf_{n\in\bN}\frac{1}{n}a_n\le\liminf_{n\to+\infty}\frac{1}{n}a_n$,
and the converse inequality holds trivially if $x=+\infty$.
Assume therefore that $x<+\infty$. Then for every $\ep>0$, there exists
an $n_{\ep}\in\bN$ such that $a_{n_{\ep}}/n_{\ep}<x+\ep$, and 
\begin{align*}
\liminf_{n\to+\infty}\frac{1}{n}a_n
\le
\liminf_{k\to+\infty}\frac{1}{kn_{\ep}}a_{kn_{\ep}}
\le
\frac{a_{n_{\ep}}}{n_{\ep}}\le x+\ep,
\end{align*}
where the second inequality follows from subadditivity. 
This completes the proof of \eqref{Fekete limit0}.

The implications \ref{Fekete1}$\imp$\ref{Fekete2-1}$\imp$\ref{Fekete2} are trivial, 
and \ref{Fekete3}$\imp$\ref{Fekete1} follows immediately in view of 
\eqref{Fekete limit0}.
If the sequence is monotone non-increasing then 
\ref{Fekete1} holds trivially.
The non-trivial part of \ref{Fekete2}$\imp$\ref{Fekete1} follows from subadditivity due to the fact that any 
$n\ge n_0:=km$ can be written as
$n=ak+bm$ with some $a,b\in\bN$. 
Indeed, the numbers $ak$, $a=0,\ldots,m-1$ all have different remainders modulo $m$, and hence any $n\in\bN$ can be written as $n=ak+bm$ with some $a\in\{0,\ldots,m-1\}$ and $b\in\bZ$, 
and if $n\ge km$ then necessarily $b\ge 0$. 

The implication \ref{Fekete1}$\imp$\ref{Fekete3} is well known, but we provide the simple proof for readers' convenience.
If $a_n=+\infty$ for every $n\in\bN$ then 
\ref{Fekete3} holds trivially. If $a_n<+\infty$ for every $n\ge n_0$ then the simple observation that for every $m\in\bN$ and $n\ge n_0$, 
\begin{align*}
\frac{1}{n}a_n
&\le \frac{1}{n}\left\lfloor\frac{n-n_0}{m}\right\rfloor a_m+\frac{1}{n}\max_{n_0\le i\le n_0+m-1}a_i
\le a_m\left\{\begin{array}{ll}
\bz\frac{1}{m}-\frac{n_0}{nm}\jz,&a_m\ge 0,\\
\s\\
\bz\frac{1}{m}-\frac{1}{n}\bz\frac{n_0}{m}+1\jz\jz,&a_m<0
\end{array}\right\}
+
\frac{1}{n}\max_{n_0\le i\le n_0+m-1}a_i,
\end{align*}
yields
\begin{align*}
\limsup_{n\to+\infty}\frac{1}{n}a_n\le 
\inf_{m\in\bN}\frac{1}{m}a_m,
\end{align*}
from which \ref{Fekete3} follows immediately due to 
$\inf_{m\in\bN}\frac{1}{m}a_m\le\liminf_{n\to+\infty}\frac{1}{n}a_n$.
\end{proof}

\begin{rem}
Note that the sequence $a_{2k+1}:=+\infty$, $a_{2k}:=0$, $k\in\bN$, is subadditive, but 
\begin{align*}
\inf_{n\in\bN}\frac{1}{n}{a_n}=\liminf_{n\to+\infty}\frac{1}{n}a_n=0<+\infty=
\limsup_{n\to+\infty}\frac{1}{n}a_n,
\end{align*}
giving a simple demonstration of the significance of the conditions
\ref{Fekete1}--\ref{Fekete3} in Lemma \ref{lemma:Fekete}. 
\end{rem}

\section{Quantum R\'enyi divergences}
\label{sec:qdiv}

\subsection{Classical and quantum divergences}
\label{sec:multidiv}

Let $\X$ be an arbitrary non-empty set. By an \ki{$\X$-variable quantum divergence} $\divv$
we mean a function on collections of matrices
\begin{align*}
\cD(\divv)\subseteq\cup_{d\in\bN}\,\B(\X,\bC^d),
\ds\ds
\divv:\,\cD(\divv)
\to\bR\cup\{\pm\infty\},
\end{align*}
that is \ki{invariant under isometries}, i.e., if $V:\,\bC^{d_1}\to\bC^{d_2}$ is an isometry then 
\begin{align*}
&V\bz\cD(\divv)\cap\B(\X,\bC^{d_1})\jz V^*\subseteq\cD(\divv),\ds\text{and}\ds\\
&\divv\bz V\W V^*\jz
=
\divv\bz \W\jz,\ds\ds\ds \W\in\cD(\divv)\cap\B(\X,\bC^{d_1}).
\end{align*}
Due to the isometric invariance, $\divv$ may be extended to collections of operators
on any finite-dimensional Hilbert space $\hil$, by 
\begin{align*}
&\cD_{\hil}(\divv):=\left\{W\in\B(\X,\hil):\,\exists\,V:\,\hil\to\bC^d\text{ isometry: }VWV^*\in\cD(\divv)\right\},\\
&\divv(\W):=\divv(V\W V^*),\ds\ds\ds \W\in\cD_{\hil}(\divv),
\end{align*}
where $V$ in the second line is any isometry $V:\,\hil\to\bC^d$ such that 
$VWV^*\in\cD(\divv)$. The isometric invariance of $\divv$ on $\cD(\divv)$ 
guarantees that this extension is well-defined, in the sense
that the value of $\divv(V\W V^*)$
is independent of the choice of $d$ and $V$. 
Clearly, this extension is again invariant under isometries, i.e.,
for any $\W\in\cD_{\hil}(\divv)$ and 
$V:\,\hil\to\kil$ isometry,
$VWV^*\in\cD_{\kil}(\divv)$, and
$\divv(V\W V^*)=\divv(W)$.
Note that this implies that $\divv$ is invariant under extensions by pure states, i.e., 
for any $W\in\cD_{\hil}(\divv)$ and unit vector $\psi\in\kil$, 
we have 
\begin{align}\label{pure tensor}
W\otimes\pr{\psi}\in\cD_{\hil\otimes\kil}(\divv)
\ds\ds\text{and}\ds\ds
\divv(\W\otimes\pr{\psi})=\divv(\W).
\end{align}

\begin{rem}\label{rem:divv partiso inv}
A quantum divergence $\divv$ is also automatically invariant under 
partial isometries in the sense that if 
$W\in\cD_{\hil}(\divv)$ and 
$V:\,\hil\to\kil$ is a partial isometry such that 
$(V^*V)W(V^*V)=W$ 
and 
$\tilde V^*W\tilde V\in\cD_{\tilde\hil}(\divv)$, where 
$\tilde V:\,\tilde\hil:=(\ker V)^{\perp}\to\hil$
is the natural embedding,
then 
\begin{align*}
\divv\bz V\W V^*\jz=\divv\bz \W\jz.
\end{align*}
Indeed, 
\begin{align*}
\divv(VWV^*)
&=
\divv(V(\tilde V\tilde V^*)W(\tilde V\tilde V^*)V^*)
=
\divv(\tilde V^*W\tilde V)
=
\divv(\tilde V\tilde V^*W\tilde V\tilde V^*)
=
\divv(W),
\end{align*}
where the first and the last equalities follow from $(\tilde V\tilde V^*)W(\tilde V\tilde V^*)=W$, since $\tilde V\tilde V^*=V^*V$,
the second equality from the invariance of $\divv$ under the isometry 
$V\tilde V$, due to the assumption that $\tilde V^*W\tilde V\in\cD_{\tilde\hil}(\divv)$,
and the third equality from the invariance of $\divv$ under the isometry 
$\tilde V$.
\end{rem}

Analogously, an $\X$-variable \ki{classical divergence} $\divv$ is a function
\begin{align*}
\cD(\divv)\subseteq\cup_{d\in\bN}\,\F(\X,[d]),\ds\ds\ds
\divv:\,\cD(\divv)
\to\bR\cup\{\pm\infty\},
\end{align*}
that is invariant under injective maps $f:\,[d]\to[d']$, in the sense that for any such map, and any $w=(w_x)_{x\in\X}\in\cD(\divv)$, 
\begin{align*}
f_*w\in \cD(\divv),\ds\ds\ds
\divv(f_*w)=\divv(w),
\end{align*}
where $f_*w:=(f_*w_x)_{x\in\X}$ with 
$(f_*w_x)(i):=w_x(f\inv(i))$, $i\in\ran f$, and $(f_*w_x)(i):=0$ otherwise.
A classical divergence can be uniquely extended to 
collections of functions 
on any non-empty finite set $\I$ as
\begin{align*}
&\cD_{\I}(\divv):=\left\{w\in\F(\X,\I):\,\exists\,f:\,\I\to[d]\text{ injective: }
f_*w\in\cD(\divv)\right\}\\
&\divv(w):=\divv(f_*w),\ds\ds\ds w\in\cD_{\I}(\divv),
\end{align*}
where $f$ in the second line is any injective map such that 
$f_*w\in\cD(\divv)$.
The invariance property of $\divv$ on $\cD(\divv)$ guarantees that 
this extension is well-defined, and also invariant under injective functions, i.e., 
for any $w\in\cD_{\I}(\divv)$ and injective function 
$f:\,\I\to\J$ into a finite set $\J$, 
$f_*w\in\cD_{\J}(\divv)$ and $\divv(f_*w)=\divv(w)$.

\begin{rem}
Divergences in quantum information theory primarily serve as statistically motivated 
measures of how far away a collection of (classical or quantum) states are from each other; 
in this case the natural domain is all possible collections of density operators 
on an arbitrary finite-dimensional Hilbert space $\hil$ indexed by 
a fixed set $\X$. More generally, one may allow (non-zero) PSD operators as arguments;  
this is motivated 
partly by the fact that formulas e.g., for the R\'enyi divergences extend naturally to that 
setting, while from a more operational point of view, subnormalized states, for instance, 
might model the outputs of probabilistic protocols. 
For technical convenience, one may also define a divergence only on
collections of PSD operators with equal supports, as is often done in matrix analysis; however, 
this is usually not well justified in quantum information theory,
and is normally only done as an intermediate step in the definition of a divergence before extending it to more general collections of PSD operators by some procedure.

Many of the conditions on the domain of a divergence considered below (e.g., convexity in a fixed dimension, closedness under tensor products, etc.) are automatically satisfied with the above 
described domains.
\end{rem}

\begin{rem}
An $\X$-variable divergence $\divv$ with $\X=[n]$ for some $n\in\bN$, 
is called an \ki{$n$-variable divergence} (classical or quantum). 
In particular, in the case $\X=[2]=\{0,1\}$, we call $\divv$ a 
\ki{a binary divergence}, and
use the notations $\rho:=W_0$ and $\sigma:=W_1$, and
\begin{align*}
\divv(\rho\|\sigma):=\divv(W_0\|W_1):=\divv((W_0,W_1)).
\end{align*}
\end{rem}

\begin{definition}\label{def:q-extension}
Let $\divv$ be an $\X$-variable classical divergence. 
We say that a quantum divergence $\divv^q$ defined on a domain 
$\cD(\divv^q)\supseteq\cD$ is a 
\ki{quantum $\divv$-divergence on $\cD$}, if for any classical
$W\in\cD$, and any orthonormal basis 
$(e_i)_{i\in\I}$ jointly diagonalizing all $W_x$, $x\in\X$,
\begin{align}\label{quantum ext}
\bz\bz\inner{e_i}{W_xe_i}\jz_{i\in\I}\jz_{x\in\X}\in\cD_{\I}(\divv)\ds\text{ and }\ds
\divv^q\bz W\jz
=
\divv\bz\bz\bz\inner{e_i}{W_xe_i}\jz_{i\in\I}\jz_{x\in\X}\jz.
\end{align}
We say that $\divv^q$ is a \ki{quantum extension} of 
$\divv$, if for any $d\in\bN$ and $w\in\cD(\divv)\cap\F(\X,[d])$,
\begin{align}\label{quantum ext2}
\bz\sum_{i=0}^{d-1}w_x(i)\pr{i}\jz_{x\in\X}\in\cD(\divv^q)\ds\text{ and }\ds
\divv^q\bz\bz\sum_{i=0}^{d-1}w_x(i)\pr{i}\jz_{x\in\X}\jz
=
\divv(w),
\end{align}
where $(\ket{i})_{i\in[d]}$ is the canonical orthonormal basis of $\bC^{[d]}$. 
\end{definition}

\begin{rem}
Clearly, the condition in Definition \ref{def:q-extension} 
for $\divv^q$ being a quantum extension of $\divv$ 
is equivalent to \eqref{quantum ext2} holding with 
$\cD_{\hil}(\divv^q)$ in place of $\cD(\divv^q)$, 
for any non-empty finite set $\I$, any
$w\in\cD_{\I}(\divv)$, and 
any orthonormal system $(\ket{i})_{i\in\I}$ 
in an arbitrary finite-dimensional Hilbert space $\hil$ with $\dim\hil\ge |\I|$.
\end{rem}

Any classical divergence $\divv$ 
has a unique extension to collections of jointly diagonalizable operators.
To see this, we will need the following simple observation:

\begin{lemma}\label{lemma:commuting extension}
Let $\W\in\B(\X,\hil)$ be classical, and let 
$(e_i)_{i\in\I}$ and $(f_i)_{i\in\I}$ be orthonormal bases jointly diagonalizing 
all $\W_x$. Then for any classical divergence $\divv$, 
\begin{align*}
\bz(\inner{e_i}{W_x e_i})_{i\in\I}\jz_{x\in\X}\in\cD_{\I}(\divv)\ds\iff\ds
\bz(\inner{f_i}{W_x f_i})_{i\in\I}\jz_{x\in\X}\in\cD_{\I}(\divv),
\end{align*}
and if $\bz(\inner{e_i}{W_x e_i})_{i\in\I}\jz_{x\in\X}\in\cD_{\I}(\divv)$ then 
\begin{align}\label{comm ext}
\divv\bz\bz(\inner{e_i}{W_x e_i})_{i\in\I}\jz_{x\in\X}\jz
=
\divv\bz\bz(\inner{f_i}{W_x f_i})_{i\in\I}\jz_{x\in\X}\jz.
\end{align}
\end{lemma}
\begin{proof}
For every $x\in\X$, let $(\lambda_{x,j})_{j\in\J_x}$ be the different eigenvalues of 
$\Wx{x}$. For every $\vecc{j}\in\times_{x\in\X}\J_x$, let 
\begin{align*}
\I_{\vecc{j}}:=\{i\in\I:\,\inner{e_i}{\Wx{x}e_i}=\lambda_{x,j_x},\,x\in\X\},\ds\ds
\I'_{\vecc{j}}:=\{i\in\I:\,\inner{f_i}{\Wx{x}f_i}=\lambda_{x,j_x},\,x\in\X\}. 
\end{align*}
Then 
\begin{align*}
|\I_{\vecc{j}}|=\dim\bz\cap_{x\in\X}\ker(\lambda_{x,j_x} I-W_x)\jz
=|\I'_{\vecc{j}}|,
\end{align*}
and therefore there exists a bijection $h:\,\I\to\I$ such that 
for all $\vecc{j}\in\times_{x\in\X}\J_x$, $i\in\I_{\vecc{j}}$ $\iff$
$h(i)\in\I'_{\vecc{j}}$.
The assertions then follow by the invariance of 
$\divv$ under bijections. 
\end{proof}

Consider now an $\X$-variable classical divergence $\divv$, and define
its \ki{commutative extension} $\divv^{\comm}$ as
\begin{align}
&\cD(\divv^{\comm}):=\bigcup_{d\in\bN}\left\{
\bz\sum_{i=0}^{d-1}w_x(i)\pr{e_i}\jz_{x\in\X}:\,
w\in\cD(\divv)\cap\F(\X,[d]),\,(e_i)_{i=0}^{d-1}\text{ ONB in }\bC^d\right\},\label{commext1}\\
&\divv^{\comm}\bz\bz\sum_{i=0}^{d-1}w_x(i)\pr{e_i}\jz_{x\in\X}\jz:=\divv(w),\ds\ds\ds
w\in\cD(\divv)\cap\F(\X,[d]),\,d\in\bN,\label{commext2}
\end{align}
where $(e_i)_{i=0}^{d-1}$ in \eqref{commext2} is an arbitrary orthonormal basis in $\bC^d$.
It follows from Lemma \ref{lemma:commuting extension} that
\eqref{commext2} is well-defined, and 
it is easy to see that $\divv^{\comm}$ is a quantum divergence that is a 
quantum extension of $\divv$ in the sense of Definition \ref{def:q-extension}; in fact, it is the unique quantum extension with the smallest domain among all quantum extensions. 
Clearly, \eqref{commext1}--\eqref{commext2} can be equivalently rewritten as
\begin{align}
&\cD(\divv^{\comm}):=\bigcup_{d\in\bN}\left\{W\in\B(\X,\bC^d):\,W\text{ classical},\,
((\inner{e_i}{W_xe_i})_{i\in[d]})_{x\in\X}\in\D(\divv)\right\},\\
&\divv^{\comm}(W):=\divv\bz((\inner{e_i}{W_xe_i})_{i\in[d]})_{x\in\X}\jz,\ds\ds\ds
W\in\cD(\divv^{\comm})\cap\B(\X,\bC^d),\,d\in\bN,
\end{align}
where $(e_i)_{i=0}^{d-1}$ in the above is any orthonormal basis 
jointly diagonalizing all $W_x$, $x\in\X$.

In the converse direction, if $\divv$ is any quantum divergence 
such that every $W\in\cD(\divv)$ is classical then 
\begin{align}
&\cD(\divv^{\cl}):=\bigcup_{d\in\bN}\left\{((\inner{e_i}{W_xe_i})_{i\in[d]})_{x\in\X}:\,W\in\cD(\divv)\cap\B(\X,\bC^d),\,(e_i)_{i=0}^{d-1}\text{ $W$-basis}\right\},
\label{commext3}\\
&\divv^{\cl}\bz((\inner{e_i}{W_xe_i})_{i\in[d]})_{x\in\X}\jz
:=\divv(W),\ds\ds\ds W\in\cD(\divv)\cap\B(\X,\bC^d),\,d\in\bN,
\label{commext4}
\end{align}
where $(e_i)_{i=0}^{d-1}$ in \eqref{commext4} is any $W$-basis, 
defines a classical divergence. It is easy to see that for any classical divergence 
$\divv$, $(\divv^{\comm})^{\cl}=\divv$, and for any quantum divergence $\divv$
with commutative domain, $(\divv^{\cl})^{\comm}=\divv$. 
Thus, classical divergences can be uniquely identified with quantum divergences with 
commutative domain, whence for the rest we will also call the latter
classical divergences. 
When we want to specify in which sense we use the terminology 
``classical divergence'', we might write 
``classical divergence on functions'' or ``classical divergence on operators''. 
Under the above identification, every classical divergence $\divv$ is 
identified with $\divv^{\comm}$, and hence for the rest
we will also denote $\divv^{\comm}$ by $\divv$.

Different quantum extensions of a classical divergence may be compared according to the following:
\begin{definition}\label{def:divorder}
For two $\X$-variable quantum divergences $\divv_1$ and $\divv_2$, we write
\begin{align}\label{divorder def1}
\divv_1\le\divv_2\ds\text{on}\ds \cD,\ds\ds\ds\text{if}\ds\ds\ds
\divv_1(W)\le\divv_2(W),\ds\ds W\in\cD,
\end{align}
where we assume that $\cD\subseteq\cD(\divv_1)\cap\cD(\divv_2)$.
When $\cD=\cD(\divv_1)\cap\cD(\divv_2)$, we omit it and write simply that 
$\divv_1\le\divv_2$.
\end{definition}

For $2$-variable divergences on pairs of non-zero PSD operators we also introduce the following strict ordering that will be useful when comparing quantum relative entropies and 
R\'enyi divergences:

\begin{definition}\label{def:divorder2}
Let $\divv_1,\divv_2$ be quantum divergences with 
$\cD(\divv_1)=\cD(\divv_2)=\cup_{d\in\bN}(\B(\bC^d)\pne\times\B(\bC^d)\pne)$.
We write
\begin{align}\label{divorder def2}
\divv_1<\divv_2\ds\ds\text{if}\ds\ds
\rho^0\le\sigma^0,\s\rho\sigma\ne\sigma\rho\ds\imp\ds
\divv_1(\rho\|\sigma)<\divv_2(\rho\|\sigma).
\end{align}
\end{definition}

\begin{rem}
For the rest, by expressions like 
``for all $W\in\B(\X,\hil)\p$'', or  
``for all $\rho,\sigma\in\B(\hil)\pne$'' we mean that 
the given property holds for any finite-dimensional Hilbert space $\hil$, and we write it out explicitly when something is only supposed to be valid for a specific Hilbert space.
\end{rem}

Before further discussing quantum extensions of classical divergences, we review a few important properties of general divergences. 
Note that by the identification of classical divergences and quantum divergences defined on jointly diagonalizable families of operators, the following considerations
apply to classical divergences as well; e.g., the following define
monotonicity, joint convexity, etc., also for classical divergences.

Let $\divv$ be an $\X$-variable quantum divergence. We say that $\divv$ is
\begin{itemize}
\item
\ki{non-negative} if 
$\divv(W)\ge 0$
for all collections of density operators $W\in\cD_{\hil}(\divv)\cap\S(\X,\hil)$, 
and it is 
\ki{strictly positive} if it is non-negative and 
$\divv(W)=0$ $\iff$ $W_x=W_y$, $x,y\in\X$,
again for density operators;

\item
\ki{monotone under a given map $\map:\,\B(\hil)\to\B(\kil)$} 
for some finite-dimensional Hilbert spaces $\hil,\kil$, 
if 
\begin{align*}
W\in\cD_{\hil}(\divv)\ds\imp\ds\map(W)\in\cD_{\kil}(\divv),\ds\ds
\divv(\map(W))\le\divv(W),
\end{align*}
where $\map(W):=(\map(W_x))_{x\in\X}$;
in particular, it is \ki{monotone under CPTP maps/PTP maps/pinchings} if 
monotonicity holds for any map in the given class for any two finite-dimensional Hilbert spaces $\hil,\kil$, and it is \ki{trace-monotone}, if 
\begin{align}\label{tracemon def}
W\in\cD_{\hil}(\divv)\ds\imp\ds \Tr W:=(\Tr W_x)_{x\in\X}\in\cD(\divv),\ds\ds
\divv(\Tr W)\le\divv(W);
\end{align}

\item
\ki{jointly convex} if for all $W^{(k)}\in\cD_{\hil}(\divv)$, $k\in[r]$, and 
probability distribution $(p_k)_{k\in[r]}$,
\begin{align}\label{jointly convex}
\sum_{k\in[r]} p_kW^{(k)}\in\cD_{\hil}(\divv),\ds\ds
\divv\Bigg(\sum_{k\in[r]} p_kW^{(k)}\Bigg)\le\sum_{k\in[r]} p_k\divv\big( W^{(k)}\big),
\end{align}
and it is jointly concave if $-\divv$ is jointly convex;

\item
\ki{additive}, if for all $W^{(k)}\in\cD_{\hil^{(k)}}(\divv)$, $k=1,2$,
\begin{align*}
W^{(1)}\potimes W^{(2)}\in\cD_{\hil^{(1)}\otimes \hil^{(2)}}(\divv),\ds\ds
\divv\big( W^{(1)}\potimes W^{(2)}\big)=
\divv\big( W^{(1)}\big)+
\divv\big( W^{(2)}\big),
\end{align*}
and \ki{subadditive (superadditive)} if 
LHS $\le$ RHS (LHS $\ge$ RHS) holds above;

\item
\ki{weakly additive}, if for all $W\in\cD_{\hil}(\divv)$,
\begin{align*}
W^{\potimes n}\in\cD_{\hil^{\otimes n}}(\divv),\ds\ds
\divv\big( W^{\potimes n}\big)=n\divv(W),\ds\ds\ds n\in\bN,
\end{align*}
and \ki{weakly subadditive (superadditive)} if LHS $\le$ RHS (LHS $\ge$ RHS) holds above;

\item
\ki{block subadditive}, if for any $W\in\cD_{\hil}(\divv)$,
and any sequence of projections
$P_0,\ldots,P_{r-1}\in\bP(\hil)$ summing to $I$,
if $P_iWP_i\in\cD_{\hil}(\divv)$, $i\in[r]$, then
\begin{align}\label{block subadd}
\divv\bz\sum_{i=0}^{r-1} P_i\W P_i\jz\le \sum_{i=0}^{r-1} \divv(P_i\W P_i).
\end{align}
Conversely, if the inequality in the above always holds in the opposite direction then 
$\divv$ is called \ki{block superadditive}, and if it is always an equality, then 
$\divv$ is \ki{block additive}.
\end{itemize}

\begin{rem}
We say that $\divv$ is \ki{strictly trace-monotone}, if equality in \eqref{tracemon def}
implies the existence of a state $\omega$ and numbers $\lambda_x$, $x\in\X$ such
that $W_x=\lambda_x\omega$, $x\in\X$.
\end{rem}

\begin{rem}
Note that if $W$ is jointly convex or jointly concave, then for any 
$W\in\cD_{\hil}(\divv)$ and any density operator $\sigma\in\S(\kil)$, 
$W\otimes\sigma\in\cD_{\hil\otimes\kil}(\divv)$, according to 
\eqref{pure tensor} and \eqref{jointly convex}.
\end{rem}

\begin{rem}\label{rem:scaling}
Note that $z\otimes X\mapsto zX$ gives a canonical identification between 
$\bC\otimes\B(\hil)$ and $\B(\hil)$. In particular,  
any additive quantum divergence $\divv$ 
satisfies the \ki{scaling law}
\begin{align}
\divv\bz(t_xW_x)_{x\in\X}\jz=\divv(t)+\divv(W),\ds\ds\ds
W\in\cD_{\hil}(\divv),\,t\in\cD(\divv)\cap\F(\X,[1]),\label{general scaling}
\end{align}
where $t=(t_x=t_x(0))_{x\in\X}$.
\end{rem}

\begin{rem}\label{rem:classical mon}
Note that in the classical case, 
monotonicity under PTP maps is equivalent to monotonicity under CPTP maps, and it means that 
for any stochastic matrix $T\in\bR^{d_1\times d_2}$
and for 
any $w\in\cD(\divv)\cap\F(\X,[d_1])$, we have 
$(\sum_{j\in[d_2]}\sum_{i\in[d_1]}w_x(i)T_{i,j}\egy_{\{j\}})_{x\in\X}\in\cD(\divv)$, and 
\begin{align*}
\divv\bz\bz\sum_{j\in[d_2]}\sum_{i\in[d_1]}w_x(i)T_{i,j}\egy_{\{j\}}\jz_{x\in\X}\jz
\le
\divv(w).
\end{align*}
Here, the matrix $T$ being stochastic means that  
$T_{i,j}\ge 0$ for all $i,j$, and for all $i$, $\sum_{j\in[d_2]}T_{i,j}=1$.
\end{rem}

A quantum divergence $\divv$ is called \ki{(positive) homogeneous}, if 
for every $W\in\cD(\divv)$ and $t\in(0,+\infty)$, 
$tW\in\cD(\divv)$, and
\begin{align*}
\divv\bz(t\Wx{x})_{x\in\X}\jz=t\divv\bz (\Wx{x})_{x\in\X}\jz.
\end{align*}
It is well known that a $2$-variable quantum divergence that is monotone 
non-decreasing under partial 
traces is jointly concave whenever it has the additional properties of block 
superadditivity and 
homogeneity. The extension to multi-variable divergences is straightforward; we give a detailed proof in Appendix \ref{sec:jointconc} for completeness.

%

\begin{lemma}\label{lemma:jointconc from mon}
(i) 
Let $\divv$ be an $\X$-variable quantum divergence that is homogeneous, and 
for every $\hil$, $\cD_{\hil}(\divv)$ is convex.
If, moreover, $\divv$ is block superadditive (subadditive) and
monotone non-decreasing (non-increasing) under partial traces 
then 
$\divv$ is jointly concave (convex) and jointly superadditive (subadditive).

(ii) 
Let $\divv$ be an $\X$-variable quantum divergence that
is stable under tensoring with the maximally mixed state, i.e., for any $W\in\cD_{\hil}(\divv)$ and any $\kil$, 
\begin{align*}
\divv(W\otimes(I_{\kil}/\dim\kil))=\divv(W).
\end{align*}
If, moreover, $\divv$ is jointly concave (convex) 
then for any $W\in\cD_{\hil}(\divv)$ and any CPTP map 
$\map:\,\B(\hil)\to\B(\kil)$ such that 
$\map(W)\in\cD_{\kil}(\divv)$, we have 
$\divv(\map(W))\ge\divv(W)$
($\divv(\map(W))\le\divv(W)$). 
\end{lemma}
\medskip

Any monotone classical divergence admits two canonical quantum extensions, the minimal and the maximal ones:

\begin{example}\label{ex:minmax extension}
For a classical divergence $\divv$, let
\begin{align}
&\cD(\divv^{\meas}):=\cup_{d\in\bN}\B(\X,\bC^d),\ds\ds\text{and for any}\ds W\in\cD(\divv^{\meas}),\\
&\divv^{\meas}(W):=\sup\{\divv\bz\M(W)\jz:\,M\in\povm(\bC^d,[m]),\,m\in\bN,\,
\M(W)=(\M(W_x))_{x\in\X}\in\cD(\divv)\}.
\label{meas div def}
\end{align}
It is easy to see that $\divv^{\meas}$ is a quantum divergence 
that is monotone under PTP maps, and if 
$\divv$ is monotone then $\divv^{\meas}$ is a quantum extension of 
$\divv$, called the \ki{measured}, or \ki{minimal} extension. 

As introduced in \cite{Matsumoto_newfdiv} in the $2$-variable case, a \ki{reverse test} for 
$W\in\B(\X,\hil)$
is a pair $(w,\rt)$ with $w\in\F(\X,\I)$ for some finite set $\I$, 
and $\rt:\,\ell^{\infty}(\I)\to\B(\hil)$ a (completely) positive trace-preserving map
such that $\rt(w)=W$.
For a classical divergence $\divv$, let
\begin{align}
&\cD(\divv^{\max}):=\cup_{d\in\bN}\B(\X,\bC^d),
\ds\ds\text{and for any}\ds W\in\cD(\divv^{\max}),\\
&\divv^{\max}(W):=\inf\left\{\divv(w):\,(w,\rt)\text{ is a reverse test for }W
\text{ with }
w\in\cD(\divv)\right\}.
\label{maxdiv def}
\end{align}
It is easy to see that $\divv^{\max}$ is a quantum divergence that is monotone under PTP maps, and if 
$\divv$ is monotone then $\divv^{\max}$
gives a quantum extension of $\divv$, called the \ki{maximal} extension.

Assume for the rest that $\divv$ is monotone, so that both 
$\divv^{\meas}$
and $\divv^{\max}$ are quantum extensions of $\divv$ that are monotone under PTP maps. 

It is straightforward to verify that for any 
quantum extension $\divv^q$ of $\divv$ that is monotone under CPTP maps, 
\begin{align*}
\divv^{\meas}(W)\le\divv^q(W)\le\divv^{\max}(W),\ds\ds\ds
W\in\cD(\divv^q)
\end{align*}
holds. 

It is also clear that if $\divv$ is additive then 
$\divv^{\max}$ is subadditive, and 
$\divv^{\meas}$ is superadditive, and 
the \ki{regularized measured} and the \ki{regularized maximal} $\divv$-divergences
\begin{align}
\oll{\divv}^{\meas}(W)
&:=
\sup_{n\in\bN}\frac{1}{n}\divv^{\meas}(W^{\potimes n}),
\label{reg meas div def}\\
\oll{\divv}^{\max}(W)&:=\inf_{n\in\bN}\frac{1}{n}\divv^{\max}(W^{\potimes n}),
\end{align}
are quantum extensions of $\divv$.
Here 
\begin{align*}
\cD(\oll{\divv}^{\meas}):=\cD(\oll{\divv}^{\max}):=\cup_{d\in\bN}\B(\X,\bC^d),
\end{align*}
i.e., $\oll{\divv}^{\meas}$ and $\oll{\divv}^{\max}$ are defined on the maximal possible domain.

Obviously, $\oll{\divv}^{\meas}$ and $\oll{\divv}^{\max}$ are monotone under CPTP maps, 
and for any 
quantum extension $\divv^q$ of $\divv$ that is monotone under CPTP maps and is additive, 
\begin{align*}
\oll{\divv}^{\meas}(W)\le\divv^q(W)\le\oll{\divv}^{\max}(W),\ds\ds\ds
W\in\cD(\divv^q)
\end{align*}
holds. 
\end{example}

\begin{rem}
As with any quantum divergence, all the above can be extended to the case where
the arguments are collections of operators on an arbitrary finite-dimensional Hilbert space $\hil$ rather than on $\bC^d$.
\end{rem}

\begin{rem}
If $\divv$ is monotone non-decreasing under stochastic maps then 
its measured and maximal versions are naturally defined for $-\divv$ instead, or equivalently, with inf and sup instead of sup and inf in 
\eqref{meas div def} and \eqref{maxdiv def}, respectively.
\end{rem}

\begin{prop}\label{prop:regmeas regmax as limits}
Let $\divv$ be a monotone and additive classical divergence. For any
$W\in\cup_{d\in\bN}\B(\X,\hil)$, the following are true:
\begin{align}
(i)\ds& \bN\ni n\mapsto \divv^{\meas}(W^{\potimes n})\ds\text{is monotone increasing and superadditive}.\label{meas superadd}\\
(ii)\ds&
\bN\ni n\mapsto \divv^{\max}(W^{\potimes n})\ds\text{is monotone increasing and subadditive}.\label{max subadd}\\
(iii)\ds&
\oll{\divv}^{\meas}(W)
=
\lim_{n\to+\infty}\frac{1}{n}\divv^{\meas}(W^{\potimes n}).\label{regmeas as limit}\\
(iv)\ds&
\text{Either}\ds
\divv^{\max}(W^{\potimes n})<+\infty,\ds n\in\bN,\ds\text{ or }\ds
\divv^{\max}(W^{\potimes n})=+\infty,\ds n\in\bN.\label{divmax (in)finite}\\
(v)\ds&
\oll{\divv}^{\max}(W)
=
\lim_{n\to+\infty}\frac{1}{n}\divv^{\max}(W^{\potimes n}).\label{regmax as limit}
\end{align}
\end{prop}
\begin{proof}
The superadditivity/subadditivity stated in \eqref{meas superadd} and \eqref{max subadd},
respectively,
are special cases of the more general superadditivity/subadditivity stated already in 
Example \ref{ex:minmax extension}.

Note that if there exists no POVM $M$ such that $\M(W^{\potimes n})\in\cD(\divv)$
then $\divv^{\meas}(W^{\potimes n})=-\infty$ by definition, and 
$\divv^{\meas}(W^{\potimes n'})\ge\divv^{\meas}(W^{\potimes n})$
holds trivially for any $n'\ge n$. 
Assume that $M\in\povm(\hil^{\otimes n},[m])$ is such that 
$\M(W^{\potimes n})\in\cD(\divv)$. Then for any $n'\ge n$, 
$M^{(n')}:=(M_i\otimes I_{\hil^{\otimes (n'-n)}})_{i\in[m]}\in\povm(\hil^{\otimes n'},[m])$,
and $\M^{[n']}(W^{\potimes n'})=\M(W^{\potimes n})\in\cD(\divv)$, whence
\begin{align*}
\divv^{\meas}(W^{\potimes n'})\ge
\divv\bz\M^{[n']}(W^{\potimes n'})\jz
=
\divv\bz\M(W^{\potimes n})\jz.
\end{align*}
Taking the supremum of the rightmost quantity over all 
 $M\in\povm(\hil^{\otimes n},[m])$ such that 
$\M(W^{\potimes n})\in\cD(\divv)$ yields 
$\divv^{\meas}(W^{\potimes n'})\ge\divv^{\meas}(W^{\potimes n})$. This proves the monotonicity in \eqref{meas superadd}. The identity in \eqref{regmeas as limit}
follows from \eqref{meas superadd} by Lemma \ref{lemma:Fekete}
applied to $a_n:=-\divv^{\meas}(W^{\potimes n})$.

To prove the monotonicity in \eqref{max subadd}, 
note again that if there exists no reverse test $(w^{(n)}, \rt^{(n)})$ for 
$W^{\potimes n}$ such that $w^{(n)}\in\cD(\divv)$ then 
$\divv^{\meas}(W^{\potimes n})=+\infty$ by definition, and thus
$\divv^{\meas}(W^{\potimes n'})\le\divv^{\meas}(W^{\potimes n})$
holds trivially for any $n'\le n$. Assume that for some $n>1$, 
$(w^{(n)}, \rt^{(n)})$ is a reverse test for 
$W^{\potimes n}$ such that $w^{(n)}\in\cD(\divv)$. Then for any $n'< n$, 
$(w^{(n')},\rt^{(n')}):=(w^{(n)}, \Tr_{[n'+1,\ldots,n]}\circ\rt^{(n)})$ is a reverse test for 
$W^{\potimes n'}$ such that $w^{(n')}\in\cD(\divv)$, whence
\begin{align*}
\divv^{\max}(W^{\potimes n'})\le
\divv\big( w^{(n')}\big)
=
\divv\big( w^{(n)}\big).
\end{align*}
Taking the infimum of the rightmost quantity over all reverse tests
$(w^{(n)}, \rt^{(n)})$ as above, we get that 
$\divv^{\max}(W^{\potimes n'})\le\divv^{\max}(W^{\potimes n})$, as desired. 
This completes the proof of \eqref{max subadd}.

If $\divv^{\max}(W^{\potimes n})<+\infty$ for every $n\in\bN$ then 
Lemma \ref{lemma:Fekete} applied to $a_n:=\divv^{\max}(W^{\potimes n})$ yields immediately 
\eqref{regmax as limit}, due to the subadditivity established in 
\eqref{max subadd}. Assume for the rest that there exists some $n\in\bN$ such that 
$\divv^{\max}(W^{\potimes n})=+\infty$. By the monotonicity established in 
\eqref{max subadd}, we have $\divv^{\max}(W^{\potimes n'})=+\infty$ for every $n'\ge n$, while
the subadditivity established in \eqref{max subadd} yields that 
$\divv^{\max}(W^{\potimes n'})=+\infty$ for every $n'\le n$. 
This completes the proof of \eqref{divmax (in)finite}, and also of 
\eqref{regmax as limit}, since the latter holds trivially
when $\divv^{\max}(W^{\potimes n})=+\infty$ for every $n\in\bN$. 
\end{proof}

\begin{cor}\label{cor:regmeas regmax weakly additive}
Let $\divv$ be a monotone and additive classical divergence. Then 
$\oll{\divv}^{\meas}$ and $\oll{\divv}^{\max}$ are 
CPTP-monotone and
weakly additive.
\end{cor}
\begin{proof}
CPTP-monotonicity is immediate from the definitions, and weak additivity follows immediately from \eqref{regmeas as limit} and \eqref{regmax as limit}.
\end{proof}

\begin{cor}\label{cor:regmeas regmax extremal}
Let $\divv$ be a monotone and additive classical divergence and 
$\cD\subseteq\cup_{d\in\bN}\B(\X,\bC^d)$ be closed under tensor products. 
Let $\divv^q$ be a CPTP-monotone quantum extension of  
$\divv$ on $\cD$.
Then 
\begin{align*}
\oll{\divv}^{\meas}(W)\le
\liminf_{n\to+\infty}\frac{1}{n}\divv^q(W^{\potimes n})
\le\limsup_{n\to+\infty}\frac{1}{n}\divv^q(W^{\potimes n})\le\oll{\divv}^{\max}(W)
,\ds\ds\ds
W\in\cD.
\end{align*}
In particular, if
$\divv$ is weakly additive then 
\begin{align*}
\oll{\divv}^{\meas}(W)\le\divv^{q}(W)\le\oll{\divv}^{\max}(W),\ds\ds\ds
W\in\cD.
\end{align*}
That is, $\oll{\divv}^{\meas}$ is the smallest, and 
$\oll{\divv}^{\max}$ is the largest weakly additive and CPTP-monotone quantum extension of 
$\divv$ on $\cD$.
\end{cor}
\bigskip

We will furthermore consider properties that only 
concern one variable of a divergence. We formulate these only for the case when this is the second variable of a $2$-variable divergence; the definitions in the general case can be obtained by straightforward modifications.
In particular, we say that a $2$-variable quantum divergence
$\divv$ is
\begin{itemize}
\item
\ki{anti-monotone in the second argument (AM)}, if for 
all $\rho,\sigma_1,\sigma_2\in\B(\hil)$ such that 
$(\rho,\sigma_1),(\rho,\sigma_2)\in\cD_{\hil}(\divv)$,
\begin{align}\label{antimon}
\sigma_1\le\sigma_2\ds\imp\ds \divv(\rho\|\sigma_1)\ge\divv(\rho\|\sigma_2);
\end{align}

\item
\ki{weakly anti-monotone in the second argument}, if for any 
$(\rho,\sigma)\in\cD_{\hil}(\divv)$, we have 
$(\rho,\sigma+\ep I)\in\cD_{\hil}(\divv)$, $\ep\in(0,\kappa_{\rho,\sigma})$,
with some $\kappa_{\rho,\sigma}>0$, and
\begin{align}\label{weakmon}
[0,\kappa_{\rho,\sigma})\ni\ep\mapsto \divv(\rho\|\sigma+\ep I)\ds\text{is decreasing};
\end{align}

\item
\ki{regular}, if for any 
$(\rho,\sigma)\in\cD_{\hil}(\divv)$, we have 
$(\rho,\sigma+\ep I)\in\cD_{\hil}(\divv)$, $\ep\in(0,\kappa_{\rho,\sigma})$,
with some $\kappa_{\rho,\sigma}>0$, and
\begin{align}\label{regularity}
\divv(\rho\|\sigma)=\lim_{\ep\searrow 0}\divv(\rho\|\sigma+\ep I);
\end{align}

\item
\ki{strongly regular}, if for any $(\rho,\sigma)\in\cD_{\hil}(\divv)$, and any sequence of operators $(\sigma_n)_{n\in\bN}$ converging decreasingly to $\sigma$ such that 
$(\rho,\sigma_n)\in\cD_{\hil}(\divv)$, $n\in\bN$, we have
\begin{align}\label{strong regularity}
\divv(\rho\|\sigma)=\lim_{n\to+\infty}\divv(\rho\|\sigma_n);
\end{align}
\end{itemize}

\begin{rem}\label{rem:AM regular}
Note that 
\begin{align*}
\text{AM + regularity}\ds\ds\imp\ds\ds
\text{strong regularity.}
\end{align*}
Indeed, assume that $\divv$ is regular and anti-monotone in its second argument.
Let $(\sigma_n)_{n\in\bN}$ be a sequence of operators converging decreasingly to $\sigma$, 
such that 
$(\rho,\sigma)\in\cD_{\hil}(\divv)$, 
$(\rho,\sigma_n)\in\cD_{\hil}(\divv)$, $n\in\bN$. 
Then for every $n\in\bN$,
\begin{align*}
\sigma\le\sigma_n=\sigma+\sigma_n-\sigma\le\sigma+
\underbrace{\norm{\sigma_n-\sigma}_{\infty}}_{=:c_n}I
=
\sigma+c_nI,
\end{align*}
whence 
\begin{align*}
\divv(\rho\|\sigma)\ge \divv(\rho\|\sigma_n)\ge \divv(\rho\|\sigma+c_nI),\ds\ds\ds n\in\bN.
\end{align*}
By the regularity assumption, the RHS above tends to $\divv(\rho\|\sigma)$ as $n\to+\infty$,
whence also $\lim_{n\to+\infty}\divv(\rho\|\sigma_n)=\divv(\rho\|\sigma)$. 
Thus, $\divv$ is strongly regular.
\end{rem}

\begin{rem}\label{rem:classical definitions}
All the above properties may be defined also for classical divergences, by 
replacing the operators $\rho,\sigma$, etc. with functions, or equivalently, by
requiring the operators in the definitions to be jointly diagonalizable.
\end{rem}

\subsection{Classical R\'enyi divergences}
\label{sec:cl Renyi}

The classical divergences of particular importance to us are the $2$-variable divergences called
\ki{relative entropy}, or \ki{Kullback-Leibler divergence},
and the \ki{classical R\'enyi divergences}. 
For a finite set $\I$, and $\rho,\sigma\in\F(\I)\pne$ the relative entropy of 
$\rho$ and $\sigma$ is defined as
\begin{align}\label{KL def}
D(\rho\|\sigma):=
\begin{cases}
\sum_{i\in\I}\left[\rho(i)\nlog\rho(i)-\rho(i)\nlog\sigma(i)\right]
,& \supp\rho\subseteq\supp\sigma,\\
+\infty,&\text{otherwise}.
\end{cases}
\end{align}
The classical R\'enyi $\alpha$-divergences \cite{Renyi} are
defined for $\rho,\sigma\in\F(\I)\pne$ and $\alpha\in(0,1)\cup(1,+\infty)$ as
\begin{align}
D_{\alpha}(\rho\|\sigma)&:=
\frac{1}{\alpha-1}\underbrace{\log Q_{\alpha}(\rho\|\sigma)}_{=:\psi_{\alpha}(\rho\|\sigma)}-\frac{1}{\alpha-1}\log\sum_{i\in\I}\rho(i),
\label{classical Renyi def}\\
Q_{\alpha}(\rho\|\sigma)&:=
\lim_{\ep\searrow 0}\sum_{i\in\I}(\rho(i)+\ep)^{\alpha}(\sigma(i)+\ep)^{1-\alpha}\nn\\
&=\begin{cases}
\sum_{i\in\I}\rho(i)^{\alpha}\sigma(i)^{1-\alpha}&\alpha\in(0,1)\text{ or }\supp\rho\subseteq\supp\sigma,\\
+\infty,&\text{otherwise}.
\end{cases}\label{classical Q def}
\end{align}
For $\alpha\in\{0,1,+\infty\}$, the R\'enyi $\alpha$-divergence is defined by the corresponding limit, and it is easy to see that 
\begin{align}
D_0(\rho\|\sigma)&:=\lim_{\alpha\searrow 0}D_{\alpha}(\rho\|\sigma)=
-\log\underbrace{\sum_{i\in\supp\rho}\sigma(i)}_{=:Q_0(\rho\|\sigma)}+\log\sum_{i\in\I}\rho(i),\label{0-Renyi cl}\\
D_1(\rho\|\sigma)&:=\lim_{\alpha\to 1}D_{\alpha}(\rho\|\sigma)=
\frac{1}{\sum_{i\in\I}\rho(i)}D(\rho\|\sigma),\label{1-Renyi cl}\\
D_{+\infty}(\rho\|\sigma)&:=\lim_{\alpha\to +\infty}D_{\alpha}(\rho\|\sigma)=
\log\inf\{\lambda>0:\,\rho\le\lambda\sigma\}.
\label{classical Dmax}
\end{align}
In particular, the R\'enyi $1$-divergence is the same as the relative entropy up to normalization.

We extend the definitions 
of the R\'enyi divergences 
to the case when the second argument is zero as
\begin{align}\label{zero argument cl1}
D_{\alpha}(\rho\|0):=+\infty,\ds\ds\ds\rho\gneq 0,\ds\alpha\in[0,+\infty),
\end{align}
and the definition of the relative entropy to the case when one or both arguments are zero as
\begin{align}\label{zero argument cl2}
\D(0\|\sigma):=0\ds\ds\sigma\ge 0,\ds\ds\ds\ds\ds
\D(\rho\|0):=+\infty\ds\ds\rho\gneq 0.
\end{align}

For the study and applications of the (classical) R\'enyi divergences, the relevant quantity 
is actually $Q_{\alpha}$ (equivalently, $\psi_{\alpha}$); the normalizations 
in \eqref{classical Renyi def} are somewhat arbitrary, and are mainly relevant only for the limits in \eqref{0-Renyi cl}--\eqref{classical Dmax}. For instance, one could alternatively use the \ki{symmetrically normalized R\'enyi $\alpha$-divergences} defined for any $\rho,\sigma\in\F(\I)\pne$ and $\alpha\in(0,1)\cup(1,+\infty)$ as 
\begin{align*}
\sD_{\alpha}(\rho\|\sigma)
&:=
\frac{1}{\alpha(1-\alpha)}\left[
-\log Q_{\alpha}(\rho\|\sigma)+\alpha\log\sum_{i\in\I}\rho(i)
+(1-\alpha)\log\sum_{i\in\I}\sigma(i)\right]\\
&=-\frac{1}{\alpha(1-\alpha)}\log Q_{\alpha}\bz\frac{\rho}{\sum_{i\in\I}\rho(i)}\Bigg\|
\frac{\sigma}{\sum_{i\in\I}\sigma(i)}\jz.
\end{align*}
For $\alpha\in\{0,1\}$ these give
\begin{align*}
\sD_{1}(\rho\|\sigma)&:=
\lim_{\alpha\to 1}\sD_{\alpha}(\rho\|\sigma)=
D\bz\frac{\rho}{\sum_{i\in\I}\rho(i)}\Bigg\|
\frac{\sigma}{\sum_{i\in\I}\sigma(i)}\jz=
D_1(\rho\|\sigma)-\log\sum_{i\in\I}\rho(i)+\log\sum_{i\in\I}\sigma(i),\\
\sD_{0}(\rho\|\sigma)&:=
\lim_{\alpha\searrow 0}\sD_{\alpha}(\rho\|\sigma)=
D\bz\frac{\sigma}{\sum_{i\in\I}\sigma(i)}\Bigg\|\frac{\rho}{\sum_{i\in\I}\rho(i)}\jz
=
D_1(\sigma\|\rho)+\log\sum_{i\in\I}\rho(i)-\log\sum_{i\in\I}\sigma(i),
\end{align*}
while 
\begin{align*}
\sD_{+\infty}(\rho\|\sigma)&:=
\lim_{\alpha\to +\infty}\sD_{\alpha}(\rho\|\sigma)=0
\end{align*}
is not very interesting. 

As mentioned already in the Introduction, the R\'enyi $\alpha$-divergences
with $\alpha\in(0,1)\cup(1,+\infty)$ can be recovered from the relative entropy as
\begin{align}\label{classical variational2}
-\log Q_{\alpha}(\rho\|\sigma)=
\inf_{\omega}\left\{
\alpha D(\omega\|\rho)+(1-\alpha)D(\omega\|\sigma)\right\},
\end{align}
where the infimum is taken over all $\omega\in\P(\I)$ with 
$\supp\omega\subseteq\supp\rho$, and it is uniquely attained at 
\begin{align}\label{classical optimal state}
\omega_{\alpha}(\rho\|\sigma):=
\sum_{i\in S}\frac{\rho(i)^{\alpha}\sigma(i)^{1-\alpha}}{\sum_{j\in S}\rho(j)^{\alpha}\sigma(j)^{1-\alpha}}\egy_{\{i\}},
\end{align}
where $S:=\supp\rho\cap\supp\sigma$, 
provided that $\supp\rho\subseteq\supp\sigma$, 
or $\supp\rho\cap\supp\sigma\ne\emptyset$ and $\alpha\in(0,1)$.
The case $\alpha\in(0,1)$ was discussed in \cite{CsM2003} in the more general setting 
where $\I$ is not finite,
while the case $\alpha>1$ was discussed in the finite-dimensional quantum case in 
\cite{MO-cqconv-cc}; see also Section \ref{sec:barycentric} below.

It is natural to ask whether the concept of R\'enyi divergences can be generalized to more than two variables. 
Formulas \eqref{classical Q def} and \eqref{classical variational2} offer two different approaches to do that. In a very general setting, one may 
consider a set $\X$ equipped with a $\sigma$-algebra $\A$. Then for any measurable
$w\in\F(\X,\I)\pne$ and signed measure $P$ on $\A$ with $P(\X)=1$, one may consider
\begin{align}\label{cl multiR limit1}
\hat Q_P(w):=\lim_{\ep\searrow 0}\sum_{i\in\I}\exp\bz\int_{\X}\log(w_x(i)+\ep)\,dP(x)\jz,
\end{align}
or 
\begin{align}\label{cl multiR limit2}
\hat Q_P(w)
:=
\lim_{\ep\searrow 0}\sum_{i\in\I}
\exp\bz\int_{\X}\log((1-\ep)w_x(i)+\ep/|\I|)\,dP(x)\jz,
\end{align}
where the latter is somewhat more natural when the $w_x$ are probability density functions on $\I$.
In the most general case, various issues regarding the existence of the integrals and the limits arise, which are important from a mathematical, but not particularly relevant from a conceptual point, and hence for the rest we will restrict our attention to the case where 
$P$ is finitely supported. 
In that case the integrals always exist, and the 
$\ep\searrow 0$ limit can be easily determined as
\begin{align}
\hat Q_P(w)=
\sum_{i\in\I}\bz\Bigg(\prod_{x:\,w_x(i)>0}w_x(i)^{P(x)}\Bigg)\cdot
\left\{\begin{array}{ll}
0,&\text{if }\sum_{x:\,w_{x}(i)=0}P(x)>0,\\
1,&\text{if } \sum_{x:\,w_{x}(i)=0}P(x)=0,\\
+\infty,&\text{if } \sum_{x:\,w_{x}(i)=0}P(x)<0,
\end{array}\right\}\jz
\label{cl multiR1}
\end{align}
independently of whether \eqref{cl multiR limit1} or \eqref{cl multiR limit2} is used.

Alternatively, one may define
\begin{align*}
\tilde Q_P^{\bary,\cl}(w):=\sup_{\tau\in[0,+\infty)^{\I}}\left\{\sum_{i\in\I}\tau(i)-\int_{\X}D(\tau\|w_x)\,dP(x)\right\},
\end{align*}
which is well-defined at least when $P$ is a probability measure, all the 
$w_x$ are probability density functions, and
$\X\ni x\mapsto D(\tau\|w_x)$ is measurable.
Again, we restrict to the case when $P$ is finitely supported, but allow it to be a signed probability measure, in which case we use a slight modification of the above to define
\begin{align}\label{cl multiR2}
Q_P^{\bary,\cl}(w):=
\sup_{\substack{\tau\in[0,+\infty)^{\I}\\ \supp\tau\subseteq\bigcap_{x:\,P(x)>0}\supp w_x}}
\left\{\sum_{i\in\I}\tau(i)-\sum_{x\in\X}P(x)D(\tau\|w_x)\right\}.
\end{align}
We will show in Section \ref{sec:defs} that this is equivalent to \eqref{classical variational2} 
when $\X=\{0,1\}$, $P(0)=\alpha\in(0,1)\cup(1,+\infty)$. 
It is not too difficult to see that with the definition in \eqref{cl multiR2}, we have 
\begin{align*}
Q_P^{\bary,\cl}(w)=+\infty\ds\iff\ds
\bigcap_{x:\,P(x)>0}\supp w_x\not\subseteq\bigcap_{x:\,P(x)<0}\supp w_x;
\end{align*}
see Proposition \ref{prop:multi psi bounds} for the more general quantum case.
Thus, while \eqref{cl multiR1} and \eqref{cl multiR2} coincide when 
$P$ is a probability measure, they may differ when $P$ can take negative values. 
The following is easy to verify from \eqref{cl multiR1} and \eqref{cl multiR2}:

\begin{lemma}\label{lemma:cl Q uniqueness}
Let $P\in\P_f^{\pm}(\X)$ and $w\in\F(\X,\I)\pne$, and assume that at least one of the following holds true:
\begin{enumerate}
\item
$\supp w_x=\supp w_{x'}$, $x,x'\in\supp P$;
\item\label{Q uniqueness ii}
$P\in\P_f(\X)\cup\P_{f,1}^{\pm}(\X)$, i.e., 
either $P(x)\ge 0$, $x\in\X$, or 
there exists a unique $x_+$ with $P(x_+)>0$ and 
$P(x)\le 0$, $x\in\X\setminus\{x_+\}$. 
\end{enumerate}
Then $\hat Q_P(w)=Q_P^{\bary,\cl}(w)$. 
\end{lemma}

\begin{definition}\label{def:classical QP}
For any $P\in\P_f^{\pm}(\X)$ and any $w\in\F(\X,\I)\pne$ such that 
$\hat Q_P(w)=Q_P^{\bary,\cl}(w)$, we call this common value the 
\ki{multi-variate R\'enyi $Q_P$ of $w$}, and denote it by $Q_P(w)$.

For any $P\in\bz\P_f(\X)\cup\P_{f,1}^{\pm}(\X)\jz\setminus\{\egy_{\{x\}}:\,x\in\X\}$, 
we define the 
\ki{(symmetrically normalized) classical $P$-weighted R\'enyi-divergence}
of $w\in\F(\X,\I)\pne$ as 
\begin{align*}
D_P(w)&:=\frac{1}{\prod_{x\in\X}(1-P(x))}
\bz -\log Q_P(w)+\sum_{x\in\X}P(x)\log\sum_{i\in\I}w_x(i)\jz\\
&=
\frac{1}{\prod_{x\in\X}(1-P(x))}
\bz -\log Q_P\bz\bz\frac{w_x}{\sum_{i\in\I}w_x(i)}\jz_{x\in\X}\jz\jz.
\end{align*}
In this case we also define
\begin{align*}
\signed{Q}_P(w):=s(P)Q_P(w),
\end{align*}
where
\begin{align*}
s(P):=\begin{cases}
-1,&P\in\P_f(\X),\\
1,&P\in \P_{f,1}^{\pm}(\X).
\end{cases}
\end{align*}
\end{definition}

Lemma \ref{lemma:cl Q uniqueness} and 
\eqref{cl multiR1} yield immediately the following:
\begin{cor}
Let $P\in\P_f(\X)\cup\P_{f,1}^{\pm}(\X)$ and $w\in\F(\X,\I)\pne$. Then 
\begin{align}\label{cl multiR smoothing}
Q_P(w)=\lim_{\ep\searrow 0}Q_P(w+\ep)=\lim_{\ep\searrow 0}Q_P((1-\ep)w+\ep/|\I|).
\end{align}
\end{cor}

\begin{rem}
In the case when $P$ is a probability measure, $Q_P(w)$ was introduced in 
classical decision theory, and called the 
\ki{$P$-weighted 
Hellinger transform of $w$};
see \cite{Strasser1985}. The case where $P(x)>0$ for exactly one $x$ was very recently considered in \cite{Mu_Econometrica2020} in the context of (classical) Blackwell dominance of experiments, 
and in
\cite{farooq2023asymptotic} in the case where all $w$ are strictly positive, in the context of 
classical state convertibility. 
\end{rem}

Note that in the case when $\X=\{0,1\}$ and $\alpha:=P(0)\in(0,1)\cup(1,+\infty)$,
condition \ref{Q uniqueness ii} in Lemma \ref{lemma:cl Q uniqueness} is always satisfied, and we have
\begin{align*}
Q_P(w)=Q_{\alpha}(w_0\|w_1).
\end{align*}
That is, the multi-variate $Q_P$ give
multi-variate extensions of the $Q_{\alpha}$ quantities.

\begin{rem}
Note that when $\X=\{0,1\}$ and $\alpha:=P(0)=0$, neither 
$\hat Q_P(w)$ nor $Q_P^{\bary,\cl}(w)$ coincides with $Q_{0}(w_0\|w_1)$
in general. The reason for this in the case of 
$\hat Q_P(w)$ is that the limits $\ep\searrow 0$ and 
$\alpha\searrow 0$ are not interchangeable, while in the case of 
$Q_P^{\bary,\cl}(w)$, it is clear that it only depends on 
$(w_x)_{x\in\supp P}$, while $Q_0$ depends on $w_0$ (or at least its support) 
even though $0\notin\supp P=\supp(0,1)=\{1\}$.
\end{rem}

Recall that classical divergences can be identified with quantum divergences defined on 
commuting operators; in particular, monotonicity under (completely) positive trace-preserving maps makes sense for the former. 
For the purposes of applications, it is monotone divergences that are relevant, 
and monotonicity is closely related to joint convexity.
The following is easy to verify; see, e.g., 
the supplementary material to \cite[Lemma 8]{Mu_Econometrica2020}.

\begin{lemma}\label{lemma:classical QP jointly convex}
Let $P$ be a finitely supported signed probability measure on $\X$. The following are equivalent:
\begin{enumerate}
\item\label{cl multiR convex1}
$\signed{Q}_P$ is jointly convex on $\F(\X,\I)\pne$ for any/some finite non-empty $\I$; 
\item\label{cl multiR convex2}
$\signed{Q}_P$ is jointly convex on $\F(\X,\I)\pp$ for any/some finite non-empty $\I$; 
\item\label{cl multiR convex3}
$P\in\P_f(\X)\cup\P_{f,1}^{\pm}(\X)$. 
\end{enumerate}
\end{lemma}
\begin{proof}
The equivalence of \ref{cl multiR convex2} and \ref{cl multiR convex3}
is due to \cite[Lemma 8]{Mu_Econometrica2020}, and the equivalence of 
\ref{cl multiR convex2} and \ref{cl multiR convex1} is immediate from 
\eqref{cl multiR smoothing}.
\end{proof}

\begin{cor}
For any $P\in\P_f(\X)\cup\P_{f,1}^{\pm}(\X)$, $\signed{Q}_P$ is 
jointly convex and
monotone under positive trace-preserving maps, and 
$D_P$ is monotone under positive trace-preserving maps whenever
$P\in\bz\P_f(\X)\cup\P_{f,1}^{\pm}(\X)\jz\setminus\{\egy_{\{x\}}:\,x\in\X\}$.
\end{cor}
\begin{proof}
It is straightforward to verify that $\signed{Q}_P$ is homogeneous, block additive, and 
stable under tensoring with an arbitrary state, and hence the assertion follows from 
Lemma \ref{lemma:classical QP jointly convex} and 
Lemma \ref{lemma:jointconc from mon}.
\end{proof}

\subsection{Quantum R\'enyi divergences}
\label{sec:qRenyi}

In this section we give a brief review of the ($2$-variable) quantum R\'enyi divergences most commonly used in the literature, which will also play an important role in the rest of the paper. We will discuss various ways to define multi-variate quantum R\'enyi divergences in 
Sections \ref{sec:nc gm} and \ref{sec:barycentric}.

\begin{definition}\label{def:quantumdiv}
For any $\alpha\in[0,+\infty]$,
by a \ki{quantum R\'enyi $\alpha$-divergence} we mean a 
quantum divergence that is a 
quantum extension of the classical R\'enyi $\alpha$-divergence. 
Similarly, by a \ki{quantum relative entropy} we mean a quantum extension of the relative entropy. 
\end{definition}

\begin{rem}\label{rem:D-Q bijection}
Note that for any $\alpha\in[0,1)\cup(1,+\infty)$, there is an obvious bijection between 
quantum extensions of $Q_{\alpha}$ and quantum extensions of $D_{\alpha}$, 
given by
\begin{align*}
D_{\alpha}^q(\rho\|\sigma)&=\frac{1}{\alpha-1}\log Q_{\alpha}^q(\rho\|\sigma)-\frac{1}{\alpha-1}\log\Tr\rho,\\
Q_{\alpha}^q(\rho\|\sigma)&=(\Tr\rho)\exp\bz(\alpha-1)D_{\alpha}^q(\rho\|\sigma)\jz.
\end{align*}
\end{rem}

\begin{rem}\label{rem:zero argument q}
Since $0$ commutes with any other operator, any quantum R\'enyi $\alpha$-divergence 
$D_{\alpha}^q$ must satisfy
\begin{align}\label{zero argument q1}
D_{\alpha}^q(\rho\|0)=+\infty,\ds\ds\ds\rho\gneq 0,
\end{align}
according to \eqref{zero argument cl1}, and 
any quantum relative  entropy $D^q$ must satisfy
\begin{align}\label{zero argument q2}
\D^q(0\|\sigma)=0\ds\ds\sigma\ge 0,\ds\ds\ds\ds\ds
\D^q(\rho\|0)=+\infty\ds\ds\rho\gneq 0,
\end{align}
according to \eqref{zero argument cl2}.

Since these values are fixed by definition, in the discussion 
of different quantum R\'enyi divergences and relative entropies below,
it is sufficient to consider non-zero arguments most of the time. 
\end{rem}

\begin{rem}
Note that there is a bijection between quantum extensions of the R\'enyi $1$-divergence 
and quantum extensions of the relative entropy, given in one direction by 
$D^q(\rho\|\sigma):=(\Tr\rho)D_1^q(\rho\|\sigma)$, and 
in the other direction by $D_1^q(\rho\|\sigma):=D^q(\rho\|\sigma)/\Tr\rho$, 
for any non-zero $\rho$. 
\end{rem}

The following examples of quantum R\'enyi $\alpha$-divergences are well studied in the literature. We review them in some detail for later use.

\begin{example}\label{ex:az Renyi}
For any $\alpha\in[0,1)\cup(1,+\infty)$ and $z\in(0,+\infty)$, 
the \ki{R\'enyi $(\alpha,z)$-divergence} 
of $\rho,\sigma\in\B(\hil)\pne$ is defined as
\cite{AD}
\begin{align}
D_{\alpha,z}(\rho\|\sigma)&:=\frac{1}{\alpha-1}\log Q_{\alpha,z}(\rho\|\sigma)
-\frac{1}{\alpha-1}\log\Tr\rho,\nn\\
Q_{\alpha,z}(\rho\|\sigma)&:=
\begin{cases}
\Tr\bz\rho^{\frac{\alpha}{2z}}\sigma^{\frac{1-\alpha}{z}}
\rho^{\frac{\alpha}{2z}}\jz^z,&\alpha\in[0,1)\text{ or }
\rho^0\le\sigma^0,\\
+\infty,&\text{otherwise.}
\end{cases}
\label{az def2}
\end{align}
It is easy to see that it
defines a quantum R\'enyi $\alpha$-divergence in the sense of Definition 
\ref{def:quantumdiv}.
$D_{\alpha,1}(\rho\|\sigma)$ is called the \ki{Petz-type} (or \ki{standard})
R\'enyi $\alpha$-divergence \cite{P86} of $\rho$ and $\sigma$, and 
$D\nw_{\alpha}(\rho\|\sigma):=D_{\alpha,\alpha}(\rho\|\sigma)$ their \ki{sandwiched R\'enyi $\alpha$-divergence} \cite{Renyi_new,WWY}.
The limit
\begin{align}
D_{\alpha,+\infty}(\rho\|\sigma)
&:=
\lim_{z=\to+\infty}D_{\alpha,z}(\rho\|\sigma)\\
&=
\begin{cases}
\frac{1}{\alpha-1}\log\underbrace{\Tr Pe^{\alpha P(\nlog\rho)P+(1-\alpha)P(\nlog\sigma)P}}_{=:Q_{\alpha,+\infty}(\rho\|\sigma))}
-\frac{1}{\alpha-1}\log\Tr\rho,&\alpha\in(0,1)\text{ or }\rho^0\le\sigma^0,\\
+\infty,&\text{otherwise},
\end{cases}\label{D alpha infty}
\end{align}
where $P:=\rho^0\wedge\sigma^0$, is also a quantum R\'enyi 
$\alpha$-divergence, often referred to as the \ki{log-Euclidean} R\'enyi 
$\alpha$-divergence \cite{AD,HP_GT,MO}.
It is known \cite{LinTomamichel15,Hiai_Mosonyi_cont_2023} that for any function $z:\,(1-\delta,1+\delta)\to(0,+\infty]$ such that 
$\liminf_{\alpha\to 1}z(\alpha)>0$, and for any
$\rho,\sigma\in\B(\hil)\pne$, 
\begin{align*}
\lim_{\alpha\to 1}D_{\alpha,z(\alpha)}(\rho\|\sigma)=\frac{1}{\Tr\rho}\DU(\rho\|\sigma)
=:D_1^{\Um}(\rho\|\sigma),
\end{align*}
where the Umegaki relative entropy $\DU(\rho\|\sigma)$ is defined 
as
\begin{align}\label{Umegaki def}
\DU(\rho\|\sigma):=
\begin{cases}
\Tr(\rho\nlog\rho-\rho\nlog\sigma),
& \rho^0\le\sigma^0,\\
+\infty,&\text{otherwise.}
\end{cases}
\end{align}
In particular, for any $z\in(0,+\infty]$, we define
$D_{1,z}(\rho\|\sigma):=D_1^{\Um}(\rho\|\sigma)$.

The R\'enyi $(\alpha,z)$-divergence is strictly positive 
for $\alpha\in(0,1)$ and $z\ge\max\{\alpha,1-\alpha\}$, and for 
$\alpha\ge 1$ and $z>0$; see
\cite[Corollary A.31]{MO-cqconv-cc}.
The range of 
$(\alpha,z)$-values for which $D_{\alpha,z}$ is monotone under CPTP maps was 
studied in a series of works 
\cite{Beigi,CFL_AD_conjecture,FL,Hiai-convexity,P86}, and was finally characterized completely 
in \cite{Zhang2018}. 
It is clear from their definitions that for every $\alpha\in(0,+\infty)$
and $z\in(0,+\infty]$, the R\'enyi $(\alpha,z)$-divergence is additive on tensor products.
\end{example}

\begin{example}\label{ex:meas Renyi}
For any $\alpha\in[0,+\infty]$ and $\rho,\sigma\in\B(\hil)\pne$, their 
\ki{measured R\'enyi $\alpha$-divergence} $D_{\alpha}^{\meas}(\rho\|\sigma)$ is defined as 
a special case of \eqref{reg meas div def} with $\divv=D_{\alpha}$, and their \ki{measured relative entropy}
$D^{\meas}(\rho\|\sigma)=(\Tr\rho)D_1^{\meas}(\rho\|\sigma)$ is also a special case of 
\eqref{meas div def} with $\divv=D$. We have 
$D_0^{\meas}=D_{0,1}$,  
$D_{1/2}^{\meas}=D_{1/2,1/2}$ (see \cite[Chapter 9]{NC}) and
for any $\rho,\sigma\in\B(\hil)\pne$, 
\begin{align}\label{maxrelentr}
D_{+\infty}^{\meas}(\rho\|\sigma)=D_{+\infty}\nw(\rho\|\sigma):=
D_{+\infty,+\infty}(\rho\|\sigma):=\lim_{\alpha\to+\infty}D_{\alpha,\alpha}(\rho\|\sigma)
=\log\inf\{\lambda\ge 0:\,\rho\le\lambda\sigma\},
\end{align}
where the last quantity was introduced in \cite{Datta} under the name 
\ki{max-relative entropy}, and its equality to the limit above has been shown in 
\cite[Theorem 5]{Renyi_new}.
No explicit expression is known for $D_{\alpha}^{\meas}$ for other $\alpha$ values.

Similarly, the \ki{regularized measured R\'enyi $\alpha$-divergence} 
$\oll D_{\alpha}^{\meas}(\rho\|\sigma)$ is obtained as a 
special case of \eqref{reg meas div def}. Surprisingly, it has a closed formula
for every $\alpha\in[0,+\infty]$, given by 
\begin{align}
\oll{D}_{\alpha}^{\meas}(\rho\|\sigma)
=
\begin{cases}
D_{\alpha,\alpha}(\rho\|\sigma),&\alpha\in[1/2,+\infty],\\
\frac{\alpha}{1-\alpha}D_{1-\alpha,1-\alpha}(\sigma\|\rho)
+\frac{1}{\alpha-1}\log\frac{\Tr\rho}{\Tr\sigma}
=
D_{\alpha,1-\alpha}(\rho\|\sigma),&\alpha\in(0,1/2),\\
D_{\alpha,1-\alpha}(\rho\|\sigma),&\alpha=0;\\
\end{cases}
\label{regmeasured explicit}
\end{align}
see \cite{HP} for $\alpha=1$, \cite{MO} for $\alpha\in(1,+\infty)$, and \cite{HT14}
for $\alpha=(1/2,1)$;
the last expression for $\alpha\in(0,1/2)$ above was first observed by P\'eter Vrana in August 2022, to the best of 
our knowledge. The case $\alpha=0$ follows from $D_{0}^{\meas}=D_{0,1}$ and the additivity of the latter.

For every $\alpha\in(0,1)$, strict positivity of 
$D_{\alpha}^{\meas}$ is immediate from the strict positivity of the classical R\'enyi
$\alpha$-divergence, which is a straightforward corollary of H\"older's inequality,
and strict positivity of 
$D_{\alpha}^{\meas}$ for $\alpha\in[1,+\infty]$ follows from this and the easily verifiable fact that $\alpha\mapsto D_{\alpha}^{\meas}$ is monotone increasing. 
Strict positivity of $\oll{D}_{\alpha}^{\meas}$
follows from $D_{\alpha}^{\meas}\le \oll{D}_{\alpha}^{\meas}$.

For any $\alpha\in[0,+\infty]$, the measured R\'enyi $\alpha$-divergence is superadditive 
on tensor products
(see Example \ref{ex:minmax extension}), 
but not additive unless $\alpha\in\{0,1/2,+\infty\}$; see, e.g., \cite[Remark 4.27]{HiaiMosonyi2017}
and \cite[Proposition III.13]{MH-testdiv} for the latter. On the other hand, 
for every $\alpha\in[0,+\infty]$, the regularized measured 
R\'enyi $\alpha$-divergence is not only weakly additive 
(see Corollary \ref{cor:regmeas regmax weakly additive}),
but even additive on tensor products, according to Example \ref{ex:az Renyi}
and \eqref{regmeasured explicit}.

Monotonicity of 
$D_{\alpha}^{\meas}$ under PTP maps and of
$\oll{D}_{\alpha}^{\meas}$ 
under CPTP maps
is obvious by definition for every $\alpha\in[0,+\infty]$ 
(see Example \ref{ex:minmax extension}). Moreover,
$\oll{D}_{\alpha}^{\meas}$, $\alpha\in[0,+\infty]$, are monotone even under PTP maps, according to \cite{Beigi,Jencova_NCLpII,MHR} and \eqref{regmeasured explicit}.
In particular, the Umegaki relative entropy $\DU$ is monotone under PTP maps,
as observed in \cite{MHR}.
\end{example}

In the following example, we use (properties of) the absolutely continuous part of PSD operators and the Kubo-Ando geometric means; see Appendix \ref{sec:oppersp} for details.
Moreover, the quantities below are special cases of the maximal $f$-divergences discussed in Appendix \ref{sec:maxfdiv}. 

\begin{example}\label{ex:max Renyi}
For any $\alpha\in[0,+\infty]$ and $\rho,\sigma\in\B(\hil)\pne$, their 
\ki{maximal R\'enyi $\alpha$-divergence} $D_{\alpha}^{\max}(\rho\|\sigma)$ is defined as 
a special case of \eqref{maxdiv def} with $\divv=D_{\alpha}$, and their \ki{maximal relative entropy}
$D^{\max}(\rho\|\sigma)=(\Tr\rho)D_1^{\max}(\rho\|\sigma)$ is also a special case of 
\eqref{maxdiv def} with $\divv=D$.

Let $\rho_{\sigma,\mathrm{ac}}:=
\max\{0\le C\le \rho:\,C^0\le\sigma^0\}=\sigmasupp\rho \sigmasupp-\sigmasupp\rho (\sigmasupp^{\perp}\rho \sigmasupp^{\perp})\inv \rho \sigmasupp$
be the \ki{absolutely continuous part} of $\rho$ w.r.t.~$\sigma$ \cite{AndersonTrapp1975}, where 
$\sigmasupp:=\sigma^0$, and let  
$\lambda_i$, $i\in[r]$, be the different
eigenvalues of $\sigma^{-1/2}\rho_{\sigma,\mathrm{ac}}\sigma^{-1/2}$ with corresponding 
spectral projections $E_i$.
Let $\I:=[r]\cup\{r+1\}$, let $\tau_0\in\S(\hil)$ be arbitrary, and 
\begin{align*}
\tau_1:=\begin{cases}
\frac{\rho-\rho_{\sigma,\mathrm{ac}}}{\Tr(\rho-\rho_{\sigma,\mathrm{ac}})},&\rho^0\not\le\sigma^0,\\
\tau_0,&\text{otherwise.}
\end{cases}
\end{align*}
According to \cite{Matsumoto_newfdiv}, 
\begin{align}
&\hat p(i):=\begin{cases}
\lambda_i\Tr \sigma E_i,&i\in[r],\\
\Tr(\rho-\rho_{\sigma,\mathrm{ac}}),&i=r+1,
\end{cases}\ds\ds
\hat q(i):=
\begin{cases}
\Tr \sigma E_i,&i\in[r],\\
0,&i=r+1,
\end{cases}\label{optimal reverse test}\\
&\hat\rt(\egy_{\{i\}}):=\begin{cases}
\frac{\sigma^{1/2}E_i\sigma^{1/2}}{\Tr\sigma E_i},&i\in[r],\,\Tr\sigma E_i\ne 0,\\
\tau_0,&i\in[r],\,\Tr\sigma E_i=0,\\
\tau_1,&i=r+1,
\end{cases}\label{optimal reverse test gamma}
\end{align}
is a reverse test for $(\rho,\sigma)$ that is 
optimal for every $D_{\alpha}^{\max}(\rho\|\sigma)$, $\alpha\in[0,2]\cup\{+\infty\}$, and 
\begin{align}
Q_{\alpha}^{\max}(\rho\|\sigma)&=\what Q_{\alpha}(\rho\|\sigma):=Q_{\alpha}(\hat p\|\hat q)
=
\Tr\persp{(f_{\alpha})}(\rho\|\sigma)\label{max Q explicit}\\
&=
\begin{cases}
\Tr\sigma \bz\sigma^{-1/2}\rho_{\sigma,\mathrm{ac}}\sigma^{-1/2}\jz^{\alpha}
=\Tr\sigma\#_{\alpha}\rho,
&\alpha\in[0,1),\\
\Tr\sigma \bz\sigma^{-1/2}\rho\sigma^{-1/2}\jz^{\alpha}=\Tr\sigma\#_{\alpha}\rho,&\alpha\in(1,2],\,
\rho^0\le\sigma^0,\\
+\infty,&\alpha\in(1,2],\,
\rho^0\nleq\sigma^0,
\end{cases}
\label{max Renyi explicit}\\
D^{\max}(\rho\|\sigma)&=
(\Tr\rho)D_1^{\max}(\rho\|\sigma)=
D(\hat p\|\hat q)
=
\Tr\persp{\eta}(\rho,\sigma)\label{BS def1}\\
&=
\begin{cases}
\Tr\sigma^{1/2}\rho\sigma^{-1/2}\nlog\bz\sigma^{-1/2}\rho\sigma^{-1/2}\jz=
\Tr\rho\nlog\bz\rho^{1/2}\sigma\inv\rho^{1/2}\jz,&\rho^0\le\sigma^0,\\
+\infty,&\text{otherwise},
\end{cases}\label{BS def2}\\
D_{+\infty}^{\max}(\rho\|\sigma)&=D_{+\infty}(\hat p\|\hat q)=D_{+\infty}\nw(\rho\|\sigma),
\label{maxRenyi infty}
\end{align}
where $f_{\alpha}:=\id_{[0,+\infty)}^{\alpha}$, $\eta(x):=x\log x$, $x\ge 0$.
The expression in \eqref{BS def2} is called the \ki{Belavkin-Staszewski relative entropy} 
\cite{BS}.
Note that the optimal reverse test above is independent of $\alpha$; 
in particular, 
\begin{align*}
D_1^{\max}(\rho\|\sigma)=\lim_{\alpha\to 1}D_{\alpha}^{\max}(\rho\|\sigma)=\frac{1}{\Tr\rho}D^{\max}(\rho\|\sigma)
\end{align*}
holds for any $\rho,\sigma\in\B(\hil)\pne$, according to \eqref{1-Renyi cl}.

No explicit expression is known for $D_{\alpha}^{\max}$ when $\alpha\in(2,+\infty)$,
in which case the above reverse test is known not to be optimal. 
Indeed, also for $\alpha\in(2,+\infty)$, we have 
$\what Q_{\alpha}(\rho\|\sigma):=Q_{\alpha}(\hat p\|\hat q)=\Tr\sigma \bz\sigma^{-1/2}\rho\sigma^{-1/2}\jz^{\alpha}$ when $\rho^0\le\sigma^0$, and 
$\what Q_{\alpha}(\rho\|\sigma):=Q_{\alpha}(\hat p\|\hat q)=+\infty$, otherwise. However, 
since $f_{\alpha}$ is not operator convex for $\alpha>2$, 
this quantity is not convex in $\rho$, and therefore not CPTP-monotone, either, according to 
Lemma \ref{lemma:jointconc from mon}, while $Q_{\alpha}^{\max}$ is (C)PTP-monotone.

Strict positivity of $D_{\alpha}^{\max}$ for all $\alpha\in(0,+\infty]$ follows from that of $D_{\alpha}^{\meas}$ and the inequality
$D_{\alpha}^{\meas}\le D_{\alpha}^{\max}$, which is due to the monotonicity of the
classical R\'enyi divergences under stochastic maps.

According to Corollary \ref{cor:maxfdiv regularity} (see also 
 \cite{Hiai_fdiv_max,HiaiMosonyi2017,Matsumoto_newfdiv}), 
for any $\rho,\sigma\in\B(\hil)\pne$ and $\alpha\in(0,2]$, we have
\begin{align}\label{max Renyi smooth}
D_{\alpha}^{\max}(\rho\|\sigma)=
\lim_{\ep\searrow 0}D_{\alpha}^{\max}(\rho+\ep I\|\sigma+\ep I)
=
\lim_{\ep\searrow 0}D_{\alpha}^{\max}(\rho\|\sigma+\ep I);
\end{align}
in particular, $D_{\alpha}^{\max}$ is regular in its second argument.
For any $\alpha\in[0,2]$, $D_{\alpha}^{\max}$ is anti-monotone in its second argument, according to 
Corollary \ref{cor:maxfdiv regularity} and the operator monotonicity of 
$(1-\alpha)\id_{(0,+\infty)}^{1-\alpha}$ and of $\eta$ (for $\alpha=1$).
By Remark \ref{rem:AM regular}, $D_{\alpha}^{\max}$ is also strongly regular in its second argument for every $\alpha\in(0,2]$.

It is immediate from their definition that $D_{\alpha}^{\max}$, $\alpha\in[0,+\infty]$, are subadditive on tensor products (see Example \ref{ex:minmax extension}). 
For $\alpha\in[0,2]\cup\{+\infty\}$, $D_{\alpha}^{\max}$ is even additive, 
according to \eqref{max Renyi explicit} and \eqref{mean tensor}
for $\alpha\in[0,2]$, and to \eqref{maxrelentr} and \eqref{maxRenyi infty} for $\alpha=+\infty$.
However, additivity 
of $D_{\alpha}^{\max}$ is not known for $\alpha\in(2,+\infty)$. 
In particular, we have
\begin{align*}
\oll{D}_{\alpha}^{\max}\begin{cases}
=D_{\alpha}^{\max},&\alpha\in[0,2]\cup\{+\infty\},\\
\le D_{\alpha}^{\max},&\alpha\in(2,+\infty).
\end{cases}
\end{align*}
\end{example}

\begin{rem}
Note that, with the notations of Example \ref{ex:max Renyi},
\begin{align}\label{0 max Renyi1}
\lim_{\alpha\searrow 0}Q_{\alpha}^{\max}(\rho\|\sigma)
=
\Tr\sigma(\sigma^{-1/2}\rho_{\sigma,\mathrm{ac}}\sigma^{-1/2})^0
=
\Tr\sigma\sum_{i:\,\lambda_i>0}E_i,
\end{align}
while 
\begin{align}\label{0 max Renyi2}
Q_0^{\max}(\rho\|\sigma)=Q_0(\hat p\|\hat q)=\sum_i\hat p(i)^0\hat q(i)
=
\sum_{i:\,\lambda_i\Tr\sigma E_i>0}\Tr\sigma E_i
=
\Tr\sigma\sum_{i:\,\lambda_i\Tr\sigma E_i>0}E_i.
\end{align}
Since 
\begin{align*}
E_i\sigma^{-1/2}\rho_{\sigma,\mathrm{ac}}\sigma^{-1/2}E_i=\lambda_i E_i,
\end{align*}
we see that $\lambda_i>0$ $\imp$ $E_i\not\perp\sigma^0$ $\iff$ $\Tr\sigma E_i>0$, 
and hence \eqref{0 max Renyi1} and \eqref{0 max Renyi2} are equal to each other, 
i.e., 
\begin{align*}
\lim_{\alpha\searrow 0}D_{\alpha}^{\max}(\rho\|\sigma)
=
D_0^{\max}(\rho\|\sigma),\ds\ds\ds\rho,\sigma\in\B(\hil)\pne.
\end{align*}
\end{rem}

\begin{example}
For any quantum R\'enyi $\alpha$-divergence $D_{\alpha}^q$, its \ki{regularization} on a pair 
$\rho,\sigma\in\B(\hil)\pne$ is defined as 
\begin{align*}
\oll{D}_{\alpha}^q(\rho\|\sigma)
:=
\lim_{n\to+\infty}\frac{1}{n}D_{\alpha}^q(\rho^{\otimes n}\|\sigma^{\otimes n}),
\end{align*}
whenever the limit exists. If the limit exists for all 
$\rho,\sigma\in\B(\hil)\pne$, then $\oll{D}_{\alpha}^q$ is a quantum 
R\'enyi $\alpha$-divergence that is weakly additive, and if $D_{\alpha}^q$ is monotone under CPTP maps then so is $\oll{D}_{\alpha}^q$.
\end{example}

\begin{rem}
Note that if $D_1^q$ is an additive quantum $1$-divergence then 
the corresponding quantum relative entropy $D^q$ is not additive; instead, it satisfies
\begin{align*}
D^q(\rho_1\otimes\rho_2\|\sigma_1\otimes\sigma_2)
=
(\Tr\rho_2)D^q(\rho_1\|\sigma_1)+
(\Tr\rho_1)D^q(\rho_2\|\sigma_2)
\end{align*}
for any $\rho_k,\sigma_k\in\B(\hil_k)\pne$, $k=1,2$. Thus, the natural notion of regularization for a quantum relative entropy $D^q$ 
on a pair $\rho,\sigma\in\B(\hil)\pne$ is 
\begin{align*}
\oll{D}(\rho\|\sigma):=(\Tr\rho)\oll{D}_1^q(\rho\|\sigma),
\end{align*}
which is well-defined whenever $\oll{D}_1^q(\rho\|\sigma)$ is. 
Clearly, if $\oll{D}(\rho\|\sigma)$ is well-defined for all $\rho,\sigma\in\B(\hil)\pne$
then it gives a quantum relative entropy that is weakly additive, and if 
$D^q$ is monotone under CPTP maps then so is $\oll{D}^q$.
\end{rem}

\begin{rem}\label{rem:relentr maxmin}
According to Remark \ref{ex:minmax extension}, for any given 
$\alpha\in[0,+\infty]$, and any quantum R\'enyi $\alpha$-divergence 
$D_{\alpha}^q$ that is monotone under CPTP maps, 
\begin{align}\label{Renyi maxmin}
D_{\alpha}^{\meas}\le D_{\alpha}^q\le D_{\alpha}^{\max}.
\end{align}
If the regularization of $D_{\alpha}^q$ is well-defined then we further have 
\begin{align}\label{Renyi maxmin2}
\oll{D}_{\alpha}^{\meas}\le \oll{D}_{\alpha}^q\le \oll{D}_{\alpha}^{\max};
\end{align}
in particular, this is the case if $D_{\alpha}^q$ is additive, when we also have 
$\oll{D}_{\alpha}^q=D_{\alpha}^q$.

Likewise, for any quantum relative entropy $D^q$ that is monotone under CPTP maps, 
\begin{align}\label{relentr maxmin}
D^{\meas}\le D^q\le D^{\max},
\end{align}
and if the regularization of $D^q$ is well-defined then we further have 
\begin{align}\label{relentr maxmin2}
\DU=\oll{D}^{\meas}\le \oll{D}^q\le \oll{D}^{\max}=D^{\max};
\end{align}
in particular, this is the case if $D_{1}^q$ is additive, when we also have 
$\oll{D}^q=D^q$.

It is also known that 
\begin{align}\label{relentr strict inequalities}
D^{\meas}<\DU<D^{\max};
\end{align}
see \cite[Theorem 4.18]{HiaiMosonyi2017} for the first inequality 
(also \cite{BFT_variational} for a slightly weaker statement), and
\cite[Theorem 4.3]{HiaiMosonyi2017} for the second inequality.
\end{rem}

\begin{rem}
Note that $D_{+\infty}^{\meas}=\oll{D}_{+\infty}^{\meas}=
\oll{D}_{+\infty}^{\max}=
D_{+\infty}^{\max}$ is the unique quantum extension of $D_{+\infty}$ that is monotone under 
(completely) positive trace-preserving maps, as it was observed in 
\cite{TomamichelBook}, and this unique extension also happens to be additive. 
On the other hand, for any other $\alpha\in[0,+\infty)$, there are infinitely many different monotone and additive
quantum R\'enyi $\alpha$-divergences; see, e.g., Example \ref{ex:az Renyi}.
\end{rem}

\begin{rem}
According to Remark \ref{rem:scaling}, any additive quantum R\'enyi $\alpha$-divergence 
$D_{\alpha}^q$ satisfies the scaling law
\begin{align}\label{Renyi scaling}
D_{\alpha}^q(t\rho\|s\sigma)=D_{\alpha}^q(\rho\|\sigma)+D_{\alpha}(t\|s)
=
D_{\alpha}^q(\rho\|\sigma)+\log t-\log s.
\end{align}
In particular, this holds for $D_{\alpha,z}$, $\alpha\in[0,+\infty)$, $z\in(0,+\infty]$, 
and $D_{\alpha}^{\max}$, $\alpha\in[0,2]\cup\{+\infty\}$. It is easy to verify that 
$D_{\alpha}^{\max}$ also satisfies \eqref{Renyi scaling} 
for every $\alpha\in(2,+\infty)$, where additivity is not known, and 
$D_{\alpha}^{\meas}$ also satisfies \eqref{Renyi scaling}
for every $\alpha\in[0,+\infty]$, even though they are not additive unless 
$\alpha\in\{0,1/2,+\infty\}$. 

Note that a quantum R\'enyi $1$-divergence $D_1^q$ satisfies the scaling law \eqref{Renyi scaling} if and only if the corresponding quantum relative entropy $D^q$ 
satisfies the scaling law
\begin{align}\label{relentr scaling}
D^q(t\rho\|s\sigma)=tD^q(\rho\|\sigma)+(\Tr\rho)D(t\|s),
\end{align}
which in turn equivalent to 
\begin{align}
\D^q(t\rho\|\sigma)&=(t\log t)\Tr\rho+t\D^q(\rho\|\sigma),\label{scaling1}\\
\D^q(\rho\|s\sigma)&=\D^q(\rho\|\sigma)-(\log s)\Tr\rho.\label{scaling2}
\end{align}
\end{rem}

\begin{rem}\label{rem:tracemon}
By definition, a quantum R\'enyi $\alpha$-divergence $D_{\alpha}^q$
is trace-monotone, if 
\begin{align}\label{Renyi tracemon}
D_{\alpha}^q(\rho\|\sigma)\ge D_{\alpha}(\Tr\rho\|\Tr\sigma)\ds(=\log\Tr\rho-\log\Tr\sigma)
\end{align}
for any $\rho,\sigma\in\B(\hil)\pne$, 
and it is strictly trace-monotone if 
equality holds in \eqref{Renyi tracemon} if and only if $\rho=\sigma$.
Likewise, a quantum relative entropy $D^q$ is 
trace-monotone, if 
\begin{align}\label{relentr tracemon}
D^q(\rho\|\sigma)\ge D(\Tr\rho\|\Tr\sigma)\ds(=(\Tr\rho)\log\Tr\rho-(\Tr\rho)\log\Tr\sigma).
\end{align}
for any $\rho,\sigma\in\B(\hil)\pne$,
and it is strictly trace-monotone if 
equality holds in  \eqref{relentr tracemon} if and only if $\rho=\sigma$.
Obviously, any trace-monotone R\'enyi $\alpha$-divergence or relative entropy is non-negative. Moreover, it is easy to see that if  
a quantum R\'enyi $\alpha$-divergence $D_{\alpha}^q$ satisfies the scaling law 
\eqref{Renyi scaling} then it is non-negative (strictly positive) if and only if it is 
(strictly) trace-monotone, 
and similarly, if a quantum relative entropy $D^q$ satisfies the scaling law 
\eqref{relentr scaling} then it is non-negative (strictly positive)
if and only if it is (strictly)
trace-monotone.
\end{rem}

\begin{rem}\label{rem:relentr pos}
If a quantum relative entropy $D^q$ satisfies the trace monotonicity 
\eqref{relentr tracemon} then for any $\tau,\sigma\in\B(\hil)\pne$,
\begin{align}
\D^q(\tau\|\sigma)\ge -(\Tr\tau)\log\frac{\Tr\sigma}{\Tr\tau}\ge \Tr\tau\bz 1-\frac{\Tr\sigma}{\Tr\tau}\jz=\Tr\tau-\Tr\sigma, \label{Dq lower bound}
\end{align}
and equality holds everywhere when $\tau=\sigma$.
As an immediate consequence of this, for any $\sigma\in\B(\hil)\pne$,
\begin{align}
\Tr\sigma
&=
\max_{\tau\in\B(\hil)\p}\{\Tr\tau-\D^q(\tau\|\sigma)\}
=
\max_{\tau\in\B(\sigma^0\hil)\p}\{\Tr\tau-\D^q(\tau\|\sigma)\},\label{Tropp variational}\\
\log\Tr\sigma
&=
\max_{\tau\in\B(\hil)\pne}\left\{\log\Tr\tau-\frac{1}{\Tr\tau}\D^q(\tau\|\sigma)\right\}
=
\max_{\tau\in\B(\sigma^0\hil)\pne}\left\{\log\Tr\tau-\frac{1}{\Tr\tau}\D^q(\tau\|\sigma)\right\}.
\label{Tropp variational2}
\end{align}
Note that 
$\tau$ is a maximizer for \eqref{Tropp variational} if and only if 
$\Tr\tau=\Tr\sigma$ and $D^q(\tau\|\sigma)=0$
(since the second inequality in \eqref{Dq lower bound} holds as an equality 
if and only if $\Tr\tau=\Tr\sigma$), and if 
$D^q$ also satisfies the scaling property \eqref{relentr scaling} then  
$\tau$ is a maximizer for \eqref{Tropp variational2} if and only if 
$D^q\bz\frac{\tau}{\Tr\tau}\big\|\frac{\sigma}{\Tr\sigma}\jz=0$.
If $D^q$ is strictly trace monotone then $\tau=\sigma$ is the unique maximizer 
for all the expressions in \eqref{Tropp variational}--\eqref{Tropp variational2}.

The variational formula \eqref{Tropp variational} has already been noted in \cite[Lemma 6]{Tropp}
in the case $D^q=\DU$.
\end{rem}

\begin{rem}\label{rem:relentr lsc}
It is easy to see from their definitions that 
$D^{\meas}$, $\DU$, and $D^{\max}$
are all
regular and anti-monotone in their second argument (due to the operator monotonicity of $\log$ and operator anti-monotonicity of the inverse \cite{Bhatia}), i.e., 
\begin{align}\label{antimon regular}
D^q(\rho\|\sigma+\ep I)\nearrow 
D^q(\rho\|\sigma)\ds\ds\text{as}\ds\ds \ep\searrow 0.
\end{align}
By Remark \ref{rem:AM regular}, they are also strongly regular.
It is clear from \eqref{Umegaki def} and \eqref{BS def2} that for any fixed $\ep>0$, 
$\B(\hil)\pne^2\ni(\rho,\sigma)\mapsto D^q(\rho\|\sigma+\ep I)$ is continuous when $q=\Um$ or $q=\max$. Hence, by \eqref{antimon regular}, $\DU$ and $D^{\max}$ are both jointly lower semi-continuous in their arguments. In particular, the classical relative entropy is jointly lower semi-continuous,
whence $D^{\meas}$, as the supremum of lower semi-continuous functions, is also 
jointly lower semi-continuous.
\end{rem}

\begin{rem}
It is clear from \eqref{Umegaki def} and \eqref{BS def2} that $\DU$ and $\D^{\max}$ are block additive.
For $\D^{\meas}$, we have block subadditivity.
Indeed, let $\rho_k,\sigma_k\in\B(\hil_k)\pne$, $k=1,2$.
For any $(M_i)_{i=0}^{n-1}\in\povm(\hil_1,[n])$, $(N_i)_{i=0}^{n-1}\in\povm(\hil_2,[n])$,
\begin{align*}
&D\bz(\Tr \rho_1 M_i)_{i=0}^{n-1}\big\|(\Tr \sigma_1 M_i)_{i=0}^{n-1}\jz
+
D\bz(\Tr \rho_2  N_i)_{i=0}^{n-1}\big\|(\Tr \sigma_2N_i)_{i=0}^{n-1}\jz\\
&\ds\ge
D\bz (\Tr \rho_1 M_i+\Tr \rho_2 N_i)_{i=0}^{n-1}\big\|(\Tr \sigma_1M_i+\Tr \sigma_2 N_i)_{i=0}^{n-1}  \jz\\
&\ds=
D\bz (\Tr (\rho_1\oplus\rho_2)(M_i\oplus N_i))_{i=0}^{n-1} \big\|(\Tr (\sigma_1\oplus\sigma_2)
(M_i\oplus N_i))_{i=0}^{n-1}  \jz,
\end{align*}
where the inequality follows from the joint subadditivity of the relative entropy
(a consequence of joint convexity and homogeneity).
Taking the supremum over 
$(M_i)_{i=0}^{n-1}\in\povm(\hil_1,[n])$, $(N_i)_{i=0}^{n-1}\in\povm(\hil_2,[n])$, and then over $n$, we get 
\begin{align*}
&D^{\meas}(\rho_1\|\sigma_1)+D^{\meas}(\rho_2\|\sigma_2)\\
&\ds\ge 
\sup_{n\in\bN}\sup_{M,N}
D\bz (\Tr (\rho_1\oplus\rho_2)(M_i\oplus N_i))_{i=0}^{n-1} \big\|(\Tr (\sigma_1\oplus\sigma_2)
(M_i\oplus N_i))_{i=0}^{n-1}  \jz\\
&\ds=
\sup_{n\in\bN}\sup_{T\in\povm(\hil_1\oplus\hil_2,[n])}
D\bz((\Tr(\rho_1\oplus\rho_2)T_i)_{i=0}^{n-1}\big\|(\Tr(\sigma_1\oplus\sigma_2)T_i)_{i=0}^{n-1}\jz\\
&\ds=
D^{\meas}(\rho_1\oplus\rho_2\|\sigma_1\oplus\sigma_2).
\end{align*}
The first equality above follows from the simple fact that for any
$T\in\B(\hil_1\oplus\hil_2)$, 
$\Tr(\rho_1\oplus\rho_2)T=\Tr(\rho_1\oplus\rho_2)(T_1\oplus T_2)$, where 
$T_k=P_kTP_k$, with $P_k$ the projection onto $\hil_k$ in 
$\hil_1\oplus\hil_2$.
\end{rem}

\begin{example}\label{ex:minmax multi}
The measured and the maximal extensions of the multi-variate R\'enyi divergences can
be obtained as special cases of Example \ref{ex:minmax extension}. In detail, 
for any $P\in\P_f(\X)\cup\P_{f,1}^{\pm}(\X)$,
\begin{align*}
\signed{Q}_P^{\meas}(W)&:=\sup\{\signed{Q}_P(\M(W)):\,M\in\povm(\hil,[n]),\,n\in\bN\},\\
\signed{Q}_P^{\max}(W)&:=\inf\{\signed{Q}_P(w):\,(w,\rt)\text{ reverse test for }w\},
\ds\ds\ds\ds\ds\ds W\in\B(\X,\hil)\pne,
\end{align*}
give the \ki{measured} and the \ki{maximal} $P$-weighted R\'enyi 
$\signed{Q}_P$-divergences.
Clearly, $\signed{Q}_P^{\meas}$ is the smallest and 
$\signed{Q}_P^{\max}$ is the largest monotone quantum extension of the classical $Q_P$.

The measured and the maximal $P$-weighted R\'enyi divergences can be expressed as
\begin{align*}
D_P^{\meas}(W)&:=\sup\left\{D_P(\M(W)):\,M\in\povm(\hil,[n]),\,n\in\bN\right\}\\
&=
\frac{1}{\prod_{x\in\X}(1-P(x))}
\bz -\log s(P)\signed{Q}_P^{\meas}\bz\bz\frac{W_x}{\Tr W_x}\jz_{x\in\X}\jz\jz,\\
D_P^{\max}(W)&:=\inf\left\{D_P(w):\,(w,\rt)\text{ reverse test for }w\right\}\\
&=
\frac{1}{\prod_{x\in\X}(1-P(x))}
\bz -\log s(P)\signed{Q}_P^{\max}\bz\bz\frac{W_x}{\Tr W_x}\jz_{x\in\X}\jz\jz
\end{align*}
for any $P\in\P_f(\X)\cup\P_{f,1}^{\pm}(\X)\setminus\{\egy_{\{x\}}:\,x\in\X\}$
and
$W\in\B(\X,\hil)\pne$. It is clear that $D_P^{\meas}$ is superadditive and 
$D_P^{\max}$ is subadditive under tensor products, and hence
their regularized versions are given as
\begin{align*}
\oll{D}_P^{\meas}&:=
\sup_{n\in\bN}\frac{1}{n}D_P^{\meas}(W^{\potimes n})
=
\lim_{n\to+\infty}\frac{1}{n}D_P^{\meas}(W^{\potimes n}),\\
\oll{D}_P^{\max}&:=\inf_{n\in\bN}\frac{1}{n}D_P^{\max}(W^{\potimes n})
=
\lim_{n\to+\infty}\frac{1}{n}D_P^{\max}(W^{\potimes n}).
\end{align*}
(See Proposition \ref{prop:regmeas regmax as limits} for the equalities.)
Clearly, $D_P^{\meas}$ is the smallest and $D_P^{\max}$ is the largest quantum extension of 
the classical $D_P$ that is monotone under (completely) positive trace-preserving maps,
and if $\oll{D}^q_P(W):=\lim_n \frac{1}{n}D^q_P(W^{\potimes n})$ exists 
for some monotone quantum extension of 
$D_P$ and some $W\in\B(\X,\hil)\pne$ then 
\begin{align*}
\oll{D}_P^{\meas}(W)\le \oll{D}^q_P(W)\le\oll{D}^{\max}_P(W).
\end{align*}
\end{example}

\subsection{Weighted geometric means and induced divergences}
\label{sec:nc gm}

For a collection of functions $w\in\F(\X,\I)\p$ on a finite set $\I$, and a finitely supported probability distribution $P\in\P_f(\X)$, the $P$-weighted geometric mean 
of $w$ is naturally defined as
\begin{align}\label{classical geom mean}
G_P(w):=\bz\prod_{x\in\supp P}w_x(i)^{P(x)}\jz_{i\in\I}
=
\bz\exp\bz\sum_{x\in\X}P(x)\log w_x(i)\jz\jz_{i\in\I},
\end{align}
where in the second expression we use the convention \eqref{zero times infty}.
This can be extended in the obvious way to the case where $P\in\P_f^{\pm}(\X)$ is a signed probability measure, provided that $w_x(i)>0$ for all $i\in\I$ whenever $P(x)<0$. 
Extending this notion to non-commutative variables has been the subject of intensive research in matrix analysis; without completeness, we refer to 
\cite{AC2011,AndoLiMathias2004,BGJ2019,Bhatia_Holbrook2006,BJL2018,HiaiLim2020,KimLee2015,
Lawson_Lim2011,LawsonLim2014,Lim_Palfia2012,Moakher_matrixmean,Petz_Temesi2005,
PV_Hellinger} and references therein. 
Normally, the definition of a matrix geometric mean includes a number of desirable properties 
(e.g., monotonicity in the arguments or joint concavity, at least when all weights are positive), whereas for our purposes the following minimalistic definition is more suitable:

\begin{definition}\label{def:nc geom multi-mean}
For a non-empty set $\X$ and a finitely supported signed probability measure
$P\in\P_f^{\pm}(\X)$, a 
\ki{non-commutative $P$-weighted geometric mean} is a
function 
\begin{align*}
G_{P}^q:\,
\cD(G_{P}^q):=\cup_{d\in\bN}\left\{W\in\B(\X,\bC^d)\p:\,W_x^0=W_y^0,\,x,y\in\supp P\right\}
\to \cup_{d\in\bN}\B(\bC^d)\p,
\end{align*}
such that 
\begin{enumerate}
\item
$W\in\cD(G_{P}^q)\cap\B(\X,\bC^d)\p$ $\imp$
$G_P^q(W)\in\B(\bC^d)\p$;
\item
$W,\tilde W\in\cD(G_{P}^q)$, $W_x=\tilde W_x$, $x\in\supp P$ \ds$\imp$\ds
$G_P^q(W)=G_P^q(\tilde W)$;
\item
$G_{P}^q$
is covariant under isometries, i.e., if $V:\,\bC^{d_1}\to\bC^{d_2}$ is an isometry then 
\begin{align*}
G_{P}^q\bz VWV^*\jz=VG_{P}^q(W)V^*,\ds\ds\ds
W\in \cD(G_P^q)\cap\B(\X,\bC^{d_1})\p;
\end{align*}
\item
if $W\in\cD(G_{P}^q)$ is such that $W_x=\sum_{i=1}^d w_x(i)\pr{i}$, $x\in\X$,
are diagonal in the same ONB $(\ket{i})_{i=1}^d$ then 
\begin{align*}
G_{P}^q\bz W \jz=\sum_{i=1}^d G_P(w)(i)\pr{i},
\end{align*}
with $G_P(w)$ given in \eqref{classical geom mean}.
\end{enumerate}
If a non-commutative $P$-weighted geometric mean $G_P^q$
is \ki{regular} in the sense that 
\begin{align*}
G_P^q(W)=\lim_{\ep\searrow 0}G_P^q\bz(W_x+\ep W_P^0)_{x\in\X}\jz,\ds\ds\ds
W\in\cD(G_P^q),
\end{align*}
where 
\begin{align*}
W_P^0:=\bigvee_{x\in\supp P}W_x^0,
\end{align*}
then it is automatically extended to collections of PSD operators with possibly different supports as
\begin{align}\label{geommean regularization}
G_P^q(W):=\lim_{\ep\searrow 0}G_P^q\bz(W_x+\ep W_P^0)_{x\in\X}\jz,
\end{align}
for any $W\in\B(\X,\bC^d)$ such that the limit above exists;
the collection of all such $W$ is denoted by 
$\oll{\cD}(G_P^q)$. 
\end{definition}

\begin{rem}
Using the isometric invariance, 
a non-commutative $P$-weighted geometric mean
can be uniquely extended to any
$W\in\B(\X,\hil)\p$ such that 
$W_x^0=W_y^0$, $x,y\in\supp P$ (we denote the collection of all such 
$W$ by $\cD_{\hil}(G_P^q)$),
and then further extended using
\eqref{geommean regularization}, provided that $G_P^q$ is regular;
we denote this extension also by $G_{P}^q$, and  
the set of $W\in\B(\X,\hil)\p$ for which it is well-defined by 
$\oll{\cD}_{\hil}(G_P^q)$. 
\end{rem}

\begin{rem}\label{rem:geommean partiso inv}
A non-commutative geometric mean $G_P^q$ is also automatically covariant under 
partial isometries on $\oll{\cD}_{\hil}(G_P^q)$
in the sense that if 
$W\in\oll{\cD}_{\hil}(G_P^q)$ and 
$V:\,\hil\to\kil$ is a partial isometry such that 
$W_P^0\le V^*V$ then 
\begin{align*}
G_P^q\bz V\W V^*\jz=VG_P^q(\W)V^*.
\end{align*}
The proof in the case $W\in\cD_{\hil}(G_P^q)$
goes by a straightforward modification of the argument in Remark 
\ref{rem:divv partiso inv}, and the extension to the case
$W\in\oll{\cD}_{\hil}(G_P^q)$ is obvious.
\end{rem}

\begin{rem}
In the $2$-variable case, i.e., when $\X=\{0,1\}$, we identify $P$ 
with the number $\gamma:=P(0)$, and use the notation 
$G_{\gamma}^q(W_0\|W_1)$ instead of $G_{P}^q(W)$. 
\end{rem}

\begin{example}
For any $\gamma\in\bR$, $z\in(0,+\infty)$, and $\rho,\sigma\in\B(\hil)\pne$ with 
$\rho^0=\sigma^0$, let 
\begin{align*}
G_{\gamma,z}(\rho\|\sigma)
&:=
\bz\rho^{\frac{\gamma}{2z}}\sigma^{\frac{1-\gamma}{z}}\rho^{\frac{\gamma}{2z}}\jz^z,\\
\wtilde G_{\gamma,z}(\rho\|\sigma)
&:=
\bz\sigma^{\frac{1-\gamma}{2z}}\rho^{\frac{\gamma}{z}}\sigma^{\frac{1-\gamma}{2z}}\jz^z,\\
\what G_{\gamma,1}(\rho\|\sigma)
&:=\sigma\#_{\gamma}\rho:=\sigma^{1/2}\bz\sigma^{-1/2}\rho\sigma^{-1/2}\jz^{\gamma}\sigma^{1/2},\\
\what G_{\gamma,z}(\rho\|\sigma)
&:=
\bz\sigma^{\frac{1}{z}}\#_{\gamma}\rho^{\frac{1}{z}}\jz^z,\\
\what G_{\gamma,+\infty}(\rho\|\sigma)
&:=\lim_{z\to+\infty}\what G_{\gamma,z}(\rho,\sigma)
=
\rho^0 e^{\gamma \logn\rho+(1-\gamma)\logn\sigma}.
\end{align*}
For $\gamma\in(0,1)$, $\#_{\gamma}$ is the Kubo-Ando $\gamma$-weighted geometric mean 
\cite{KA} (see Section \ref{sec:gamma relentr} for a more detailed exposition), and 
the last equality is due to  \cite[Lemma 3.3]{HP_GT}.
It is easy to see that these all define $2$-variable non-commutative $\gamma$-weighted geometric means. The last one of the above quantities can be 
immediately extended to more than two variables as 
\begin{align}\label{multi log mean}
\what G_{P,+\infty}(W):=S e^{\sum_xP(x)\logn W_x},
\end{align}
where $S=W_x^0$, $x\in\supp P$.
\end{example}

For any non-commutative $P$-weighted geometric mean $G_{P}^q$, 
\begin{align}\label{multiQ from G}
Q_P^{q}(W):=Q_P^{G_{P}^q}(W):=\Tr G_{P}^q(W),\ds\ds\ds W\in\oll{\cD}_{\hil}(G_P^q)
\end{align}
defines an extension of the classical $Q_{P}$ to 
$\oll{\cD}_{\hil}(G_P^q)$ for any finite-dimensional Hilbert space $\hil$; 
see Definition \ref{def:classical QP} for the definition of the classical quantity. 
This may potentially be further extended by 
\begin{align}\label{multi geommean Renyi extension}
Q_P^{q}(W):=Q_P^{G_{P}^q}(W):=
\lim_{\ep\searrow 0}\Tr G_{P}^q((W_x+\ep W_P^0)_{x\in\X}),
\end{align}
provided that the limit exists; this is obviously the case when 
$W\in\oll{\cD}_{\hil}(G_P^q)$, but there might be other $W$ for which the 
limit in \eqref{multi geommean Renyi extension} exists, even though the limit in 
\eqref{geommean regularization} does not.
Among others, the R\'enyi $(\alpha,z)$-divergences, 
the log-Euclidean R\'enyi divergences and the maximal R\'enyi divergences 
can be given in this way for the following parameter ranges:
\begin{align}
&Q_{\alpha,z}(\rho\|\sigma)=Q_{\alpha}^{G_{\alpha,z}}(\rho\|\sigma)=Q_{\alpha}^{\tilde G_{\alpha,z}}(\rho\|\sigma),
& &\alpha\in(0,1)\cup(1,+\infty),\ds z\in(0,+\infty),\\
&Q_{\alpha,+\infty}(\rho\|\sigma)=
Q_{\alpha}^{\what G_{\alpha,+\infty}}(\rho\|\sigma),
& &\alpha\in(0,1)\cup(1,+\infty),\\
&Q_{\alpha}^{\max}(\rho\|\sigma)=
Q_{\alpha}^{\what G_{\alpha,1}}(\rho\|\sigma),& & 
\label{Q max}
\alpha\in(0,1)\cup(1,2],
\end{align}
for any $\rho,\sigma\in\B(\hil)\pne$; see Examples \ref{ex:az Renyi} and \ref{ex:max Renyi}, 
and \cite{MO-cqconv}.
The equality in \eqref{Q max} follows from \eqref{max Renyi explicit} and the fact that 
$\sigma\#_{\alpha}\rho=\sigma\#_{\alpha}\rho_{\sigma,\ac}$, $\alpha\in(0,1)$.

\begin{rem}
For every $\alpha\in(0,1)$, 
$[1,+\infty]\ni z\mapsto Q_{\alpha}^{\what G_{\alpha,z}}(\rho\|\sigma)$ interpolates monotone increasingly between $Q_{\alpha}^{\max}(\rho\|\sigma)$ and $Q_{\alpha,+\infty}(\rho\|\sigma)$, according to \cite[Corollary 2.4]{AndoHiai1994}, and hence
\begin{align*}
[1,+\infty]\ni z\mapsto 
D_{\alpha}^{\what G_{\alpha}^{\alpha,z}}(\rho\|\sigma):=\frac{1}{\alpha-1}\log Q_{\alpha}^{\what G_{\alpha,z}}(\rho\|\sigma)-
\frac{1}{\alpha-1}\log\Tr\rho
\end{align*}
interpolates monotone decreasingly between $D_{\alpha}^{\max}(\rho\|\sigma)$ and $D_{\alpha,+\infty}(\rho\|\sigma)$.
\end{rem}

Next, we explore some ways 
to define (potentially) new quantum R\'enyi divergences and relative entropies from 
combining some given quantum R\'enyi divergences/relative entropies
with some non-commutative weighted geometric means. 

The simplest way to do so is to take any $\Y$-variable 
$\tilde P$-weighted quantum R\'enyi quantity 
$Q_{\tilde P}^q$, for every $y\in\Y$, a
$P^{(y)}\in\P_f^{\pm}(\X)$ and 
a non-commutative $P^{(y)}$-weighted geometric mean
$G_{P^{(y)}}^{q_y}$,
and define
\begin{align*}
Q_{P}^{q,\qv}(W):=
Q_{\tilde P}^q\bz\bz G_{P^{(y)}}^{q_y}(W)\jz_{y\in\Y}\jz,\ds\ds\ds
W\in\bigcap_{y\in\supp \tilde P}\oll{\cD}(G_{P^{(y)}}^{q_y}),
\end{align*}
where 
\begin{align*}
P(x):=\sum_{y\in\supp \tilde P}\tilde P(y)P^{(y)}(x),\ds\ds\ds x\in\X.
\end{align*}
This can be further extended as in \eqref{multi geommean Renyi extension}, or as
\begin{align*}
Q_{P}^{q,\qv}(W):=
\lim_{\ep\searrow 0}Q_{\tilde P}^q\bz\bz G_{P^{(y)}}^{q_y}(W+\ep W_{P^{(y)}}^0)\jz_{y\in\Y}\jz,\ds\ds\ds
W\in\bigcup_{d\in\bN}\B(\X,\bC^d)\p,
\end{align*}
provided that the limit exists.

Such quantities have been defined before in \cite{FuruyaLashkariOuseph2023} for $\X=[n+1]$ as
\begin{align*}
Q_P^{\gamma,z}(W):=Q_{P(0),z}\bz W_0\Bigg\|\bz W_1^{\frac{P(1)}{z\kappa_1}}\#_{\gamma_1}
\bz W_2^{\frac{P(2)}{z\kappa_2}}\#_{\gamma_2}\ldots\bz W_{n-1}^{\frac{P(n-1)}{z\kappa_{n-1}}}\#_{\gamma_{n-1}}W_n^{\frac{P(n)}{z\kappa_n}}\jz\jz\jz^{\frac{z}{1-P(0)}}\jz,
\end{align*}
where $P\in\P([n+1])$
is a probability distribution, $\gamma_i\in(0,1)$, $i\in\{1,\ldots,n-1\}$ are arbitrary parameters,
$\gamma_n:=0$, and $\kappa_i:=\gamma_1\ccdots\gamma_{i-1}\cdot(1-\gamma_i)$, $i=1,\ldots,n$.
Here, $\Y=[2]$, $P^{(0)}=\egy_{\{0\}}$, $P^{(1)}(0)=0$,
$P^{(1)}(k)=P(k)/(1-P(0))$, $k=1,\ldots,n$.
For the range of parameter values for which these quantities are monotone under CPTP maps, see \cite{FuruyaLashkariOuseph2023}.

Similar quantities were considered in \cite{bunth2021equivariant} in the study of relative 
submajorization of (subnormalized) state families; namely 
\begin{align*}
D_{\alpha,\alpha}(W_0\|G_P^q(W_1,\ldots,W_{n+1})),
\end{align*}
where $D_{\alpha,\alpha}$ is the sandwiched R\'enyi $\alpha$-divergence, and 
$G_P^q$ is any non-commutative geometric mean satisfying a number of desirable properties.
(In fact, \cite{bunth2021equivariant} dealt with a more general setting where 
$P$ need not be finitely supported, and hence infinitely many $W_x$ may enter the expression in the second argument of $D_{\alpha,\alpha}$.)
A new family of ($2$-variable) quantum R\'enyi $\alpha$-divergences was also proposed in 
\cite{bunth2021equivariant} as
\begin{align*}
D^{*,\gamma}_{\alpha}(\rho\|\sigma):=
\frac{1}{1-\gamma}D_{\frac{\alpha-\gamma}{1-\gamma},\frac{\alpha-\gamma}{1-\gamma}}
\bz\rho\|\sigma\#_{\gamma}\rho\jz,
\end{align*}
where $\alpha>1$ and $\gamma\in(0,1)$. This concept can be immediately generalized as
\begin{align}\label{gamma-weighted Renyi}
D_{\alpha}^{\qo,G^{\qn}_{\gamma}}(\rho\|\sigma):=\frac{1}{1-\gamma}D^{\qo}_{\frac{\alpha-\gamma}{1-\gamma}}
\bz\rho\|G^{\qn}_{\gamma}(\rho\|\sigma)\jz
\end{align}
for any $\alpha\in(0,+\infty)$ and $\gamma\in(0,\min\{1,\alpha\})$, 
where $D_{\frac{\alpha-\gamma}{1-\gamma}}^{\qo}$ is a quantum R\'enyi 
$\frac{\alpha-\gamma}{1-\gamma}$-divergence, and 
$G^{\qn}_{\gamma}$ is an arbitrary non-commutative $\gamma$-weighted geometric mean.
The special case $\alpha=1$ gives that any quantum relative entropy $D^{\qo}$ and any 
non-commutative $\gamma$-weighted geometric mean $G^{\qn}_{\gamma}$ with some $\gamma\in(0,1)
$ defines a quantum relative entropy via
\begin{align}\label{gamma-weighted relentr}
D^{\qo,G^{\qn}_{\gamma}}(\rho\|\sigma):=\frac{1}{1-\gamma}D^{\qo}(\rho\|G^{\qn}_{\gamma}(\rho\|\sigma)).
\end{align}
In Section \ref{sec:gamma relentr} we will study in detail the relative entropies of the form 
\eqref{gamma-weighted relentr} when $G^{\qn}_{\gamma}=\#_{\gamma}$
is the Kubo-Ando $\gamma$-weighted geometric mean.

\begin{rem}
If $D_{\frac{\alpha-\gamma}{1-\gamma}}^{\qo}$ satisfies the scaling law \eqref{Renyi scaling}
then for any $\rho,\sigma\in\B(\hil)\pne$ such that 
$Q_{G^{\qn}_{\gamma}(\rho\|\sigma)}=\Tr G^{\qn}_{\gamma}(\rho\|\sigma)\ne 0$, 
\eqref{gamma-weighted Renyi} can be rewritten as 
\begin{align*}
D_{\alpha}^{\qo,G^{\qn}_{\gamma}}(\rho\|\sigma)
&=\frac{1}{1-\gamma}\left[D^{\qo}_{\frac{\alpha-\gamma}{1-\gamma}}
\bz\frac{\rho}{\Tr\rho}\Big\|\frac{G^{\qn}_{\gamma}(\rho\|\sigma)}{\Tr G^{\qn}_{\gamma}(\rho\|\sigma)}\jz
+\log \Tr\rho-\log \Tr G^{\qn}_{\gamma}(\rho\|\sigma)\right]\\
&=
\frac{1}{1-\gamma}D^{\qo}_{\frac{\alpha-\gamma}{1-\gamma}}
\bz\frac{\rho}{\Tr\rho}\Big\|\frac{G^{\qn}_{\gamma}(\rho\|\sigma)}{\Tr G^{\qn}_{\gamma}(\rho\|\sigma)}\jz+D_{\gamma}^{G^{\qn}_{\gamma}}(\rho\|\sigma).
\end{align*}
Likewise, if $D^{\qo}$ satisfies the scaling law \eqref{relentr scaling}, and 
$\Tr G^{\qn}_{\gamma}(\rho\|\sigma)\ne 0$, then \eqref{gamma-weighted relentr} 
can be rewritten as 
\begin{align}
D^{\qo,G^{\qn}_{\gamma}}(\rho\|\sigma)
&=
\frac{\Tr\rho}{1-\gamma}\left[
D^{\qo}\bz\frac{\rho}{\Tr\rho}\Big\|\frac{G^{\qn}_{\gamma}(\rho\|\sigma)}{\Tr G^{\qn}_{\gamma}(\rho\|\sigma)}\jz+\log\Tr\rho-\log\Tr G^{\qn}_{\gamma}(\rho\|\sigma)\right]\\
&=
(\Tr\rho)\left[
\frac{1}{1-\gamma}D^{\qo}\bz\frac{\rho}{\Tr\rho}\Big\|\frac{G^{\qn}_{\gamma}(\rho\|\sigma)}{\Tr G^{\qn}_{\gamma}(\rho\|\sigma)}\jz
+D_{\gamma}^{G^{\qn}_{\gamma}}(\rho\|\sigma)\right].
\label{gamma-weighted relentr2}
\end{align}
In particular, if $D_{\gamma}^{G^{\qn}_{\gamma}}$ is strictly positive and
$D^{\qo}_{\frac{\alpha-\gamma}{1-\gamma}}$
(resp.~$D^{\qo}$) is non-negative, then 
$D_{\alpha}^{\qo,G^{\qn}_{\gamma}}$ (resp.~$D^{\qo,G^{\qn}_{\gamma}}$)
is strictly positive.
\end{rem}
\medskip

Another way of obtaining multi-variate quantum R\'enyi divergences using non-commutative geometric means is by starting with quantum relative entropies $D^{\qv}:=(D^{q_x})_{x\in\X}$
and non-commutative $P$-weighted geometric means
$G_P^{\tilde\qv}:=(G_{P}^{\tilde q_x})_{x\in\X}$,
and defining
\begin{align}\label{geometric Renyi}
-\log Q_P^{\tilde\qv,\qv}(W)
&:=
\sum_{x\in\X}P(x)D^{q_x}\bz\frac{G_P^{\tilde q_x}(W)}{\Tr G_P^{\tilde q_x}(W)}\Bigg\|W_x\jz.
\end{align}
One can easily verify that for a classical $W$ this indeed reduces to the classical 
$-\log Q_P(W)$.
It was shown in \cite[Theorem 3.6]{MO} that for $D^{\qn}=D^{\qo}=\DU$ and
$G^{\qn}_{\alpha}=G^{\qo}_{\alpha}=\what G_{\alpha,+\infty}$, the quantum R\'enyi $\alpha$-divergence obtained by \eqref{geometric Renyi} 
is equal to $D_{\alpha,+\infty}$ in \eqref{D alpha infty}. We give an extension of this to the multi-variate case in Proposition \ref{prop:Um Hellinger}. 
We also show in Proposition \ref{prop:max Renyi as geometric} below that in the case 
$D^{\qn}=D^{\qo}=D^{\max}$ and
$G^{\qn}_{\alpha}=G^{\qo}_{\alpha}=\what G_{\alpha,1}$, the quantum R\'enyi $\alpha$-divergence obtained by \eqref{geometric Renyi} 
is equal to $D_{\alpha}^{\max}$.

Instead of writing normalized geometric means in the first arguments of the quantum 
relative entropies in \eqref{geometric Renyi}, one may optimize over arbitrary states, which leads to 
\begin{align}\label{barycentyric def0}
-\log Q_{P}^{\bary,\qv}(W)
&:=
\inf_{\omega\in\S(\hil)}\sum_xP(x) D^{q_x}\bz\omega\|W_x\jz.
\end{align}
(When the $W_x$ may have different supports, we take the infimum on a restricted set of states; see Section \ref{sec:defs} for the precise definition.)
This again gives a multivariate quantum R\'enyi divergence,
which we call the \ki{barycentric R\'enyi $\alpha$-divergence}
corresponding to $D^{\qv}$. Moreover, as we show in Section \ref{sec:defs},
\eqref{barycentyric def0} can be equivalently written as
\begin{align}\label{barycentyric def00}
Q_{P}^{\bary,\qv}(W)
&:=
\sup_{\tau\in\B(\hil)\pne}\left\{\Tr\tau-\sum_xP(x) D^{q_x}\bz\tau\|W_x\jz\right\},
\end{align}
(the precise form again given in Section \ref{sec:defs}), and 
any optimal $\tau$ (if exists) gives a multi-variate $P$-weighted geometric mean 
$G_P^{D^{\qv}}(W)$
of $W$, for which 
\begin{align}
Q_{P}^{\bary,\qv}(W)=\Tr G_P^{D^{\qv}}(W),
\end{align}
(see Section \ref{sec:baryR is quantumR}), whence this approach can be formally seen as a special case of \eqref{multiQ from G}.
However, the initial data in the two approaches are different:
a non-commutative $P$-weighted geometric mean in \eqref{multiQ from G},
and a collection of relative entropies in 
\eqref{barycentyric def0}. Note also that while every 
non-commutative $P$-weighted geometric mean $G_P^q$ defines a quantum extension of $Q_P$ by
\eqref{multiQ from G}, it is not clear whether every quantum R\'enyi quantity
$Q_P^q$ defines a
non-commutative $P$-weighted geometric mean $G_P^q$ for which 
\eqref{multiQ from G} holds, and hence this is a non-trivial feature of the barycentric 
R\'enyi divergences.

The main goal of this paper is the detailed study of the barycentric R\'enyi divergences
(in Sections \ref{sec:barycentric} and \ref{sec:ex}.)
We note that these give new quantum R\'enyi divergences even in the $2$-variable case;
indeed, we show in Section \ref{sec:ex}
that under very mild conditions on $D^{\qn},D^{\qo}$ and 
$\rho,\sigma$,
$D_{\alpha,+\infty}(\rho\|\sigma)<D_{\alpha}^{\bary,\qv}(\rho\|\sigma)<D_{\alpha}^{\max}(\rho\|\sigma)$ holds for every $\alpha\in(0,1)$. 

Other approaches to non-commutative geometric means and multi-variate R\'enyi divergences
are outlined in Figure \ref{fig:multiRenyi}.

\section{$\gamma$-weighted geometric relative entropies}
\label{sec:gamma relentr}

Recall the definition of the Kubo-Ando weighted geometric means $\#_{\gamma}$ from Appendix \ref{sec:KA}.
The following is a special case of \eqref{gamma-weighted relentr}:

\begin{definition}
Let $D^q$ be a quantum relative entropy. For 
every $\gamma\in[0,1)$, and 
every $\rho,\sigma\in\B(\hil)\pne$, let 
\begin{align*}
D^{q,\#_\gamma}(\rho\|\sigma):=
\begin{cases}
\frac{1}{1-\gamma}D^q(\rho\|\sigma\#_{\gamma}\rho),&\gamma\in(0,1),\\
D^q(\rho\|\sigma),&\gamma=0.
\end{cases}
\end{align*}
We call $D^{q,\#_\gamma}$ the 
\ki{$\gamma$-weighted geometric relative entropy derived from $D^q$}, or
\ki{$\gamma$-weighted geometric $D^q$} for short.
\end{definition}

\begin{rem}
Note that by \eqref{wmean support}, 
$\sigma\#_{\gamma}\rho=0$ might happen even if both $\rho$ and $\sigma$
are quantum states (thus, in particular, are non-zero); in this case the value 
of $D^{q,\#_\gamma}(\rho\|\sigma)$ is $+\infty$, according to 
Remark \ref{rem:zero argument q}.
\end{rem}

\begin{rem}\label{rem:weighted relentr sum expression}
Assume that $D^q$ satisfies the scaling property \eqref{relentr scaling}.
Then for any $\rho,\sigma\in\B(\hil)\pne$ such that 
$\rho^0\wedge\sigma^0\ne 0$, 
\begin{align*}
D^{q,\#_{\gamma}}(\rho\|\sigma)=
(\Tr\rho)\left[
\frac{1}{1-\gamma}D^q\bz\frac{\rho}{\Tr\rho}\Big\|\frac{\sigma\#_{\gamma}\rho}{\Tr \sigma\#_{\gamma}\rho}\jz
+D_{\gamma}^{\max}(\rho\|\sigma)\right],\ds\ds\ds\gamma\in(0,1).
\end{align*}
This is a special case of \eqref{gamma-weighted relentr2}.
\end{rem}

It is clear that for any 
$\gamma\in(0,1)$, $D^{q,\#_\gamma}$ is a quantum relative entropy.
In the following we show how certain properties of $D^q$ are inherited by $D^{q,\#_\gamma}$.
\begin{enumerate}
\item
If $D_1^q$ is additive on tensor products then so is $D_1^{q,\#_\gamma}$, due to 
\eqref{mean tensor}.

\item
If $D^q$ is block subadditive/superadditive/additive, then so are $D^{q,\#_\gamma}$, 
$\gamma\in(0,1)$.
This follows immediately from the fact that for any sequence of projections 
$P_1,\ldots,P_r$ summing to $I$, we have
$\bz\sum_{i=1}^r P_i\rho P_i\jz\#_{\gamma}\bz\sum_{i=1}^r P_i\sigma P_i\jz
=
\sum_{i=1}^r(P_i\rho P_i)\#_{\gamma}(P_i\sigma P_i)$.

\item\label{gamma weighted support}
If $D^q$ satisfies the support condition
\begin{align*}
D^q(\rho\|\sigma)<+\infty\ds\iff\ds\rho^0\le\sigma^0
\end{align*}
then so do $D^{q,\#_{\gamma}}$, $\gamma\in(0,1)$. Indeed, 
by \eqref{wmean support}, $(\sigma\#_{\gamma}\rho)^0=\rho^0\wedge\sigma^0$, and clearly, 
$\rho^0\le\rho^0\wedge\sigma^0$ $\iff$ $\rho^0\le\sigma^0$, whence
$D^{q,\#_\gamma}(\rho\|\sigma)<+\infty$ $\iff$ $\rho^0\le\sigma^0$. 

\item
If $D^q$ satisfies 
either of the scaling laws \eqref{scaling1} or \eqref{scaling2} then 
so do $D^{q,\#_{\gamma}}$, $\gamma\in(0,1)$. As a consequence, if 
$D^q$ satisfies 
the scaling law \eqref{relentr scaling}
then so do $D^{q,\#_{\gamma}}$, $\gamma\in(0,1)$. 
These can be verified by straightforward computations, which we omit.

\item\label{weighted trmon imp strict pos}
If $D^q$ satisfies trace monotonicity \eqref{relentr tracemon} then 
$D^{q,\#_{\gamma}}$, $\gamma\in(0,1)$, also satisfy trace monotonicity \eqref{relentr tracemon}; moreover, they are strictly positive. Indeed,
if $\sigma\#_{\gamma}\rho\ne 0$ then 
\begin{align}
D^{q,\#_\gamma}(\rho\|\sigma)
&=
\frac{1}{1-\gamma}D^q(\rho\|\sigma\#_{\gamma}\rho)
\label{weighted strict pos4}\\
&\ge
\frac{1}{1-\gamma}\big[(\Tr\rho)\log\Tr\rho-(\Tr\rho)\log
\underbrace{\Tr(\sigma\#_{\gamma}\rho)}_{\le(\Tr\rho)^{\gamma}(\Tr\sigma)^{1-\gamma}}
\big]
\label{weighted strict pos3}\\
&\ge
\frac{1}{1-\gamma}[(\Tr\rho)\log\Tr\rho-(\Tr\rho)\log\bz(\Tr\rho)^{\gamma}(\Tr\sigma)^{1-\gamma}\jz]\label{weighted strict pos2}\\
&=(\Tr\rho)\log\Tr\rho-(\Tr\rho)\log\Tr\sigma,\label{weighted strict pos}
\end{align}
where the first inequality follows from 
the assumed trace monotonicity of $D^q$, and the second inequality follows from
Lemma \ref{lemma:KA Holder}. 
If $\sigma\#_{\gamma}\rho=0$ then
either $\rho=0$ and thus 
$D^{q,\#_{\gamma}}(\rho\|\sigma)=
\frac{1}{1-\gamma}D^q(0\|0)=0=D(\Tr\rho\|\Tr\sigma)$, or 
$\rho\ne 0$, whence
$D^{q,\#_\gamma}(\rho\|\sigma)=
\frac{1}{1-\gamma}
D^q(\rho\|0)=+\infty\ge D(\Tr\rho\|\Tr\sigma)$ holds trivially. 
This shows that $D^{q,\#_{\gamma}}$, $\gamma\in(0,1)$, satisfy trace monotonicity, and hence they are also non-negative, according to Remark \ref{rem:tracemon}.

Assume now that $\Tr\rho=\Tr\sigma=1$ and $\rho\ne\sigma$.
If $\rho^0\wedge\sigma^0=0$ then 
$D^{q,\#_\gamma}(\rho\|\sigma)=
\frac{1}{1-\gamma}
D^q(\rho\|0)=+\infty>0$ holds trivially. Otherwise 
\eqref{weighted strict pos4}--\eqref{weighted strict pos} hold, 
and the inequality in 
\eqref{weighted strict pos2} is strict, according to Lemma \ref{lemma:KA Holder}, 
whence $D^{q,\#_{\gamma}}(\rho\|\sigma)>0$ for every $\gamma\in(0,1)$. This proves that 
$D^{q,\#_{\gamma}}$, $\gamma\in(0,1)$, are strictly positive. 

\item\label{scaling 2 imp strict pos}
If $D^q$ satisfies the scaling law \eqref{scaling2} and it is non-negative then 
$D^{q,\#_{\gamma}}$, $\gamma\in(0,1)$, are strictly positive. Indeed, let 
$\rho,\sigma\in\S(\hil)$ be unequal states. 
If $\sigma\#_{\gamma}\rho=0$ then 
$D^{q,\#_{\gamma}}(\rho\|\sigma)
=\frac{1}{1-\gamma}D^q(\rho\|0)=+\infty>0$ is obvious. Assume for the rest that 
 $\sigma\#_{\gamma}\rho\ne 0$. Then
\begin{align*}
(1-\gamma)D^{q,\#_{\gamma}}(\rho\|\sigma)
=D^q\bz\rho\Big\|\frac{\sigma\#_{\gamma}\rho}{\Tr \sigma\#_{\gamma}\rho}
\Tr \sigma\#_{\gamma}\rho\jz
=
\underbrace{D^q\bz\rho\Big\|\frac{\sigma\#_{\gamma}\rho}{\Tr \sigma\#_{\gamma}\rho}\jz}_{\ge 0}
-(\Tr\rho)\log\underbrace{\Tr\sigma\#_{\gamma}\rho}_{<1}
>0,
\end{align*} 
where the strict inequality is due to Lemma \ref{lemma:KA Holder}.

\item\label{item:strong regularity}
If $D^q$ is strongly regular then so are $D^{q,\#_\gamma}$, $\gamma\in(0,1)$. This follows immediately from the fact that if $\sigma_n$, $n\in\bN$, is a sequence of PSD operators converging monotone decreasingly to $\sigma$ then $\sigma_n\#_{\gamma}\rho$ converges monotone decreasingly to 
$\sigma\#_{\gamma}\rho$; see \eqref{geom moncont}.

\item\label{item:anti-monotone}
If $D^q$ is anti-monotone in its second argument, then so are 
$D^{q,\#_\gamma}$, $\gamma\in(0,1)$. Indeed, this follows immediately from \eqref{wmean montone}.

\item\label{item:CPTP-mon}
Assume that $D^q$ is anti-monotone in its second argument. 
If $\rho,\sigma\in\B(\hil)\p$ and $\map:\,\B(\hil)\to\B(\kil)$ is a positive linear map 
such that 
\begin{align*}
D^q(\map(\rho)\|\map(\sigma))\le D^q(\rho\|\sigma)\ds\ds\text{then}\ds\ds
D^{q,\#_{\gamma}}(\map(\rho)\|\map(\sigma))\le D^{q,\#_{\gamma}}(\rho\|\sigma),\ds
\gamma\in(0,1).
\end{align*}
Indeed,
\begin{align*}
D^{q,\#_\gamma}(\map(\rho)\|\map(\sigma))
=
\frac{1}{1-\gamma}D^q(\map(\rho)\|\underbrace{\map(\sigma)\#_{\gamma}\map(\rho)}_{\ge
\map(\sigma\#_{\gamma}\rho)})
&\le
\frac{1}{1-\gamma}D^q(\map(\rho)\|\map(\sigma\#_{\gamma}\rho))\\
&\le
\frac{1}{1-\gamma}D^q(\rho\|\sigma\#_{\gamma}\rho),
\end{align*}
where the first inequality follows from Lemma \ref{lemma:geommean monotone}, the second 
inequality from anti-monotonicity of $D^q$ in its second argument, and the last inequality
from the monotonicity of $D^q$ under the given class of maps.
In particular, $D^{q,\#_{\gamma}}$, $\gamma\in(0,1)$, are monotone under any class of positive maps under which $D^q$ is monotone. 

\item
Since the $\gamma$-weighted geometric means are jointly concave in their arguments
\cite{KA}, it is easy to see that 
if $D^q$ is anti-monotone in its second argument and it is jointly convex, then so 
are $D^{q,\#_\gamma}$, $\gamma\in(0,1)$. 

\item\label{weighted ordering}
For any two quantum relative entropies $D^{q_1},D^{q_2}$,
\begin{align}
D^{q_1}\le D^{q_2}\ds\ds\imp\ds\ds D^{q_1,\#_{\gamma}}\le D^{q_2,\#_{\gamma}},\ds\ds\ds
\gamma\in(0,1),\label{geom ordering1}
\end{align}
which follows immediately by definition. Moreover, 
\begin{align}
&\forall\,\rho,\sigma\in\B(\hil)\pp,\,\rho\sigma\ne\sigma\rho:\,
D^{q_1}(\rho\|\sigma)< D^{q_2}(\rho\|\sigma)\nn\\
&\ds\ds\imp\ds\ds 
\forall\,\rho,\sigma\in\B(\hil)\pp,\,\rho\sigma\ne\sigma\rho:\,
D^{q_1,\#_{\gamma}}(\rho\|\sigma)< D^{q_2,\#_{\gamma}}(\rho\|\sigma),\ds\ds\ds
\gamma\in(0,1).
\label{geom ordering2}
\end{align}
Indeed, for this we only need that for any $\rho,\sigma\in\B(\hil)\pp$,
$\rho\sigma\ne\sigma\rho$ $\imp$
$\rho(\sigma\#_{\gamma}\rho)\ne(\sigma\#_{\gamma}\rho)\rho$.
This is straightforward to 
verify; indeed, by \eqref{gamma mean3},
\begin{align*}
\rho(\sigma\#_{\gamma}\rho)=(\sigma\#_{\gamma}\rho)\rho
&\ds\iff\ds
\rho \rho^{1/2}\bz \rho^{-1/2}\sigma\rho^{-1/2}\jz^{1-\gamma}\rho^{1/2}
=\rho^{1/2}\bz \rho^{-1/2}\sigma\rho^{-1/2}\jz^{1-\gamma}\rho^{1/2}
\rho\\
&\ds\iff\ds
\rho\bz \rho^{-1/2}\sigma\rho^{-1/2}\jz^{1-\gamma}
=
\bz \rho^{-1/2}\sigma\rho^{-1/2}\jz^{1-\gamma}\rho\\
&\ds\iff\ds
\rho\bz \rho^{-1/2}\sigma\rho^{-1/2}\jz=\bz \rho^{-1/2}\sigma\rho^{-1/2}\jz\rho\\
&\ds\iff\ds
\rho\sigma=\sigma\rho.
\end{align*}
\end{enumerate}

\begin{rem}
By Remark \ref{rem:weighted relentr sum expression} and the strict positivity of $D_{\gamma}^{\max}$, if 
$D^q$ satisfies the scaling law \eqref{relentr scaling} and it is non-negative then 
$D^{q,\#_{\gamma}}$ is strictly positive. Note that 
\ref{weighted trmon imp strict pos} above establishes strict positivity of 
$D^{q,\#_{\gamma}}$ under a slightly weaker condition;  
see Remark \ref{rem:tracemon}.
\end{rem}

\begin{rem}
According to Remark \ref{rem:relentr maxmin}, if 
$D^{q,\#_{\gamma}}$ is monotone under 
CPTP maps for some $\gamma\in(0,1)$
 then 
\begin{align*}
D^{\meas}\le D^{q,\#_{\gamma}}\le D^{\max},
\end{align*}
and if $D^{q,\#_{\gamma}}$ is also additive on tensor products then 
\begin{align}\label{Dgamma bounds}
\DU\le D^{q,\#_{\gamma}}\le D^{\max}.
\end{align}
In particular,
\begin{align}\label{DU gamma larger}
\DU\le D^{\Um,\#_{\gamma}}\le D^{\max,\#_{\gamma}}\le D^{\max},
\end{align}
where we also used \ref{weighted ordering} above.
In particular, \eqref{Dgamma bounds}--\eqref{DU gamma larger} hold for every $\gamma\in(0,1)$ when 
$D^q$ is monotone under CPTP maps and anti-monotone in its second argument, according to 
\ref{item:CPTP-mon}.
\end{rem}

\begin{lemma}\label{lemma:Dq gamma as limit}
Assume that $D^q$ is anti-monotone in its second argument and regular. 
Let $\rho,\sigma\in\B(\hil)\p$ be the limits of monotone decreasing sequences
$(\rho_{\ep})_{\ep>0}$ and $(\sigma_{\ep})_{\ep>0}$, respectively. Then, for
every $\gamma\in(0,1)$, 
\begin{align}
\frac{1}{1-\gamma}D^q(\rho\|\sigma_{\ep}\#_{\gamma}\rho_{\ep})\nearrow 
D^{q,\#_\gamma}(\rho\|\sigma)\ds\ds\text{as}\ds\ds
\ep\searrow 0,\label{Dq gamma as limit}\\
D^{q,\#_\gamma}(\rho\|\sigma_{\ep})\nearrow D^{q,\#_\gamma}(\rho\|\sigma)
\ds\ds\text{as}\ds\ds
\ep\searrow 0.\label{Dq gamma as limit2}
\end{align}
\end{lemma}
\begin{proof}
By the assumptions and Remark \ref{rem:AM regular}, $D^q$ is strongly regular. 
By \eqref{geom moncont}, both 
$\sigma_{\ep}\#_{\gamma}\rho_{\ep}$ and 
$\sigma_{\ep}\#_{\gamma}\rho$ converge monotone decreasingly to 
$\sigma\#_{\gamma}\rho$ as $\ep\searrow 0$.
From these, the assertions follow immediately.
\end{proof}

\begin{rem}
\eqref{Dq gamma as limit2} is a special case of the 
anti-monotonicity and
strong regularity stated in
\ref{item:anti-monotone} and \ref{item:strong regularity} above.
\end{rem}

\begin{cor}\label{cor:joint lsc}
If $D^q$ is anti-monotone in its second argument, regular, and jointly lower semi-continuous 
in its arguments, then $D^{q,\#_\gamma}$ has the same properties for every $\gamma\in(0,1)$.
\end{cor}
\begin{proof}
Anti-monotonicity and (strong) regularity have been covered in 
\ref{item:anti-monotone} and \ref{item:strong regularity} above
(see also Remark \ref{rem:AM regular}).
For every $\ep>0$, 
$(\rho,\sigma)\mapsto (\sigma+\ep I)\#_{\gamma}(\rho+\ep I)$ is continuous
(due to the continuity of functional calculus), and thus
$(\rho,\sigma)\mapsto \frac{1}{1-\gamma}D^q(\rho\|(\sigma+\ep I)\#_{\gamma}(\rho+\ep I))$
is lower semi-continuous. Thus, by \eqref{Dq gamma as limit}, 
$(\rho,\sigma)\mapsto D^{q,\#_\gamma}(\rho\|\sigma)$ is the supremum of lower semi-continuous functions, and therefore is itself lower semi-continuous.
\end{proof}

\begin{example}\label{ex:Dmax gamma sregular}
Combining Remark \ref{rem:relentr lsc} and Lemma \ref{lemma:Dq gamma as limit} yields that 
\eqref{Dq gamma as limit}--\eqref{Dq gamma as limit2}
hold for $D^q=D^{\meas},D^{\Um},D^{\max}$, and by 
Corollary \ref{cor:joint lsc}, 
$D^{\meas,\#_{\gamma}}$, 
$D^{\Um,\#_{\gamma}}$, and
$D^{\max,\#_{\gamma}}$, are jointly lower semi-continuous in their arguments for every 
$\gamma\in(0,1)$.
\end{example}
\medskip

For a given relative entropy $D^q$, the $\gamma$-weighted relative entropies $D^{q,\#_{\gamma}}$, $\gamma\in(0,1)$, give (potentially) new quantum 
relative entropies. On the other hand, as Proposition \ref{prop:iteration} below shows, 
iterating this procedure does not give further quantum relative entropies.
We will need the following simple lemma for the proof.

\begin{lemma}\label{lemma:gamma mean comp}
For any $\rho,\sigma\in\B(\hil)\p$ and any $\gamma_1,\gamma_2\in(0,1)$,
\begin{align}\label{gamma mean comp}
\bz\sigma\#_{\gamma_2}\rho\jz\#_{\gamma_1}\rho
&=
\sigma\#_{1-(1-\gamma_1)(1-\gamma_2)}\rho.
\end{align}
\end{lemma}
\begin{proof}
First, assume that $\rho,\sigma$ are positive definite. Then the statement follows as
\begin{align}
\bz\sigma\#_{\gamma_2}\rho\jz\#_{\gamma_1}\rho
&=
\rho^{1/2}\bz \rho^{-1/2}\bz\sigma\#_{\gamma_2}\rho\jz\rho^{-1/2}\jz^{1-\gamma_1}\rho^{1/2}
\nn\\
&=
\rho^{1/2}\bz \rho^{-1/2}\bz
\rho^{1/2}\bz \rho^{-1/2}\sigma\rho^{-1/2}\jz^{1-\gamma_2}\rho^{1/2}\jz
\rho^{-1/2}\jz^{1-\gamma_1}\rho^{1/2}\nn\\
&=
\rho^{1/2}\bz \rho^{-1/2}\sigma\rho^{-1/2}\jz^{(1-\gamma_1)(1-\gamma_2)}\rho^{1/2}\nn\\
&=
\sigma\#_{1-(1-\gamma_1)(1-\gamma_2)}\rho.
\label{gamma mean comp proof1}
\end{align}
Next, we consider the general case. 
Since $(\sigma+\ep I)\#_{\gamma_2}(\rho+\ep I)$
converges to $\sigma\#_2\rho$ monotone decreasingly, the continuity of $\#_{\gamma_1}$
under monotone decreasing sequences (see \eqref{geom moncont}) yields
\begin{align*}
\bz\sigma\#_{\gamma_2}\rho\jz\#_{\gamma_1}\rho
&=\lim_{\ep\searrow 0}\bz(\sigma+\ep I)\#_{\gamma_2}(\rho+\ep I)\jz\#_{\gamma_1}(\rho+\ep I)\\
&=\lim_{\ep\searrow 0}(\sigma+\ep I)\#_{1-(1-\gamma_1)(1-\gamma_2)}(\rho+\ep I)\\
&=
\sigma\#_{1-(1-\gamma_1)(1-\gamma_2)}\rho,
\end{align*}
where the second equality is by \eqref{gamma mean comp proof1}, and the last equality is by definition.
\end{proof}

\begin{prop}\label{prop:iteration}
For any quantum relative entropy $D^q$, and any $\gamma_1,\gamma_2\in(0,1)$, 
\begin{align*}
D^{(q,\#_{\gamma_1}),\#_{\gamma_2}}
&=
D^{q,\#_{1-(1-\gamma_1)(1-\gamma_2)}}.
\end{align*}
\end{prop}
\begin{proof}
For any $\rho,\sigma\in\B(\hil)\p$,
\begin{align*}
D^{(q,\#_{\gamma_1}),\#_{\gamma_2}}(\rho\|\sigma)
&=
\frac{1}{1-\gamma_2}D^{q,\#_{\gamma_1}}(\rho\|\sigma\#_{\gamma_2}\rho)\\
&=
\frac{1}{1-\gamma_2}\cdot\frac{1}{1-\gamma_1}
D^{q}(\rho\|(\sigma\#_{\gamma_2}\rho)\#_{\gamma_1}\rho)\\
&=
\frac{1}{1-\gamma_2}\cdot\frac{1}{1-\gamma_1}
D^{q}(\rho\|(\sigma\#_{1-(1-\gamma_1)(1-\gamma_2)}\rho)\\
&=
D^{q,\#_{1-(1-\gamma_1)(1-\gamma_2)}}(\rho\|\sigma),
\end{align*}
where all equalities apart from the third one are by definition, and the third equality 
is due to Lemma \ref{lemma:gamma mean comp}.
\end{proof}
\medskip

Next, we study the $\gamma$-weighted geometric relative entropies corresponding to the 
largest and the smallest additive and CPTP-monotone quantum relative entropies, 
i.e., to $D^{\max}$ and $\DU$. The first case turns out to be trivial, while 
in the second case we get that the $\gamma$-weighted geometric Umegaki relative entropies essentially 
interpolate between $\DU$ and $D^{\max}$ in a monotone increasing way.

\begin{prop}\label{prop:Dmax stable}
For any $\gamma\in(0,1)$, 
\begin{align*}
D^{\max,\#_{\gamma}}=D^{\max}.
\end{align*}
\end{prop}
\begin{proof}
We need to show that for any $\rho,\sigma\in\B(\hil)\pne$, 
\begin{align*}
\frac{1}{1-\gamma}D^{\max}(\rho\|\sigma\#_{\gamma}\rho)=
D^{\max}(\rho\|\sigma).
\end{align*}
If $\rho^0\nleq\sigma^0$ then both sides are $+\infty$ (see \ref{gamma weighted support} above), and hence for the rest we assume that $\rho^0\le\sigma^0$.

Assume first that $\rho$ and $\sigma$ are invertible. Then
\begin{align}
\frac{1}{1-\gamma}D^{\max}(\rho\|\sigma\#_{\gamma}\rho)
&=
\frac{1}{1-\gamma}\Tr\rho\log\bz\rho^{1/2}(\sigma\#_{\gamma}\rho)\inv\rho^{1/2}\jz\nn\\
&=
\frac{1}{1-\gamma}\Tr\rho\log\bz\rho^{1/2}
(\rho^{-1/2}(\rho^{1/2}\sigma\inv\rho^{1/2})^{1-\gamma}\rho^{-1/2})\rho^{1/2}\jz\nn\\
&=
\frac{1}{1-\gamma}\Tr\rho\log(\rho^{1/2}\sigma\inv\rho^{1/2})^{1-\gamma}\nn\\
&=
\Tr\rho\log(\rho^{1/2}\sigma\inv\rho^{1/2})
=
D^{\max}(\rho\|\sigma),
\label{Dmax gamma proof1}
\end{align}
where the first equality is due to \eqref{BS def2}, 
the second equality follows from \eqref{gamma mean3}, and the rest are obvious.
If we only assume that $\rho^0\le\sigma^0$ then we have
\begin{align*}
D^{\max,\#_{\gamma}}(\rho\|\sigma)
&=
\lim_{\ep\searrow 0}
\frac{1}{1-\gamma}D^{\max}(\rho\|(\sigma+\ep I)\#_{\gamma}(\rho+\ep I))\\
&=
\lim_{\ep\searrow 0}\Tr\rho\log\bz(\rho+\ep I)^{1/2}(\sigma+\ep I)\inv(\rho+\ep I)^{1/2}\jz\\
&=
\Tr\rho\log(\rho^{1/2}\sigma\inv\rho^{1/2})
=
D^{\max}(\rho\|\sigma),
\end{align*}
where the first equality was stated in Example \ref{ex:Dmax gamma sregular}, the second equality follows
as in
\eqref{Dmax gamma proof1}, the third equality is easy to verify, and the last equality 
follows by \eqref{BS def2}.
\end{proof}

\begin{lemma}\label{lemma:monotone in gamma}
Assume that $D^q$ is a quantum relative entropy such that 
$D^q\le D^{q,\#_{\gamma}}$ for every $\gamma\in(0,1)$. Then 
$(0,1)\ni\gamma\mapsto D^{q,\#_{\gamma}}$ is monotone increasing. 
\end{lemma}
\begin{proof}
Let $0<\gamma_1\le\gamma_2<1$, and let $\gamma:=\frac{\gamma_2-\gamma_1}{1-\gamma_1}$,
so that $1-\gamma=\frac{1-\gamma_2}{1-\gamma_1}$. For any $\rho,\sigma\in\B(\hil)\p$, 
\begin{align*}
D^{q,\#_{\gamma_1}}(\rho\|\sigma)
&=
\frac{1}{1-\gamma_1}D^q(\rho\|\sigma\#_{\gamma_1}\rho)\\
&\le
\frac{1}{1-\gamma_1}D^{q,\#_{\gamma}}(\rho\|\sigma\#_{\gamma_1}\rho)\\
&=
\frac{1}{1-\gamma_1}\frac{1}{1-\gamma}D^q(\rho\|(\sigma\#_{\gamma_1}\rho)\#_{\gamma}\rho)\\
&=
\frac{1}{1-\gamma_1}\frac{1}{1-\gamma}D^q(\rho\|(\sigma\#_{\gamma_2}\rho)\\
&=
\frac{1}{1-\gamma_2}D^q(\rho\|(\sigma\#_{\gamma_2}\rho)\\
&=
D^{q,\#_{\gamma_2}}(\rho\|\sigma),
\end{align*}
where the first, the second, and the last equalities are by definition, 
the inequality is by assumption, 
while the third 
and the fourth equalities follow from Lemma \ref{lemma:gamma mean comp} and the definition of $\gamma$.
\end{proof}

\begin{prop}\label{prop:DU Dmax interpolation}
Let $\rho,\sigma\in\B(\hil)\pne$. Then 
\begin{align}\label{DU gamma mon}
(0,1)\ni\gamma\mapsto D^{\Um,\#_{\gamma}}(\rho\|\sigma)\ds\ds\text{is monotone increasing,}
\end{align}
and if $\rho,\sigma$ are positive definite then 
\begin{align}
&D^{\Um,\#_{\gamma}}(\rho\|\sigma)\searrow \DU(\rho\|\sigma),
\ds\ds\ds\text{as}\ds \gamma\searrow 0,\label{DU gamma limit0}\\
&D^{\Um,\#_{\gamma}}(\rho\|\sigma)\nearrow D^{\max}(\rho\|\sigma),
\ds\ds\ds\text{as}\ds \gamma\nearrow 1.\label{DU gamma limit1}
\end{align}
\end{prop}
\begin{proof}
Since $\DU$ is anti-monotone in its second argument and monotone under CPTP maps, so are $D^{\Um,\#_{\gamma}}$, $\gamma\in(0,1)$, according to \ref{item:CPTP-mon} above, and hence, by \eqref{DU gamma larger}, 
$\DU\le D^{\Um,\#_{\gamma}}$, $\gamma\in(0,1)$. By Lemma \ref{lemma:monotone in gamma}, 
$(0,1)\ni \gamma\mapsto D^{\Um,\#_{\gamma}}$ is monotone increasing. 

Let $\rho,\sigma\in\B(\hil)\pp$. Then \eqref{DU gamma limit0} follows simply by the continuity of functional calculus, so we only have to prove \eqref{DU gamma limit1}.
Note that the definition of $\sigma\#_{\gamma}\rho$ in \eqref{gamma mean1},
and hence also the definition of $\DU(\rho\|\sigma\#_{\gamma}\rho)$
make sense for any $\gamma\in\bR$, and both are (infinitely many times) differentiable 
functions of $\gamma$. Moreover, 
$\DU(\rho\|\sigma\#_{1}\rho)=\DU(\rho\|\rho)=0$.
Hence,
\begin{align*}
\lim_{\gamma\nearrow 1}D^{\Um,\#_{\gamma}}(\rho\|\sigma)
&=
-\lim_{\gamma\nearrow 1}\frac{D^{\Um}(\rho\|\sigma\#_{\gamma}\rho)-D^{\Um}(\rho\|\sigma\#_1\rho)}{\gamma-1}\\
&=
-\frac{d}{d\gamma}\DU(\rho\|\sigma\#_{\gamma}\rho)\Big\vert_{\gamma=1}
=
-\frac{d}{d\gamma}\bz\Tr\rho\log\rho-\Tr\rho\log(\sigma\#_{\gamma}\rho)\jz\Big\vert_{\gamma=1}
\\
&=
\Tr\rho\, (D\log)[\sigma\#_{1}\rho]\bz\frac{d}{d\gamma}(\sigma\#_{\gamma}\rho)\Big\vert_{\gamma=1}\jz\\
&=
\sum_{i,j}\log^{[1]}(\lambda_i,\lambda_j)\Tr \rho P_i
\underbrace{\sigma^{1/2}(\sigma^{-1/2}\rho\sigma^{-1/2})}_{
=\rho\sigma^{-1/2}}\log(\sigma^{-1/2}\rho\sigma^{-1/2})\sigma^{1/2}P_j\\
&=
\sum_{i,j}\log^{[1]}(\lambda_i,\lambda_j)
\Tr \underbrace{P_j\rho P_i\rho}_{=\delta_{i,j}\lambda_i^2 P_i}
\sigma^{-1/2}
\log(\sigma^{-1/2}\rho\sigma^{-1/2})\sigma^{1/2}\\
&=
\sum_{i}\underbrace{\log^{[1]}(\lambda_i,\lambda_i)\lambda_i^2}_{=\lambda_i}
\Tr P_i\sigma^{-1/2}
\log(\sigma^{-1/2}\rho\sigma^{-1/2})\sigma^{1/2}\\
&=
\Tr\sigma^{1/2}\rho\sigma^{-1/2}
\log(\sigma^{-1/2}\rho\sigma^{-1/2})\\
&=
D^{\max}(\rho\|\sigma),
\end{align*}
where $P_1,\ldots,P_r\in\bP(\hil)$ and $\lambda_1,\ldots,\lambda_r\in(0,+\infty)$ are such that 
$\sum_{i=1}^rP_i=I$, $\sum_{i=1}^r\lambda_i P_i=\rho$, and
in the fifth equality we used \eqref{opfunction derivative}, while the rest of the steps are straightforward.
\end{proof}

\begin{cor}
Let $D^q$ be a quantum relative entropy that is monotone under CPTP maps and additive on tensor products. For any $\rho,\sigma\in\B(\hil)\pp$,
\begin{align}\label{general q limit 1}
\lim_{\gamma\nearrow 1}D^{q,\#_{\gamma}}(\rho\|\sigma)=
D^{\max}(\rho\|\sigma).
\end{align}
If, moreover, $D^q$ is continuous on $\B(\hil)\pp\times\B(\hil)\pp$ then for any $\rho,\sigma\in\B(\hil)\pp$,
\begin{align}\label{general q limit 2}
\lim_{\gamma\searrow 0}D^{q,\#_{\gamma}}(\rho\|\sigma)=
D^{q}(\rho\|\sigma).
\end{align}
In particular, $(D^{q,\#_{\gamma}})_{\gamma\in(0,1)}$ continuously interpolates between 
$D^q$ and $D^{\max}$ when the arguments are restricted to be invertible.
\end{cor}
\begin{proof}
By \eqref{relentr maxmin2} and the preservation of ordering stated in \ref{weighted ordering} above, we have
\begin{align*}
D^{\Um,\#_{\gamma}}(\rho\|\sigma)\le
D^{q,\#_{\gamma}}(\rho\|\sigma)
\le
D^{\max,\#_{\gamma}}(\rho\|\sigma)
=
D^{\max}(\rho\|\sigma),
\end{align*}
where the last equality follows from Proposition \ref{prop:Dmax stable}. 
Taking the limit $\gamma\nearrow 1$ and using \eqref{DU gamma limit1} yields
\eqref{general q limit 1}. The limit in \eqref{general q limit 2} is obvious 
from the assumed continuity and that $\lim_{\gamma\searrow 0}\rho\#_{\gamma}\sigma=\sigma$ when $\rho$ and $\sigma$ are invertible.
\end{proof}

\begin{rem}\label{rem:DU Dmax interpolation}
Let $\sigma\in\B(\hil)\pne$ and $\psi\in\ran\sigma$ be a unit vector.
Then 
\begin{align*}
\sigma\#_{\gamma}\pr{\psi}
=
\sigma^{1/2}\bz\sigma^{-1/2}\pr{\psi}\sigma^{-1/2}\jz^{\gamma}\sigma^{1/2}
=
\pr{\psi}\inner{\psi}{\sigma\inv\psi}^{\gamma-1},
\end{align*}
whence
\begin{align*}
D^{\Um,\#_{\gamma}}(\pr{\psi}\|\sigma)
&=
-\frac{1}{1-\gamma}\Tr\pr{\psi}\log\bz\pr{\psi}\inner{\psi}{\sigma\inv\psi}^{\gamma-1}\jz\\
&=
\log\inner{\psi}{\sigma\inv\psi}
=\Tr\pr{\psi}\log\bz\pr{\psi}\sigma\inv\pr{\psi}\jz\\
&=
D^{\max}(\pr{\psi}\|\sigma)
\end{align*}
for every $\gamma\in(0,1)$,
while
\begin{align*}
\DU(\pr{\psi}\|\sigma)=
-\Tr\pr{\psi}\log\sigma=\inner{\psi}{(\log\sigma\inv)\psi}.
\end{align*}
Thus, we get that 
\begin{align*}
\DU(\pr{\psi}\|\sigma)\le D^{\Um,\#_{\gamma}}(\pr{\psi}\|\sigma)
=
D^{\max}(\pr{\psi}\|\sigma),\ds\ds\ds\gamma\in(0,1),
\end{align*}
and the inequality is strict if $\psi$ is not an eigenvector of $\sigma$.
In particular, this shows that the condition that $\rho$ and $\sigma$ are invertible cannot be completely 
omitted in \eqref{DU gamma limit0}.
\end{rem}

\begin{rem}
Obviously, $[0,1]\ni t\mapsto D^{(t)}:=(1-t)\DU+tD^{\max}$ interpolates continuously and monotone increasingly between $\DU$ and $D^{\max}$, and the same is true for 
$[0,1]\ni\gamma\mapsto D^{\Um,\#_{\gamma}}$, according to 
Proposition \ref{prop:DU Dmax interpolation}. 
The two families, however, are different. Indeed, 
the example in Remark \ref{rem:DU Dmax interpolation} shows that if a unit vector
$\psi\in\ran \sigma$ is not an eigenvector of $\sigma$ then 
\begin{align*}
D^{(t)}(\pr{\psi}\|\sigma)<D^{\max}(\pr{\psi}\|\sigma)=
D^{\Um,\#_{\gamma}}(\pr{\psi}\|\sigma),\ds\ds\ds t,\gamma\in(0,1).
\end{align*}
Since $D^q(\rho\|\sigma)=\lim_{\ep\searrow 0}D^q(\rho+\ep I\|\sigma+\ep I)$ holds for 
both $D^q=\DU$ and $D^q=D^{\max}$ and any $\rho,\sigma\in\B(\hil)\pne$, the above argument also shows that for any $t,\gamma\in(0,1)$ there exist invertible $\rho,\sigma$ such that 
$D^{(t)}(\rho\|\sigma)<D^{\Um,\#_{\gamma}}(\rho\|\sigma)$.
\end{rem}

\section{Barycentric R\'enyi divergences}
\label{sec:barycentric}

In the rest of the paper (i.e., in the present section and in Section \ref{sec:ex}) we 
use the term ``quantum relative entropy'' in a more restrictive (though still very general)
sense than in the previous sections. Namely, a quantum divergence $D^q$ will be called a quantum relative entropy if, on top of being a quantum extension of the classical relative entropy, it is also non-negative, it satisfies the scaling law \eqref{relentr scaling}, and the following 
\ki{support condition}:
\begin{align}\label{supp const}
D^q(\rho\|\sigma)<+\infty\ds\iff\ds \rho^0\le\sigma^0.
\end{align}
Note that by Remark \ref{rem:tracemon}, any quantum relative entropy in the above sense is also trace-monotone. In particular, no quantum relative entropy can take the value $-\infty$.

\begin{example}\label{ex:qrelative entropies}
It is easy to verify that that 
$\DU$, $D^{\meas}$ and $D^{\max}$ are 
all quantum relative entropies in the above more restrictive sense.
\end{example}

\subsection{Definitions}
\label{sec:defs}

\begin{definition}\label{defin:weighted Q}
Let $W\in\B(\X,\hil)\pne$ be a gcq channel, let $P\in\P_f^{\pm}(\X)$ and
\begin{align*}
S_+:=
\bigwedge\limits_{x:\,P(x)>0}W_x^0,\ds\ds\ds
S_-:=
\bigwedge\limits_{x:\,P(x)<0}W_x^0,
\end{align*}
and for every $x\in \X$, let $\D^{q_x}$ be a quantum relative entropy.
We define
\begin{align}
\QP{W}&:=
\sup_{\tau\in\B(S_+\hil)_{\ge 0}}\left\{\Tr\tau-\sum_{x\in\X}P(x)\D^{q_x}(\tau\|\ch{x})\right\},\label{weighted Q}\\
\psiP{W}&:=\log \QP{W},\label{gcq psi}\\
R_{\D^{\qv},\bal}(W,P)&:=
\inf_{\omega\in\S(S_+\hil)}\sum_{x\in\X}P(x)\D^{q_x}(\omega\|\ch{x}).
\label{left divrad}
\end{align}
Here, $\qv:=(q_x)_{x\in\X}$, 
$\D^{\qv}:=(D^{q_x})_{x\in\X}$, and 
$R_{\D^{\qv},\bal}(W,P)$ is the \ki{$P$-weighted left $\D^{\qv}$-radius of $W$}.
We call any $\omega$ attaining the infimum in \eqref{left divrad}
a \ki{$P$-weighted left $D^{\qv}$-center for $W$.}

When $P\notin\{\egy_{\{x\}}:\,x\in\X\}$, we also define the 
\ki{$P$-weighted barycentric R\'enyi-divergence of $W$ corresponding to $D^{\qv}$} as
\begin{align*}
D_P^{\bary,\qv}(W):=\frac{1}{\prod_{x\in\X}(1-P(x))}\bz-\log Q_P^{\bary,\qv}\bz\bz\frac{W_x}{\Tr W_x}\jz_{x\in\X}\jz\jz.
\end{align*}
\end{definition}

\begin{rem}
Since we almost exclusively consider only left divergence radii and left divergence centers in this paper, we will normally omit ``left'' from the terminology.
\end{rem}

\begin{rem}
Note that by definition, 
\begin{align*}
P(x)\ge 0,\,x\in\X\ds\imp\ds S_-=I,\ds\ds\ds
P(x)\le 0,\,x\in\X\ds\imp\ds S_+=I.
\end{align*}
\end{rem}

\begin{definition}\label{def:barycentric Renyi}
Let $D^{\qv}=(D^{\qn},D^{\qo})$ be quantum relative entropies. 
For any two non-zero PSD operators $\rho,\sigma\in\B(\hil)_{\gneq 0}$, 
and any $\alpha\in[0,+\infty)$, let 
\begin{align}
\Q_{\alpha}^{\bary,\qv}(\rho\|\sigma)&:=\sup_{\tau\in\B(\rho^0\hil)_{\ge 0}}\left\{\Tr\tau-
\alpha\D^{\qn}(\tau\|\rho)-(1-\alpha)\D^{\qo}(\tau\|\sigma)\right\},
\label{barycentric Qalpha def}\\
\psi_{\alpha}^{\bary,\qv}(\rho\|\sigma)&:=\log\Q_{\alpha}^{\bary,\qv}(\rho\|\sigma),
\label{barycentric psialpha def}\\
\D_{\alpha}^{\bary,\qv}(\rho\|\sigma)&:=
\frac{1}{\alpha-1}\log\Q_{\alpha}^{\bary,\qv}(\rho\|\sigma)-\frac{1}{\alpha-1}\log\Tr\rho,
\label{barycentric Dalpha def}
\end{align}
where we define the last quantity only for $\alpha\in[0,1)\cup(1,+\infty)$.
$\D_{\alpha}^{\bary,\qv}(\rho\|\sigma)$ is called the 
\ki{barycentric R\'enyi $\alpha$-divergence} of $\rho$ and $\sigma$
corresponding to $D^{\qv}$. 
\end{definition}

\begin{rem}
When $D^{\qn}=\D^{\qo}=D^q$, we will use the simpler notation 
$D_{\alpha}^{\bary,q}$ instead of 
$D_{\alpha}^{\bary,(\qn,\qo)}$.
\end{rem}

\begin{rem}
It is easy to see that when $P$ is a probability measure, the supremum in 
\eqref{weighted Q} and the infimum in \eqref{left divrad} can be equivalently taken over
$\B(\hil)\p$ and $\S(\hil)$, respectively, 
i.e., 
\begin{align}
\QP{W}
&=
\sup_{\tau\in\B(\hil)\p}\left\{\Tr\tau-\sum_{x\in\X}P(x)\D^{q_x}(\tau\|\ch{x})\right\},
\label{barycentric Qalpha def support}\\
R_{\D^{\qv},\bal}(W,P)
&=
\inf_{\omega\in\S(\hil)}\sum_{x\in\X}P(x)\D^{q_x}(\omega\|\ch{x}).
\label{barycentric Dalpha def support}
\end{align}
Likewise, in the $2$-variable case we have 
\begin{align}
\Q_{\alpha}^{\bary,\qv}(\rho\|\sigma)
&=
\sup_{\tau\in\B(\hil)_{\ge 0}}\left\{\Tr\tau-
\alpha\D^{\qn}(\tau\|\rho)-(1-\alpha)\D^{\qo}(\tau\|\sigma)\right\}\\
&=
\sup_{\tau\in\B((\rho^0\wedge\sigma^0)\hil)_{\ge 0}}\left\{\Tr\tau-
\alpha\D^{\qn}(\tau\|\rho)-(1-\alpha)\D^{\qo}(\tau\|\sigma)\right\},\ds\ds\ds
\alpha\in(0,1).
\label{barycentric Qalpha def3}
\end{align}
In the general case, the restriction $\tau^0\le S_+$ is introduced to avoid the appearance
of infinities of opposite signs in $\sum_{x\in\X}P(x)\D^{q_x}(\tau\|\ch{x})$.
In the $2$-variable case \eqref{barycentric Qalpha def}, the restriction 
$\tau^0\le\rho^0$ also serves to guarantee that  $Q_{\alpha}^{\bary,\qv}$ is a quantum 
extension of $Q_{\alpha}^{\cl}$ for $\alpha>1$, which would not be true, for instance, 
if it was replaced by $\tau^0\le\rho^0\wedge\sigma^0$; see, 
e.g., Corollary \ref{cor:Renyi bq infty}.
\end{rem}

\begin{rem}
Note that with the choice 
$\X=\{0,1\}$, $W_0=\rho$, $W_1=\sigma$, and $P(0)=\alpha$,
\eqref{barycentric Qalpha def} and \eqref{barycentric psialpha def}
are special cases of \eqref{weighted Q} and \eqref{gcq psi}, respectively, when 
$\alpha\in(0,1)$, and we will show in Lemma \ref{lemma:psi rep} that also 
\eqref{barycentric Dalpha def} is a special case of \eqref{left divrad} in this case. 
When $\alpha=0$, the restriction $\tau^0\le S_+$ in 
\eqref{weighted Q} would give $\tau^0\le\sigma^0$, while we use 
$\tau^0\le\rho^0$ in \eqref{barycentric Qalpha def}. The reason for 
this is to guarantee the continuity of $D_{\alpha}^{\bary,q}$ at $0$; see 
Proposition \ref{prop:limits}.
\end{rem}

\begin{rem}\label{rem:psi 1}
Note that \eqref{barycentric Qalpha def} can be seen as a $2$-variable extension of the variational formula
\eqref{Tropp variational}. In particular, 
we have
\begin{align}\label{psi proof 2}
Q^{\bary,\qv}_{1}(\rho\|\sigma)
=
\max_{\tau\in\B(\rho^0\hil)\p}\{\Tr\tau-\D^{\qn}(\tau\|\rho)\}=\Tr\rho,
\end{align}
where the first equality is by definition \eqref{barycentric Qalpha def}, 
and the second equality is due to \eqref{Tropp variational}.
Thus, 
\begin{align}\label{psi 1}
\psi^{\bary,\qv}_{1}(\rho\|\sigma)=\log\Tr\rho, 
\end{align}
and for every $\alpha\in[0,1)\cup(1,+\infty)$,
\begin{align}
\D_{\alpha}^{\bary,\qv}(\rho\|\sigma)
&=
\frac{1}{\alpha-1}\log\Q_{\alpha}^{\bary,\qv}(\rho\|\sigma)-\frac{1}{\alpha-1}\log\Tr\rho\nn\\
&=\frac{\psi_{\alpha}^{\bary,\qv}(\rho\|\sigma)-\psi_{1}^{\bary,\qv}(\rho\|\sigma)}{\alpha-1}\,.
\label{barycentric Dalpha def2}
\end{align}
By Remark \ref{rem:relentr pos}, the maximum in \eqref{psi proof 2} is attained at $\tau$ 
if and only if $\Tr\tau=\Tr\rho$ and $D^{\qn}(\tau\|\rho)=0$;
in particular, 
$\tau=\rho$ is the unique maximizer in \eqref{psi proof 2}
when $D^{\qn}$ is strictly positive.

At $\alpha=0$, \eqref{barycentric Qalpha def} and \eqref{Tropp variational} give
\begin{align}\label{barycentric 0-Renyi}
\sigma^0\le\rho^0\ds\imp\ds Q^{\bary,\qv}_{0}(\rho\|\sigma)=\Tr\sigma\ds\imp\ds
D^{\bary,\qv}_{0}(\rho\|\sigma)=\log\Tr\rho-\log\Tr\sigma.
\end{align}
In general the above equalities do not hold; see Proposition \ref{prop:barycentric nonneg}.
\end{rem}

\begin{lemma}\label{lemma:psi rep}
(i) In the setting of Definition \ref{defin:weighted Q},
\begin{align}
-\log Q_P^{\bary,\qv}(W)
&=
R_{\D^{\qv},\bal}(W,P).
\label{psi optimization}
\end{align}
Moreover, if $S_+\le S_-$ then 
a $\tau\in\B(S_+\hil)\p$ is optimal in \eqref{weighted Q} if and only if
\begin{align}\label{Tr argmax=max}
\QP{W}=\Tr\tau
\ds\ds\ds\text{and}\ds\ds\ds
\sum_{x\in\X}P(x)\D^{q_x}(\tau\|\ch{x})=0,
\end{align}
and if $\tau\ne 0$ is optimal in \eqref{weighted Q} then $\omega:=\tau/\Tr\tau$ is optimal in \eqref{left divrad}. 
Conversely, for any $\omega$ that is optimal in \eqref{left divrad},
$\tau:=e^{-R_{\D^{\qv},\bal}(W,P)}\omega$
is optimal in \eqref{weighted Q}.

(ii) In the setting of Definition \ref{def:barycentric Renyi},
\begin{align}
-\log Q_{\alpha}^{\bary,\qv}(\rho\|\sigma)&=\inf_{\omega\in\S(\rho^0\hil)}
\{\alpha\D^{\qn}(\omega\|\rho)+(1-\alpha)\D^{\qo}(\omega\|\sigma)\},
\ds\ds\alpha\in[0,+\infty).\label{psi alpha}
\end{align}
Assume for the rest that $\alpha\in[0,1]$ or $\rho^0\le\sigma^0$. Then $\tau$ is optimal in \eqref{barycentric Qalpha def} if and only if
\begin{align}\label{Tr argmax=max2}
Q_{\alpha}^{\bary,\qv}(\rho\|\sigma)=\Tr\tau
\ds\ds\ds\text{and}\ds\ds\ds
\alpha\D^{\qn}(\tau\|\rho)+(1-\alpha)\D^{\qo}(\tau\|\sigma)=0,
\end{align}
and if $\tau\ne 0$ is optimal in \eqref{barycentric Qalpha def} then $\omega:=\tau/\Tr\tau$ is optimal in 
\eqref{psi alpha}. Conversely, for any $\omega$ that is optimal in \eqref{psi alpha},
$\tau:=e^{\psi_{\alpha}^{\bary,\qv}(\rho\|\sigma)}\omega$ optimal in \eqref{barycentric Qalpha def}.
\end{lemma}
\begin{proof}
(i)
Assume first that $S_+=0$. Then the only admissible $\tau\in\B(\S_+\hil)\p$ 
in \eqref{weighted Q} is 
$\tau=0$, whence $\QP{W}=0$, according to \eqref{zero argument q2}, and thus 
$\psiP{W}=-\infty$. On the other hand, the infimum in \eqref{left divrad} is taken over the empty set, and hence it is equal to $+\infty$. Thus, \eqref{psi optimization} and 
\eqref{Tr argmax=max} hold. 

Assume next that $S_+\ne 0$. 
If there exists an $x\in\X$ such that $P(x)<0$ and $S_+\nleq W_x^0$ then 
taking $\tau:=\omega:=S_+/\Tr S_+$ yields
\begin{align*}
\QP{W}\ge
\Tr\tau-\underbrace{\sum_{x:\,P(x)>0}P(x)\D^{q_x}(\tau\|\ch{x})}_{\in\bR}
-\underbrace{\sum_{x:\,P(x)<0}P(x)\D^{q_x}(\tau\|\ch{x})}_{=-\infty}=+\infty,
\end{align*}
and
\begin{align*}
R_{\D^{\qv},\bal}(W,P)\le
\underbrace{\sum_{x:\,P(x)>0}P(x)\D^{q_x}(\omega\|\ch{x})}_{\in\bR}
+\underbrace{\sum_{x:\,P(x)<0}P(x)\D^{q_x}(\omega\|\ch{x})}_{=-\infty}=-\infty,
\end{align*}
(where we used that $D^{q_x}$ does not take the value $-\infty$), whence 
\eqref{psi optimization} holds.

Finally, if $0\ne S_+\le S_-$ 
then the proof follows easily from representing a positive semi-definite operator 
$\tau\in\B(S_+\hil)\p$ as a pair $(\omega,t)\in\S(S_+\hil)\times[0,+\infty)$. 
Indeed, we have
\begin{align}
\QP{W}&=
\sup_{\omega\in\S(S_+\hil)}\sup_{t\in[0,+\infty)}\left\{\Tr t\omega
-\sum_{x\in\X}P(x)\D^{q_x}(t\omega\|\ch{x})\right\}\nonumber\\
&=
\sup_{\omega\in\S(S_+\hil)}\sup_{t\in[0,+\infty)}\Big\{t-t\log t-t
\underbrace{\sum_{x\in\X}P(x)\D^{q_x}(\omega\|\ch{x})}_{=:c(\omega)}\Big\},
\label{psi proof 1}
\end{align}
where the first equality is by definition, and the second equality follows from 
the scaling property \eqref{scaling1}.
Note that $c(\omega)\ne\pm\infty$ by assumption, and 
the inner supremum in \eqref{psi proof 1} is equal to 
$e^{-c(\omega)}$, attained at $t=e^{-c(\omega)}$, according to Lemma \ref{lemma:optimization}.
From these, all the remaining assertions in (i) follow immediately.

The assertions in (ii) are special cases of the corresponding ones in (i)
when $\alpha\in(0,+\infty)$ (also taking into 
account \eqref{barycentric Qalpha def3} when $\alpha\in(0,1)$).
The case $\alpha=0$ can be verified analogously to the above; we omit the easy 
details. 
\end{proof}

\begin{rem}
Clearly, when $\alpha>1$ and $\rho^0\nleq\sigma^0$ then 
the set of optimal $\tau$ operators in \eqref{barycentric Qalpha def} is exactly 
$\{\tau\in\B(\rho^0\hil)\pne:\,\tau^0\nleq\sigma^0\}$, and the set of optimal $\omega$ states in 
\eqref{psi alpha} is exactly 
$\{\omega\in\S(\rho^0\hil):\,\omega^0\nleq\sigma^0\}$.
\end{rem}

\begin{cor}
Assume that $S_+\le S_-$. Then 
\begin{align*}
\QP{W}=\max\left\{\Tr\tau:\,\tau\in\B(S_+\hil)\p,\,\sum_{x\in\X}P(x)\D^{q_x}(\tau\|\ch{x})=0\right\}.
\end{align*}
Likewise, if $\alpha\in[0,1]$ or $\rho^0\le\sigma^0$, 
then 
\begin{align*}
Q_{\alpha}^{\bary,\qv}(\rho\|\sigma)
=
\max\left\{\Tr\tau:\,\tau\in \B(\rho^0\hil)\p,\,
\alpha \D^{\qn}(\tau\|\rho)+(1-\alpha)\D^{\qo}(\tau\|\sigma)=0\right\}.
\end{align*}
\end{cor}
\begin{proof}
Immediate from the characterizations of the optimal $\tau$ in \eqref{Tr argmax=max} and \eqref{Tr argmax=max2}.
\end{proof}

\begin{rem}
Note that in the case $S_+\le S_-$, the condition $\sum_{x\in\X}P(x)\D^{q_x}(\tau\|\ch{x})=0$ is necessary for the optimality of $\tau$, but not sufficient. Indeed, it is easy to see from the scaling property \eqref{scaling1} that 
\begin{align*}
&\left\{\tau\in\B(S_+\hil)\p:\,\sum_{x\in\X}P(x)\D^{q_x}(\tau\|\ch{x})=0\right\}\\
&\ds\ds=
\left\{\exp\bz-\sum_{x\in\X}P(x)\D^{q_x}(\tau\|\ch{x})\jz\omega:\,\omega\in\S(S_+\hil)\right\}\cup\{0\}.
\end{align*}
On the other hand, each $\tau\in\B(S_+\hil)\pne$ with
$\sum_{x\in\X}P(x)\D^{q_x}(\tau\|\ch{x})=0$ has the extremality property
\begin{align*}
\Tr(\lambda\tau)-\sum_{x\in\X}P(x)\D^{q_x}(\lambda\tau\|\ch{x})
&=
(\lambda-\lambda\log\lambda)\Tr\tau\\
&<\Tr\tau=\Tr(\tau)-\sum_{x\in\X}P(x)\D^{q_x}(\tau\|\ch{x})
\end{align*}
for every $\lambda\in(0,1)\cup(1,+\infty)$, where the first equality is again due to the scaling property \eqref{scaling1}.
\end{rem}

\begin{rem}
Under the conditions given in Lemma \ref{lemma:psi rep},
for the supremum in \eqref{weighted Q} to be a maximum, it is sufficient if
the infimum in \eqref{left divrad} is a minimum.
For the latter, a natural sufficient condition is that each $D^{q_x}$ with $x\in\supp P$
is lower semi-continuous in its first argument (when $P$ is a probability measure),
or continuous in its first argument with its support dominated by the support of a
fixed second argument (when $P$ can take negative values),
since the domain of optimization,
namely, $\S(S_+\hil)$, is a compact set. 

Examples of quantum relative entropies that are lower semi-continuous in their first argument (in fact, in both of their arguments), include
$D^{\meas}$, $\DU$, and their $\gamma$-weighted versions, as well
$D^{\max}$, and obviously, all possible convex combinations of these. 
$\DU$ and $D^{\max}$ are also clearly continuous in their first argument when its support is dominated by the support of a fixed second argument. 
\end{rem}

\begin{rem}\label{rem:Dalpha rep}
Using 
Lemma \ref{lemma:psi rep}, we get 
that for every $\alpha\in[0,1)\cup(1,+\infty)$,
\begin{align}
\D_{\alpha}^{\bary,\qv}(\rho\|\sigma)
&=
\frac{1}{\alpha-1}\log\Q_{\alpha}^{\bary,\qv}(\rho\|\sigma)-\frac{1}{\alpha-1}\log\Tr\rho
\label{Renyi bq1}\\
&=
\frac{1}{1-\alpha}
\inf_{\omega\in\S(\rho^0\hil)}\{\alpha\D^{\qn}(\omega\|\rho)
+(1-\alpha)\D^{\qo}(\omega\|\sigma)\}
-\frac{1}{\alpha-1}\log\Tr\rho\label{Renyi bq2}\\
&=
\frac{1}{1-\alpha}
\inf_{\omega\in\S(\rho^0\hil)}\left\{\alpha\D^{\qn}\bz\omega\Big\|\frac{\rho}{\Tr\rho}\jz
+(1-\alpha)\D^{\qo}\bz\omega\|\sigma\jz\right\}
+\log\Tr\rho\label{Renyi bq3}\\
&=\frac{1}{1-\alpha}
\inf_{\omega\in\S(\rho^0\hil)}\left\{\alpha\D^{\qn}\bz\omega\Big\|\frac{\rho}{\Tr\rho}\jz
+(1-\alpha)\D^{\qo}\bz\omega\Big\|\frac{\sigma}{\Tr\sigma}\jz\right\}
+\log\Tr\rho-\log\Tr\sigma,\label{Renyi bq4}
\end{align}
where the first equality is by definition,
the second equality follows from \eqref{psi alpha},
and the third and the fourth equalities from the scaling law
\eqref{scaling2}. Moreover, for $\alpha\in(0,1)$, the infimum can be 
taken over $\S(\hil)$, i.e., 
\begin{align}
\D_{\alpha}^{\bary,\qv}(\rho\|\sigma)
&=
\frac{1}{1-\alpha}
\inf_{\omega\in\S(\hil)}\{\alpha\D^{\qn}(\omega\|\rho)
+(1-\alpha)\D^{\qo}(\omega\|\sigma)\}
-\frac{1}{\alpha-1}\log\Tr\rho\label{Renyi bq6}\\
&=\frac{1}{1-\alpha}
\inf_{\omega\in\S(\hil)}\left\{\alpha\D^{\qn}\bz\omega\Big\|\frac{\rho}{\Tr\rho}\jz
+(1-\alpha)\D^{\qo}\bz\omega\Big\|\frac{\sigma}{\Tr\sigma}\jz\right\}
+\log\Tr\rho-\log\Tr\sigma,\label{Renyi bq5}
\end{align}
because if $\omega^0\nleq\rho^0$ then 
$\D^{\qn}(\omega\|\rho)=\D^{\qn}\bz\omega\Big\|\frac{\rho}{\Tr\rho}\jz=+\infty$.
The situation is different for $\alpha=0$; see, e.g., \eqref{Renyi bq limit0}.

The above formulas explain the term
``barycentric R\'enyi divergence''.
\end{rem}

\begin{definition}
For $\alpha\in(0,1)$, any $\omega$ attaining the infimum in \eqref{Renyi bq2} will be called an
\ki{$\alpha$-weighted (left) $D^{\qv}$-center for $(\rho,\sigma)$}. 
\end{definition}

\subsection{Barycentric R\'enyi divergences are quantum R\'enyi divergences}
\label{sec:baryR is quantumR}

In this section we show that the barycentric R\'enyi $\alpha$-divergences
are quantum R\'enyi divergences for every $\alpha\in(0,1)$, provided that the defining 
quantum relative entropies are monotone under pinchings. This latter condition does not pose a serious restriction; indeed, 
all the concrete quantum relative entropies that we consider in this paper 
(e.g., measured, Umegaki, maximal, and the $\gamma$-weighted versions of these)
are monotone
under PTP maps, and hence also under pinchings.

Isometric invariance holds even without this mild restriction, and also for 
$\alpha>1$:

\begin{lemma}\label{lemma:isoinv}
All the quantities in \eqref{weighted Q}--\eqref{barycentric Dalpha def}
are invariant under isometries, and hence they are all quantum divergences.
\end{lemma}
\begin{proof}
We prove the statement only for  $\Q_P^{\bary,\qv}$, as for the other quantities it 
either follows from that, or the proof goes the same way.
Let $W\in\B(\X,\hil)\pne$ be a gcq channel, $P\in\P_f^{\pm}(\X)$, and 
$V:\,\hil\to\kil$ be an isometry. 
Obviosuly, $\tilde S_+:=\bigwedge_{x:\,P(x)>0}(VW_xV^*)^0=V(\bigwedge_{x:\,P(x)>0}W_x^0)V^*
=VS_+V^*$, and for any $\tau\in\B(\tilde S_+\kil)\p$ there exists a unique 
$\hat\tau\in\B(S_+\hil)\p$ such that $\tau=V\hat\tau V^*$. Thus,
\begin{align}\label{isoinv proof 1}
\Q_P^{\bary,\qv}(VWV^*)&=
\sup_{\tau\in\B(\tilde S_+\kil)\p}
\Bigg\{\Tr\tau-\sum_{x\in\X}P(x)\D^{q_x}(\tau\|V\ch{x}V^*)\Bigg\}\\
&=
\sup_{\hat\tau\in\B(S_+\hil)\p}
\Bigg\{\Tr V\hat\tau V^*-\sum_{x\in\X}P(x)\D^{q_x}(V\hat\tau V^*\|V\ch{x}V^*)\Bigg\}\\
&=
\sup_{\hat\tau\in\B(S_+\hil)\p}
\Bigg\{\Tr \hat\tau-\sum_{x\in\X}P(x)\D^{q_x}(\hat\tau \|\ch{x})\Bigg\}\\
&=
\Q_P^{\bary,\qv}(W),
\end{align}
where the third  equality follows by the isometric invariance of the relative entropies. 
\end{proof}

\bigskip

Recall that $D^q$ is said to be monotone under pinchings if 
\begin{align*}
D^q\bz\sum_{i=1}^r P_i\rho P_i\Big\|\sum_{i=1}^r P_i\sigma P_i\jz
\le
D^q(\rho\|\sigma)
\end{align*}
for any $\rho,\sigma\in\B(\hil)\pne$ and $P_1,\ldots,P_r\in\bP(\hil)$ such that 
$\sum_{i=1}^r P_i=I$. 

\begin{lemma}\label{lemma:commuting divrad}
Let $W\in\B(\X,\hil)\pne$ be a gcq channel that is classical on the support of some 
$P\in\P_f(\X)$, i.e., 
there exists an ONB $(e_i)_{i=0}^{d-1}$ in $\hil$ such that 
$W_x=\sum_{i=0}^{d-1}\wtilde W_x(i)\pr{e_i}$, 
where $\wtilde W_x(i):=\inner{e_i}{W_xe_i}$, $i\in[d]$, 
$x\in\supp P$.
If all $D^{q_x}$, $x\in\supp P$, are monotone under pinchings then 
\begin{align}\label{Q(W,P) classical}
\QP{W}=\sum_{i\in \tilde S}\prod_{x\in\supp P}\wtilde W_x(i)^{P(x)},
\end{align}
where $\displaystyle{\tilde S:=\medcap_{x\in\supp P}\supp \wtilde W_x}$
and $\supp \wtilde W_x=\{i\in[d]:\,\wtilde W_x(i)>0\}$; moreover, 
there exists a unique optimal $\tau$ in \eqref{weighted Q}, given by 
\begin{align}\label{commuting tau}
\tau_P^{\qv}(W):=\tau_P(\wtilde W):=\sum_{i\in \tilde S}\pr{e_i}\prod_{x\in\supp P}\wtilde W_x(i)^{P(x)}\ds.
\end{align}
\end{lemma}
\begin{proof}
If $S_+=0$ then $\QP{W}=0$, and the RHS of \eqref{Q(W,P) classical} is an empty sum, whence the equality in \eqref{Q(W,P) classical} holds trivially. 

Thus, for the rest we assume that $S_+\ne 0$. 
Let $\E(\valt):=\sum_{i=0}^{d-1}\pr{e_i}(\valt)\pr{e_i}$ be the pinching corresponding to 
the joint eigenbasis of the $W_x$, $x\in\supp P$, guaranteed by the classicality assumption. For any $\tau\in\B(S_+\hil)\p$,
\begin{align*}
\Tr\tau-\sum_{x\in\X}P(x)\underbrace{\D^{q_x}(\tau\|\ch{x})}_{\ge \D^{q_x}(\E(\tau)\|\E(\ch{x}))}
&\le
\underbrace{\Tr\tau}_{=\Tr\E(\tau)}-\sum_{x\in\X}P(x)\D^{q_x}(\E(\tau)\|\underbrace{\E(\ch{x})}_{=\ch{x}})\\
&=
\Tr\E(\tau)-\sum_{x\in\X}P(x)\D^{q_x}(\E(\tau)\|\ch{x})),
\end{align*}
where the inequality follows from the monotonicity of the $D^{q_x}$ under pinchings.
Thus, the supremum in \eqref{weighted Q} can be restricted to $\tau$ operators that can be written as 
$\tau=\sum_{i=1}^d\tilde\tau(i)\pr{e_i}$ with some $\tilde\tau(i)\in[0,+\infty)$, $i\in[d]$.
Clearly, $\tau^0\le S_+$ is equivalent to $\supp\tilde\tau\subseteq\tilde S$.
For any such $\tau$,
\begin{align*}
\Tr\tau-\sum_{x\in\X}P(x)\D^{q_x}(\tau\|\ch{x})
&=
\Tr\tau-\sum_{x\in\X}P(x)
\sum_{i\in \tilde S}[\tilde\tau(i)\log\tilde\tau(i)-\tilde\tau(i)\log\ch{x}(i)]\\
&=
\sum_{i\in \tilde S}\Big[\tilde\tau(i)-\tilde\tau(i)\log\tilde\tau(i)+\tilde\tau(i)\sum_{x\in\supp P}P(x)\log\wtilde W_x(i)
\Big].
\end{align*}
The supremum of this over all such $\tau$ is
\begin{align*}
\sum_{i\in \tilde S}e^{\sum_{x\in\supp P}P(x)\log\wtilde W_x(i)}=
\sum_{i\in \tilde S}\prod_{x\in\supp P}\wtilde W_x(i)^{P(x)},
\end{align*}
which is uniquely attained at the $\tau=\tau_P^{\qv}(W)$ given in \eqref{commuting tau},
according to Lemma \ref{lemma:optimization}.
This proves \eqref{Q(W,P) classical}.
\end{proof}

\begin{cor}\label{cor:commuting radius}
In the setting of Lemma \ref{lemma:commuting divrad},
the $P$-weighted left $D^{\qv}$-radius of $W$ can be given explicitly as 
\begin{align*}
R_{\D^{\qv},\bal}(W,P)=-\log \sum_{i\in \tilde S}\prod_{x\in\supp P}\wtilde W_x(i)^{P(x)},
\end{align*}
and if $\tilde S\ne \emptyset$ then there is a unique 
$P$-weighted left $D^{\qv}$-center for $W$,
given by 
\begin{align}\label{classical center}
\omega^{\qv}_P(W):=\frac{\tau^{\qv}(W,P)}{\Tr\tau^{\qv}(W,P)}=\sum_{i\in \tilde S}\pr{e_i}\frac{\prod_{x\in\supp P}\wtilde W_x(i)^{P(x)}}{\sum_{j\in \tilde S}
\prod_{x\in\supp P}\wtilde W_x(j)^{P(x)}}=:\omega_P(\wtilde W).
\end{align}
\end{cor}
\begin{proof}
Immediate from Lemmas \ref{lemma:psi rep} and \ref{lemma:commuting divrad}.
\end{proof}

Lemma \ref{lemma:commuting divrad} yields immediately the following:

\begin{cor}\label{cor:commuting Qalpha}
Assume that $\rho,\sigma\in\B(\hil)\pne$ commute, and hence can be written as
$\rho=\sum_{i=1}^d\tilde\rho(i)\pr{e_i}$,
$\sigma=\sum_{i=1}^d\tilde\sigma(i)\pr{e_i}$,
in some ONB $(e_i)_{i=1}^d$. If $D^{\qn}$ and $D^{\qo}$ are monotone under pinchings then 
\begin{align*}
Q_{\alpha}^{\bary,\qv}(\rho\|\sigma)=
Q_{\alpha}(\tilde\rho\|\tilde\sigma)=
\sum_{i=1}^d\tilde\rho(i)^{\alpha}\tilde\sigma(i)^{1-\alpha},\ds\ds\ds\alpha\in(0,1),
\end{align*}
and there exists a unique 
optimal $\tau$ in \eqref{barycentric Qalpha def}, 
given by 
\begin{align}\label{optimal comm tau}
\tau^{\qv}_{\alpha}(\rho\|\sigma):=
\tau_{\alpha}(\tilde\rho\|\tilde\sigma):=\sum_{i=1}^d\pr{e_i}\tilde\rho(i)^{\alpha}\tilde\sigma(i)^{1-\alpha}.
\end{align}
\end{cor}

As a special case of Corollary \ref{cor:commuting radius}, we get the following:

\begin{cor}
In the setting of Corollary \ref{cor:commuting Qalpha}, if 
$\rho^0\wedge\sigma^0\ne 0$ then for every $\alpha\in(0,1)$
there exists a unique 
$\alpha$-weighted $D^{\qv}$-center for $(\rho,\sigma)$,
given by 
\begin{align}\label{classical Hellinger}
\omega_{\alpha}^{\qv}(\rho\|\sigma):=\frac{\tau^{\qv}_{\alpha}(\rho\|\sigma)}{\Tr \tau^{\qv}_{\alpha}(\rho\|\sigma)}=
\sum_{i=1}^d\pr{e_i}\frac{\tilde\rho(i)^{\alpha}\tilde\sigma(i)^{1-\alpha}}{\sum_{j=1}^d
\tilde\rho(j)^{\alpha}\tilde\sigma(j)^{1-\alpha}}
=:\omega_{\alpha}(\tilde\rho\|\tilde\sigma).
\end{align}
\end{cor}
\begin{proof}
Immediate from Corollary \ref{cor:commuting Qalpha} and Lemma \ref{lemma:psi rep}.
\end{proof}

\begin{rem}
Note that $\tau_P^{\qv}(W)$ in \eqref{commuting tau} and 
$\omega^{\qv}_P(W)$ in \eqref{classical center}
are independent of $D^{\qv}$, as long as all $D^{q_x}$, $x\in\supp P$,
are monotone
under pinchings. Likewise,
$\tau_{\alpha}^{\qv}(\rho\|\sigma)$ in \eqref{optimal comm tau} and 
$\omega_{\alpha}^{\qv}(\rho\|\sigma)$ in \eqref{classical Hellinger}
are independent of $D^{\qn}$ and $D^{\qo}$, as long as both of them are monotone
under pinchings.
\end{rem}

Lemmas \ref{lemma:isoinv} and \ref{lemma:commuting divrad}, and Corollary \ref{cor:commuting Qalpha} together give the following:

\begin{prop}\label{prop:barycentric quantum Renyi}
If $D^{q_x}$, $x\in\supp P$, are quantum relative entropies that are 
monotone under 
pinchings then $Q_P^{\bary,\qv}$ is a quantum extension 
(in the sense of Definition \ref{def:quantumdiv}) of the classical $Q_P$
given in Definition \ref{def:classical QP}.

Likewise, if $D^{\qn}$ and $D^{\qo}$ are two quantum relative entropies that are monotone under 
pinchings then for every $\alpha\in(0,1)$ 
the corresponding barycentric R\'enyi $\alpha$-divergence 
$D_{\alpha}^{\bary,\qv}$ is a quantum R\'enyi $\alpha$-divergence
in the sense of Definition \ref{def:quantumdiv}.
\end{prop}

\begin{rem}
Note that in the classical case the barycentric R\'enyi $\alpha$-divergence is equal to the 
unique classical R\'enyi $\alpha$-divergence also for $\alpha>1$; see
\eqref{classical variational2}.
On the other hand, if $D^{\qn}\ne D^{\qo}$ then it may happen that $D_{\alpha}^{\bary,\qv}$ is not a quantum 
R\'enyi $\alpha$-divergence for some $\alpha>1$; see Remark \ref{rem:notqRenyi}.
\end{rem}
\medskip

Note that for a fixed $i\in\cap_{x\in\supp P}\supp \wtilde W_x$, the expression 
$\prod_{x\in\supp P}\wtilde W_x(i)^{P(x)}$ in \eqref{Q(W,P) classical} is the weighted geometric mean 
of $(\wtilde W_x(i))_{x\in\supp P}$ with weights $(P(x))_{x\in\supp P}$.
This motivates the following:

\begin{definition}\label{def:Dq-mean}
If $D^{\qv}$, $W\in\B(\X,\hil)\pne$, and $P\in\P_f^{\pm}(\X)$ are such that there exists a unique optimizer
$\tau=:\tau^{\qv}_P(W)$ in \eqref{weighted Q} then this $\tau$ is called the 
$P$-weighted $D^{\qv}$-geometric mean of $W$, and is also 
denoted by $G_P^{D^{\qv}}(W):=\tau^{\qv}_P(W)$.

Similarly, if there exists a unique optimizer $\tau=:\tau_{\alpha}^{\qv}(\rho\|\sigma)$ in \eqref{barycentric Qalpha def} then it is called the 
$\alpha$-weighted $D^{\qv}$-geometric mean of $\rho$ and $\sigma$, and it is also denoted by 
$G^{D^{\qv}}_{\alpha}(\rho\|\sigma):=\tau_{\alpha}^{\qv}(\rho\|\sigma)$.
\end{definition}

\begin{rem}
Note that if $G_P^{D^{\qv}}(W)$ exists then by definition and by Lemma \ref{lemma:psi rep},
\begin{align*}
Q_P^{\bary,\qv}(W)=\Tr G_P^{D^{\qv}}(W)=Q_P^{G_P^{D^{\qv}}}(W)
\end{align*}
(in particular, if 
$G^{D^{\qv}}_{\alpha}(\rho\|\sigma)$ exists then 
$Q^{\bary,\qv}_{\alpha}(\rho\|\sigma)=\Tr G^{D^{\qv}}_{\alpha}(\rho\|\sigma)$),
which can be seen as a special case of \eqref{multiQ from G}.
\end{rem}

In classical statistics, the family of states 
$(\omega_{\alpha}(\tilde\rho\|\tilde\sigma))_{\alpha\in(0,1)}$ given in 
\eqref{classical Hellinger} is called the \ki{Hellinger arc}. 
(Note that 
if $\tilde\rho$ and $\tilde\sigma$ are probability distributions with equal supports then 
the Hellinger arc connects them in the sense that
$\lim_{\alpha\searrow 0}\omega_{\alpha}(\tilde\rho\|\tilde\sigma)=\tilde\rho$, 
$\lim_{\alpha\nearrow 1}\omega_{\alpha}(\tilde\rho\|\tilde\sigma)=\tilde\sigma$.)
This motivates the following:

\begin{definition}
Assume that $\rho,\sigma\in\B(\hil)\pne$ and $D^{\qn}$, $D^{\qo}$ are such that 
for every $\alpha\in(0,1)$ there exists a unique 
$\alpha$-weighted $D^{\qv}$-center $\omega_{\alpha}^{\qv}(\rho\|\sigma)$
for $(\rho,\sigma)$.
Then 
$(\omega_{\alpha}^{\qv}(\rho\|\sigma))_{\alpha\in(0,1)}$
is called the \ki{$D^{\qv}$-Hellinger arc} for $\rho$ and $\sigma$.

More generally, if 
$W$ and $D^{\qv}$ are such that 
for every $P\in\P_f(\X)$
there exists a unique 
$P$-weighted $D^{\qv}$-center $\omega_{P}^{\qv}(W)$ for $W$ 
then we call  
$(\omega_{P}^{\qv}(W))_{P\in\P_f(\X)}$
the \ki{$D^{\qv}$-Hellinger body} for $W$.
\end{definition}

\begin{rem}\label{rem:mean divcenter relations}
Note that by Lemma \ref{lemma:psi rep}, for given $P\in\P_f^{\pm}(\X)$, 
$W\in\B(\X,\hil)\pne$, and $D^{\qv}$, 
there exists a unique non-zero $P$-weighted $D^{\qv}$-geometric mean 
$G_P^{D^{\qv}}(W)=\tau_P^{\qv}(W)$ if and only if there exists a unique 
$P$-weighted $D^{\qv}$-center $\omega_P^{\qv}(W)$ for $W$, and in this case we have
\begin{align*}
\omega_P^{\qv}(W)&=\frac{G_P^{D^{\qv}}(W)}{\Tr G_P^{D^{\qv}}(W)},\\
Q_P^{\bary,\qv}(W)&=\Tr G_P^{D^{\qv}}(W),\\
0&=\sum_{x\in\X}P(x)D^{q_x}(G_P^{D^{\qv}}(W)\|W_x),\\
-\log Q_P^{\bary,\qv}(W)&=\sum_{x\in\X}P(x)D^{q_x}(\omega_P^{\qv}(W)\|W_x).
\end{align*}
\end{rem}
\medskip

The following is a multi-variate generalization of \cite[Theorem 3.6]{MO-cqconv}:

\begin{prop}\label{prop:Um Hellinger}
Let $P\in\P_f^{\pm}(\X)$, $W\in\B(\X,\hil)\pne$, and for every $x\in\supp P$, let $D^{q_x}=\DU$. 
Assume that $0\ne S_+\le S_-$. 
Then there exist a unique 
$P$-weighted $\DU$-geometric mean 
$G_P^{\DU}(W):=\tau^{\Um}_P(W)$
of $W$ and a unique $P$-weighted $\DU$-center
$\omega_P^{\Um}$ for $W$, given by 
\begin{align}
G_P^{\DU}(W)=\tau^{\um}_P(W)
&=
S_+e^{\sum_{x\in\supp P}P(x)S_+(\nlog W_x)S_+},\\
\omega^{\um}_P(W)
&=
\frac{G_P^{\DU}(W)}{\Tr G_P^{\DU}(W)}=
\frac{S_+e^{\sum_{x\in\supp P}P(x)S_+(\nlog W_x)S_+}}{\Tr S_+e^{\sum_{x\in\supp P}P(x)S_+(\nlog W_x)S_+}}\s,\label{Umegaki Hellinger}
\end{align}
respectively, and 
\begin{align*}
-\log Q_P^{\bary,\Um}(W)=-\log\Tr G_P^{\DU}(W)=\sum_{x\in\X}\DU(\omega^{\um}_P(W)\|W_x).
\end{align*}
\end{prop}
\begin{proof}
Note that for $\sigma:=S_+e^{\sum_{x\in\supp P}P(x)S_+(\nlog W_x)S_+}$ and any 
$\tau\in\B(S_+\hil)\p$, we have 
\begin{align}
\DU(\tau\|\sigma)
&=
\Tr\tau\log\tau-\Tr\tau\nlog\bz S_+e^{\sum_{x\in\supp P}P(x)S_+(\nlog W_x)S_+}\jz
\nn\\
&=
\Tr\tau\log\tau-\Tr\tau\sum_{x\in\supp P}P(x)\nlog W_x
\nn\\
&=
\sum_{x\in\supp P}P(x)\underbrace{\bz\Tr\tau\log\tau-\Tr\tau\nlog W_x\jz}_{=\DU(\tau\|W_x)}.
\label{DU center proof1}
\end{align}
Thus,
\begin{align}
Q_P^{\bary,\Um}(W)
&=
\sup_{\tau\in\B(S_+\hil)\p}\left\{\Tr\tau-\sum_{x\in\supp P}P(x) \DU(\tau\|W_x)\right\}
\nn\\
&=
\max_{\tau\in\B(S_+\hil)\p}\{\Tr\tau-\DU(\tau\|\sigma)\}
\label{DU center proof2}\\
&=
\Tr\sigma,\nn
\end{align}
where the first equality is by definition, the second equality is by \eqref{DU center proof1}, and
last equality is due to \eqref{Tropp variational}. 
Moreover, since $\DU$ is strictly trace monotone (see, e.g., 
\cite[Proposition A.4]{HMPB}), Remark \ref{rem:relentr pos} yields that 
$\tau=\sigma$ is the unique state attaining the maximum 
in \eqref{DU center proof2}. This proves the assertion about the 
$P$-weighted $\DU$-geometric mean, and the rest of the assertions follow from this according to Remark \ref{rem:mean divcenter relations}.
\end{proof}

\begin{cor}
Let $D^{q_x}=\DU$, $x\in\X$, and let $W\in\B(\X,\hil)\pne$ be such that 
$\wedge_{x\in\X_0}W_x^0\ne 0$ for any finite subset $W_0\subseteq W$. Then 
the $\DU$-Hellinger body for $W$ exists. 
\end{cor}

\begin{rem}\label{rem:Um geomeric mean}
Note that for $P\in\P_f(\X)$, 
\begin{align*}
G_P^{\DU}(W)=\what G_{P,+\infty}(W),
\end{align*}
where the latter was defined in 
\eqref{multi log mean}.
\end{rem}

\begin{rem}\label{rem:Umebary=a-infty}
Note that in the $2$-variable case, 
\begin{align*}
D_{\alpha}^{\bary,\Um}(\rho\|\sigma)=D_{\alpha,+\infty}(\rho\|\sigma),\ds\ds\ds
\alpha\in(0,1)\cup(1,+\infty),
\end{align*}
where the latter was given in \eqref{D alpha infty}. This was already proved in
\cite[Theorem 3.6]{MO-cqconv}.
\end{rem}

\subsection{Homogeneity and scaling}

Note that the normalized relative entropies
$D^{q_0}_1$ and $D^{\qo}_1$ 
satisfy the scaling property \eqref{Renyi scaling} by assumption. This property is inherited by all the 
corresponding barycentric R\'enyi divergences $D_{\alpha}^{\qv}$. 
More generally, we have the following:

\begin{lemma}\label{lemma:barycentric sclaing}
For any $P\in\P_f^{\pm}(\X)$, 
any gcq channel $W\in\B(\X,\hil)\pne$ and any $t\in(0,+\infty)^{\X}$, 
\begin{align}
Q_P^{\bary,\qv}\bz(t_xW_x)_{x\in\X}\jz&=\bz\prod_{x\in\supp P}t_x^{P(x)}\jz\QP{W},
\label{multiQ scaling}\\
-\log Q_P^{\bary,\qv}\bz(t_xW_x)_{x\in\X}\jz
&=
-\log Q_P^{\bary,\qv}(W)-\sum_xP(x)\log t_x.
\label{multipsi scaling}
\end{align}
In particular, $Q_{\alpha}^{\bary,\qv}$ is homogeneous.
\end{lemma}
\begin{proof}
\eqref{multipsi scaling} is 
straightforward to verify from 
\eqref{psi optimization} and the scaling law \eqref{scaling2},
and \eqref{multiQ scaling} follows immediately from it.
\end{proof}

\begin{cor}\label{cor:barycentric sclaing}
The barycentric R\'enyi divergences satisfy the scaling law \eqref{Renyi scaling}, i.e., 
\begin{align}\label{barycentric scaling}
\D^{\bary,\qv}_{\alpha}(t\rho\|s\sigma)
=
\D^{\bary,\qv}_{\alpha}(\rho\|\sigma)+\log t-\log s,
\end{align}
for every 
$\rho,\sigma\in\B(\hil)\pne$, $t,s\in(0,+\infty)$, $\alpha\in[0,+\infty]$.
\end{cor}
\begin{proof}
Immediate from Lemma \ref{lemma:barycentric sclaing}, or alternatively, from
\eqref{Renyi bq4}.
\end{proof}

\subsection{Monotonicity in $\alpha$ and limiting values}
\label{sec:limits}

Monotonicity in the parameter $\alpha$ is a characteristic property 
of the classical R\'enyi divergences, which is inherited by the measured, the regularized measured, and the maximal R\'enyi divergences, and it also holds for the Petz-type R\'enyi divergences. 
The representations in Lemma \ref{lemma:psi rep} and Remark \ref{rem:Dalpha rep} 
show that barycentric R\'enyi divergences have the same monotonicity property.

\begin{prop}\label{prop:monconvex}
(i) For any $W\in\B(\X,\hil)\pne$, the maps 
\begin{align*}
P\mapsto Q_P^{\bary,\qv}(W)\ds\ds\text{and}\ds\ds
P\mapsto \log Q_P^{\bary,\qv}(W)
\end{align*}
are convex, and 
\begin{align*}
P \mapsto R_{\D^{\qv},\bal}(W,P)
\end{align*}
is concave, on $\P_f^{\pm}(\X)$. 

(ii)
For any fixed $\rho,\sigma\in\B(\hil)\pne$, 
\begin{align*}
&\alpha\mapsto\log Q_{\alpha}^{\bary,\qv}(\rho\|\sigma)\ds\ds\text{is convex on }[0,+\infty),\ds\text{and}\\
&\alpha\mapsto\D_{\alpha}^{\bary,\qv}(\rho\|\sigma)\ds\ds\text{is monotone increasing on }
[0,1)\cup(1,+\infty).
\end{align*}
\end{prop}
\begin{proof}
(i)
By the definition in \eqref{left divrad}, 
$P \mapsto R_{\D^{\qv},\bal}(W,P)$ is the infimum of affine functions in $P$,
and hence it is concave. By \eqref{psi optimization}, this implies the convexity of 
$P\mapsto \log Q_P^{\bary,\qv}(W)$, from which the convexity of 
$P\mapsto Q_P^{\bary,\qv}(W)$ follows immediately.

(ii)
By \eqref{psi alpha}, $\alpha\mapsto\log Q_{\alpha}^{\bary,\qv}(\rho\|\sigma)$ 
is the supremum of affine functions, and hence convex.
(The convexity on $(0,1)$ also follows as a special case of the above.)
The monotonicity of $\alpha\mapsto\D_{\alpha}^{\bary,\qv}(\rho\|\sigma)$ 
follows from this convexity by \eqref{barycentric Dalpha def2}.
\end{proof}

Let us introduce the following limiting values:
\begin{align}
\D^{\bary,\qv}_{1}(\rho\|\sigma)&:=
\sup_{\alpha\in(0,1)}\D^{\bary,\qv}_{\alpha}(\rho\|\sigma)=
\lim_{\alpha\nearrow 1}\D^{\bary,\qv}_{\alpha}(\rho\|\sigma),
\label{Db 1 def}\\
\D^{\bary,\qv}_{1^+}(\rho\|\sigma)&:=
\inf_{\alpha>1}\D^{\bary,\qv}_{\alpha}(\rho\|\sigma)=
\lim_{\alpha\searrow 1}\D^{\bary,\qv}_{\alpha}(\rho\|\sigma),\\
\D^{\bary,q}_{\infty}(\rho\|\sigma)&:=
\sup_{\alpha>1}\D^{\bary,q}_{\alpha}(\rho\|\sigma)=
\lim_{\alpha\nearrow +\infty}\D^{\bary,q}_{\alpha}(\rho\|\sigma),
\end{align}
where the equalities follow from the monotonicity established in Proposition \ref{prop:monconvex}. Using the representations in Remark \ref{rem:Dalpha rep}, it is easy to show the following:

\begin{prop}\label{prop:limits}
Let $\rho,\sigma\in\B(\hil)\pne$.
\smallskip

\noindent (i)
We have
\begin{align}
\D^{\bary,\qv}_{0}(\rho\|\sigma)&=
\inf_{\alpha\in(0,1)}\D^{\bary,\qv}_{\alpha}(\rho\|\sigma)=
\lim_{\alpha\searrow 0}\D^{\bary,\qv}_{\alpha}(\rho\|\sigma)\label{Renyi bq limit00}\\
&=
\inf_{\omega\in\S(\rho^0\hil)}\left\{\D^{\qo}\bz\omega\Big\|\frac{\sigma}{\Tr\sigma}\jz\right\}
+\log\Tr\rho-\log\Tr\sigma,
\label{Renyi bq limit0}\\
\D^{\bary,\qv}_{\infty}(\rho\|\sigma)
&=
\sup_{\omega\in\S(\rho^0\hil)}\left\{\D^{\qo}\bz\omega\Big\|\frac{\sigma}{\Tr\sigma}\jz
-\D^{\qn}\bz\omega\Big\|\frac{\rho}{\Tr\rho}\jz\right\}
+\log\Tr\rho-\log\Tr\sigma\label{Renyi bq limit infty1}\\
&=
\sup_{\omega\in\S(\rho^0\hil)}\{\D^{\qo}(\omega\|\sigma)-\D^{\qn}(\omega\|\rho)\}.
\label{Renyi bq limit infty2}
\end{align}

\noindent (ii)
If $\D^{\qn}$ is strictly positive and $\D^{\qn}\big(\valt\big\|\frac{\rho}{\Tr\rho}\big)$ 
and $D^{\qo}\bz\valt\|\sigma\jz$ are both
lower semi-continuous on $\S(\rho^0\hil)$
then 
\begin{align}\label{Renyi bq limit11}
\D^{\bary,\qv}_{1}(\rho\|\sigma)=\frac{1}{\Tr\rho}\D^{\qo}(\rho\|\sigma).
\end{align}
If, moreover, $\D^{\qo}(\valt\|\sigma)$ is
upper semi-continuous on $\S(\rho^0\hil)$ then 
\begin{align}
\D^{\bary,\qv}_{1}=\D^{\bary,\qv}_{1^+}(\rho\|\sigma).
\label{Renyi bq limit12}
\end{align}
\end{prop}
\begin{proof}
(i)
The second equality in \eqref{Renyi bq limit00} is obvious from the monotonicity established in Proposition \ref{prop:monconvex}, and the rest of the equalities in 
\eqref{Renyi bq limit00}--\eqref{Renyi bq limit0}
follow as
\begin{align*}
\inf_{\alpha\in(0,1)}\D^{\bary,\qv}_{\alpha}(\rho\|\sigma)
&=
\inf_{\alpha\in(0,1)}
\inf_{\omega\in\S(\rho^0\hil)}\left\{\frac{\alpha}{1-\alpha}\D^{\qn}\bz\omega\Big\|\frac{\rho}{\Tr\rho}\jz+\D^{\qo}\bz\omega\Big\|\frac{\sigma}{\Tr\sigma}\jz\right\}
+\log\Tr\rho-\log\Tr\sigma\\
&=
\inf_{\omega\in\S(\rho^0\hil)}\inf_{\alpha\in(0,1)}\left\{\frac{\alpha}{1-\alpha}
\D^{\qn}\bz\omega\Big\|\frac{\rho}{\Tr\rho}\jz+\D^{\qo}\bz\omega\Big\|\frac{\sigma}{\Tr\sigma}\jz\right\}
+\log\Tr\rho-\log\Tr\sigma\\
&=\inf_{\omega\in\S(\rho^0\hil)}\left\{\D^{\qo}\bz\omega\Big\|\frac{\sigma}{\Tr\sigma}\jz\right\}
+\log\Tr\rho-\log\Tr\sigma\\
&=\inf_{\omega\in\S(\rho^0\hil)}\left\{\D^{\qo}\bz\omega\Big\|\sigma\jz\right\}
+\log\Tr\rho\\
&=
\D^{\bary,\qv}_{0}(\rho\|\sigma),
\end{align*}
where the first equality is due to \eqref{Renyi bq4}, the second equality is trivial,
the third equality follows from the non-negativity of $\D^{\qn}$, 
the fourth equality follows from the scaling law \eqref{scaling2}, 
and the last equality follows from \eqref{psi alpha}.

The equalities in \eqref{Renyi bq limit infty1}--\eqref{Renyi bq limit infty2} follow as
\begin{align*}
\sup_{\alpha>1}\D^{\bary,\qv}_{\alpha}(\rho\|\sigma)
&=
\sup_{\alpha>1}
\sup_{\omega\in\S(\rho^0\hil)}\left\{\frac{\alpha}{1-\alpha}\D^{\qn}\bz\omega\Big\|\frac{\rho}{\Tr\rho}\jz+\D^{\qo}\bz\omega\Big\|\frac{\sigma}{\Tr\sigma}\jz\right\}
+\log\Tr\rho-\log\Tr\sigma\\
&=
\sup_{\omega\in\S(\rho^0\hil)}\sup_{\alpha>1}\left\{\frac{\alpha}{1-\alpha}\D^{\qn}\bz\omega\Big\|\frac{\rho}{\Tr\rho}\jz+\D^{\qo}\bz\omega\Big\|\frac{\sigma}{\Tr\sigma}\jz\right\}
+\log\Tr\rho-\log\Tr\sigma\\
&=\sup_{\omega\in\S(\rho^0\hil)}\left\{\D^{\qo}\bz\omega\Big\|\frac{\sigma}{\Tr\sigma}\jz
-\D^{\qn}\bz\omega\Big\|\frac{\rho}{\Tr\rho}\jz\right\}
+\log\Tr\rho-\log\Tr\sigma\\
&=
\sup_{\omega\in\S(\rho^0\hil)}\left\{\D^{\qo}\bz\omega\|\sigma\jz
-\D^{\qn}\bz\omega\|\rho\jz\right\},
\end{align*}
where the first equality is due to 
\eqref{Renyi bq4}, the second one is trivial, the third one follows from the non-negativity 
of $\D^{\qn}\bz\omega\Big\|\frac{\rho}{\Tr\rho}\jz$, and the last equality is due to 
the scaling property \eqref{scaling2}.

(ii) The equality in \eqref{Renyi bq limit11} follows as
\begin{align*}
\sup_{\alpha\in(0,1)}\D^{\bary,\qv}_{\alpha}(\rho\|\sigma)
&=
\sup_{\alpha\in(0,1)}
\inf_{\omega\in\S(\rho^0\hil)}\left\{\frac{\alpha}{1-\alpha}\D^{\qn}\bz\omega\Big\|\frac{\rho}{\Tr\rho}\jz+\D^{\qo}\bz\omega\|\sigma\jz\right\}
+\log\Tr\rho\\
&=
\inf_{\omega\in\S(\rho^0\hil)}\sup_{\alpha\in(0,1)}\left\{\frac{\alpha}{1-\alpha}\D^{\qn}\bz\omega\Big\|\frac{\rho}{\Tr\rho}\jz+\D^{\qo}\bz\omega\|\sigma\jz\right\}
+\log\Tr\rho\\
&=\D^{\qo}\bz\frac{\rho}{\Tr\rho}\Big\|\sigma\jz+\log\Tr\rho
=
\frac{1}{\Tr\rho}\D^{\qo}(\rho\|\sigma),
\end{align*}
where the first equality is due to \eqref{Renyi bq3}, and the second one follows from the minimax theorem in Lemma \ref{lemma:minimax2}, using the fact that $\alpha\mapsto \frac{\alpha}{1-\alpha}$ is monotone increasing on $(0,1)$.
The third equality follows from the fact that $\sup_{\alpha\in(0,1)}\frac{\alpha}{1-\alpha}=+\infty$, and that 
$\D^{\qn}$ is strictly positive, and the last equality is immediate from the scaling law \eqref{scaling1}.

Finally, to prove \eqref{Renyi bq limit12}, first note that if 
$\rho^0\nleq\sigma^0$ then $\D^{\bary,\qv}_{\alpha}(\rho\|\sigma)=+\infty$ for every $\alpha>1$; 
indeed, by \eqref{Renyi bq2},
\begin{align}\label{bary infty for a>1}
\D^{\bary,\qv}_{\alpha}(\rho\|\sigma)\ge
\frac{\alpha}{1-\alpha}\underbrace{\D^{\qn}(\rho\|\rho)}_{=0}
+\underbrace{\D^{\qo}(\rho\|\sigma)}_{=+\infty}
-\frac{1}{\alpha-1}\log\Tr\rho=+\infty.
\end{align}
(See also Corollary \ref{cor:Renyi bq infty}.)
Hence,
\begin{align*}
\D^{\bary,\qv}_{1^+}(\rho\|\sigma)=\inf_{\alpha>1}\D^{\bary,\qv}_{\alpha}(\rho\|\sigma)=
\inf_{\alpha>1}+\infty=+\infty=\frac{1}{\Tr\rho}\D^{\qo}(\rho\|\sigma)
=\D^{\bary,\qv}_{1}(\rho\|\sigma),
\end{align*}
where the fourth equality is due to the support condition \eqref{supp const}, and the 
last equality is \eqref{Renyi bq limit11}.
Hence, for the rest we assume that $\rho^0\le\sigma^0$, so that 
$\D^{\qn}\big(\omega\big\|\frac{\rho}{\Tr\rho}\big)$ and 
$\D^{\qo}(\omega\|\sigma)$ are both 
finite for $\omega\in\S(\rho^0\hil)$.
Moreover, by assumption, 
$\omega\mapsto\frac{\alpha}{1-\alpha}\D^{\qn}\big(\omega\big\|\frac{\rho}{\Tr\rho}\big)
+\D^{\qo}\bz\omega\|\sigma\jz$ is upper semi-continous on 
$\S(\rho^0\hil)$ for every $\alpha>1$. 
Then 
\begin{align*}
\D^{\bary,\qv}_{1^+}(\rho\|\sigma)&=
\inf_{\alpha>1}\D^{\bary,\qv}_{\alpha}(\rho\|\sigma)\\
&=
\inf_{\alpha>1}
\sup_{\omega\in\S(\rho^0\hil)}\left\{\frac{\alpha}{1-\alpha}\D^{\qn}\bz\omega\Big\|\frac{\rho}{\Tr\rho}\jz+\D^{\qo}\bz\omega\|\sigma\jz\right\}
+\log\Tr\rho\\
&=
\sup_{\omega\in\S(\rho^0\hil)}\inf_{\alpha>1}
\left\{\frac{\alpha}{1-\alpha}\D^{\qn}\bz\omega\Big\|\frac{\rho}{\Tr\rho}\jz
+\D^{\qo}\bz\omega\|\sigma\jz\right\}
+\log\Tr\rho\\
&=\D^{\qo}\bz\frac{\rho}{\Tr\rho}\Big\|\sigma\jz+\log\Tr\rho
=
\frac{1}{\Tr\rho}\D^{\qo}(\rho\|\sigma)=\D^{\bary,\qv}_{1}(\rho\|\sigma),
\end{align*}
where the second equality is due to \eqref{Renyi bq3}, the third one follows from the 
minimax theorem in Lemma \ref{lemma:minimax2},
the fourth equality follows from the strict positivity of $\D^{\qn}$ and the fact that 
$\inf_{\alpha>1}\alpha/(1-\alpha)=-\infty$, the fifth one is due to the scaling law
\eqref{scaling1}, and the last equality is \eqref{Renyi bq limit11}.
\end{proof}

\begin{rem}\label{rem:notqRenyi}
Clearly, for any $\rho\in\B(\hil)\pne$ and any 
quantum R\'enyi $\alpha$-divergence $D_{\alpha}^q$, 
$D_{\alpha}^q(\rho\|\rho)=0$.
On the other hand, if $D^{\qn}$ and $D^{\qo}$ are such that 
for any two invertible non-commuting $\omega_1,\omega_2\in\B(\hil)\pp$,
$D^{\qo}(\omega_1\|\omega_2)>D^{\qn}(\omega_1\|\omega_2)$, then 
by \eqref{Renyi bq limit infty2},
$D_{\infty}^{\bary,\qv}(\rho\|\rho)>0$, $\rho\in\B(\hil)\pp$, and hence
$D_{\alpha}^{\bary,\qv}$ is not a quantum R\'enyi $\alpha$-divergence
for any large enough $\alpha$. For instance, this is the case if
$D^{\qn}=\DU$ and $D^{\qo}=D^{\BS}$;
see, e.g., \cite[Proposition 4.7.]{HiaiMosonyi2017}.

We leave open the question whether in the above setting for every $\alpha>1$ there exists a $\rho\in\B(\hil)\p$ such that $D_{\infty}^{\bary,\qv}(\rho\|\rho)>0$.
\end{rem}

\begin{rem}\label{rem:dual}
It is well known and easy to verify that for commuting states $\rho$ and $\sigma$, 
the unique R\'enyi $\alpha$-divergences satisfy
\begin{align*}
Q_{\alpha}(\rho\|\sigma)=Q
_{1-\alpha}(\sigma\|\rho),
\ds\ds\ds\alpha\in(0,1).
\end{align*}
As a consequence, if 
$D_{1-\alpha}^q$ is a quantum R\'enyi $(1-\alpha)$-divergence for some 
$\alpha\in(0,1)$ then
\begin{align*}
\tilde D_{\alpha}^q(\rho\|\sigma)
&:=
\frac{1}{1-\alpha}\left[
\alpha D_{1-\alpha}^q(\sigma\|\rho)+\log\Tr\rho-\log\Tr\sigma
\right]
\end{align*}
defines a quantum R\'enyi $\alpha$-divergence. Given a collection 
$(D_{\alpha}^q)_{\alpha\in(0,1)}$ of quantum R\'enyi $\alpha$-divergences, 
we call $(\tilde D_{\alpha}^q)_{\alpha\in(0,1)}$ the \ki{dual collection}.
The measured, the regularized measured and the maximal R\'enyi divergences are easily seen to be self-dual, 
as are the 
Petz-type (or standard) R\'enyi divergences 
(see Section \ref{sec:qRenyi} for the definitions).

For the barycentric R\'enyi divergences, it is straightforward to verify from 
\eqref{Renyi bq6} that
\begin{align*}
\tilde\D^{\bary,(\qn,\qo)}_{\alpha}(\rho\|\sigma)
=
\D^{\bary,(\qo,\qn)}_{\alpha}(\rho\|\sigma). 
\end{align*}
In particular, $(\D_{\alpha}^{\bary,(\qn,\qo)})_{\alpha\in(0,1)}$ is self-dual when $D^{\qn}=D^{\qo}$. 

Combining this duality with \eqref{Renyi bq limit11} we obtain
that if $D^{\qo}$ is strictly positive and 
$\D^{\qo}\big(\valt\big\|\frac{\rho}{\Tr\rho}\big)$ 
and $D^{\qn}\bz\valt\|\sigma\jz$ are both
lower semi-continuous on $\S(\rho^0\hil)$
then
\begin{align}
\frac{1}{\Tr\rho}D^{\qn}(\rho\|\sigma)
&=
\lim_{\alpha\nearrow 1} D_{\alpha}^{\bary,(\qo,\qn)}(\rho\|\sigma)
=
\lim_{\alpha\nearrow 1}\tilde D_{\alpha}^{\bary,(\qn,\qo)}(\rho\|\sigma)\nn\\
&=
\lim_{\alpha\searrow 0}
\frac{1}{\alpha}\left[
(1-\alpha) D_{\alpha}^{\bary,(\qn,\qo)}(\sigma\|\rho)+\log\Tr\rho-\log\Tr\sigma
\right].\label{lim at 0}
\end{align}
Due to this and Proposition \ref{prop:limits}, 
if both $D^{\qn}$ and $D^{\qo}$ are strictly positive and 
lower semi-continuous in their first variable, then 
they can both be recovered from $(D_{\alpha}^{\bary,\qv})_{\alpha\in(0,1)}$ by taking limits at $\alpha\searrow 0$ and at $\alpha\nearrow 1$, respectively.
In particular, if $(D^{\qn},D^{\qo})\ne (D^{\tilde q_0},D^{\tilde q_1})$
then $D_{\alpha}^{\bary,(\qn,\qo)}\ne D_{\alpha}^{\bary,(\tilde q_0,\tilde q_1)}$
for $\alpha$ close enough to $0$ from above or to $1$ from below.
\end{rem}

\subsection{Non-negativity and finiteness}

Here we show that the barycentric R\'enyi divergences are pseudo-distances in the sense that 
$D_{\alpha}^{\bary,\qv}$ is non-negative for every $\alpha\in[0,+\infty]$, and 
it is strictly positive for every $\alpha\in(0,+\infty]$ under some mild conditions on 
$D^{\qn}$ and $D^{\qo}$. 

We start more generally with giving bounds on the multi-variate R\'enyi quantities in Definition \ref{defin:weighted Q}, for which the following easy observation will be useful.

\begin{lemma}
Let $D^q$ be a quantum relative entropy. 
If $D^q$ is trace-monotone then for any $\sigma\in\B(\hil)\pne$ and 
any state $\omega\in\S(\hil)$, 
\begin{align}\label{relentr lower bound}
-\log\Tr\sigma\le D^q(\omega\|\sigma).
\end{align}
If $D^q$ is anti-monotone in its second argument then 
\begin{align}\label{relentr bounds}
D^q(\omega\|\sigma)\begin{cases}
\le -\log\lambda_{\min}(\sigma)
+\Tr\omega\log\omega\le -\log\lambda_{\min}(\sigma),&\omega^0\le\sigma^0,\\
=+\infty,&\text{otherwise,}
\end{cases}
\end{align}
where $\lambda_{\min}(\sigma)$ denotes the smallest non-zero eigenvalue of $\sigma$.
\end{lemma}
\begin{proof}
The inequality in \eqref{relentr lower bound} is immediate from the trace monotonicity of 
$D^q$, and $\omega^0\nleq\sigma^0\imp D^q(\omega\|\sigma)=+\infty$ 
in \eqref{relentr bounds}
follows from the support condition \eqref{supp const}. Hence, we are left to prove
the upper bound in \eqref{relentr bounds}.

By assumption, $\D^q$ is anti-monotone in its second argument, whence
$\sigma\ge\lambda_{\min}(\sigma)\sigma^0$ implies that 
\begin{align*}
\D^q(\omega\|\sigma)
&\le\D^q(\omega\|\lambda_{\min}(\sigma)\sigma^0)\\
&=
-\log\lambda_{\min}(\sigma)+\underbrace{\D^q(\omega\|\sigma^0)}_{=D(\omega\|\sigma^0)}\\
&=
-\log\lambda_{\min}(\sigma)+
\underbrace{D(\omega\|\sigma^0)}_{=\Tr\omega\log\omega}\\
&=
-\log\lambda_{\min}(\sigma)
+\Tr\omega\log\omega,
\end{align*}
where $D$ is the unique extension of the classical relative entropy to commuting pairs of operators (see Lemma \ref{lemma:commuting extension}),
the first equality follows from the scaling property 
\eqref{scaling2}, the second equality from the fact that 
$\omega$ and $\sigma^0$ commute due to the assumption that $\omega^0\le\sigma^0$,
and the last equality is by the definition of $D$.
\end{proof}

\begin{prop}\label{prop:multi psi bounds}
(i) Let $P\in\P_f(\X)$. Then, for any $W\in\B(\X,\hil)\pne$, 
\begin{align}\label{psi lower bound}
-\sum_xP(x)\log\Tr\Wx{x}
\le
-\psi_P^{\bary,\qv}(W),
\end{align}
and if each $D^{q_x}$, $x\in\supp P$, is anti-monotone in its second argument then 
\begin{align}\label{psi upper bounds}
-\psi_P^{\bary,\qv}(W)
\begin{cases}
\le\sum_x P(x)(-\log\lambda_{\min}(\Wx{x}))<+\infty,&S_+\ne 0,\\
=+\infty,&\text{otherwise.}
\end{cases}
\end{align}
In particular, 
\begin{align}\label{multipsi infty}
-\psi_P^{\bary,\qv}(W)=+\infty\ds\iff\ds S_+=0\ds\iff\ds
Q_P^{\bary,\qv}(W)=0.
\end{align}

(ii) Assume that $P\in\P_f^{\pm}(\X)$ is such that $P(x)<0$ for some $x\in\X$. 
If for each $x$ such that $P(x)>0$, $D^{q_x}$ is anti-monotone in its second argument  
then 
\begin{align}\label{psi lower bound2}
\sum_{x:\,P(x)>0}P(x)\log\lambda_{\min}(\Wx{x})
+\sum_{x:\,P(x)<0}P(x)\log\Tr\Wx{x}
\le
\psi_P^{\bary,\qv}(W).
\end{align}
If for each $x$ such that $P(x)<0$, $D^{q_x}$ is anti-monotone in its second argument  
then 
\begin{align}\label{psi upper bounds2}
\psi_P^{\bary,\qv}(W)
\begin{cases}
\le
\sum_{x:\,P(x)>0} P(x)\log\Tr \Wx{x}
+
\sum_{x:\,P(x)<0} P(x)\log\lambda_{\min}(\Wx{x}),&S_+\le S_-,\\
=+\infty,&\text{otherwise;}
\end{cases}
\end{align}
in particular, 
\begin{align}\label{multipsi infty2}
\psi_P^{\bary,\qv}(W)=+\infty\ds\iff\ds S_+\nleq S_-\ds\iff\ds
Q_P^{\bary,\qv}(W)=+\infty.
\end{align}
\end{prop}
\begin{proof}
The inequalities in \eqref{psi lower bound}, \eqref{psi upper bounds},
\eqref{psi lower bound2}, and \eqref{psi upper bounds2}
are obvious from \eqref{psi optimization}, \eqref{relentr lower bound}, and
\eqref{relentr bounds}, and \eqref{multipsi infty} 
and \eqref{multipsi infty2} follow immediately. 
\end{proof}

As a special case, we get the following characterization of the finiteness of the 
$2$-variable barycentric R\'enyi divergences. This gives an easy way to show that certain 
quantum R\'enyi divergences cannot be represented as barycentric R\'enyi divergences, as we show in Section \ref{sec:finiteness}.

\begin{cor}\label{cor:Renyi bq infty}
We have
\begin{align}\label{Renyi bq infty}
\D^{\bary,\qv}_{\alpha}(\rho\|\sigma)=+\infty
\begin{cases}
\iff\rho^0\wedge\sigma^0=0,&\text{when }\alpha\in[0,1),\\
\pmi\rho^0\nleq\sigma^0,&\text{when }\alpha>1.
\end{cases}
\end{align}
If $\D^{\qo}$ is anti-monotone in its second argument then the one-sided implication above is also an equivalence.
\end{cor}
\begin{proof}
The case $\alpha\in(0,+\infty)$ is immediate from Proposition \ref{prop:multi psi bounds}.
The case $\alpha=0$ follows 
immediately from \eqref{Renyi bq4}.
\end{proof}

Finally, we turn to the strict positivity of the $2$-variable barycentric R\'enyi divergences.
\renewcommand{\theenumi}{(\alph{enumi})}

\begin{prop}\label{prop:barycentric nonneg}
Let $\rho,\sigma\in\B(\hil)\pne$.
\smallskip

\noindent(i) For every $\alpha\in[0,+\infty]$, 
\begin{align}
\D^{\bary,\qv}_{\alpha}(\rho\|\sigma)\ge \log\Tr\rho-\log\Tr\sigma.\label{Renyi bq strictpos}
\end{align}

\noindent(ii) If $\sigma^0\le\rho^0$ then 
\eqref{Renyi bq strictpos} holds with equality for $\alpha=0$. 
If $D^{\qo}$ is strictly positive and 
$D^{\qo}(\valt\|\sigma)$ is lower semi-continuous on $\S(\rho^0\hil)$ 
then $\sigma^0\le\rho^0$
is also necessary for equality in \eqref{Renyi bq strictpos} for $\alpha=0$.
\smallskip

\noindent(iii)  
We have \ref{lower bound sat1}$\imp$\ref{lower bound sat2}$\pmi$\ref{lower bound sat3}$\pmi$\ref{lower bound sat4} and \ref{lower bound sat1}$\imp$\ref{lower bound sat3}
in the following:
\begin{enumerate}
\item\label{lower bound sat1}
Equality holds in \eqref{Renyi bq strictpos} for every $\alpha\in[0,+\infty]$.

\item\label{lower bound sat2}
Equality holds in \eqref{Renyi bq strictpos} for some $\alpha\in(0,+\infty]$.

\item\label{lower bound sat3}
Equality holds in \eqref{Renyi bq strictpos} for every $\alpha\in[0,1]$.

\item\label{lower bound sat4}
$\rho/\Tr\rho=\sigma/\Tr\sigma$.
\end{enumerate}

If $D^{\qn}$ and $D^{\qo}$ are strictly positive, and
$D^{\qn}(\valt\|\rho/\Tr\rho)$ and $D^{\qo}(\valt\|\sigma/\Tr\sigma)$ are lower semi-continuous on 
$\S(\rho^0\hil)$, then we have
\ref{lower bound sat1}$\imp$\ref{lower bound sat2}$\iff$\ref{lower bound sat3}$\iff$\ref{lower bound sat4}.

On the other hand, if $D^{\qn}=D^{\qo}$ then the implication 
\ref{lower bound sat4}$\imp$\ref{lower bound sat1} holds.
\end{prop}

\begin{proof}
(i) The inequality in \eqref{Renyi bq strictpos} for $\alpha\in[0,1)$ is immediate from 
\eqref{Renyi bq4} and the non-negativity of $D^{\qn}$ and $D^{\qo}$. 
From this 
\eqref{Renyi bq strictpos} follows also for $\alpha\in[1,+\infty]$, using the monotonicity 
in Proposition \ref{prop:monconvex}. 

(ii) Let $\alpha=0$. If $\sigma^0\le\rho^0$ then 
\eqref{Renyi bq strictpos}
holds with equality, according to \eqref{barycentric 0-Renyi}. Assume now that $D^{\qo}(\valt\|\sigma)$ is lower semi-continuous on 
$\S(\rho^0\hil)$, so that there exists an $\omega_0\in\S(\rho^0\hil)$ such that 
$\inf_{\omega\in\S(\rho^0\hil)}\D^{\qo}\big(\omega\big\|\frac{\sigma}{\Tr\sigma}\big)=\D^{\qo}\big(\omega_0\big\|\frac{\sigma}{\Tr\sigma}\big)$.
Now, if $D^{\qo}$ is strictly positive then, by
\eqref{Renyi bq limit0}, 
$D_0^{\bary,\qv}(\rho\|\sigma)>\log\Tr\rho-\log\Tr\sigma$ unless
$\omega_0=\sigma/\Tr\sigma$, which in turn implies that $\sigma^0=\omega_0^0\le\rho^0$. 
\smallskip

(iii)
The implications 
\ref{lower bound sat1}$\imp$\ref{lower bound sat2}$\pmi$\ref{lower bound sat3} 
and \ref{lower bound sat1}$\imp$\ref{lower bound sat3} are obvious. If 
$\rho/\Tr\rho=\sigma/\Tr\sigma$ then choosing 
$\omega:=\rho/\Tr\rho$ yields that the infimum in 
\eqref{Renyi bq4} is zero, and hence equality holds in \eqref{Renyi bq strictpos}, for every $\alpha\in[0,1)$, and also for $\alpha=1$, by taking the limit.
This proves \ref{lower bound sat4}$\imp$\ref{lower bound sat3}.

Next, we prove \ref{lower bound sat2}$\imp$\ref{lower bound sat4}
under the assumption that 
$D^{\qn}$ and $D^{\qo}$ are strictly positive, and
$D^{\qn}(\valt\|\rho/\Tr\rho)$ and $D^{\qo}(\valt\|\sigma/\Tr\sigma)$ are lower semi-continuous on $\S(\rho^0\hil)$.
Note that equality in \eqref{Renyi bq strictpos} for some 
$\alpha\in(0,+\infty]$ implies equality 
for every $\beta\in[0,\alpha]$, according to (i) and 
the monotonicity in 
Proposition \ref{prop:monconvex}.
In particular, the 
infimum in \eqref{Renyi bq4} is zero. By the semi-continuity assumptions,
there exists an $\omega_0$ where that infimum is attained.
Strict positivity of $D^{\qn}$ and $D^{\qo}$ then implies
$\omega_0=\rho/\Tr\rho$ and $\omega_0=\sigma/\Tr\sigma$, proving 
\ref{lower bound sat4}.

Finally, if $D^{\qn}=D^{\qo}$ and $\rho/\Tr\rho=\sigma/\Tr\sigma$ then 
\eqref{Renyi bq strictpos}, \eqref{Renyi bq limit infty1}, and the monotonicity 
established in Proposition \ref{prop:monconvex}, yield that for every $\alpha\in[0,+\infty]$,
\begin{align*}
\log\Tr\rho-\log\Tr\sigma\le 
D_{\alpha}^{\bary,(\qn,\qo)}(\rho\|\sigma)\le
D_{\infty}^{\bary,(\qn,\qo)}(\rho\|\sigma)=
\log\Tr\rho-\log\Tr\sigma,
\end{align*}
proving 
the implication 
\ref{lower bound sat4}$\imp$\ref{lower bound sat1}.
\end{proof}

\begin{rem}
Note that by Remark \ref{rem:notqRenyi},
the condition $D^{\qn}=D^{\qo}$ cannot be completely omitted for
the implication 
\ref{lower bound sat4}$\imp$\ref{lower bound sat1} in (iii) to hold.
\end{rem}

\renewcommand{\theenumi}{(\roman{enumi})}

Proposition \ref{prop:barycentric nonneg} yields immediately the following:

\begin{cor}
For every $\alpha\in[0,+\infty]$, $D_{\alpha}^{\bary,\qv}$ is non-negative. If 
$D^{\qn}$ and $D^{\qo}$ are both strictly positive and lower semi-continuous in their 
first arguments then $D_{\alpha}^{\bary,\qv}$ is strictly positive for every 
$\alpha\in(0,+\infty]$.
\end{cor}

\begin{rem}
Using the scaling property \eqref{barycentric scaling}, \eqref{Renyi bq strictpos} can be 
rewritten equivalently as
\begin{align*}
D_{\alpha}^{\bary,\qv}\bz\frac{\rho}{\Tr\rho}\Big\|\frac{\sigma}{\Tr\sigma}\jz
\ge 0,\ds\ds\ds\alpha\in[0,+\infty].
\end{align*}
\end{rem}

\begin{rem}
Note that \eqref{Renyi bq strictpos} tells exactly that the trace-monotonicity of 
$D^{\qn}$ and $D^{\qo}$ is inherited by all the generated barycentric R\'enyi divergences
$D_{\alpha}^{\bary,\qv}$, $\alpha\in[0,+\infty]$.

Since the barycentric R\'enyi divergences satisfy the scaling property \eqref{Renyi scaling} according to Corollary \ref{cor:barycentric sclaing}, 
their non-negativity is in fact equivalent to trace-monotonicity, as noted in
Remark \ref{rem:tracemon}.
\end{rem}

\subsection{Monotonicity under CPTP maps and joint convexity}

One of the most important properties of quantum divergences is monotonicity under quantum 
operations (i.e., CPTP maps).
Many of the important quantum divergences are monotone under more general trace-preserving
maps, e.g., dual Schwarz maps in the case of Petz-type R\'enyi divergences for 
$\alpha\in[0,2]$ \cite{P86}, or 
PTP maps in the case of the sandwiched R\'enyi divergences for $\alpha\ge 1/2$ 
\cite{Beigi,Jencova_NCLpII,MHR},
and the measured as well as the maximal R\'enyi divergences for $\alpha\in[0,+\infty]$, by definition.
It is easy to see that for $\alpha\in[0,1]$, 
the barycentric R\'enyi $\alpha$-divergences are monotone under the same class
of PTP maps as their generating quantum relative entropies. 
More generally, we have the following:

\begin{prop}\label{prop:weighted PTP-mon}
If all $D^{q_x}$, $x\in\X$, are monotone non-increasing
under a trace non-decreasing positive map $\map\in\pplus(\hil,\kil)$ then
$Q_P^{\bary,\qv}$ is 
monotone non-decreasing, and 
$R_{D^{\qv},\bal}$ is 
monotone non-increasing under $\map$, i.e., for every 
gcq channel $W\in\B(\X,\hil)\pne$ and every $P\in\P_f(\X)$, 
\begin{align}
Q^{\bary,\qv}_P\bz \map(W)\jz&\ge Q^{\bary,\qv}_P\bz W\jz,\label{Qmon}\\
R_{D^{\qv},\bal}\bz \map(W),P\jz
&\le
R_{D^{\qv},\bal}\bz W,P\jz.\label{Rmon}
\end{align}
\end{prop}
\begin{proof}
We have
\begin{align*}
Q^{\bary,\qv}_P\bz \map(W)\jz&=
\sup_{\tau\in\B(\kil)\p}\left\{
\Tr\tau-\sum_{x\in\X}P(x)D^{q_x}(\tau\|\map(W_x))\right\}\\
&\ge 
\sup_{\tilde\tau\in\B(\hil)\p}\Big\{
\underbrace{\Tr\map(\tilde\tau)}_{\ge\Tr\tilde\tau}-
\sum_{x\in\X}P(x)\underbrace{D^{q_x}(\map(\tilde\tau)\|\map(W_x))}_{
\le D^{q_x}(\tilde\tau\|W_x)}\Big\}\\
&\ge
\sup_{\tilde\tau\in\B(\hil)\p}\left\{
\Tr\tilde\tau-\sum_{x\in\X}P(x)D^{q_x}(\tilde\tau\|W_x)\right\}\\
&=
Q^{\bary,\qv}_P\bz W\jz,
\end{align*}
where the equalities are by definition \eqref{weighted Q} and by \eqref{barycentric Qalpha def support}, the first inequality is obvious, and the second one follows from the assumptions. 
This proves \eqref{Qmon}, and \eqref{Rmon} follows immediately by 
Lemma \ref{lemma:psi rep}.
\end{proof}

\begin{prop}\label{prop:PTP monotonicity}
If $D^{\qn}$ and $D^{\qo}$ are monotone 
under a trace non-decreasing map $\map\in\pplus(\hil,\kil)$ then
for every $\alpha\in[0,1]$, $Q_{\alpha}^{\bary,\qv}$ 
is monotone non-decreasing 
under $\map$, i.e., for every 
$\rho,\sigma\in\B(\hil)\pne$, 
\begin{align}
Q_{\alpha}^{\bary,\qv}(\map(\rho)\|\map(\sigma))
&\ge
Q_{\alpha}^{\bary,\qv}(\rho\|\sigma).\label{Qmon bary}
\end{align}
If, moreover, $\map$ is trace-preserving, then for every $\alpha\in[0,1]$, 
$D_{\alpha}^{\bary,\qv}$
is monotone non-increasing under $\map$, i.e., for every 
$\rho,\sigma\in\B(\hil)\pne$, 
\begin{align}\label{Dmon bary}
D_{\alpha}^{\bary,\qv}(\map(\rho)\|\map(\sigma))
&\le
D_{\alpha}^{\bary,\qv}(\rho\|\sigma).
\end{align}
Vice versa, if  $D^{\qn}$ and $D^{\qo}$ are strictly positive, and lower semi-continuous 
in their first variables, 
and $D_{\alpha}^{\bary,\qv}$, $\alpha\in(0,1)$, are monotone non-increasing 
under a trace-preserving positive map $\map\in\pplus(\hil,\kil)$, then 
$D^{\qn}$ and $D^{\qo}$ are monotone non-increasing under the same map.
\end{prop}
\begin{proof}
Proposition \ref{prop:weighted PTP-mon} yields \eqref{Qmon bary} 
as a special case when $\alpha\in(0,1]$, and the case $\alpha=0$ follows by a trivial modification of the proof. 
From this, \eqref{Dmon bary} follows immediately. 
The last assertion follows due to \eqref{Renyi bq limit11} and \eqref{lim at 0}.
\end{proof}

\begin{rem}
By assumption, both $D^{\qn}$ and $D^{\qo}$ are trace-monotone, and hence so are 
$D_{\alpha}^{\bary,\qv}$, $\alpha\in[0,1]$, according to 
Proposition \ref{prop:PTP monotonicity} above. This gives an alternative proof of \eqref{Renyi bq strictpos}.
\end{rem}

\begin{rem}
The above proof for Proposition \ref{prop:weighted PTP-mon} only works when 
$P$ is a probability measure (i.e., there is no $x$ such that $P(x)<0$), which translates to 
$\alpha\in[0,1]$ in Proposition \ref{prop:PTP monotonicity}.
These conditions cannot be removed in general; indeed, it was shown in 
\cite[Lemma 3.17]{MO-cqconv} that $D^{\bary,\Um}_{\alpha}$ is not monotone under CPTP maps
(in fact, not even under pinchings) for 
any $\alpha>1$, even though $\DU$ is monotone. 
\end{rem}

\begin{prop}\label{prop:bary jointconv}
Let $P\in\P_f(\X)$, and 
assume that at least one of the following holds:
\begin{enumerate}
\item\label{jointconv sufficient1}
$D^{q_x}$, $x\in\supp P$, are monotone under partial traces 
and are block subadditive. 
\item\label{jointconv sufficient3}
$D^{q_x}$, $x\in\supp P$, are jointly convex in their variables. 
\end{enumerate}
Then $\W\mapsto Q_P^{\bary,\qv}(W)$ and 
$\W\mapsto \psi_P^{\bary,\qv}(W)$
are concave, and 
$\W\mapsto R_{D^{\qv},\bal}\bz W,P\jz$ is convex 
on $\B(\X,\hil)\pne$. 
\end{prop}
\begin{proof}
Since all $D^{q_x}$ satisfy the scaling law 
\eqref{relentr scaling}, they are in particular homogeneous, and thus, 
by Lemma \ref{lemma:jointconc from mon},
\ref{jointconv sufficient1} implies \ref{jointconv sufficient3}.
Assume therefore  \ref{jointconv sufficient3}. Then 
\begin{align*}
\B(\hil)\times\B(\X,\hil)\pne\ni(\tau,W)\mapsto\Tr\tau-\sum_x P(x)D^{q_x}(\tau\|\Wx{x})
\end{align*}
is jointly concave in $\tau$ and $W$, and hence its supremum over  
$\tau$ is concave in $\W$.
This proves the concavity of $\W\mapsto Q_P^{\bary,\qv}(W)$.
The concavity of 
$\W\mapsto \psi_P^{\bary,\qv}(W)$ then follows immediately, which in 
turn implies the convexity of $\W\mapsto R_{D^{\qv},\bal}\bz W,P\jz$
due to \eqref{psi optimization}.
\end{proof}

As a special case, we get the following:

\begin{cor}
Let $\alpha\in(0,1)$, and 
assume that at least one of the following holds:
\begin{enumerate}
\item\label{jointconv sufficient1-2}
$D^{q_x}$, $x\in\{0,1\}$, are monotone under partial traces 
and are block subadditive. 
\item\label{jointconv sufficient3-2}
$D^{q_x}$, $x\in\{0,1\}$, are jointly convex in their variables. 
\end{enumerate}
Then $(\rho,\sigma)\mapsto Q_{\alpha}^{\bary,\qv}(\rho,\sigma)$ and 
$(\rho,\sigma)\mapsto \psi_{\alpha}^{\bary,\qv}(\rho,\sigma)$
are concave, and 
$(\rho,\sigma)\mapsto D_{\alpha}^{\bary,\qv}(\rho,\sigma)$ is convex 
on $\B(\X,\hil)\pne $. 
The same conclusions hold  
a) if $\alpha=0$ and even if the properties in \ref{jointconv sufficient1-2}
or  \ref{jointconv sufficient3-2} are assumed only for $D^{q_1}$,
b) for $D_1^{\bary,\qv}$, 
c) for $Q_1^{\bary,\qv}$ and $\psi_1^{\bary,\qv}$ without any assumption on 
$D^{\qn}$ and $D^{\qo}$.
\end{cor}
\begin{proof}
The case $\alpha\in(0,1)$ is immediate from Proposition \ref{prop:bary jointconv}, and the case b) follows from it according to the definition of $D_1^{\bary,\qv}$ in 
\eqref{Db 1 def}.
The case c) is trivial from Remark \ref{rem:psi 1}, and the case a)
follows in the same way as in the proof of Proposition \ref{prop:bary jointconv}.
\end{proof}

\subsection{Lower semi-continuity and regularity}

Note that when $\X$ is finite then $\B(\X,\hil)=\B(\hil)^{\X}$ is a finite-dimensional vector 
space, and hence there exists a unique norm topology on it, which is what we implicitly 
refer to in statements about (semi-)continuity on this space.

\begin{prop}\label{prop:multiR lsc}
If $P\in\P_f(\X)$ and $D^{q_x}$, $x\in\supp P$, are all jointly lower semi-continuous in their variables then $W\mapsto Q_P^{\bary,\qv}(W)$ is upper semi-continuous, and 
$W\mapsto R_{\D^{\qv},\bal}(W,P)$ is lower semi-continuous on 
$\B(\supp P,\hil)\pne$.
\end{prop}
\begin{proof}
By assumption, 
\begin{align*}
\B(\supp P,\hil)\pne\times\S(\hil)\ni(\W,\omega)\mapsto
\sum_{x\in\X}P(x)D^{q_x}(\omega\|\Wx{x})
\end{align*}
is lower semi-continuous, and hence, by Lemma \ref{lemma:usc},
its infimum over $\omega\in\S(\hil)$ is lower semi-continuous, too, 
proving the assertion about $R_{\D^{\qv},\bal}$. 
The assertion about $Q_P^{\bary,\qv}$ then follows immediately due to 
\eqref{psi optimization}.
\end{proof}

\begin{cor}\label{cor:barycentric Renyi lsc}
If $D^{\qn}$ and $D^{\qo}$ are jointly lower semi-continuous in their variables then 
for any $\alpha\in(0,1]$, 
$Q_{\alpha}^{\bary,\qv}$ is jointly upper semi-continous and 
$D_{\alpha}^{\bary,\qv}$ is jointly lower semi-continous in their arguments.
\end{cor}
\begin{proof}
The case $\alpha\in(0,1)$ follows as a special case of Proposition \ref{prop:multiR lsc}.
For $\alpha=1$, $Q_1^{\bary,\qv}(\rho\|\sigma)=\Tr\rho$ by \eqref{psi proof 2}, and continuity holds trivially, while 
$D_{1}^{\bary,\qv}$ is the supremum of lower semi-continuous functions according to 
the above and \eqref{Db 1 def}, and hence is itself lower semi-continuous. 
\end{proof}

\begin{prop}\label{prop:multi smoothing}
Let $P\in\P_f(\X)$ and assume that 
$D^{q_x}$, $x\in\supp P$, are weakly anti-monotone in their second arguments. Then 
\begin{align}\label{smoothing monotonicity2}
(0,+\infty)\ni\ep\mapsto Q_P^{\bary,\qv}\bz(\Wx{x}+\ep I)_{x\in\X}\jz
\ds\ds\text{is monotone increasing,}\\
(0,+\infty)\ni\ep\mapsto R_{D^{\qv},\bal}^{\bary}\bz(\Wx{x}+\ep I)_{x\in\X}\jz
\ds\ds\text{is monotone decreasing}\label{smoothing monotonicity3}
\end{align}
for any $\W\in\B(\X,\hil)\pne$.
If, moreover, $D^{q_x}$, $x\in\supp P$, are regular in the sense of \eqref{regularity}, and 
lower semi-continuous in their first arguments, then 
for any $\W\in\B(\X,\hil)\pne$,
\begin{align}
Q_P^{\bary,\qv}(W)
&=
\lim_{\ep\searrow 0}Q_P^{\bary,\qv}\bz(\Wx{x}+\ep I)_{x\in\X}\jz
=
\inf_{\ep>0}Q_P^{\bary,\qv}\bz(\Wx{x}+\ep I)_{x\in\X}\jz,
\label{multi Q smoothing}\\
R_{D^{\qv},\bal}^{\bary}(W)
&=
\lim_{\ep\searrow 0}R_{D^{\qv},\bal}^{\bary}\bz(\Wx{x}+\ep I)_{x\in\X}\jz
=\sup_{\ep> 0}R_{D^{\qv},\bal}^{\bary}\bz(\Wx{x}+\ep I)_{x\in\X}\jz.
\label{R smoothing}
\end{align}
\end{prop}
\begin{proof}
The monotonicity assertions in \eqref{smoothing monotonicity2}--\eqref{smoothing monotonicity3} are obvious, as are the second equalities in 
\eqref{multi Q smoothing}--\eqref{R smoothing}. 
The first equality in \eqref{R smoothing} follows as 
\begin{align*}
&\sup_{\ep> 0}R_{D^{\qv},\bal}^{\bary}\bz(\Wx{x}+\ep I)_{x\in\X}\jz\\
&\ds=
\sup_{\ep> 0}\inf_{\omega\in\S(S_+\hil)}\sum_xP(x)D^{q_x}(\omega\|\Wx{x}+\ep I)\\
&\ds=
\inf_{\omega\in\S(S_+\hil)}\sup_{\ep> 0}\sum_xP(x)D^{q_x}(\omega\|\Wx{x}+\ep I)\\
&\ds=
\inf_{\omega\in\S(S_+\hil)}\sum_xP(x)D^{q_x}(\omega\|\Wx{x})\\
&\ds= R_{D^{\qv},\bal}^{\bary}(W),
\end{align*}
where the first and the last equalities are by definition, 
the second equality follows from the minimax theorem in Lemma \ref{lemma:minimax2},
and the third equality from the regularity of the $D^{q_x}$. 
From this, the first equality in \eqref{multi Q smoothing}
follows by \eqref{psi optimization}.
\end{proof}

In the $2$-variable case we have the following:

\begin{cor}\label{cor:smoothing}
Assume that $D^{\qn}$ and $D^{\qo}$ are weakly anti-monotone in their second arguments. Then 
\begin{align}\label{smoothing monotonicity}
(0,+\infty)\ni\ep\mapsto Q_{\alpha}^{\bary,\qv}(\rho+\ep I\|\sigma+\ep I)
\ds\ds\text{is monotone increasing }
\end{align}
for any $\rho,\sigma\in\B(\hil)\pne$ and $\alpha\in[0,1]$. 
If, moreover, $D^{\qn}$ and $D^{\qo}$ are regular in the sense of \eqref{regularity}, and 
lower semi-continuous in their first arguments, then 
for any $\rho,\sigma\in\B(\hil)\pne$,
\begin{align}
Q_{\alpha}^{\bary,\qv}(\rho\|\sigma)
&=
\lim_{\ep\searrow 0}Q_{\alpha}^{\bary,\qv}(\rho+\ep I\|\sigma+\ep I)
=
\inf_{\ep>0}Q_{\alpha}^{\bary,\qv}(\rho+\ep I\|\sigma+\ep I),
&\alpha\in(0,1],
\label{Qalpha smoothing}\\
D_{\alpha}^{\bary,\qv}(\rho\|\sigma)
&=
\lim_{\ep\searrow 0}D_{\alpha}^{\bary,\qv}(\rho+\ep I\|\sigma+\ep I),
&\alpha\in(0,1).
\label{Dalpha smoothing}
\end{align}
\end{cor}
\begin{proof}
The monotonicity assertion in \eqref{smoothing monotonicity} is again obvious by definition,
and the equalities in \eqref{Qalpha smoothing}--\eqref{Dalpha smoothing} 
for $\alpha\in(0,1)$
follow as special cases of 
\eqref{multi Q smoothing}--\eqref{R smoothing}, while
\eqref{Qalpha smoothing} for $\alpha=1$ is trivial from \eqref{psi proof 2}.
\end{proof}

\begin{rem}\label{rem:notregular at 0}
Note that \eqref{Qalpha smoothing} does not hold for $\alpha=0$ in general. Indeed, for commuting $\rho$ and $\sigma$, $Q_0^{\bary,\qv}(\rho\|\sigma)=\Tr\rho^0\sigma$, while 
$\lim_{\ep\searrow 0}Q_{\alpha}^{\bary,\qv}(\rho+\ep I\|\sigma+\ep I)=\Tr\sigma$, which are different whenever $\rho^0\not\geq\sigma^0$. This immediately shows that 
\eqref{Dalpha smoothing} does not hold for $\alpha=0$ in general.
On the other hand, \eqref{Dalpha smoothing} trivially holds for $\alpha=1$ whenever 
\eqref{Renyi bq limit11} is satisfied and $D^{\qo}$ has the regularity property
$D^{\qo}(\rho\|\sigma)=\lim_{\ep\searrow 0}D^{\qo}(\rho+\ep I\|\sigma+\ep I)$; this holds, for instance, for $\DU$ and $D^{\max}$, as well as for $D^{\Um,\#_{\gamma}}$, $\gamma\in(0,1)$. 
\end{rem}

\begin{rem}\label{rem:notmonotone}
By \eqref{smoothing monotonicity}, 
$\ep\mapsto\frac{1}{\alpha-1}\log Q_{\alpha}^{\bary,\qv}(\rho+\ep I\|\sigma+\ep I)$
is monotone decreasing, while the same need not hold for 
$\ep\mapsto D_{\alpha}^{\bary,\qv}(\rho+\ep I\|\sigma+\ep I)$. 
Intuitively, the reason for this is that 
$\frac{1}{1-\alpha}\log\Tr(\rho+\ep I)$ is monotone
increasing in $\ep$. As a concrete example, consider the
$1$-dimensional commutative case
$\rho=r\in\B(\bC)\pne$ and $\sigma=s\in\B(\bC)\pne$, so that  
\begin{align*}
D_{\alpha}^{\bary,\qv}(\rho+\ep I\|\sigma+\ep I)=
D_{\alpha}(\rho+\ep I\|\sigma+\ep I)=
\log\frac{r+\ep}{s+\ep}
=
\log\bz 1+\frac{r-s}{s+\ep}\jz,
\end{align*}
which is monotone increasing in $\ep$ if $r<s$, and monotone decreasing if $r>s$.
\end{rem}

\begin{prop}\label{prop:smoothing2}
Assume that $D^{\qo}$ is weakly anti-monotone in its second argument. Then 
\begin{align}
&(0,+\infty)\ni\ep\mapsto Q_{\alpha}^{\bary,\qv}(\rho\|\sigma+\ep I)
\ds\ds\text{is monotone increasing},\\
&(0,+\infty)\ni\ep\mapsto D_{\alpha}^{\bary,\qv}(\rho\|\sigma+\ep I)
\ds\ds\text{is monotone decreasing}\label{barycentric antimon}
\end{align}
for any $\rho,\sigma\in\B(\hil)\pne$ and $\alpha\in[0,1]$. 
If, moreover, $D^{\qo}$ is regular in the sense of \eqref{regularity}, and 
lower semi-continuous in its first argument, then 
for any $\rho,\sigma\in\B(\hil)\pne$,
\begin{align}
Q_{\alpha}^{\bary,\qv}(\rho\|\sigma)
&=
\lim_{\ep\searrow 0}Q_{\alpha}^{\bary,\qv}(\rho\|\sigma+\ep I)
=
\inf_{\ep>0}Q_{\alpha}^{\bary,\qv}(\rho\|\sigma+\ep I),
&\alpha\in[0,1],
\label{Qalpha smoothing2}\\
D_{\alpha}^{\bary,\qv}(\rho\|\sigma)
&=
\lim_{\ep\searrow 0}D_{\alpha}^{\bary,\qv}(\rho\|\sigma+\ep I)
=
\sup_{\ep>0}D_{\alpha}^{\bary,\qv}(\rho\|\sigma+\ep I)
&\alpha\in[0,1].
\label{Dalpha smoothing2}
\end{align}
\end{prop}
\begin{proof}
The proof is essentially the same as for Proposition \ref{prop:multi smoothing} 
and Corollary \ref{cor:smoothing}, and hence we omit
most of it, and only mention that the $\alpha=1$ case in 
\eqref{barycentric antimon} and in \eqref{Dalpha smoothing2}
follow from the respective statements for $\alpha\in[0,1)$ using 
\eqref{Db 1 def}.
\end{proof}

\subsection{Finiteness and non-examples}
\label{sec:finiteness}

Corollary \ref{cor:Renyi bq infty} gives an easily verifiable condition 
for a quantum R\'enyi $\alpha$-divergence not being a barycentric R\'enyi $\alpha$-divergence, as follows: 

\begin{prop}\label{prop:notbary}
Let $D_{\alpha}^q$ be a quantum R\'enyi $\alpha$-divergence for some $\alpha\in(0,1)$
with the property that $D_{\alpha}^q(\rho\|\sigma)=+\infty$ $\iff$ $\rho\perp\sigma$.
Then there exist no quantum relative entropies $D^{\qn}$ and $D^{\qo}$ with which 
$D_{\alpha}^q=D_{\alpha}^{\bary,\qv}$.
\end{prop}
\begin{proof}
Let $\rho,\sigma\in\B(\hil)\pne$ be such that $\rho^0\wedge\sigma^0=0$ and 
$\rho\not\perp\sigma$. Then 
\begin{align}\label{notbary proof}
D_{\alpha}^{\bary,\qv}(\rho\|\sigma)=+\infty>D_{\alpha}^{q}(\rho\|\sigma)
\end{align}
for any quantum relative entropies $D^{\qn}$ and $D^{\qo}$, according to 
Corollary \ref{cor:Renyi bq infty}.
Since such pairs exist in any dimension larger than $1$, we get that 
$D_{\alpha}^{\bary,\qv}\ne D_{\alpha}^q$.
\end{proof}

\begin{cor}
$D_{\alpha,z}$ is not a barycentric R\'enyi $\alpha$-divergence for any 
$\alpha\in(0,1)$ and $z\in(0,+\infty)$. 
\end{cor}
\begin{proof}
It is obvious by definition that for any $\alpha\in(0,1)$ and $z\in(0,+\infty)$,
and any $\rho,\sigma\in\B(\hil)\pne$, 
$D_{\alpha,z}(\rho\|\sigma)=+\infty$ $\iff$ $\rho\perp\sigma$, and hence the assertion follows immediately from Proposition \ref{prop:notbary}.
\end{proof}

\begin{cor}\label{cor:measured notbary}
The measured R\'enyi $\alpha$-divergence $D_{\alpha}^{\meas}$ is not a barycentric 
R\'enyi $\alpha$-divergence for any $\alpha\in(0,1)$. 
\end{cor}
\begin{proof}
According to Proposition \ref{prop:notbary} we only need to prove that for any 
$\alpha\in(0,1)$ and any $\rho,\sigma\in\B(\hil)\pne$, 
$D_{\alpha}^{\meas}(\rho\|\sigma)=+\infty$ $\iff$ $\rho\perp\sigma$.
This is well known and easy to verify, but we give the details for the readers' convenience. 
If $\rho\perp\sigma$ then the measurement $M_0:=\rho^0$, $M_1:=I-\rho^0$
gives $D_{\alpha}^{\meas}(\rho\|\sigma)\ge D_{\alpha}(\M(\rho)\|\M(\sigma))=+\infty$.
If $\rho\not\perp\sigma$ then we have $D_{\alpha}^{\meas}(\rho\|\sigma)
\le D_{\alpha,1}(\rho\|\sigma)<+\infty$, where the first inequality is due to the 
monotonicity of the Petz-type R\'enyi $\alpha$-divergence under measurements
\cite{P86}.
\end{proof}

One might have the impression that the strict inequality in \eqref{notbary proof} 
is the result of some pathology, and would not happen if the operators had full support, 
and both R\'enyi divergences took finite values on them. This, however, is not the case, at least if we assume some mild and very natural continuity and regularity properties of $D_{\alpha}^q$ and 
the quantum relative entropies $D^{\qn}$ and $D^{\qo}$. 

Indeed, Proposition \ref{prop:notbary} and 
Corollary \ref{cor:barycentric Renyi lsc} yield the following:

\begin{cor}\label{cor:notbary}
Let $D_{\alpha}^q$ be a quantum R\'enyi $\alpha$-divergence for some $\alpha\in(0,1)$, such that 
$D_{\alpha}^q(\rho\|\sigma)=+\infty$ $\iff$ $\rho\perp\sigma$, and assume that 
$D_{\alpha}^q$ is jointly continuous in its arguments.
Let $D^{\qn}$ and $D^{\qo}$ be quantum relative entropies that are
jointly lower semi-continuous in their arguments.
Then for any two $\rho_0,\sigma_0\in\B(\hil)\pne$ such that $\rho_0^0\wedge\sigma_0^0=0$ 
and
$\rho_0\not\perp\sigma_0$, and for any $\rho,\sigma\in\B(\hil)\pne$ in a neighbourhood of 
$(\rho_0,\sigma_0)$,
\begin{align}\label{notbary strict ineq1}
D_{\alpha}^{\bary,\qv}(\rho\|\sigma)>D_{\alpha}^q(\rho\|\sigma),
\end{align} 
and $D_{\alpha}^{\bary,\qv}(\rho\|\sigma)<+\infty$ for any pair of invertible elements in the neighbourhood.
In particular, $D_{\alpha}^q\ne D_{\alpha}^{\bary,\qv}$.
\end{cor}
\begin{proof}
Let $M>D_{\alpha}^q(\rho_0\|\sigma_0)$ be a finite number.  
By the (semi-)continuity assumptions, 
$\{(\rho,\sigma)\in\B(\hil)\pne^2:\,D_{\alpha}^q(\rho\|\sigma)<M<D_{\alpha}^{\bary,\qv}(\rho\|\sigma)\}$ is an open subset of $\B(\hil)\p^2$
containing $(\rho_0,\sigma_0)$, and 
for any of its elements $(\rho,\sigma)$, the inequality \eqref{notbary strict ineq1} holds.
The assertion about the invertible pairs is obvious from 
Corollary \ref{cor:Renyi bq infty}.
\end{proof}

Likewise, Proposition \ref{prop:notbary} and Corollary \ref{cor:smoothing} yield the following:

\begin{cor}\label{cor:notbary2}
Let 
$D_{\alpha}^q$ be a quantum R\'enyi $\alpha$-divergence for some $\alpha\in(0,1)$, such that 
$D_{\alpha}^q(\rho\|\sigma)=+\infty$ $\iff$ $\rho\perp\sigma$, and 
such that $D_{\alpha}^q$ is regular in the sense that 
$D_{\alpha}^q(\rho\|\sigma)=\lim_{\ep\searrow 0}D_{\alpha}^q(\rho+\ep I\|\sigma+\ep I)$ for any $\rho,\sigma\in\B(\hil)\pne$.
Let $D^{\qn}$ and $D^{\qo}$ be quantum relative entropies 
that are lower semi-continuous in their first arguments,
weakly anti-monotone in their second arguments, and regular. 
Then for any $\rho_0,\sigma_0\in\B(\hil)\pne$ such that $\rho_0^0\wedge\sigma_0^0=0$ 
and
$\rho_0\not\perp\sigma_0$, and for any $\ep> 0$ small enough, 
\begin{align}\label{notbary strict2}
+\infty>D_{\alpha}^{\bary,\qv}(\rho_0+\ep I\|\sigma_0+\ep I)
>D_{\alpha}^q(\rho_0+\ep I\|\sigma_0+\ep I).
\end{align} 
In particular, $D_{\alpha}^q\ne D_{\alpha}^{\bary,\qv}$.
\end{cor}
\begin{proof}
Let $M>D_{\alpha}^q(\rho_0\|\sigma_0)$ be a finite number.  
By Corollary \ref{cor:Renyi bq infty} and Corollary
\ref{cor:smoothing}, there exists some $\ep_1>0$ such that 
$+\infty>D_{\alpha}^{\bary,\qv}(\rho_0+\ep I\|\sigma_0+\ep I)>M$ for every 
$0<\ep\le \ep_1$. By the assumed regularity of $D_{\alpha}^q$, 
there exists some $\ep_2>0$ such that 
$D_{\alpha}^q(\rho_0+\ep I\|\sigma_0+\ep I)<M$ for every $0<\ep<\ep_2$. Hence
\eqref{notbary strict2} holds for every $0<\ep<\min\{\ep_1,\ep_2\}$.
\end{proof}

\begin{example}
For every $\alpha\in(0,1)$ and $z\in(0,+\infty)$, $D_{\alpha,z}$ satisfies the conditions in 
Corollaries \ref{cor:notbary} and \ref{cor:notbary2}, and hence both corollaries 
apply to it.
\end{example}

\begin{example}
As discussed above, for every $\alpha\in(0,1)$, $D_{\alpha}^{\meas}(\rho\|\sigma)=+\infty$
$\iff$ $\rho\perp\sigma$. Hence, Corollaries \ref{cor:notbary} and \ref{cor:notbary2}
can be applied to $D_{\alpha}^{\meas}$ if it is jointly continuous in its arguments. 
To show this, let us fix a finite-dimensional Hilbert space $\hil$ and an 
orthonormal system $(e_i)_{i=1}^d$ in it. Let $\bU(\hil)$ denote the set of unitaries on 
$\hil$. Continuity of the classical R\'enyi $\alpha$-divergence $D_{\alpha}$ yields that the function 
\begin{align*}
\B(\hil)\pne\times\B(\hil)\pne\times\bU(\hil)\ni(\rho,\sigma,U)
\mapsto
D_{\alpha}\bz\bz\inner{e_i}{U^*\rho Ue_i}\jz_{i=1}^d\Big\|\bz\inner{e_i}{U^*\rho Ue_i}\jz_{i=1}^d\jz
\end{align*}
is continuous. Hence,
\begin{align}\label{projective measured}
D_{\alpha}^{\meas}(\rho\|\sigma)=\sup_{U\in\bU(\hil)}D_{\alpha}\bz\bz\inner{e_i}{U^*\rho Ue_i}\jz_{i=1}^d\Big\|\bz\inner{e_i}{U^*\rho Ue_i}\jz_{i=1}^d\jz
\end{align}
is jointly continuous in $\rho$ and $\sigma$ according to Lemma \ref{lemma:usc}.
(For the equality in \eqref{projective measured}, see \cite[Theorem 4]{BFT_variational}.)
\end{example}

\section{Examples of barycentric R\'enyi divergences}
\label{sec:ex}

In this section we consider the relations among various known quantum R\'enyi 
$\alpha$-divergences and barycentric R\'enyi $\alpha$-divergences obtained from 
specific quantum relative entropies. 
Our main results in this respect are Corollary \ref{cor:strictly between Umand max} and Theorem \ref{thm:maxRenyi strictly larger}; we illustrate these in Figure \ref{fig:barycentric}.

\subsection{General relations}

Recall the definition of the ordering of quantum divergences given in Definition 
\ref{def:divorder}, and the definition of strict ordering of $2$-variable 
quantum divergences defined on pairs of non-zero PSD operators
from Definition \ref{def:divorder2}.

We start with two simple observations.
The first one is trivial by definition. 

\begin{lemma}\label{lemma:trivorder}
If $P\in\P_f(\X)$ and
$D^{q_x}\le D^{\tilde q_x}$, $x\in\supp P$, then 
\begin{align}\label{trivorder2}
-\log Q_{P}^{\bary,\qv}\le
-\log Q_{P}^{\bary,\tilde\qv}.
\end{align}
In particular, if $\X=\{0,1\}$ then
\begin{align}\label{trivorder}
D^{\qn}\le D^{\tilde q_0},\s
D^{\qo}\le D^{\tilde q_1}
\ds\imp\ds
D^{\bary,\qv}_{\alpha}\le D^{\bary,\tilde\qv}_{\alpha},\ds\ds\ds\alpha\in[0,1].
\end{align}
\end{lemma}

The following is also easy to verify:

\begin{prop}\label{prop:trivorder}
Let $P\in\P_f(\X)$, and let $D^{q_x}$, $x\in\supp P$, be monotone under CPTP maps. Then 
\begin{align}\label{trivial monotone ordering2}
-\log Q_P^{\meas}
\le 
-\log Q_P^{\bary,\meas}\le
-\log Q_P^{\bary,\qv}\le
-\log Q_P^{\bary,\max}
\le -\log Q_P^{\max}.
\end{align}
In particular, if $D^{\qn}$ and $D^{\qo}$ are quantum relative entropies that are monotone under CPTP maps then 
\begin{align}\label{trivial monotone ordering}
D_{\alpha}^{\meas}\le D_{\alpha}^{\bary,\meas}\le
D_{\alpha}^{\bary,\qv}\le
D_{\alpha}^{\bary,\max}
\le D_{\alpha}^{\max},\ds\ds\ds\alpha\in[0,1].
\end{align}
\end{prop}
\begin{proof}
The second and the third inequalities in \eqref{trivial monotone ordering2} 
are immediate from \eqref{trivorder2} and \eqref{relentr maxmin}.
Since $D^{\meas}$ and $D^{\max}$ are monotone under CPTP maps, so are 
$-\log Q_P^{\bary,\meas}$ and $-\log Q_P^{\bary,\max}$ as well, 
according to 
Proposition \ref{prop:weighted PTP-mon}.
Hence, the first and the last inequalities in 
\eqref{trivial monotone ordering2} follow immediately from 
Example \ref{ex:minmax multi} and Proposition \ref{prop:barycentric quantum Renyi}.
The inequalities in \eqref{trivial monotone ordering} follow the same way.
\end{proof}

We have seen in Corollary \ref{cor:measured notbary} that the first inequality in \eqref{trivial monotone ordering} is not an equality, i.e., 
for any $\alpha\in(0,1)$, 
the smallest barycentric R\'enyi $\alpha$-divergence generated by
CPTP-monotone quantum relative entropies is above and not equal to the smallest 
CPTP-monotone quantum R\'enyi $\alpha$-divergence. We will show in Section \ref{sec:maximal Renyi} that the same 
happens ``at the top of the spectrum'', i.e., the last inequality in 
\eqref{trivial monotone ordering} is not an equality, either.

One might expect that the strict  ordering of relative entropies yields 
a strict ordering of the corresponding variational R\'enyi divergences. This, however, is not true in complete generality, as Example \ref{ex:equality} below shows.
On the other hand, 
$D_{\alpha}^{\bary,\qv}(\rho\|\sigma)<D_{\alpha}^{\bary,\tilde\qv}(\rho\|\sigma)$
might nevertheless hold if some extra conditions are imposed on 
the inputs, as we show in Sections \ref{sec:Um and smaller}--\ref{sec:maximal Renyi} below.
Here we make the following general observation:

\begin{lemma}\label{lemma:strict ineq for barycentric}
Let $P\in\P_f(\X)\setminus\{\egy_{\{x\}}:\,x\in\X\}$ and $D^{\qv},D^{\tilde\qv}$ be such that 
$D^{q_x}\le D^{\tilde q_x}$, $x\in\supp P$, and 
$D^{q_y}<D^{\tilde q_y}$ for some $y\in\supp P$. 
Let $W\in\B(\X,\hil)\pne$,
and assume that one of the following holds:
\begin{enumerate}
\item\label{strict ineq cond1}
There exists a $P$-weighted $D^{\tilde\qv}$-center for $W$
that does not commute with $W_y$.
\item\label{strict ineq cond2}
There exists a $P$-weighted $D^{\tilde\qv}$-center 
for $W$, and no state that commutes with $W_y$ is a 
$P$-weighted $D^{\qv}$-center for $W$.
\end{enumerate}
Then $-\log Q_P^{\bary,\qv}(W)<-\log Q_P^{\bary,\tilde\qv}(W)$.
\end{lemma}
\begin{proof}
Assume \ref{strict ineq cond1}, and let $\omega$ be 
a $P$-weighted $D^{\tilde\qv}$-center for $W$
that does not commute with $W_y$. Then 
\begin{align*}
-\log Q_P^{\bary,\tilde\qv}(W)
&=
\sum_{x\in\X}P(x)D^{\tilde q_x}(\omega\|W_x)\\
&=
\underbrace{D^{\tilde q_y}(\omega\|W_y)}_{> D^{q_y}(\omega\|W_y)}+
\sum_{x\in\X\setminus\{y\}}P(x)\underbrace{D^{\tilde q_x}(\omega\|W_x)}_{\ge D^{q_x}(\omega\|W_x)}\\
&>
\sum_{x\in\X}P(x)D^{q_x}(\omega\|W_x)\\
&\ge
-\log Q_P^{\bary,\qv}(W),
\end{align*}
where the first equality and the last inequality are by definition,
and the strict inequality follows from the assumption that 
$\omega W_y\ne W_y\omega$ and that $D^{q_y}<D^{\tilde q_y}$.

Assume now \ref{strict ineq cond2}, and let $\tilde\omega$ be a  
$P$-weighted $D^{\tilde\qv}$-center for $W$. If 
$\tilde\omega$ does not commute with $W_y$ then 
$-\log Q_P^{\bary,\qv}(W)<-\log Q_P^{\bary,\tilde\qv}(W)$ by the previous point. 
If $\tilde\omega$ commutes with $W_y$ then 
it cannot be a
$P$-weighted $D^{\qv}$-center for $W$ by assumption, and hence there exists an 
$\omega\in\S(S_+\hil)$ such that 
\begin{align*}
\sum_{x\in\X}P(x)D^{q_x}(\tilde\omega\|W_x)
>
\sum_{x\in\X}P(x)D^{q_x}(\omega\|W_x).
\end{align*}
Thus,
\begin{align*}
-\log Q_P^{\bary,\tilde\qv}(W)
&=
\sum_{x\in\X}P(x)D^{\tilde q_x}(\tilde\omega\|W_x)\\
&\ge
\sum_{x\in\X}P(x)D^{q_x}(\tilde\omega\|W_x)\\
&>
\sum_{x\in\X}P(x)D^{q_x}(\omega\|W_x)\\
&\ge
-\log Q_P^{\bary,\qv}(W),
\end{align*}
proving the assertion.
\end{proof}

\begin{rem}
The condition in \ref{strict ineq cond2} of 
Lemma \ref{lemma:strict ineq for barycentric}
that there exists a $P$-weighted $D^{\tilde\qv}$-center 
for $W$ is very mild; indeed, it is satisfied whenever 
for every $x\in\supp P$, $D^{q_x}$ is lower semi-continuous
in its first variable.
\end{rem}

\subsection{Umegaki relative entropy and smaller/larger relative entropies}
\label{sec:Um and smaller}

\begin{prop}\label{prop:Umegaki and smaller}
Let $D^{\qn},D^{\qo}\le \DU$, and assume that 
at least one of them is strictly smaller than $\DU$. 
Then for any two non-commuting invertible positive operators
$\rho,\sigma\in\B(\hil)\pp$, 
\begin{align}\label{Um strict inequality}
D_{\alpha}^{\bary,\qv}(\rho\|\sigma)<
D_{\alpha}^{\bary,\um}(\rho\|\sigma)\ds (=D_{\alpha,+\infty}(\rho\|\sigma)),
\ds\ds\ds \alpha\in(0,1).
\end{align}
In particular, if $D^{\meas}\le D^{\qn},D^{\qo}\le \DU$ then 
for any two non-commuting invertible positive operators
\begin{align}\label{Um strict inequality2}
D_{\alpha}^{\bary,\meas}(\rho\|\sigma)\le\left\{
\begin{array}{ll}
D_{\alpha}^{\bary,(\meas,\qo)}(\rho\|\sigma)\\
D_{\alpha}^{\bary,(\qn,\meas)}(\rho\|\sigma)
\end{array}
\right\}
<D_{\alpha}^{\bary,\Um}(\rho\|\sigma),\ds\ds\ds \alpha\in(0,1).
\end{align}
\end{prop}
\begin{proof}
For the equality in \eqref{Um strict inequality} see Remark \ref{rem:Umebary=a-infty}.
Recall the form of the $\alpha$-weighted $\DU$-center
$\omega_{\alpha}:=\omega_{\alpha}^{\um}(\rho\|\sigma)$
in \eqref{Umegaki Hellinger}.
It is easy to see that if it 
commutes with $\rho$ or $\sigma$ then $\rho$ and $\sigma$ have to commute with each other. Indeed, 
assume that $\omega_{\alpha}$ commutes with $\rho$; then $\rho$ also commutes with any function of 
$\omega_{\alpha}$, in particular, with $(1-\alpha)\log\rho+\alpha\log\sigma$, and hence it also commutes with $\sigma$. The same argument works if $\omega_{\alpha}$ commutes with $\sigma$. 
The strict inequality in \eqref{Um strict inequality} follows from this by Lemma \ref{lemma:strict ineq for barycentric} (using condition \ref{strict ineq cond1}).
The strict inequality in \eqref{Um strict inequality2} follows from 
\eqref{Um strict inequality} and
\eqref{relentr strict inequalities}.
Finally, the first inequality in \eqref{Um strict inequality2} is obvious from 
Lemma \ref{lemma:trivorder} and \eqref{relentr strict inequalities}.
\end{proof}

\begin{rem}\label{rem:Umegaki and smaller}
Note that due to the assumption that $\rho,\sigma>0$
and that in this case $\omega_{\alpha}^{\Um}(\rho\|\sigma)>0$,
the strict inequality condition
$D^{\qn}<\DU$ or $D^{\qo}<\DU$ in Proposition \ref{prop:Umegaki and smaller} can be weakened 
to the assumption that for any strictly positive $\omega\in\S(\hil)$ not commuting with 
$\rho$ and $\sigma$, $D^{\qn}(\omega\|\rho)<\DU(\omega\|\rho)$
or $D^{\qo}(\omega\|\sigma)<\DU(\omega\|\sigma)$ holds.
In particular, it is sufficient to assume that 
$D^{\qn}(\omega_1\|\omega_2)<\DU(\omega_1\|\omega_2)$ for every 
non-commuting invertible $\omega_1,\omega_2$, or
$D^{\qo}(\omega_1\|\omega_2)<\DU(\omega_1\|\omega_2)$ for every 
non-commuting invertible $\omega_1,\omega_2$.
\end{rem}

\begin{prop}\label{prop:Umegaki and larger}
Let $D^{\qn},D^{\qo}\ge \DU$, and assume that 
at least one of them is strictly larger than $\DU$. 
Let $\rho,\sigma\in\B(\hil)\pp$
be non-commuting invertible positive operators 
and $\alpha\in(0,1)$ be such that there exists an $\alpha$-weighted 
$D^{\qv}$-center for $(\rho,\sigma)$. Then 
\begin{align}\label{Um strict inequality3}
D_{\alpha}^{\bary,\qv}(\rho\|\sigma)>
D_{\alpha}^{\bary,\um}(\rho\|\sigma)\ds (=D_{\alpha,+\infty}(\rho\|\sigma)).
\end{align}
\end{prop}
\begin{proof}
Follows immediately from 
Lemma \ref{lemma:strict ineq for barycentric} (using condition \ref{strict ineq cond2})
and the fact that $\omega_{\alpha}^{\Um}$ commutes with neither $\rho$ nor $\sigma$, as we have seen in the proof of Proposition \ref{prop:Umegaki and smaller}.
\end{proof}

\begin{cor}\label{cor:Umegaki larger}
Let $\DU\le D^{\qn},D^{\qo}\le D^{\max}$, and let $\rho,\sigma\in\B(\hil)\pp$ be non-commuting, and $\alpha\in(0,1)$. If there exist an
$\alpha$-weighted $D^{(\qn,\max)}$-center and 
an $\alpha$-weighted $D^{(\max,\qo)}$-center for $(\rho,\sigma)$ then 
\begin{align}\label{Um strict inequality4}
D_{\alpha}^{\bary,\max}(\rho\|\sigma)\ge\left\{
\begin{array}{ll}
D_{\alpha}^{\bary,(\qn,\max)}(\rho\|\sigma)\\
D_{\alpha}^{\bary,(\max,\qo)}(\rho\|\sigma)
\end{array}
\right\}
>D_{\alpha}^{\bary,\Um}(\rho\|\sigma).
\end{align}
\end{cor}
\begin{proof}
The first inequality in \eqref{Um strict inequality4} is obvious from 
Lemma \ref{lemma:trivorder} and \eqref{relentr strict inequalities}.
The strict inequality in \eqref{Um strict inequality4} follows from 
Proposition \ref{prop:Umegaki and larger} and
\eqref{relentr strict inequalities}.
\end{proof} 

\begin{cor}
For any two non-commuting invertible $(\rho,\sigma)$, 
\begin{align*}
D_{\alpha}^{\bary,\Um}(\rho\|\sigma)<
D_{\alpha}^{\bary,\max}(\rho\|\sigma),\ds\ds\ds\alpha\in(0,1).
\end{align*}
\end{cor}
\begin{proof}
Note that $D^{\max}$ is jointly lower semi-continuous in its arguments, and
hence for any $\rho,\sigma\in\B(\hil)\pne$ 
and any
$\alpha\in(0,1)$, there exists an $\alpha$-weighted $D^{\max}$-center for 
$(\rho,\sigma)$. Thus, the assertion follows from 
Corollary \ref{cor:Umegaki larger}.
\end{proof}

\subsection{Maximal relative entropy and a smaller relative entropy}

\begin{lemma}\label{lemma:max center noncomm}
Let $\rho,\sigma\in\B(\hil)\pne$ be such that 
$\rho^0\wedge\sigma^0\ne 0$, and 
$\rho$ and $\sigma$ do not have a common eigenvector.
Then for any $\alpha\in(0,1)$, there exists an 
$\alpha$-weighted $D^{\max}$-center for $(\rho,\sigma)$ 
such that it commutes with neither
$\rho$ nor $\sigma$. 
\end{lemma}
\begin{proof}
Let $\A:=\{\rho\}'\cap\{\sigma\}'$ be the ${}^*$-subalgebra
of operators commuting with both $\rho$ and $\sigma$, and let 
$P_1,\ldots,P_r$ be a sequence of minimal projections in $\A$ summing to $I$
(in particular, $P_i\perp P_j$ for $i\ne j$).
Let $\map(\valt)=\sum_{i=1}^r P_i(\valt)P_i$ be the corresponding pinching operation.
By the lower semi-continuity of $D^{\max}$ in its first argument, there exists at least one state $\omega\in\S(\rho^0\hil)$ attaining the infimum in \eqref{Renyi bq2}
with $D^{\qn}=D^{\qo}=D^{\max}$. 
Moreover, for any $\omega\in \S(\rho^0\hil)$, we have 
\begin{align*}
&\frac{\alpha}{1-\alpha}D^{\bs}(\omega\|\rho)+D^{\bs}(\omega\|\sigma)\\
&\ds\ge
\frac{\alpha}{1-\alpha}D^{\bs}(\map(\omega)\|\underbrace{\map(\rho)}_{=\rho})+D^{\bs}(\map(\omega_{\alpha})\|\underbrace{\map(\sigma)}_{=\sigma})\\
&\ds=
\frac{\alpha}{1-\alpha}D^{\bs}(\map(\omega)\|\rho)+D^{\bs}(\map(\omega)\|\sigma),
\end{align*}
where the  inequality follows by the monotonicity of $D^{\bs}$ under CPTP maps,
and the equality is due to $\map(\rho)=\rho$ and $\map(\sigma)=\sigma$.
Hence, there exists an optimal $\omega$ that also satisfies
$\omega=\map(\omega)=\sum_iP_i\omega P_i$.
Let $\omega_{\alpha}$ be such an optimal state.
Using then the block additivity of $D^{\bs}$, we get 
\begin{align*}
D_{\alpha}^{\bary,\bs}(\rho\|\sigma)
=
\sum_{i=1}^r \frac{\alpha}{1-\alpha}D^{\bs}(P_i\omega_{\alpha}P_i\|P_i\rho P_i)+
D^{\bs}(P_i\omega_{\alpha}P_i\|P_i\sigma P_i)-\frac{1}{\alpha-1}\log\Tr\rho.
\end{align*}

Assume now that $\omega_{\alpha}$ commutes with both $\rho$ and $\sigma$, or equivalently, 
that $P_i\omega_{\alpha}P_i$ commutes with both $P_i\rho P_i$ and $P_i\sigma P_i$
for every $i=1,\ldots,r$.
Note that, by the definition of the $P_i$, the only operators of the form $P_iXP_i$ that 
commute with both $P_i\rho P_i$ and $P_i\sigma P_i$ 
are constant multiples of $P_i$. 
Therefore our assumption yields that there exists an $i$ such that 
$P_i\omega_{\alpha} P_i=cP_i$ with some $c\in(0,+\infty)$. 
For the rest we may restrict the Hilbert space to 
$\ran P_i$, and use the notations $I$ for $P_i$, $\tilde\rho$ for $P_i\rho P_i$,
and $\tilde\sigma$ for $P_i\sigma P_i$. 
Note that the assumption that $\rho^0\wedge\sigma^0\ne 0$ yields that 
$D_{\alpha}^{\bary,\bs}(\rho\|\sigma)<+\infty$, according to \eqref{Renyi bq infty}.
Thus,
$D^{\bs}(cI\|\tilde\rho)<+\infty$, $D^{\bs}(cI\|\tilde\sigma)<+\infty$, and therefore
$\tilde\rho^0=I=\tilde\sigma^0$, according to \eqref{BS def2}.
By the definition of $\omega_{\alpha}$, for any self-adjoint traceless operator $X\in\B(\ran P_i)$, and any 
$t\in(-c/\norm{X},c/\norm{X})$,
\begin{align}\label{BS larger proof2}
f_X(t):=\frac{\alpha}{1-\alpha}D^{\bs}(cI+tX\|\tilde\rho)+
D^{\bs}(cI+tX\|\tilde\sigma)
\ge\frac{\alpha}{1-\alpha}D^{\bs}(cI\|\tilde\rho)+
D^{\bs}(cI\|\tilde\sigma).
\end{align}
By the joint convexity of $D^{\bs}$ \cite{Matsumoto_newfdiv}, $f_X(\valt)$ is a convex function, and hence, by the 
above, it has a global minimum at $t=0$. 
Since it is also differentiable at $t=0$, as we show below, we get that 
\begin{align*}
f_X'(0)=0,\ds\ds\ds X\in\B(\ran P_i)_{\sa},\ds\Tr X=0.
\end{align*}
We have
\begin{align*}
&\frac{d}{dt}D^{\max}(cI+tX\|\tilde\rho)\Big\vert_{t=0}\\
&\ds=
\frac{d}{dt}
\Tr\tilde\rho^{1/2}(cI+tX)\tilde\rho^{-1/2}
\log\bz\tilde\rho^{-1/2}(cI+tX)\tilde\rho^{-1/2}\jz
\Big\vert_{t=0}\\
&\ds=
\Tr\tilde\rho^{1/2}X\tilde\rho^{-1/2}
\log\big(\underbrace{\tilde\rho^{-1/2}cI\tilde\rho^{-1/2}}_{=\,c\tilde\rho\inv}\big)
+
\Tr\underbrace{\tilde\rho^{1/2}cI\tilde\rho^{-1/2}}_{=\,cI}
\frac{d}{dt}\log\bz\tilde\rho^{-1/2}(cI+tX)\tilde\rho^{-1/2}\jz\Big\vert_{t=0}.
\end{align*}
Let $\tilde\rho=\sum_{i=1}^r\lambda_i Q_i$ be the spectral decomposition of $\tilde\rho$.
Then the Fr\'echet derivative of $\log$ at $c\tilde\rho\inv$ is the linear 
operator
\begin{align*}
\FD\log(c\tilde\rho\inv):\,
A\mapsto
\sum_{i,j=1}^r\log^{[1]}\bz c/\lambda_i,c/\lambda_j\jz Q_iAQ_j,\ds\ds\ds
A\in\B(\ran P_i),
\end{align*}
where 
\begin{align*}
\log^{[1]}(x,y)=\begin{cases}
\frac{\log x-\log y}{x-y},&x\ne y,\\
1/x,&x=y,
\end{cases}
\end{align*}
is the first divided difference function of $\log$ (see \eqref{opfunction derivative}).
Thus,
\begin{align*}
&\Tr cI\frac{d}{dt}\log\bz\tilde\rho^{-1/2}(cI+tX)\tilde\rho^{-1/2}\jz\Big\vert_{t=0}\\
&\ds=
\Tr cI\sum_{i,j=1}^r\log^{[1]}\bz c/\lambda_i,c/\lambda_j\jz 
\underbrace{Q_i \tilde\rho^{-1/2}X\tilde\rho^{-1/2} Q_j}_{=\,\lambda_i^{-1/2}\lambda_j^{-1/2}
Q_iXQ_j}\\
&\ds=
c\sum_{i,j=1}^r\log^{[1]}\bz c/\lambda_i,c/\lambda_j\jz 
\lambda_i^{-1/2}\lambda_j^{-1/2}\underbrace{\Tr Q_iXQ_j}_{=\delta_{i,j}\Tr Q_iX}\\
&\ds=
c\sum_{i=1}^r\frac{\lambda_i}{c}\frac{1}{\lambda_i} \Tr Q_iX=\Tr X.
\end{align*}
By an exactly analogous computation for 
$\frac{d}{dt}D^{\max}(cI+tX\|\tilde\sigma)\Big\vert_{t=0}$, we finally get
\begin{align*}
0=f_X'(0)
=
\Tr X\left[\frac{\alpha}{1-\alpha}\bz I+\log(c\tilde\rho\inv)\jz+I
+\log(c\tilde\sigma\inv) \right]
=
-\Tr X\left[\frac{\alpha}{1-\alpha}\log\tilde\rho+\log\tilde\sigma\right],
\end{align*}
for any $X\in\B(\ran P_i)_{\sa}$ with $\Tr X=0$. 
This is equivalent to the existence of some $\kappa\in\bR$ such that 
\begin{align*}
\frac{\alpha}{1-\alpha}\log\tilde\rho+\log\tilde\sigma=\kappa I,
\end{align*}
i.e.,
\begin{align*}
\tilde\sigma=e^{\kappa}\tilde\rho^{\frac{\alpha}{\alpha-1}}.
\end{align*}
In particular, $\tilde\rho$ and $\tilde\sigma$ have a common eigenvector 
$\psi\in\ran P_i$, which is also a common eigenvector of $\rho$ and $\sigma$, 
contradicting our initial assumptions.
\end{proof}

\begin{thm}\label{thm:BS larger}
Let $D^{\qn}$ and $D^{\qo}$ be quantum relative entropies 
such that $D^{\qn}\le D^{\max}$, $D^{\qo}\le D^{\max}$, 
and at least one of the inequalities is strict.
Let $\rho,\sigma\in\B(\hil)\pne$ be such that 
$\rho^0\wedge\sigma^0\ne 0$, and 
$\rho$ and $\sigma$ do not have a common eigenvector. Then 
\begin{align}\label{BS larger}
D_{\alpha}^{\bary,\qv}(\rho\|\sigma)
<
D_{\alpha}^{\bary,\bs}(\rho\|\sigma),
\ds\ds\ds \alpha\in(0,1).
\end{align}
\end{thm}
\begin{proof}
Immediate from Lemma \ref{lemma:max center noncomm} and 
(i) of Lemma \ref{lemma:strict ineq for barycentric}.
\end{proof}

\begin{cor}\label{cor:strictly between Umand max}
Let $D^{\qn},D^{\qo}$ be lower semi-continuous in their first arguments, and assume that 
$\DU\le D^{\qn},D^{\qo}\le D^{\max}$. Assume, moreover, that 
$\DU<D^{\qn}$ or $\DU<D^{\qo}$, and 
$D^{\qn}<D^{\max}$ or $D^{\qo}<D^{\max}$. Then for any two non-commuting 
invertible $\rho,\sigma$ 
that do not have a common eigenvector,
\begin{align*}
D_{\alpha}^{\bary,\Um}(\rho\|\sigma)
<
D_{\alpha}^{\bary,\qv}(\rho\|\sigma)
<
D_{\alpha}^{\bary,\max}(\rho\|\sigma),\ds\ds\ds\alpha\in(0,1).
\end{align*}
\end{cor}
\begin{proof}
Immediate from Proposition \ref{prop:Umegaki and larger} and Theorem \ref{thm:BS larger}.
\end{proof}

\begin{example}\label{ex:equality}
Let $D^{\qn}$  and $D^{\qo}$ be quantum relative entropies as in Theorem \ref{thm:BS larger}.
Let 
\begin{align*}
\rho:=p\oplus(1-p)\pr{\psi_1}\ds\ds\text{and}\ds\ds\sigma:=q\oplus(1-q)\pr{\psi_2}
\end{align*}
be PSD operators on 
$\hil=\bC\oplus\bC^d$ for some $d>1$, where $p,q\in(0,1)$, and 
$\psi_1,\psi_2\in\bC^d$ are unit vectors that are neither parallel nor orthogonal. 
Then $\rho^0\wedge \sigma^0=1\oplus 0$, and hence the unique optimal 
$\omega_{\alpha}$ for any $\alpha\in(0,1)$ and any barycentric R\'enyi $\alpha$-divergence is
$\omega_{\alpha}=1\oplus 0$. Thus,
\begin{align*}
D_{\alpha}^{\bary,\max}(\rho\|\sigma)
&=
\frac{\alpha}{1-\alpha}D^{\max}(\omega_{\alpha}\|\rho)+D^{\max}(\omega_{\alpha}\|\sigma)
-\frac{1}{\alpha-1}\log\Tr\rho\\
&=
\frac{\alpha}{1-\alpha}(-\log p)-\log q-\frac{1}{\alpha-1}\log\Tr\rho\\
&=
\frac{\alpha}{1-\alpha}D^{\qn}(\omega_{\alpha}\|\rho)+D^{\qo}(\omega_{\alpha}\|\sigma)
-\frac{1}{\alpha-1}\log\Tr\rho\\
&=
D_{\alpha}^{\bary,\qv}(\rho\|\sigma).
\end{align*}
This shows that the assumption that $\rho$ and $\sigma$ do not have a common eigenvector
cannot be completely omitted in Lemma \ref{lemma:max center noncomm} or
in Theorem \ref{thm:BS larger}.
\end{example}

\subsection{Maximal R\'enyi divergences vs.~the barycentric maximal R\'enyi divergences}
\label{sec:maximal Renyi}

By Proposition \ref{prop:trivorder}, for every $\alpha\in(0,1)$,
$D_{\alpha}^{\bary,\max}\le D_{\alpha}^{\max}$. Our aim in this section is to show that 
equality does not hold. In fact, we conjecture the stronger relation that for 
non-commuting invertible PSD operators $\rho,\sigma$, 
$D_{\alpha}^{\bary,\max}(\rho\|\sigma)< D_{\alpha}^{\max}(\rho\|\sigma)$, $\alpha\in(0,1)$, which is supported by numerical examples. We will prove this below in the special case where the inputs are $2$-dimensional.
Of course, this already gives at least that 
\begin{align*}
D_{\alpha}^{\bary,\max}\lneq D_{\alpha}^{\max},\ds\ds\ds\alpha\in(0,1).
\end{align*}

Let $\rho,\sigma\in\B(\hil)\pne$ be such that 
$P:=\rho^0\wedge\sigma^0\ne 0$, 
and let $(\hat p,\hat q,\hat\rt)$ be the reverse test given in 
\eqref{optimal reverse test}--\eqref{optimal reverse test gamma}, defined in terms of the spectral decomposition 
\begin{align*}
\sigma^{-1/2}\rho_{P,\ac}\sigma^{-1/2}
=
\sum_{i=1}^r\lambda_i E_i.
\end{align*}
Let 
\begin{align}
&Q_{\alpha}^{\max}:=Q_{\alpha}^{\max}(\rho\|\sigma)
=
Q_{\alpha}(\hat p\|\hat q)
=
\sum_{i=1}^r\lambda_i^{\alpha}\Tr\sigma E_i
=
\Tr\sigma (\sigma^{-1/2}\rho_{P,\ac}\sigma^{-1/2})^{\alpha}
=\Tr\sigma\#_{\alpha}\rho,\\
&\omega_{\alpha}:=\omega_{\alpha}(\hat p\|\hat q)=
\frac{1}{Q_{\alpha}(\hat p\|\hat q)}\sum_{i:\,\hat p(i)\hat q(i)>0}\hat p(i)^{\alpha}\hat q(i)^{1-\alpha}\egy_{\{i\}}
=
\sum_{i=1}^r
\frac{\lambda_i^{\alpha}\Tr\sigma E_i}{Q_{\alpha}^{\max}(\rho\|\sigma)}\egy_{\{i\}}
\label{omega alpha maxbary},
\end{align}
so that 
\begin{align}
&\hat\rt(\omega_{\alpha})=
\frac{1}{Q_{\alpha}^{\max}(\rho\|\sigma)}
\sum_i\lambda_i^{\alpha}\sigma^{1/2}E_i\sigma^{1/2}
=
\frac{1}{Q_{\alpha}^{\max}(\rho\|\sigma)}
\underbrace{\sigma^{1/2}(\sigma^{-1/2}\rho_{P,\ac}\sigma^{-1/2})^{\alpha}\sigma^{1/2}}_{=
\sigma\#_{\alpha}\rho}=:\what{\sigma\#_{\alpha}\rho},
\label{normalized alpha-mean}
\end{align}
where $\sigma\#_{\alpha}\rho$
is the $\alpha$-weighted Kubo-Ando geometric mean
of $\rho$ and $\sigma$ (see Appendix \ref{sec:KA}).

Then for any $\alpha\in(0,1)$, and also for $\alpha\in(1,2]$ if $\rho^0\le\sigma^0$,
\begin{align}
\frac{1}{\alpha-1}\log\Tr\rho+D_{\alpha}^{\max}(\rho\|\sigma)
&=
\frac{1}{\alpha-1}\log\Tr\hat p+D_{\alpha}(\hat p\|\hat q)\\
&=
\frac{\alpha}{1-\alpha}D(\omega_{\alpha}\|\hat p)+D(\omega_{\alpha}\|\hat q)
\label{maxmax inequality1}\\
&\ge
\frac{\alpha}{1-\alpha}D^{\max}(\hat\rt(\omega_{\alpha})\|\underbrace{\hat\rt(\hat p)}_{=\rho})+
D^{\max}(\hat\rt(\omega_{\alpha})\|\underbrace{\hat\rt(\hat q)}_{=\sigma})
\label{maxmax inequality2}\\
&=
\frac{\alpha}{1-\alpha}
D^{\max}\bz\frac{\sigma\#_{\alpha}\rho}{Q_{\alpha}^{\max}(\rho\|\sigma)}
\Big\|\rho\jz+
D^{\max}\bz\frac{\sigma\#_{\alpha}\rho}{Q_{\alpha}^{\max}(\rho\|\sigma)}
\Big\|\sigma\jz,
\end{align}
where the first equality is due to the optimality of the reverse test (see Example \ref{ex:max Renyi}), the second equality follows by a straightforward computation  
(see also \eqref{classical variational2}--\eqref{classical optimal state}),
the inequality is due to the monotonicity of $D^{\max}$ under positive trace-preserving maps,
and the third equality is by the definition of $\hat\rt$ 
and \eqref{normalized alpha-mean}. 

The inequality in \eqref{maxmax inequality2}
is in fact an equality, according to Proposition \ref{prop:max Renyi as geometric} below.
For the proof, we will need the following special case of \cite[Theorem 3.34]{HiaiMosonyi2017}:

\begin{lemma}\label{lemma:maxrelentr preservation}
Let $A_1,A_2\in\B(\hil)\pne$ be such that $A_1^0\le A_2^0$, and let 
$\map:\,\B(\hil)\to\B(\kil)$ be a positive trace-preserving map. Then the following are equivalent:
\begin{enumerate}
\item
$D^{\max}(\map(A_1)\|\map(A_2))=D^{\max}(A_1\|A_2)$;
\item
for some (equivalently, for all) $\beta\in(0,1)$, 
$\map(A_2)\#_{\beta}\map(A_1)=\map(A_2\#_{\beta}A_1)$.
\end{enumerate}
\end{lemma}

\begin{prop}\label{prop:max Renyi as geometric}
Let $\rho,\sigma\in\B(\hil)\pne$ be such that $\rho^0\wedge\sigma^0\ne 0$, and 
assume that $\alpha\in(0,1)$, or that $\alpha\in(1,2]$ and $\rho^0=\sigma^0$.
Then
\begin{align}\label{max Renyi var}
D_{\alpha}^{\max}(\rho\|\sigma)
&=
\frac{\alpha}{1-\alpha}D^{\max}\bz\frac{\sigma\#_{\alpha}\rho}{Q_{\alpha}^{\max}(\rho
\|\sigma)}\Big\|\rho\jz+
D^{\max}\bz\frac{\sigma\#_{\alpha}\rho}{Q_{\alpha}^{\max}(\rho\|\sigma)}\Big\|\sigma\jz
-\frac{1}{\alpha-1}\log\Tr\rho.
\end{align}
\end{prop}
\begin{proof}
Assume first that 
$\rho^0=\sigma^0$, in which case we may equivalently assume that
$\rho$
and $\sigma$ are invertible, and
recall that in this case, 
\begin{align}\label{sigma sharp rho}
\sigma\#_{\alpha}\rho
=
\rho\#_{1-\alpha}\sigma
=
\rho^{1/2}(\rho^{-1/2}\sigma\rho^{-1/2})^{1-\alpha}\rho^{1/2};
\end{align}
see \eqref{gamma mean1}--\eqref{gamma mean0}.
Thus, by \eqref{BS def1},
\begin{align}
D^{\max}\bz\frac{\sigma\#_{\alpha}\rho}{Q_{\alpha}^{\max}(\rho
\|\sigma)}\Big\|\rho\jz&=
\frac{1}{Q_{\alpha}^{\max}}\Tr\rho(\rho^{-1/2}\sigma\rho^{-1/2})^{1-\alpha}
\log\frac{(\rho^{-1/2}\sigma\rho^{-1/2})^{1-\alpha}}{Q_{\alpha}^{\max}}
\nn\\
&=
-\frac{\log Q_{\alpha}^{\max}}{Q_{\alpha}^{\max}}\underbrace{\Tr \rho(\rho^{-1/2}\sigma\rho^{-1/2})^{1-\alpha}}_{=\Tr\sigma\#_{\alpha}\rho=Q_{\alpha}^{\max}}
\nn\\
&\ds\ds
+\frac{1-\alpha}{Q_{\alpha}^{\max}}
\Tr\rho(\rho^{-1/2}\sigma\rho^{-1/2})^{1-\alpha}
\log(\rho^{-1/2}\sigma\rho^{-1/2})\nn\\
&=
-\log Q_{\alpha}^{\max}-\frac{1-\alpha}{Q_{\alpha}^{\max}}
\Tr\sigma(\sigma^{-1/2}\rho\sigma^{-1/2})^{\alpha}
\log(\sigma^{-1/2}\rho\sigma^{-1/2}),\label{max Renyi var proof 1}
\end{align}
where in the last equality we used that the 
transpose function of $f(\valt):=(\valt)^{1-\alpha}\log(\valt)$ is 
$\tilde f(\valt):=-(\valt)^{\alpha}\log(\valt)$, whence
\begin{align}\label{max Renyi var proof 3}
\Tr\rho(\rho^{-1/2}\sigma\rho^{-1/2})^{1-\alpha}
\log(\rho^{-1/2}\sigma\rho^{-1/2})=
-\Tr\sigma(\sigma^{-1/2}\rho\sigma^{-1/2})^{\alpha}
\log(\sigma^{-1/2}\rho\sigma^{-1/2}),
\end{align}
according to \eqref{persp trans}.
Similarly, 
\begin{align}
D^{\max}\bz\frac{\sigma\#_{\alpha}\rho}{Q_{\alpha}^{\max}(\rho\|\sigma)}\Big\|\sigma\jz
&=
\frac{1}{Q_{\alpha}^{\max}}\Tr\sigma(\sigma^{-1/2}\rho\sigma^{-1/2})^{\alpha}
\log\frac{(\sigma^{-1/2}\rho\sigma^{-1/2})^{\alpha}}{Q_{\alpha}^{\max}}
\nn\\
&=
-\frac{\log Q_{\alpha}^{\max}}{Q_{\alpha}^{\max}}\underbrace{
\Tr \sigma(\sigma^{-1/2}\rho\sigma^{-1/2})^{\alpha}}_{=Q_{\alpha}^{\max}}
\nn\\
&\ds\ds
+\frac{\alpha}{Q_{\alpha}^{\max}}
\Tr\sigma(\sigma^{-1/2}\rho\sigma^{-1/2})^{\alpha}
\log(\sigma^{-1/2}\rho\sigma^{-1/2}).\label{max Renyi var proof 2}
\end{align}
From \eqref{max Renyi var proof 1} and \eqref{max Renyi var proof 2} we obtain that the RHS of \eqref{max Renyi var} is 
\begin{align*}
\frac{1}{\alpha-1}\log Q_{\alpha}^{\max}-\frac{1}{\alpha-1}\log\Tr\rho
=
D_{\alpha}^{\max}(\rho\|\sigma),
\end{align*}
where the equality is by definition. 

Assume next that $\alpha\in(0,1)$. 
For any $\beta\in(0,1)$, we have
\begin{align}
\hat\rt(\hat p)\#_{\beta}\hat\rt(\omega_{\alpha})
&=
\rho\#_{\beta}\bz\frac{1}{Q_{\alpha}^{\max}}\sigma\#_{\alpha}\rho\jz
=
\frac{1}{(Q_{\alpha}^{\max})^{\beta}}
\bz\sigma\#_{\alpha}\rho\jz\#_{1-\beta}\rho
=
\frac{1}{(Q_{\alpha}^{\max})^{\beta}}
\sigma\#_{1-\beta(1-\alpha)}\rho,
\label{maxmax proof1}
\end{align}
where we used \eqref{normalized alpha-mean} in the first equality, 
\eqref{gamma mean3} in the second, and Lemma \ref{lemma:gamma mean comp} in the third.
On the other hand,
\begin{align}
\hat\rt\bz\hat p\#_{\beta}\omega_{\alpha}\jz
&=
\frac{1}{(Q_{\alpha}^{\max})^{\beta}}
\hat\rt\bz
\sum_{i:\,\Tr\sigma E_i>0}\bz\lambda_i^{\alpha}\Tr\sigma E_i\jz^{\beta}
\bz\lambda_i\Tr\sigma E_i\jz^{1-\beta}\egy_{\{i\}}\jz\nn\\
&=
\frac{1}{(Q_{\alpha}^{\max})^{\beta}}
\sum_{i=1}^{r}\lambda_i^{\alpha\beta+(1-\beta)}\sigma^{1/2} E_i\sigma^{1/2}\nn\\
&=
\frac{1}{(Q_{\alpha}^{\max})^{\beta}}
\sigma^{1/2}\bz\sigma^{-1/2}\rho_{P,\ac}\sigma^{-1/2}\jz^{\alpha\beta+(1-\beta)}\sigma^{1/2}\nn\\
&=
\frac{1}{(Q_{\alpha}^{\max})^{\beta}}
\sigma\#_{\alpha\beta+(1-\beta)}\rho,
\label{maxmax proof2}
\end{align}
where each equality follows by definition.
Since \eqref{maxmax proof1} is equal to \eqref{maxmax proof2},
Lemma \ref{lemma:maxrelentr preservation} yields 
$D^{\max}(\omega_{\alpha}\|\hat p)=
D^{\max}(\hat\rt(\omega_{\alpha})\|\hat\rt(\hat p))$.
An exactly analogous argument
yields $D^{\max}(\omega_{\alpha}\|\hat q)=
D^{\max}(\hat\rt(\omega_{\alpha})\|\hat\rt(\hat q))$.
Thus, the inequality in \eqref{maxmax inequality2} holds as an equality, proving \eqref{max Renyi var}.
\end{proof}

Our aim now is to prove that $\what{\sigma\#_{\alpha}\rho}$ is not an optimal 
$\omega$ in the variational formula \eqref{Renyi bq2} for
$D_{\alpha}^{\bary,\max}$ when $\alpha\in(0,1)$. We prove this (at least in the $2$-dimensional case)
by showing that any state $\omega$ on the line segment connecting 
$\what{\sigma\#_{\alpha}\rho}$ and the maximally mixed state 
$\pi_{\hil}:=I/d$, $d:=\dim\hil$, that is close enough to 
$\what{\sigma\#_{\alpha}\rho}$ but is not equal to it, 
gives a strictly lower value than the RHS of \eqref{max Renyi var}
when substituted into 
$\frac{\alpha}{1-\alpha}D^{\max}\bz\valt\|\rho\jz+
D^{\max}\bz\valt\|\sigma\jz
-\frac{1}{\alpha-1}\log\Tr\rho$.

\begin{lemma}\label{lemma:Dmax derivative}
Let $\rho,\sigma\in\B(\hil)\pp$, and let $P_1,\ldots,P_r\in\bP(\hil)$ and
$\lambda_1,\ldots,\lambda_r\in\bR$, be such that $\sum_{i=1}^r P_i=I$, and
\begin{align*}
\sigma^{-1/2}\rho\sigma^{-1/2}=\sum_{i=1}^r \lambda_i P_i.
\end{align*}
Then
\begin{align}
\partial_{\pi_{\hil}}&:=\frac{d}{dt}\left[\alpha D^{\max}\bz (1-t)\what{\sigma\#_{\alpha}\rho}+t\pi_{\hil}\big\|\rho\jz
+
(1-\alpha)D^{\max}\bz (1-t)\what{\sigma\#_{\alpha}\rho}+t\pi_{\hil}\big\|\sigma\jz\right]
\Bigg\vert_{t=0}\nn\\
&=
-1+\frac{1}{d}\sum_{i,j}\Tr P_i\sigma P_j\sigma\inv
\underbrace{\Big[
\alpha\log^{[1]}\bz\lambda_i^{\alpha-1},\lambda_j^{\alpha-1}\jz\lambda_i^{\alpha-1}
+(1-\alpha)\log^{[1]}\bz\lambda_i^{\alpha},\lambda_j^{\alpha}\jz\lambda_i^{\alpha}
\Big]}_{=:\Lambda_{\alpha,i,j}},
\label{Dmax derivative1}
\end{align}
where
\begin{align}\label{Lambda def}
\Lambda_{\alpha,i,j}
=
\begin{cases}\alpha(1-\alpha)
\bz\log\lambda_i-\log\lambda_j\jz\frac{
(\lambda_i-\lambda_j)}{(\lambda_i^{\alpha}-\lambda_j^{\alpha})
(\lambda_i^{1-\alpha}-\lambda_j^{1-\alpha})},&\lambda_i\ne \lambda_j,\\
1,&\lambda_i=\lambda_j.
\end{cases}
\end{align}
\end{lemma}
\begin{proof}
We defer the slightly lengthy proof to Appendix \ref{sec:proof of lemma}.
\end{proof}

Our aim is therefore to prove that $\partial_{\pi_{\hil}}<0$. For this, we will need the following:

\begin{lemma}
The following equivalent inequalities are true:
for every $\alpha\in(0,1)$, 
\begin{align}
\frac{\log\lambda-\log\eta}{\lambda-\eta}
&>
\frac{1}{\alpha}\frac{\lambda^{\alpha}-\eta^{\alpha}}{\lambda-\eta}\cdot
\frac{1}{1-\alpha}\frac{\lambda^{1-\alpha}-\eta^{1-\alpha}}{\lambda-\eta},
&\lambda,\eta\in(0,+\infty),\,\lambda\ne \eta,\label{logmean bound}\\
\frac{\log x}{x-1}
&>
\frac{1}{\alpha}\frac{x^{\alpha}-1}{x-1}
\frac{1}{1-\alpha}\frac{x^{1-\alpha}-1}{x-1},
& x\in(0,+\infty)\setminus\{1\},\\
\int_{0}^1 \frac{1}{tx+1-t}\,dt 
&>
\int_{0}^1 \frac{1}{(tx+1-t)^{\alpha}}\,dt
\int_{0}^1 \frac{1}{(tx+1-t)^{1-\alpha}}\,dt,
& x\in(0,+\infty)\setminus\{1\}.\label{integral inequality}
\end{align}
\end{lemma}
\begin{proof}
It is straightforward to verify that the above inequalities are equivalent to each other.
The inequality in \eqref{integral inequality} follows from the strict concavity of the power functions, as 
\begin{align*}
\int_{0}^1 \frac{1}{(tx+1-t)^{\gamma}}\,dt<\bz\int_{0}^1 \frac{1}{tx+1-t}\,dt\jz^{\gamma},\ds\ds\ds\gamma\in(0,1).
\end{align*}
\end{proof}

\begin{cor}\label{Lambda bound}
In the setting of Lemma \ref{lemma:Dmax derivative},
\begin{align}
\Lambda_{\alpha,i,j}>1,\ds\ds\ds i\ne j.
\end{align}
\end{cor}
\begin{proof}
Immediate from \eqref{logmean bound}.
\end{proof}

Note that we may take the $P_i$ in Lemma \ref{lemma:Dmax derivative} to be rank $1$, i.e., $P_i=\pr{e_i}$, $i=1,\ldots,d$,
for some orthonormal eigenbasis of $\sigma^{-1/2}\rho \sigma^{-1/2}$. Then 
\eqref{Dmax derivative1} can be rewritten as
\begin{align}
\partial_{\pi_{\hil}}
&=
-1+\frac{1}{d}\sum_{i,j}
\underbrace{\inner{e_i}{\sigma e_j}}_{=:S_{i,j}}
\underbrace{\inner{e_j}{\sigma\inv e_i}}_{=(S\inv)\trans_{i,j}}
\cdot
\Lambda_{\alpha,i,j}\label{Dmax derivative2}\\
&=
-1+\inner{u}{(S\star (S\inv)\trans\star\Lambda_{\alpha})u},
\label{Dmax derivative3}
\end{align}
where $u=\frac{1}{\sqrt{d}}(1,1,\ldots,1)$ and $A\star B$ denotes the component-wise 
(also called Hadamard, or Schur) product of two matrices $A$ and $B$. 

Next, note that  
\begin{align*}
(S\inv)_{j,i}=(-1)^{i+j}\frac{\det([S]_{i,j})}{\det S},
\end{align*}
where $[S]_{i,j}$ is the matrix that we get by omitting the $i$-th row and $j$-th 
column of $S$. 
Thus, \eqref{Dmax derivative2} can be rewritten as 
\begin{align*}
\partial_{\pi_{\hil}}
&=
-1+\frac{1}{d}\sum_{i=1}^d\frac{1}{\det S}
\sum_{j=1}^d (-1)^{i+j}S_{i,j}\det([S]_{i,j})\Lambda_{\alpha,i,j}.
\end{align*}
Note that for every $i$, 
\begin{align*}
\frac{1}{\det S}\sum_{j=1}^d (-1)^{i+j}S_{i,j}\det([S]_{i,j})=
(SS\inv)_{i,i}=1.
\end{align*}
\medskip

\begin{thm}\label{thm:maxRenyi strictly larger}
Let $\rho,\sigma\in\B(\hil)\pp$, where $\dim\hil=2$, and assume that
$\rho\sigma\ne\sigma\rho$. Then 
\begin{align*}
D_{\alpha}^{\bary,\max}(\rho\|\sigma)<D_{\alpha}^{\max}(\rho\|\sigma),
\ds\ds\alpha\in(0,1).
\end{align*}
\end{thm}
\begin{proof}
By Corollary \ref{cor:barycentric sclaing}, we may assume that 
$\Tr\rho=\Tr\sigma=1$.
By the above, it is sufficient to prove that $\partial_{\pi_{\hil}}<0$.
Let $(e_1,e_2)$ be an orthonormal eigenbasis of 
$\sigma^{-1/2}\rho \sigma^{-1/2}$ with corresponding eigenvalues $\lambda_1,\lambda_2$. 
By assumption, $\rho\sigma\ne\sigma\rho$, which implies that $\lambda_1\ne\lambda_2$.
(In fact, $\lambda_1=\lambda_2$ $\iff$ $\rho=c\sigma$ for some $c>0$, in which case
$c=\lambda_1=\lambda_2$.)
Writing out everything in the ONB $(e_1,e_2)$, we have
\begin{align*}
S=\half\begin{bmatrix}1+z & x-iy \\ x+iy & 1-z\end{bmatrix}
\end{align*}
with some $r:=(x,y,z)\in\bR^3$ such that $\norm{r}^2=x^2+y^2+z^2<1$, and 
\begin{align*}
(S\inv)\trans=\frac{4}{1-\norm{r}^2}\cdot
\half\begin{bmatrix}1-z & -x-iy \\ -x+iy & 1+z\end{bmatrix},
\end{align*}
whence
\begin{align*}
S\star (S\inv)\trans=\frac{1}{1-\norm{r}^2}
\begin{bmatrix}1-z^2 & -(x^2+y^2) \\ -(x^2+y^2) & 1-z^2\end{bmatrix}\,.
\end{align*}
Hence, by \eqref{Dmax derivative2} and the symmetry $\Lambda_{\alpha,1,2}=\Lambda_{\alpha,2,1}$, 
\begin{align}
\partial_{\pi_{\hil}}
&=
-1+\frac{1}{1-\norm{r}^2}\Big[1-z^2-(x^2+y^2)\Lambda_{\alpha,1,2}\Big].
\end{align}
Since $\sigma$ is not diagonal in the given ONB (otherwise it would commute with $\rho$),
we have $(x^2+y^2)>0$. Combining this with Corollary \ref{Lambda bound}, 
we get $\partial_{\pi_{\hil}}<0$, as required.
\end{proof}

\section{Conclusion and outlook}

We defined a new family of additive and monotone quantum relative entropies
in Section \ref{sec:gamma relentr}, and 
in Section \ref{sec:barycentric}, we
gave a general procedure of defining monotone multi-variate quantum R\'enyi divergences from 
monotone quantum relative entropies via a variational formula. 
For the latter, probably the biggest open question is additivity. While it is clear 
from the definition that 
if all the generating relative entropies are additive on tensor products then 
$D_P^{\bary,\qv}$ is subadditive for any $P\in\P_f(\X)$, but it is not clear whether superadditivity holds, and 
at the moment we see no general argument to establish it. Of course, additivity may be 
easily established when an explicit expression for 
$D_P^{\bary,\qv}$ is available, as is the case when all the generating relative entropies 
coincide with the Umegaki relative entropy. This leads to another open question:
to find explicit expressions for $D_P^{\bary,\qv}$ at least for the most important quantum 
relative entropies; e.g., when $D^{q_x}=D^{\max}$ for all $x$, or more generally, 
when $D^{q_x}=D^{\Um,\#_{\gamma}}$ for some $\gamma\in(0,1)$. Of course, it might turn out 
that additivity does not hold in general; in this case the natural continuation 
would be the study of the regularized quantity
\begin{align*}
\oll{D}_P^{\bary,\qv}(W)
:=
\inf_{n\in\bN}\frac{1}{n}D_P^{\bary,\qv}(W^{\potimes n})
=
\lim_{n\to+\infty}\frac{1}{n}D_P^{\bary,\qv}(W^{\potimes n}),
\end{align*}
where the equality follows from subadditivity.
While this quantity is weakly additive by definition, 
establishing its additivity is still a non-trivial problem.
We remark that there are a number of notable $2$-variable quantum R\'enyi 
divergences that are not additive, but have very important regularizations; 
for instance, $D_{\alpha}^{\#}$ from \cite{FawziFawzi2021},
or the R\'enyi divergences considered very recently in 
\cite{Frenkel_integral,HircheTomamichel_integral}.

The study of multivariate quantum R\'enyi divergences seems to be a new initiative; 
the only other paper that we are aware of dealing with the subject is 
\cite{FuruyaLashkariOuseph2023}, and partly \cite{bunth2021equivariant}. There are of course a host of interesting open problems
in this direction, the most important probably being finding 
multi-variate extensions of the ($2$-variable) Petz-type and the sandwiched R\'enyi 
divergences. These have great operational significance in quantum information theory
as quantifiers of the trade-off between the operational quantities in problems characterized 
by two competing operational quantities, 
and the goal would be to find multi-variate extensions of these R\'enyi divergences that play 
a similar role in problems characterized by multiple competing operational quantities. 
Such problems include multi-state conversion problems and state exclusion, where 
definitive results in terms of multi-variate R\'enyi divergences have been 
obtained very recently in the classical case 
\cite{farooq2023asymptotic,MishraWildeNussbaum2023}.
As for the multi-variate extension of the sandwiched R\'enyi divergences, probably the most 
natural candidate is the regularized measured R\'enyi divergence
\begin{align}\label{regmeas multiR}
\oll{D}^{\meas}_P(W):=\lim_{n\to+\infty}\frac{1}{n}D^{\meas}_P(W^{\potimes n})=
\inf_{n\in\bN}\frac{1}{n}D^{\meas}_P(W^{\potimes n}).
\end{align}
The question of course is whether this has a closed-form expression, as is the case
for two variables, and whether this quantity has analogous operational interpretations
to the $2$-variable version. 
An alternative approach could be to extend the definition given in \cite{FawziFawzi2021} to 
the multi-variate case and then take regularization; here the choice of a multi-variate 
extension of the Kubo-Ando weighted geometric means emerges as a question at the first step.

While the maximal R\'enyi divergences have no known direct operational interpretation, 
they nevertheless have a distinguished role as the largest monotone R\'enyi divergences,
as well as because of their close connection to the weighted Kubo-Ando geometric 
means, which are fundamental quantities in matrix analysis. While their definition is straightforward also in the multi-variate case,
it is an open question whether they can be given by an explicit formula
(at least when the weights are non-negative), similarly to the $2$-variable case. 
Here it may be expected that a relation of the form 
$Q_P^{\max}(W)=\Tr G_P^q(W)$ holds, where $G_P^q$ is some multi-variate extension of the 
Kubo-Ando weighted geometric means; one natural candidate would be the Karcher mean
\cite{Bhatia_Holbrook2006,Moakher_matrixmean}. 
A closely related, but not completely identical question is whether a multi-variate extension of \eqref{max Renyi var} holds as
\begin{align*}
-\log Q_P^{\max}(W)=\sum_{x\in\X}P(x)D^{\max}\bz\frac{G_P^q(W)}{\Tr G_P^q(W)}\Big\|W_x\jz
\end{align*}
with some multi-variate non-commutative weighted geometric mean $G_P^q$.

\section*{Acknowledgments}

This work was partially funded by the
National Research, Development and 
Innovation Office of Hungary via the research grants K124152, KH129601, 
K146380, FK146643, and
by the Ministry of Culture and Innovation and the National Research, Development and Innovation Office within the Quantum Information National Laboratory of Hungary (Grant No. 2022-2.1.1-NL-2022-00004).
The work of P\'eter Vrana was further supported by a Bolyai fellowship of the Hungarian Academy of Sciences, and 
by the \'UNKP-21-5 New National Excellence Program of the
Ministry for Innovation and Technology from the source of the National
Research, Development and Innovation Fund.
Gergely Bunth was partially supported by the Momentum Program of the Hungarian Academy of Sciences (grant no. LP2021-15/2021).
We thank Mark Wilde for comments on a previous version of the paper, 
and an anonymous referee for several detailed comments that greatly helped to improve the presentation.
\appendix

\section{Operator perspective and absolutely continuous part}
\label{sec:oppersp}

\subsection{Absolutely continuous part of PSD operators}
\label{sec:abscont}

Given two positive semi-definite operators $\rho,\sigma\in\B(\hil)\p$, 
the \ki{absolutely continuous part} of $\rho$ with respect to $\sigma$ is defined as
\cite{AndoLebesgue1976}
\begin{align*}
\rho_{\sigma,\ac}:=\acc{\rho}{\sigma}:=[\sigma]\rho&:=
\max\{Z\in\B(\hil)\p:\,Z^0\le \sd^0,\,Z\le \ft\}\\
&=\max\{Z\in\B(\hil)\p:\,Z^0\le \sd^0\wedge \ft^0,\,Z\le \ft\}.
\end{align*}
The existence of the maximum is not obvious. However, using the block decomposition 
\begin{align*}
\rho=
\begin{bmatrix}\sigmasupp\ft \sigmasupp & \sigmasupp\ft \sigmasupp^{\perp} \\ \sigmasupp^{\perp}\ft \sigmasupp & \sigmasupp^{\perp}\ft \sigmasupp^{\perp}\end{bmatrix}
\end{align*}
where $\sigmasupp:=\sd^0$, the variational representation for the Schur complement 
(see, e.g., \cite[Theorem 2.4]{HiaiPetzMatrix}) gives
\begin{align}
\acc{\ft}{\sd}
&=\sigmasupp\rho \sigmasupp-\sigmasupp\rho (\sigmasupp^{\perp}\rho \sigmasupp^{\perp})\inv \rho \sigmasupp,\label{ac1}
\end{align} 
where, as usual, the inverse is in the generalized sense (i.e., on the support).

\begin{rem}
Ando used the notation $[\sigma]\rho$ in \cite{AndoLebesgue1976}, while we use 
$\rho_{\sigma,\ac}$ in the main body of this paper. For some of the discussions in this Appendix, however, the notation $\acc{\rho}{\sigma}$ seems more convenient.
\end{rem}

\begin{rem}
Note that 
\begin{align}\label{ac wrt support}
\acc{\ft}{\sd}=\acc{\ft}{\sigma^0}=
\acc{\rho}{\sigma^0\wedge\rho^0};
\end{align}
in particular, it only depends on the support of $\sigma$, and not on $\sigma$ itself.
It is also obvious from the definition that 
\begin{align}\label{ac support}
\acc{\ft}{\sd}^0=\rho^0\wedge\sigma^0.
\end{align}
Another easily verifiable identity is the iteration relation
\begin{align}\label{iterated abscont}
\acc{\acc{\rho}{\sigma_1}}{\sigma_2}=\acc{\rho}{\sigma_1^0\wedge\sigma_2^0}.	
\end{align}
\end{rem}

\begin{rem}\label{rem:invertible}
It is obvious that if $V:\,\ran\rho\to\hil$ is the canonical embedding then 
\begin{align*}
\acc{\rho}{\sigma}
=V\bz\acc{V^*\rho V}{V^*(\sigma^0\wedge\rho^0)V}\jz V^*.
\end{align*}
Note that here $V^*\rho V$ is an invertible PSD operator on $\B(\ran\rho)$. 
\end{rem}

The following properties are straightforward to verify from the definition:

\begin{itemize}
\item Isometric covariance: For any $\rho,\sigma\in\B(\hil)\p$ and isometry
$V:\,\hil\to\kil$, 
\begin{align}\label{abscont isometric cov}
\acc{V\rho V^*}{V\sigma V^*}=V\acc{\rho}{\sigma}V^*.
\end{align}

\item
Block additivity: If $\rho_1^0\vee\sigma_1^0\perp\rho_2^0\vee\sigma_2^0$ then 
$\acc{\rho_1}{\sigma_1}\perp \acc{\rho_2}{\sigma_2}$, and 
\begin{align*}
\acc{\rho_1+\rho_2}{\sigma_1+\sigma_2}
=
\acc{\rho_1}{\sigma_1}+\acc{\rho_2}{\sigma_2}.
\end{align*}

\item
Joint monotonicity: 
\begin{align*}
\rho_1\le\rho_2,\,\sigma_1\le\sigma_2\ds\imp\ds
\acc{\rho_1}{\sigma_1}\le \acc{\rho_2}{\sigma_2}.
\end{align*}

\item
Joint concavity: For any
$\rho_0,\rho_1,\sigma_0,\sigma_1\in\B(\hil)\p$ and any $t\in[0,1]$, 
\begin{align*}
\acc{(1-t)\rho_0+t\rho_1}{(1-t)\sigma_0+t\sigma_1}
\ge
(1-t)\acc{\rho_0}{\sigma_0}+
t\acc{\rho_1}{\sigma_1}.
\end{align*}

\item
Monotonicity under positive maps: For any $\rho,\sigma\in\B(\hil)\p$ and 
positive linear map $\map:\,\B(\hil)\to\B(\kil)$, 
\begin{align}\label{abscont mon}
\acc{\map(\rho)}{\map(\sigma)}\ge\map\bz\acc{\rho}{\sigma}\jz.
\end{align}
\end{itemize}


It was proved recently in \cite{Kosaki_ac} that 
\begin{align}\label{ac mult}
\acc{\rho_1\otimes\rho_2}{\sigma_1\otimes \sigma_2}= \acc{\rho_1}{\sigma_1}\otimes\acc{\rho_2}{\sigma_2},
\end{align}
where $\rho_k,\sigma_k$ are bounded positive semi-definite operators on possibly infinite-dimensional Hilbert spaces. Below we give an elementary proof (that is different from the one in \cite{Kosaki_ac}) in the finite-dimensional case. 

Recall that by definition, for any $C\in\B(\hil)\p$, $\egy_{\{0\}}(C)$ is the projection onto the kernel of $C$.

\begin{lemma}\label{lemma:ac1}
For any $\rho,\sigma\in\B(\hil)\p$,
\begin{align}\label{ac2}
	\acc{\ft}{\sd}=\sigma^0\rho^{1/2}\egy_{\{0\}}(\rho^{1/2}(\sigma^0)^{\perp}\rho^{1/2})\rho^{1/2} \sigma^0.
\end{align}
\end{lemma}
\begin{proof}
Let $\sigmasupp:=\sigma^0$. Note that \eqref{ac1} can be written as 
\begin{align*}
\acc{\ft}{\sd}&=\sigmasupp\rho^{1/2}\rho^{1/2} \sigmasupp-\sigmasupp\rho^{1/2}\rho^{1/2}\sigmasupp^{\perp} (\sigmasupp^{\perp}\rho \sigmasupp^{\perp})\inv \sigmasupp^{\perp}\rho^{1/2}\rho^{1/2} \sigmasupp\\
&=
\sigmasupp\rho^{1/2}\left[I-\rho^{1/2}\sigmasupp^{\perp} (\sigmasupp^{\perp}\rho \sigmasupp^{\perp})\inv \sigmasupp^{\perp}\rho^{1/2}\right]\rho^{1/2} \sigmasupp\\
&=
\sigmasupp\rho^{1/2}\left[I-X(X^*X)\inv X^*\right]\rho^{1/2} \sigmasupp,
\end{align*}
where $X:=\rho^{1/2}\sigmasupp^{\perp}$. 
Let $f(x):=0$, $x=0$, and $f(x)=1/x$, $x>0$. Then there exists a 
(Lagrange interpolation) polynomial 
$p(x)=\sum_{k=1}^n c_kx^k$ such that $p(x)=f(x)$ at every $x\in\spec(X^*X)\cup\spec(XX^*)\cup\{0\}$, whence
$p(X^*X)=(X^*X)\inv$ and $p(XX^*)=(XX^*)\inv$. 
Thus, 
\begin{align*}
X(X^*X)\inv X^*=X\sum_{k=1}^n c_k(X^*X)^k X^*=\sum_{k=1}^n c_k(XX^*)^{k+1}
=(XX^*)\sum_{k=1}^n c_k(XX^*)^k=(XX^*)^0.
\end{align*}
Hence,
\begin{align}
\acc{\rho}{\sigma}
&=
\sigmasupp\rho^{1/2}\left[I-(XX^*)^0\right]\rho^{1/2} \sigmasupp=
\sigmasupp\rho^{1/2}\egy_{\{0\}}(\rho^{1/2}\sigmasupp^{\perp}\rho^{1/2})\rho^{1/2} \sigmasupp,
\end{align}
as stated.
\end{proof}

\begin{rem}\label{rem:kernel}
Note that $\psi\in\ker (\rho^{1/2}\sigmasupp^{\perp}\rho^{1/2})=\ker(\sigmasupp^{\perp}\rho^{1/2})$ if and only if $\sigmasupp\rho^{1/2}\psi=\rho^{1/2}\psi$. 
If, moreover, $\rho$ is invertible, then the above is further equivalent to 
$\psi\in\rho^{-1/2}(\ran \sigmasupp):=\{\rho^{-1/2}\phi:\,\phi\in\ran \sigmasupp\}$.
\end{rem}

Recall that for two projections $\sigmasupp_k\in\bP(\hil_k)$, $k=1,2$, 
\begin{align}\label{tensor}
\ran(\sigmasupp_1\otimes \sigmasupp_2)=(\ran \sigmasupp_1)\otimes(\ran \sigmasupp_2)=
\spann\{\psi_1\otimes\psi_2:\,\psi_k\in\ran \sigmasupp_k,\,k=1,2\}.
\end{align}

\begin{lemma}\label{lemma:ac2}
Let $\rho_k\in\B(\hil_k)\pp$, be invertible, and $\sigmasupp_k\in\bP(\hil_k)$, $k=1,2$. Then 
\begin{align*}
\egy_{\{0\}}((\rho_1\otimes\rho_2)^{1/2}(\sigmasupp_1\otimes \sigmasupp_2)^{\perp}(\rho_1\otimes\rho_2)^{1/2})
=
\egy_{\{0\}}(\rho_1^{1/2}\sigmasupp_1^{\perp}\rho_1^{1/2})
\otimes
\egy_{\{0\}}(\rho_2^{1/2}\sigmasupp_2^{\perp}\rho_2^{1/2}).
\end{align*}
\end{lemma}
\begin{proof}
By Remark \ref{rem:kernel},
\begin{align*}
\ker\bz(\rho_1\otimes\rho_2)^{1/2}(\sigmasupp_1\otimes \sigmasupp_2)^{\perp}(\rho_1\otimes\rho_2)^{1/2}\jz
&=
(\rho_1\otimes\rho_2)^{-1/2}(\underbrace{\ran(\sigmasupp_1\otimes \sigmasupp_2)}_{=(\ran \sigmasupp_1)\otimes(\ran \sigmasupp_2)})\\
&=
\rho_1^{-1/2}(\ran \sigmasupp_1)\otimes\rho_2^{-1/2}(\ran \sigmasupp_2),
\end{align*}
and the assertion follows by \eqref{tensor}.
\end{proof}

\begin{prop}
Let $\rho_k,\sigma_k\in\B(\hil_k)\p$, $k=1,2$. Then 
\eqref{ac mult} holds.
\end{prop}
\begin{proof}
Let $\sigmasupp_k:=\sigma_k^0$, and 
$V_k:\,\ran\rho_k\to\hil_k$, $k=1,2$, be the canonical embeddings. 
Since $V_1\otimes V_2$ is the canonical embedding of $\ran(\rho_1\otimes\rho_2)$ into 
$\hil_1\otimes\hil_2$, and 
$(\sigmasupp_1\otimes \sigmasupp_2)\wedge(\rho_1\otimes\rho_2)^0
=\bz \sigmasupp_1\wedge\rho_1^0\jz\otimes \bz \sigmasupp_2\wedge\rho_2^0\jz$,
Remark \ref{rem:invertible} shows that it is sufficient to prove 
\eqref{ac mult} in the case when $\rho_1,\rho_2$ are invertible. 
The statement then follows immediately from Lemmas \ref{lemma:ac1}
and \ref{lemma:ac2}.
\end{proof}

\subsection{Operator perspective function}
\label{sec:oppersp2}

With the help of the absolute continuous part, an explicit expression can be given for the 
limit in the definition \eqref{persp def}--\eqref{persp extension} of the 
operator perspective function of an operator convex function $f:\,(0,+\infty)\to\bR$, under 
the extra assumption that 
\begin{align}\label{limits}
f(0^+):=\lim_{x\searrow 0}f(x),\ds\ds\ds
\tilde f(0^+):=\lim_{x\searrow 0}\tilde f(x)=
\lim_{x\to+\infty}f(x)/x,
\end{align}
are finite, where $\tilde f(x):=xf(1/x)$, $x\in(0,+\infty)$, is the transpose function of $f$.
For the rest, we assume that $f:\,(0,+\infty)\to\bR$ is continuous. 
We will always use the following convention:
\begin{align*}
&\text{if}\ds f(0^+)\in\bR\ds\text{then}\ds f(0):=f(0^+),\\
&\text{if}\ds \tilde f(0^+)\in\bR\ds\text{then}\ds \tilde f(0):=\tilde f(0^+),
\end{align*}
i.e., if $f(0^+)$ is finite then we extend 
$f$ to be a function on $[0,+\infty)$ with $f(0):=f(0^+)$, and if 
$\tilde f(0^+)$ is finite then we extend 
$\tilde f$ to be a function on $[0,+\infty)$ with $\tilde f(0):=\tilde f(0^+)$.
Note that 
\begin{align}
f(0^+)\in\bR,\ds \rho^0\le\sigma^0\ds\imp\ds &\exists\,\lim_{\ep\searrow 0}
(\sigma+\ep I)^{1/2}f\bz(\sigma+\ep I)^{-1/2}(\rho+\ep I)
(\sigma+\ep I)^{-1/2}\jz(\sigma+\ep I)^{1/2}\nn\\
&=
\sigma^{1/2}f\bz\sigma^{-1/2}\rho\sigma^{-1/2}\jz\sigma^{1/2},\label{persp expp2}\\
\tilde f(0^+)\in\bR,\ds \rho^0\ge\sigma^0\ds\imp\ds &\exists\,\lim_{\ep\searrow 0}
(\rho+\ep I)^{1/2}\tilde f\bz(\rho+\ep I)^{-1/2}(\sigma+\ep I)
(\rho+\ep I)^{-1/2}\jz(\rho+\ep I)^{1/2}\nn\\
&=
\rho^{1/2}\tilde f\bz\rho^{-1/2}\sigma\rho^{-1/2}\jz\rho^{1/2},\label{persp expp4}
\end{align}
simply due to the continuity of functional calculus. 
In particular,
\begin{align}
&f(0^+)\in\bR\ds\imp\ds 
\persp{f}(\rho,\sigma+\ep I)=
(\sigma+\ep I)^{1/2}f\bz(\sigma+\ep I)^{-1/2}\rho(\sigma+\ep I)^{-1/2}\jz(\sigma+\ep I)^{1/2},\ds\ds\ep\in(0,+\infty)\label{persp expp5}\\
&\tilde f(0^+)\in\bR\ds\imp\ds 
\persp{\tilde f}(\sigma,\rho+\ep I)
=
(\rho+\ep I)^{1/2}\tilde f\bz(\rho+\ep I)^{-1/2}\sigma(\rho+\ep I)^{-1/2}\jz(\rho+\ep I)^{1/2},
\ds\ds\ep\in(0,+\infty).\label{persp expp6}
\end{align}
Moreover, 
\begin{align}\label{oppersp equal supports}
\rho^0=\sigma^0\ds\imp\ds
\persp{f}(\rho\|\sigma)
=\sigma^{1/2}f\bz\sigma^{-1/2}\rho\sigma^{-1/2}\jz\sigma^{1/2}
=
\rho^{1/2}\tilde f\bz\rho^{-1/2}\sigma\rho^{-1/2}\jz\rho^{1/2},
\end{align}
where the last two expressions are interpreted with considering $\rho$ and $\sigma$
as invertible operators on their joint support, and therefore are well-defined even if
$f(0^+)$ or $\tilde f(0^+)$ is not finite.

It has been shown in \cite[Theorem 17]{Kosaki_ac} that
if $f:\,(0,+\infty)\to[0,+\infty)$ is an operator monotone function such that 
$f(0^+)=0=\tilde f(0^+)$, then 
for any $\rho,\sigma\in\B(\hil)\p$,
\begin{align}\label{perspective abscont}
\persp{f}(\rho,\sigma)
=
\persp{f}(\rho_{\sigma,\ac},\sigma)
=
\persp{f}(\rho,\sigma_{\rho,\ac})
=
\persp{f}(\rho_{\sigma,\ac},\sigma_{\rho,\ac}),
%
\end{align}
combining which with \eqref{ac wrt support}
yields also the expression
\begin{align}\label{perspective abscont2}
	\persp{f}(\rho,\sigma)&=
	\persp{f}(\rho_{P,\ac},\sigma_{P,\ac}),
\end{align}
where $P:=\rho^0\wedge\sigma^0$.
Combining the above with \cite[Proposition 3.25(i)]{HiaiMosonyi2017} yields
the expressions \eqref{persp explicit2-7}--\eqref{persp explicit2-6} 
below for the operator perspective function; we give the 
details for readers' convenience.

\begin{cor}\label{cor:oppersp explicit2}
Let $f:\,(0,+\infty)\to\bR$ be an operator convex function such that 
$f(0^+),\tilde f(0^+)<+\infty$. Then for any $\rho,\sigma\in\B(\hil)\p$, 
\begin{align}
\exists\,\persp{f}(\rho,\sigma)
&:=
\lim_{\ep\searrow 0}
(\sigma+\ep I)^{1/2}f\bz(\sigma+\ep I)^{-1/2}(\rho+\ep I)
(\sigma+\ep I)^{-1/2}\jz(\sigma+\ep I)^{1/2}\label{persp explicit2-7}
\\
&=\lim_{\ep\searrow 0}
(\sigma+\ep I)^{1/2}f\bz(\sigma+\ep I)^{-1/2}\rho
(\sigma+\ep I)^{-1/2}\jz(\sigma+\ep I)^{1/2}\label{persp explicit2-0}\\
%
&=
\sigma^{1/2}f\bz\sigma^{-1/2}\rho_{\sigma,\ac}\sigma^{-1/2}\jz\sigma^{1/2}
+\tilde f(0^+)(\rho-\rho_{\sigma,\ac})\label{persp explicit2-1}\\
&=
\sigma_{P,\ac}^{1/2}f\bz\sigma_{P,\ac}^{-1/2}\rho_{P,\ac}\sigma_{P,\ac}^{-1/2}\jz\sigma_{P,\ac}^{1/2}+\tilde f(0^+)(\rho-\rho_{P,\ac})+f(0^+)(\sigma-\sigma_{P,\ac})
\label{persp explicit2-2}\\
&=\lim_{\ep\searrow 0}
(\rho+\ep I)^{1/2}\tilde f\bz(\rho+\ep I)^{-1/2}(\sigma+\ep I)(\rho+\ep I)^{-1/2}\jz(\rho+\ep I)^{1/2}\label{persp explicit2-3}\\
&=\lim_{\ep\searrow 0}
(\rho+\ep I)^{1/2}\tilde f\bz(\rho+\ep I)^{-1/2}\sigma(\rho+\ep I)^{-1/2}\jz(\rho+\ep I)^{1/2}\label{persp explicit2-4}\\
&=
\rho^{1/2}\tilde f\bz\rho^{-1/2}\sigma_{\rho,\ac}\rho^{-1/2}\jz\rho^{1/2}
+f(0^+)(\sigma-\sigma_{\rho,\ac})\label{persp explicit2-5}\\
&=
\rho_{P,\ac}^{1/2}\tilde f\bz\rho_{P,\ac}^{-1/2}\sigma_{P,\ac}\rho_{P,\ac}^{-1/2}\jz\rho_{P,\ac}^{1/2}+\tilde f(0^+)(\rho-\rho_{P,\ac})+f(0^+)(\sigma-\sigma_{P,\ac}),
\label{persp explicit2-6}
\end{align}
where $P:=\rho^0\wedge\sigma^0$. 
\end{cor}
\begin{proof}
Assume first that $f$ is operator monotone with $f(0^+)=0=\tilde f(0^+)$.
The existence of the limit in \eqref{persp explicit2-7} follows from the fact that $f$ is 
non-negative and operator monotone, and by definition \eqref{persp extension}, the limit is equal to 
$\persp{f}(\rho,\sigma)$. 
By the operator monotonicity of $f$,
\begin{align*}
\persp{f}(\rho,\sigma)\le
\persp{f}(\rho,\sigma+\ep I)
\le 
\persp{f}(\rho+\ep I,\sigma+\ep I).
\end{align*}
Combining this with \eqref{persp expp5} yields the equality in 
\eqref{persp explicit2-0}.
By 
\eqref{perspective abscont} we have
\begin{align*}
\persp{f}(\rho,\sigma)=\persp{f}(\rho_{\sigma,\ac},\sigma)
&=
\sigma^{1/2}f\bz\sigma^{-1/2}\rho_{\sigma,\ac}\sigma^{-1/2}\jz\sigma^{1/2},
\end{align*} 
where the second equality is immediate from \eqref{persp expp2} due to $\rho_{\sigma,\ac}^0\le\sigma^0$.
This proves the equality in \eqref{persp explicit2-1}, and the 
equality in \eqref{persp explicit2-2} follows in the same way
using \eqref{perspective abscont2}. The equality in \eqref{persp explicit2-3} follows from 
\eqref{persp explicit2-7} by \eqref{persp trans}, and the rest of the equalities follow as above. 

Next we consider the general case. Let $g(x):=f(0^+)+\tilde f(0^+)x-f(x)$, $x\in[0,+\infty)$, which, by Lemma \ref{lemma:opconv to opmon} below, is a non-negative operator monotone function, and obviously, $g(0^+)=g(0)=0=\tilde g(0^+)$. Then 
\begin{align*}
&\lim_{\ep\searrow 0}
(\sigma+\ep I)^{1/2} f\bz(\sigma+\ep I)^{-1/2}(\rho+\ep I)(\sigma+\ep I)^{-1/2}\jz
(\sigma+\ep I)^{1/2}\\
&\ds=
\lim_{\ep\searrow 0}\left\{
f(0^+)(\sigma+\ep I)+\tilde f(0^+)(\rho+\ep I)-
(\sigma+\ep I)^{1/2} g\bz(\sigma+\ep I)^{-1/2}(\rho+\ep I)(\sigma+\ep I)^{-1/2}\jz
(\sigma+\ep I)^{1/2}\right\}\\
&\ds=
f(0^+)\sigma+\tilde f(0^+)\rho-
\sigma^{1/2}g\bz\sigma^{-1/2}\rho_{\sigma,\ac}\sigma^{-1/2}\jz\sigma^{1/2}\\
&\ds=
\sigma^{1/2}f\bz\sigma^{-1/2}\rho_{\sigma,\ac}\sigma^{-1/2}\jz\sigma^{1/2}+
\tilde f(0^+)(\rho-\rho_{\sigma,\ac}),
\end{align*}
where the second equality follows from 
applying the equality between 
\eqref{persp explicit2-7} and \eqref{persp explicit2-1} to $g$ in place of $f$,
and the rest are obvious. This proves the equality between 
\eqref{persp explicit2-7} and \eqref{persp explicit2-1}, 
and the equalities in  \eqref{persp explicit2-0} and \eqref{persp explicit2-2}
follow by analogous arguments. The equalities in  
\eqref{persp explicit2-3}--\eqref{persp explicit2-6} follow again from the above by using 
\eqref{persp trans}.
\end{proof}

The following observation can be proved from an integral representation of operator convex functions given in \cite[Proposition 8.4]{HMPB}, and 
has been used e.g., in the proof of \cite[Proposition 3.25]{HiaiMosonyi2017}
and in \cite[Section 6.1]{Matsumoto_newfdiv}. Here we give a very short and elementary 
proof that does not rely on any integral representation.

\begin{lemma}\label{lemma:opconv to opmon}
Let $f:\,[0,+\infty)\to\bR$ be an operator convex function such that 
$\tilde f(0^+)<+\infty$. Then 
\begin{align}\label{g def}
g(x):=f(0^+)+\tilde f(0^+)x-f(x),
\ds\ds\ds x\in[0,+\infty),
\end{align}
defines a non-negative operator concave function on $[0,+\infty)$, which is thus also operator monotone.
\end{lemma}
\begin{proof}
Note that the convexity of $f$ implies that 
$x\mapsto \frac{f(x)-f(0^+)}{x}$ is monotone increasing, whence
\begin{align*}
\tilde f(0^+)=\lim_{y\to+\infty} \frac{f(y)-f(0^+)}{y}\ge  \frac{f(x)-f(0^+)}{x},\ds\ds\ds
x\in(0,+\infty).
\end{align*}
Thus,
\begin{align*}
g(x):=f(0^+)+\tilde f(0^+)x-f(x)=x\bz\tilde f(0^+)-\frac{f(x)-f(0^+)}{x}\jz\ge 0,
\ds\ds\ds x\in(0,+\infty),
\end{align*}
whence $g$ is non-negative and operator concave, and this is well known to imply that it is also operator monotone; see, e.g., \cite[Theorem V.2.5.]{Bhatia}.
\end{proof}

\begin{rem}
It is easy to see that if $\rho$ and $\sigma$ are commuting then the equalities in 
\eqref{persp explicit2-7}--\eqref{persp explicit2-6} hold even under the much weaker assumption that $f$ is continuous and $f(0^+)^+$ and $\tilde f(0^+)$ are finite. 
\end{rem}

\subsection{Kubo-Ando geometric means}
\label{sec:KA}

Let us now apply the above for the power functions.
For positive definite operators $\rho,\sigma\in\B(\hil)\pp$ 
and $\gamma\in\bR\setminus\{0,1\}$, let  
\begin{align}
	\sigma\#_{\gamma}\rho
	&:=
	\sigma^{1/2}\bz \sigma^{-1/2}\rho\sigma^{-1/2}\jz^{\gamma}\sigma^{1/2}\label{gamma mean1}\\
	&=
	\persp{(\id^{\gamma})}(\rho,\sigma)
	=
	\persp{(\id^{1-\gamma})}(\sigma,\rho)\label{gamma mean2}\\
	&=
	\rho^{1/2}\bz \rho^{-1/2}\sigma\rho^{-1/2}\jz^{1-\gamma}\rho^{1/2}
	=
	\rho\#_{1-\gamma}\sigma,\label{gamma mean0}
\end{align}
(see \eqref{persp trans} for the equality in \eqref{gamma mean2}).
The definition is extended to pairs of general PSD operators
as
\begin{align}\label{KA extension}
	\sigma\#_{\gamma}\rho:=\lim_{\ep\searrow 0}(\sigma+\ep I)\#_{\gamma}(\rho+\ep I),
\end{align}
whenever the limit exists. In particular, 
we have
\begin{align}
	& \rho^0\le\sigma^0\ds\imp\ds
	\sigma\#_{\gamma}\rho
	=
	\sigma^{1/2}\bz \sigma^{-1/2}\rho\sigma^{-1/2}\jz^{\gamma}\sigma^{1/2},\ds\ds
	\gamma\in(0,+\infty)\setminus\{1\},\label{gamma mean4}\\
	& \rho^0\ge\sigma^0\ds\imp\ds
	\sigma\#_{\gamma}\rho
	=
	\rho^{1/2}\bz \rho^{-1/2}\sigma\rho^{-1/2}\jz^{1-\gamma}\rho^{1/2},\ds\ds
	\gamma\in(-\infty,0),
\end{align}
where the negative powers are taken on the support of the respective operators;
see, e.g., \cite{KA} and \cite[Proposition 3.26]{HiaiMosonyi2017}.

For $\gamma\in(0,1)$, $\#_{\gamma}$ is called
the \ki{Kubo-Ando $\gamma$-weighted geometric mean} \cite{KA}.
In this case Corollary \ref{cor:oppersp explicit2} yields that the limit in \eqref{KA extension} exists for any $\rho,\sigma\in\B(\hil)\p$, and is given explicitly by 
\begin{align}
\sigma\#_{\gamma}\rho
&=
\sigma^{1/2}\bz\sigma^{-1/2}\rho_{\sigma,\ac}\sigma^{-1/2}\jz^{\gamma}\sigma^{1/2}
=
\sigma_{P,\ac}^{1/2}\bz\sigma_{P,\ac}^{-1/2}\rho_{P,\ac}\sigma_{P,\ac}^{-1/2}\jz^{\gamma}\sigma_{P,\ac}^{1/2}\label{KA explicit1}\\
&=
\rho^{1/2}\bz\rho^{-1/2}\sigma_{\rho,\ac}\rho^{-1/2}\jz^{1-\gamma}\rho^{1/2}
=
\rho_{P,\ac}^{1/2}\bz\rho_{P,\ac}^{-1/2}\sigma_{P,\ac}\rho_{P,\ac}^{-1/2}\jz^{1-\gamma}\rho_{P,\ac}^{1/2},\label{KA explicit2}
\end{align}
where again $P:=\rho^0\wedge\sigma^0$.

When $\gamma\in\{0,1\}$, we define
\begin{align}
\sigma\#_0\rho&:=\lim_{\gamma\searrow 0}\sigma\#_{\gamma}\rho=
\sigma^{1/2}\bz\sigma^{-1/2}\rho_{\sigma,\ac}\sigma^{-1/2}\jz^0\sigma^{1/2}
=
\rho^{1/2}\bz\rho^{-1/2}\sigma_{\rho,\ac}\rho^{-1/2}\jz\rho^{1/2}=\sigma_{\rho,\ac},
\label{KA0}\\
\sigma\#_1\rho&:=\lim_{\gamma\nearrow 1}\sigma\#_{\gamma}\rho=
\rho^{1/2}\bz\rho^{-1/2}\sigma_{\rho,\ac}\rho^{-1/2}\jz^0\rho^{1/2}=
\sigma^{1/2}\bz\sigma^{1/2}\rho_{\sigma,\ac}\sigma^{1/2}\jz\sigma^{1/2}
=
\rho_{\sigma,\ac},
\label{KA1}
\end{align}
where the equalities follow from \eqref{KA explicit1}--\eqref{KA explicit2}.
Note that this is different from the usual convention, where 
$\sigma\#_0\rho$ is defined to be $\sigma$ and 
$\sigma\#_1\rho$ is defined to be $\rho$, which is what we would obtain
if we defined these quantities via \eqref{gamma mean1}--\eqref{KA extension} for $\gamma=0$ and $\gamma=1$. Note, however, that with the usual convention, 
$\gamma\mapsto \sigma\#_{\gamma}\rho$ is not right continuous at $\gamma=0$
unless $\sigma^0\le\rho^0$, and 
it is not left continuous at $\gamma=1$
unless $\sigma^0\ge\rho^0$.

It is well known and easy to see that for $\gamma\in(0,1)$,
the $\gamma$-weighted geometric means are monotone continuous in the sense that for any 
functions $(0,+\infty)\ni\ep\mapsto \rho_{\ep}\in\B(\hil)\p$ and 
$(0,+\infty)\ni\ep\mapsto \sigma_{\ep}\in\B(\hil)\p$
that are monotone decreasing in the PSD order,
\begin{align}\label{geom moncont}
	(\sigma_{\ep}\#_{\gamma}\rho_{\ep})\searrow 
	(\lim\nolimits_{\ep\searrow 0}\sigma_{\ep})\#_{\gamma}(\lim\nolimits_{\ep\searrow 0}\rho_{\ep}),\ds\ds\ds
	\gamma\in(0,1),
\end{align}
as $\ep\searrow 0$. 
However, this is clearly not true for $\gamma\in\{0,1\}$ with the definitions in \eqref{KA0}--\eqref{KA1}.

It is clear from \eqref{gamma mean1}--\eqref{KA1} that for any $\gamma\in[0,1]$,
and any $\rho,\rho_1,\rho_2,\sigma,\sigma_1,\sigma_2\in\B(\hil)\p$,
\begin{align}
	& (\sigma\#_{\gamma}\rho)^0=\sigma^0\wedge\rho^0;\label{wmean support}\\
	& \sigma\#_{\gamma}\rho
	=
	\rho\#_{1-\gamma}\sigma;\label{gamma mean3}\\
	&\rho_1\le\rho_2,\,\sigma_1\le\sigma_2\ds\imp\ds
	\sigma_1\#_{\gamma}\rho_1\le
	\sigma_2\#_{\gamma}\rho_2,\label{wmean montone}
\end{align}
and \eqref{KA explicit1}--\eqref{KA1} combined with \eqref{ac mult} yield that 
for any $\rho_k,\sigma_k\in\B(\hil_k)\pne$, $k=1,2$, 
\begin{align}\label{mean tensor}
	(\sigma_1\otimes\sigma_2)\#_{\gamma}(\rho_1\otimes\rho_2)
	=
	(\sigma_1\#_{\gamma}\rho_1)\otimes(\sigma_2\#_{\gamma}\rho_2).
\end{align}
Tensor multiplicativity can also be seen by noting that
for 
any sequence $\ep_1>\ep_2>\ldots\to 0$, and any $\gamma\in(0,1)$,
\begin{align*}
	&\left[(\sigma_1+\ep_n I)\otimes(\sigma_2+\ep_n I)\right]\#_{\gamma}
	\left[(\rho_1+\ep_n I)\otimes(\rho_2+\ep_n I)\right]\\
	&\ds\ds\ds\ds=
	\left[(\sigma_1+\ep_n I)\#_{\gamma}(\rho_1+\ep_n I)\right]
	\otimes
	\left[(\sigma_2+\ep_n I)\#_{\gamma}(\rho_2+\ep_n I)\right],
\end{align*}
follows immediately from the definition \eqref{gamma mean1},
and taking the limit $n\to+\infty$ yields \eqref{mean tensor} for $\gamma\in(0,1)$, from 
which \eqref{mean tensor} follows also for $\gamma\in\{0,1\}$ by taking the respective limit according to \eqref{KA0} or \eqref{KA1}.

The following is well known:
\begin{lemma}\label{lemma:KA Holder}
	For any $\rho,\sigma\in\B(\hil)\p$ and any $\gamma\in[0,1]$, 
	\begin{align}\label{KA Holder}
		\Tr(\sigma\#_{\gamma}\rho)\le\Tr\rho^{\gamma}\sigma^{1-\gamma}
		\le
		(\Tr\rho)^{\gamma}(\Tr\sigma)^{1-\gamma},
	\end{align}
	and for $\gamma\in(0,1)$, equality holds in the second inequality if and only if 
	$\rho$ and $\sigma$ are linearly dependent. In particular, for states
	$\rho,\sigma\in\S(\hil)$ and $\gamma\in(0,1)$,
	\begin{align}\label{KA Holder2}
		\Tr(\sigma\#_{\gamma}\rho)\le 1,\ds\text{and}\ds
		\Tr(\sigma\#_{\gamma}\rho)= 1\ds\iff\ds\rho=\sigma.
	\end{align}
\end{lemma}
\begin{proof}
For $\gamma\in(0,1)$, the
second inequality in \eqref{KA Holder} and its equality condition is a simple consequence
	of H\"older's inequality,
	the first inequality in \eqref{KA Holder} follows, e.g., from 
	\cite[Example 4.5]{HiaiMosonyi2017} or 
	\cite{AndoHiai1994}, and the
	assertion in \eqref{KA Holder2} follows immediately from the above.
	The inequalities in \eqref{KA Holder} for $\gamma\in\{0,1\}$ follow from the case
	$\gamma\in(0,1)$ by taking the respective limits $\gamma\searrow 0$ and $\gamma\nearrow 1$, according to \eqref{KA0}--\eqref{KA1}.
\end{proof}

The inequality between the first and the last terms in \eqref{KA Holder} is a special case of the following lemma. For a proof of the lemma
for $\gamma\in(0,1)$ see, e.g., \cite[Proposition 3.30]{HiaiMosonyi2017};
the cases $\gamma\in\{0,1\}$ follow immediately 
by taking the respective limits $\gamma\searrow 0$ and $\gamma\nearrow 1$, according to \eqref{KA0}--\eqref{KA1}.
\begin{lemma}\label{lemma:geommean monotone}
	For any $\rho,\sigma\in\B(\hil)\pne$, any $\gamma\in[0,1]$, and any positive linear map
	$\map:\,\B(\hil)\to\B(\kil)$, 
	\begin{align*}
		\map\bz\rho\#_{\gamma}\sigma\jz\le \map(\rho)\#_{\gamma}\map(\sigma).
	\end{align*}
\end{lemma}

\subsection{Maximal $f$-divergences}
\label{sec:maxfdiv}

The observations in Appendix \ref{sec:oppersp2} can be used to prove a regularity property of Matsumoto's maximal $f$-divergences, defined in \cite{Matsumoto_newfdiv}.
Most of the what follows below have been covered in greater generality in 
\cite{Hiai_fdiv_max} in the most general von Neumann algebra setting; see also 
\cite{HiaiUeadaWada2022} for further extensions. 
Our presentation below is more elementary, aimed for the general readership with a background in matrix analysis or quantum information theory. 

Recall that 
for a convex function $f:\,(0,+\infty)\to\bR$ and 
$\rho,\sigma\in[0,+\infty)^{\I}\setminus\{0\}$, 
where $\I$ is a non-empty finite set,
the classical $f$-divergence \cite{Csiszar_fdiv} of $\rho$ and $\sigma$ is defined as
\begin{align}
D_f(\rho\|\sigma)&:=
\sum_{i\in\I}\persp{f}(\rho(i),\sigma(i))=
\lim_{\ep\searrow 0}\sum_{i\in\I}(\sigma(i)+\ep)
f\bz\frac{\rho(i)+\ep}{\sigma(i)+\ep}\jz\nn\\
&=
\sum_{i:\,\rho(i)\sigma(i)>0}\sigma(i)f\bz\frac{\rho(i)}{\sigma(i)}\jz
+f(0^+)\sum_{i:\,\rho(i)=0}\sigma(i)
+\tilde f(0^+)\sum_{i:\,\sigma(i)=0}\rho(i),
\label{classical fdiv def}
\end{align}
where $f(0^+)$ and $\tilde f(0^+)$ are given in \eqref{limits}.
For $\rho,\sigma\in\B(\hil)\pne$, Matsumoto's \ki{maximal $f$-divergence} 
\cite{Matsumoto_newfdiv} is defined as
\begin{align}\label{maxfdiv def}
D_f^{\max}(\rho\|\sigma):=\inf\{D_f(p\|q):\,(p,q,\rt)\text{ reverse test for }(\rho,\sigma)\};
\end{align}
see Example \ref{ex:minmax extension} for the definition of a reverse test.
It was shown in \cite[Theorem 3]{Matsumoto_newfdiv} that if $f$ is operator convex 
with $f(0^+)=0$ then 
\begin{align}\label{Matsumoto explicit}
D_f^{\max}(\rho\|\sigma)
=
\Tr\sigma f\bz\sigma^{-1/2}\rho_{\sigma,\ac}\sigma^{-1/2}\jz
+
\tilde f(0^+)\Tr(\rho-\rho_{\sigma,\ac}),
\end{align}
and a trivial argument shows that the same holds also in the more general case that
$f(0^+)<+\infty$.
Combining this with Corollary \ref{cor:oppersp explicit2} yields the following
regularity properties of the maximal $f$-divergences.

\begin{cor}\label{cor:maxfdiv regularity}
Let $f:\,(0,+\infty)\to\bR$ be an operator convex function such that 
$f(0^+)<+\infty$. Then 
\begin{align}
D_f^{\max}(\rho\|\sigma)
&=
\lim_{\ep\searrow 0}D_f^{\max}(\rho+\ep I\|\sigma+\ep I)\label{maxfdiv limit1}\\
&=
\lim_{\ep\searrow 0}
\Tr(\sigma+\ep I) f\bz(\sigma+\ep I)^{-1/2}(\rho+\ep I)(\sigma+\ep I)^{-1/2}\jz
\label{maxfdiv limit2}\\
&=\lim_{\ep\searrow 0}
\Tr(\rho+\ep I)^{1/2}\tilde f\bz(\rho+\ep I)^{-1/2}(\sigma+\ep I)(\rho+\ep I)^{-1/2}\jz(\rho+\ep I)^{1/2}\label{maxfdiv limit2-1}\\
&=
\lim_{\ep\searrow 0}D_f^{\max}(\rho\|\sigma+\ep I)\label{maxfdiv limit3}\\
&=
\lim_{\ep\searrow 0}
\Tr(\sigma+\ep I) f\bz(\sigma+\ep I)^{-1/2}\rho(\sigma+\ep I)^{-1/2}\jz.
\label{maxfdiv limit4}
\end{align}
\end{cor}
\begin{proof}
The equality of \eqref{maxfdiv limit1} and \eqref{maxfdiv limit2} is obvious from 
\eqref{Matsumoto explicit} whenever the limit in 
\eqref{maxfdiv limit2} exists, and the same holds for the equality of 
\eqref{maxfdiv limit3} and \eqref{maxfdiv limit4}.
The equality of \eqref{maxfdiv limit2} and \eqref{maxfdiv limit2-1} is immediate from 
\eqref{persp trans}, whenever either of the limits in 
\eqref{maxfdiv limit2} or \eqref{maxfdiv limit2-1} exists.
 
If $\rho^0\le\sigma^0$ then $\rho_{\sigma,\ac}=\rho$, and
the existence of the limits in  \eqref{maxfdiv limit2} and \eqref{maxfdiv limit4}
and their equality to \eqref{Matsumoto explicit} 
are obvious.

If $\tilde f(0^+)=+\infty$ and $\rho^0\nleq\sigma^0$ then 
the expression in \eqref{Matsumoto explicit} is equal to $+\infty$, 
and a straightforward computation show the same for 
\eqref{maxfdiv limit2} and \eqref{maxfdiv limit4}; see the proof of 
\cite[Proposition 3.29]{HiaiMosonyi2017} for details. 

Finally, if $\tilde f(0^+)<+\infty$ then the existence of the limits in 
\eqref{maxfdiv limit2} and \eqref{maxfdiv limit4}, as well as their equality to 
\eqref{Matsumoto explicit}, follow from Corollary \ref{cor:oppersp explicit2}.
\end{proof}


In the rest of this section, we give some comments on the formula 
\eqref{Matsumoto explicit} and some alternative expressions for $D_f^{\max}$. 

First, we note that a seemingly different quantity was studied in \cite{HiaiMosonyi2017}, defined 
for any operator convex function $f:\,(0,+\infty)\to\bR$ and $\rho,\sigma\in\B(\hil)\pne$ as
\begin{align*}
\what D_f(\rho\|\sigma):=\lim_{\ep\searrow 0}\Tr \persp{f}(\rho+\ep I,\sigma+\ep I).
\end{align*}
The existence of the limit follows immediately from (operator) convexity, as well as the joint convexity in 
$\rho$ and $\sigma$, and one can easily see that this is a quantum extension of the classical $f$-divergence in the sense of Definition \ref{def:q-extension}. It is also straightforward to verify that this quantity is stable under tensoring with the maximally mixed state, and therefore, by joint convexity, it is also monotone under CPTP maps, according to Lemma \ref{lemma:jointconc from mon}. Hence, 
\begin{align*}
D_f^{\max}(\rho\|\sigma)\ge \what D_f(\rho\|\sigma),\ds\ds\ds\rho,\sigma\in\B(\hil), 
\end{align*}
as explained in Example \ref{ex:minmax extension}. 
It is straightforward to verify that 
\begin{align*}
D_f^{\max}(\rho\|\sigma)=+\infty=\what D_f(\rho\|\sigma),\ds\ds\text{if}\ds\ds
f(0^+)=+\infty\text{ and }\rho^0\ngeq\sigma^0,\ds\text{or}\ds
\tilde f(0^+)=+\infty\text{ and }\rho^0\nleq\sigma^0,
\end{align*}
where the first equality is obvious by definition, and for the second, 
see, e.g., the proof of 
\cite[Proposition 3.29]{HiaiMosonyi2017}.
Obviously, 
$\what D_f(\rho\|\sigma)=\Tr\persp{f}(\rho,\sigma)$ whenever $\persp{f}(\rho,\sigma)$ is well-defined; in particular, we trivially have
\begin{align}
D_f^{\max}(\rho\|\sigma)\ge \what D_f(\rho\|\sigma)
&=
\begin{cases}
\Tr\sigma f\bz\sigma^{-1/2}\rho\sigma^{-1/2}\jz=
\Tr\rho\tilde f\bz\rho^{-1/2}\sigma\rho^{-1/2}\jz,&\rho^0=\sigma^0,\\
\Tr\sigma f\bz\sigma^{-1/2}\rho\sigma^{-1/2}\jz,&f(0^+)<+\infty,\,\rho^0\le\sigma^0,\\
\Tr\rho\tilde f\bz\rho^{-1/2}\sigma\rho^{-1/2}\jz,&
\tilde f(0^+)<+\infty,\,\rho^0\ge\sigma^0,
\end{cases}\label{Dmax lower bound1}
\end{align}
and if both $f(0^+)$ and $\tilde f(0^+)$ are finite then 
by Corollary \ref{cor:oppersp explicit2},
$\what D_f(\rho\|\sigma)$ is equal to the trace of any of the expressions 
in \eqref{persp explicit2-1}, \eqref{persp explicit2-2}, \eqref{persp explicit2-5},
\eqref{persp explicit2-6}, giving 
\begin{align}
D_f^{\max}(\rho\|\sigma)&\ge \what D_f\rho\|\sigma)\nn\\
&=
\Tr \sigma_{P,\mathrm{ac}}^{1/2}f\bz \sigma_{P,\mathrm{ac}}^{-1/2}\rho_{P,\mathrm{ac}}
\sigma_{P,\mathrm{ac}}^{-1/2}\jz\sigma_{P,\mathrm{ac}}^{1/2}
+f(0^+)\Tr(\sigma-\sigma_{P,\mathrm{ac}})
+\tilde f(0^+)\Tr(\rho-\rho_{P,\mathrm{ac}}).
\label{Dmax lower bound2}
\end{align}
Note the close resemblance between \eqref{Dmax lower bound2} and \eqref{classical fdiv def}.

Conversely, the inequalities in \eqref{Dmax lower bound1} and 
\eqref{Dmax lower bound2} can be shown to be equalities by 
giving reverse tests attaining the lower bound in each case, as was done 
in \cite{Matsumoto_newfdiv} for \eqref{Matsumoto explicit}.
In the case of \eqref{Dmax lower bound2}, this can be given 
by the following variant of Matsumoto's optimal reverse test \cite{Matsumoto_newfdiv}
(see Example \ref{ex:max Renyi}). 
Let $P:=\rho^0\wedge\sigma^0$, let 
$\lambda_1,\ldots,\lambda_r$ be the different eigenvalues of 
$\sigma_{P,\ac}^{-1/2}\rho_{P,\ac}\sigma_{P,\ac}^{-1/2}$ with corresponding spectral projections $E_1,\ldots,E_r$, and define  $\I:=[r+2]$.
Let $\tau_0\in\S(\hil)$ be arbitrary, and 
\begin{align*}
\tau_1:=\begin{cases}
\frac{\rho-\rho_{P,\mathrm{ac}}}{\Tr(\rho-\rho_{P,\mathrm{ac}})},&\rho^0\not\le\sigma^0,\\
\tau_0,&\text{otherwise},
\end{cases}\ds\ds\ds
\tau_2:=\begin{cases}
\frac{\sigma-\sigma_{P,\mathrm{ac}}}{\Tr(\sigma-\sigma_{P,\mathrm{ac}})},&\sigma^0\not\le\rho^0,\\
\tau_0,&\text{otherwise}.
\end{cases}\ds\ds\ds
\end{align*}
It is straightforward to verify that $(\hat p, \hat q,\hat\Gamma)$ defined as
\begin{align}\label{optimal reverse test2}
&\hat p(i):=\begin{cases}
\lambda_i\Tr \sigma_{P,\mathrm{ac}} E_i,&i\in[r],\\
\Tr(\rho-\rho_{P,\mathrm{ac}}),&i=r+1,\\
0,&i=r+2,
\end{cases}\ds\ds\ds\ds
\hat q(i):=
\begin{cases}
\Tr\sigma_{P,\mathrm{ac}} E_i,&i\in[r],\\
0,&i=r+1,\\
\Tr(\sigma-\sigma_{P,\mathrm{ac}}),&i=r+2
\end{cases}\\
&\hat\rt(\egy_{\{i\}}):=\begin{cases}
\frac{\sigma_{P,\mathrm{ac}}^{1/2}E_i\sigma_{P,\mathrm{ac}}^{1/2}}{\Tr\sigma_{P,\mathrm{ac}} P_i},&i\in[r],\,\Tr\sigma_{P,\mathrm{ac}} E_i\ne 0,\\
\tau_0,&i\in[r],\,\Tr\sigma_{P,\mathrm{ac}} E_i=0,\\
\tau_1,&i=r+1,\\
\tau_2,&i=r+2,
\end{cases}
\end{align}
gives a reverse test for $(\rho,\sigma)$, and therefore
\begin{align}
D_f^{\max}(\rho\|\sigma)&\le D_f(\hat p\|\hat q)\nn\\
&=
\Tr \underbrace{\sigma_{P,\mathrm{ac}}^{1/2}f\bz \sigma_{P,\mathrm{ac}}^{-1/2}\rho_{P,\mathrm{ac}}
\sigma_{P,\mathrm{ac}}^{-1/2}\jz\sigma_{P,\mathrm{ac}}^{1/2}}_{
\persp{f}(\rho_{P,\mathrm{ac}},\sigma_{P,\mathrm{ac}})}
+f(0^+)\Tr(\sigma-\sigma_{P,\mathrm{ac}})
+\tilde f(0^+)\Tr(\rho-\rho_{P,\mathrm{ac}}),
\label{Matsumoto upper bound}
\\
&=:\wtilde D_f(\rho\|\sigma),
\end{align}
where the equality can be verified by a straightforward computation.
Note that this is valid even if the finiteness of 
$f(0^+)$ and $\tilde f(0^+)$ are not assumed, when the expression in 
\eqref{Matsumoto upper bound} in interpreted 
using the convention explained below \eqref{oppersp equal supports}.

Finally, we show an alternative proof of the fact that 
\eqref{Matsumoto upper bound} is a lower bound on 
$D_f^{\max}$ (see \eqref{Dmax lower bound2}) which does not rely on 
the results of \cite{Kosaki_ac} and Corollary \ref{cor:oppersp explicit2}.

In view of \eqref{perspective abscont}, the following is a special case of \cite[Proposition 3.30]{HiaiMosonyi2017}. We give a streamlined proof below for the special case in consideration, which does not rely on 
\eqref{perspective abscont}.

\begin{lemma}\label{lemma:maxfdiv mon}
Let $f:\,(0,+\infty)\to(-\infty,0)$ be a non-positive operator convex function
such that $f(0^+),\tilde f(0^+)<+\infty$. 
Then for any $\rho,\sigma\in\B(\hil)\pne$ and 
positive trace-preserving map $\map:\,\B(\hil)\to\B(\kil)$,
\begin{align*}
\persp{f}(\map(\rho)_{\map(\sigma),\ac},\map(\sigma)_{\map(\rho),\ac})
\le
\map\bz\persp{f}\bz\rho_{\sigma,\ac},\sigma_{\rho,\ac}\jz\jz.
\end{align*}
\end{lemma}
\begin{proof}
Note that $\tilde f$ is also a non-positive operator convex function
(this follows from \cite{ENG,Effros,EH} and the relation between $\persp{f}$ and 
$\persp{\tilde f}$; see also \cite[Proposition A.1]{HiaiMosonyi2017}), and hence both $f$ and $\tilde f$ are 
operator non-increasing (see, e.g., \cite[Theorem V.2.5]{Bhatia}). Thus,
\begin{align*}
\persp{f}(\map(\rho)_{\map(\sigma),\mathrm{ac}},\map(\sigma)_{\map(\rho),\mathrm{ac}})
&=
\map(\sigma)_{\map(\rho),\mathrm{ac}}^{1/2}f\big(
\map(\sigma)_{\map(\rho),\mathrm{ac}}^{-1/2}
\underbrace{\map(\rho)_{\map(\sigma),\mathrm{ac}}}_{
\ge\map\bz \rho_{\sigma,\mathrm{ac}}\jz}
\map(\sigma)_{\map(\rho),\mathrm{ac}}^{-1/2}\big)
\map(\sigma)_{\map(\rho),\mathrm{ac}}^{1/2}\\
&\le
\map(\sigma)_{\map(\rho),\mathrm{ac}}^{1/2}f\big(
\map(\sigma)_{\map(\rho),\mathrm{ac}}^{-1/2}
\map(\rho_{\sigma,\mathrm{ac}})
\map(\sigma)_{\map(\rho),\mathrm{ac}}^{-1/2}\big)
\map(\sigma)_{\map(\rho),\mathrm{ac}}^{1/2}\\
&=
\map(\rho_{\sigma,\mathrm{ac}})^{1/2}
\tilde f\big(
\map(\rho_{\sigma,\mathrm{ac}})^{-1/2}
\underbrace{\map(\sigma)_{\map(\rho),\mathrm{ac}}}_{
\ge\map\bz\sigma_{\rho,\ac}\jz}
\map(\rho_{\sigma,\mathrm{ac}})^{-1/2}
\big)
\map(\rho_{\sigma,\mathrm{ac}})^{1/2}\\
&\le
\map(\rho_{\sigma,\mathrm{ac}})^{1/2}
\tilde f\big(
\map(\rho_{\sigma,\mathrm{ac}})^{-1/2}
\map\bz\sigma_{\rho,\ac}\jz
\map(\rho_{\sigma,\mathrm{ac}})^{-1/2}
\big)
\map(\rho_{\sigma,\mathrm{ac}})^{1/2}\\
&=
\map\bz\sigma_{\rho,\ac}\jz^{1/2}
f\bz
\map\bz\sigma_{\rho,\ac}\jz^{-1/2}
\map\bz\rho_{\sigma,\ac}\jz^{1/2}
\map\bz\sigma_{\rho,\ac}\jz^{-1/2}
\jz
\map\bz\sigma_{\rho,\ac}\jz^{1/2}\\
&=
\persp{f}\bz\map\bz\rho_{\sigma,\ac}\jz,
\map\bz\sigma_{\rho,\ac}\jz\jz,
\end{align*}
where all formulas are well-defined due to 
$\map(\rho_{\sigma,\ac})^0=
\map\bz\sigma_{\rho,\ac}\jz^0=\map(P)^0=
(\map(\rho)_{\map(\sigma),\mathrm{ac}})^0
=(\map(\sigma)_{\map(\rho),\mathrm{ac}})^0$, 
where $P:=\rho^0=\sigma^0$,
the inequalities follow due to \eqref{abscont mon},
and in the second and the third equalities we have used \eqref{persp trans}.
Now, using that 
\begin{align*}
\Psi(\valt):=\map\bz\sigma_{\rho,\ac}\jz^{-1/2}
\map\bz\sigma_{\rho,\ac}^{1/2}
(\valt)
\sigma_{\rho,\ac}^{1/2}\jz
\map\bz\sigma_{\rho,\ac}\jz^{-1/2}
\end{align*}
is a completely positive unital map on the commutative 
$C^*$-algebra generated by 
$\sigma_{\rho,\ac}^{-1/2}\rho_{\sigma,\ac}\sigma_{\rho,\ac}^{-1/2}$, 
the Kraus representation together with the operator Jensen inequality \cite{HansenPedersen2003} yields that 
\begin{align*}
f\bz\Psi\bz\sigma_{\rho,\ac}^{-1/2}\rho_{\sigma,\ac}\sigma_{\rho,\ac}^{-1/2}\jz\jz
\le
\Psi\bz f\bz
\sigma_{\rho,\ac}^{-1/2}\rho_{\sigma,\ac}\sigma_{\rho,\ac}^{-1/2}\jz\jz,
\end{align*}
which can be rewritten as 
\begin{align*}
\persp{f}\bz\map\bz\rho_{\sigma,\ac}\jz,
\map\bz\sigma_{\rho,\ac}\jz\jz
\le
\map\bz\persp{f}\bz\rho_{\sigma,\ac},\sigma_{\rho,\ac}\jz\jz,
\end{align*}
completing the proof.
\end{proof}

\begin{cor}\label{cor:Dfhat fdiv}
Let $f:\,(0,+\infty)\to(-\infty,0)$ be a non-positive operator convex function such that 
$f(0^+)=0=\tilde f(0^+)$. Then $\wtilde D_f$ 
defines a quantum extension of the classical $f$-divergence $D_f$
in the sense of Definition \ref{def:q-extension}, which is monotone under positive trace-preserving maps. 
\end{cor}
\begin{proof}
It is clear that $\wtilde D_f$ reduces to the classical $f$-divergence of the diagonals for jointly diagonalizable arguments. Isometric invariance of $\wtilde D_f$ follows from 
\eqref{abscont isometric cov}, and monotonicity from Lemma \ref{lemma:maxfdiv mon}.
\end{proof}

\begin{prop}
Let $f:\,(0,+\infty)\to\bR$ be an operator convex function such that 
$f(0^+),\tilde f(0^+)<+\infty$.  For any 
$\rho,\sigma\in\B(\hil)\pne$,
\begin{align}\label{Matsumoto explicit2}
D_f^{\max}(\rho\|\sigma)
&=
\wtilde D_f(\rho\|\sigma).
\end{align}
\end{prop}
\begin{proof}
According to \eqref{Matsumoto upper bound}, we only need to prove 
$D_f^{\max}(\rho\|\sigma)\ge
\wtilde D_f(\rho\|\sigma)$.
If $f(0^+)=0=\tilde f(0^+)$ then this follows immediately from Corollary \ref{cor:Dfhat fdiv} 
 as explained in Example \ref{ex:minmax extension}.

Next, assume that $f(0^+)$ and $\tilde f(0^+)$ are finite, and let 
$g$ be defined as in \eqref{g def}, so that 
$f(x)=f(0^+)+\tilde f(0^+)x-g(x)$, $x\in[0,+\infty)$. 
For any reverse test $(\hat p,\hat q)$ for $(\rho,\sigma)$, 
\begin{align*}
D_f(\hat p\|\hat q)
=
\underbrace{D_{f(0^+)}(\hat p\|\hat q)}_{=f(0^+)\Tr\sigma}
+
\underbrace{D_{\tilde f(0^+)\id}(\hat p\|\hat q)}_{=\tilde f(0^+)\Tr\rho}
+
D_{-g}(\hat p\|\hat q),
\end{align*}
whence 
\begin{align*}
D_f^{\max}(\rho\|\sigma)
&=
f(0^+)\Tr\sigma+\tilde f(0^+)\Tr\rho
+\underbrace{D_{-g}^{\max}(\rho\|\sigma)}_{=\wtilde D_{-g}(\rho\|\sigma))}\\
&=
f(0^+)\Tr\sigma+\tilde f(0^+)\Tr\rho-
\Tr \sigma_{P,\mathrm{ac}}^{1/2}g\bz \sigma_{P,\mathrm{ac}}^{-1/2}\rho_{P,\mathrm{ac}}
\sigma_{P,\mathrm{ac}}^{-1/2}\jz\sigma_{P,\mathrm{ac}}^{1/2}\\
&=
\wtilde D_f(\rho\|\sigma),
\end{align*}
where the second equality follows from the special case above, and the
last equality follows the same way as in the proof of
Corollary \ref{cor:oppersp explicit2}.
\end{proof}

\section{Proof of Lemma \ref{lemma:jointconc from mon}}
\label{sec:jointconc}

\begin{proof}
(i) Let $\W^{(i)}\in\cD_{\hil}(\divv)$,
$i\in[r]$, let $(t_i)_{i\in[r]}$ be a probability distribution, and let
$(\ket{i})_{i\in[r]}$ be an orthonormal system in some Hilbert space $\kil$. 
Then $\bz\sum_{i=0}^{r-1} t_i\Wx{x}^{(i)}\otimes\pr{i}\jz_{x\in\X}\in
\cD_{\hil\otimes\kil}(\divv)$ due to the isometric invariance of $\divv$ and the 
convexity of $\cD_{\hil\otimes\kil}(\divv)$, and 
\begin{align*}
\divv\bz\bz\sum_{i=0}^{r-1} t_i\Wx{x}^{(i)}\jz_{x\in\X}\jz
&=
\divv\bz\bz\Tr_{\kil}\sum_{i=0}^{r-1} t_i\Wx{x}^{(i)}\otimes\pr{i}\jz_{x\in\X}\jz\\
&\ge
\divv\bz\bz\sum_{i=0}^{r-1} t_i\Wx{x}^{(i)}\otimes\pr{i}\jz_{x\in\X}\jz\\
&\ge
\sum_{i=0}^{r-1}\divv\bz\bz t_i\Wx{x}^{(i)}\otimes\pr{i}\jz_{x\in\X}\jz\\
&=
\sum_{i=0}^{r-1} t_i\divv\bz\bz\Wx{x}^{(i)}\otimes\pr{i}\jz_{x\in\X}\jz\\
&=
\sum_{i=0}^{r-1} t_i\divv\bz\bz\Wx{x}^{(i)}\jz_{x\in\X}\jz,
\end{align*}
where the first equality is obvious, the first inequality is by the assumption that 
$\divv$ is monotone non-decreasing under partial traces,
the second inequality is due to the block superadditivity of $\divv$, the second equality follows from homogeneity, and the last equality is due to the 
isometric invariance of $\divv$.
This proves joint concavity, and joint superadditivity follows from it immediately due to 
homogeneity.
The assertion about joint convexity and joint subadditivity follows the same way, or by applying the above to $-\divv$ in place of $\divv$.

(ii)
Let $W\in\cD_{\hil}(\divv)$ and $\map:\,\B(\hil)\to\B(\kil)$ be a CPTP map
such that $\map(W)\in\cD_{\kil}(\divv)$.  
Let $\map(.)=\Tr_E V(.)V^*$ be a Stinespring representation of $\map$, where 
$V:\,\hil\to\hil_E\otimes\kil$ is an isometry.
Let $(U_{ab})_{a,b=0}^{d_E-1}$ be the discrete Weyl unitaries in some ONB of $\hil_E$, so that 
$(1/d_E^2)\sum_{a,b=0}^{d_E-1}U_{a,b}(.)U_{a,b}^*=(1/d_E)I_E\Tr(.)$
Then 
\begin{align*}
\divv(\map(W))
&=
\divv\bz\Tr_E VW V^*\jz\\
&=
\divv\bz (1/d_E)I_E\otimes\Tr_E VW V^*\jz\\
&=
\divv\bz \frac{1}{d_E^2}\sum_{a,b=0}^{d_E-1}(U_{a,b}\otimes I_{\kil}) VW V^* (U_{a,b}\otimes I_{\kil})^*\jz\\
&\ge
\frac{1}{d_E^2}\sum_{a,b=0}^{d_E-1}
\divv\bz (U_{a,b}\otimes I_{\kil}) VW V^*(U_{a,b}\otimes I_{\kil})^*\jz\\
&=
\frac{1}{d_E^2}\sum_{a,b=0}^{d_E-1}
\divv\bz VW V^* \jz
=
\divv\bz W \jz,
\end{align*}
where the second equality is due to stability,
the inequality is due to the joint concavity of $\divv$, 
and the fourth and the fifth equalities are due to isomeric invariance.
Again, the assertion about the opposite inequality under the assumption of joint convexity follows by applying the above to $-\divv$ in place of $\divv$.
\end{proof}

\section{Proof of Lemma \ref{lemma:Dmax derivative}.}
\label{sec:proof of lemma}

\begin{proof}
Let us introduce the notation 
\begin{align*}
(\rho/\sigma):=\sigma^{-1/2}\rho\sigma^{-1/2}=\sum_i\lambda_i P_i.
\end{align*}
Let 
\begin{align}\label{rho per sigma polar}
X:=\rho^{1/2}\sigma^{-1/2},\ds\ds\text{and}\ds\ds
X=U|X|=U\sum_i\lambda_i^{1/2} P_i
\end{align}
be its polar decomposition. Then
\begin{align}\label{sigma per rho spectral}
\rho^{-1/2}\sigma\rho^{-1/2}=(X\inv)^*(X\inv)=
U\bz\sum_i\lambda_i\inv P_i\jz U^*=\sum_i\lambda_i\inv \underbrace{UP_i U^*}_{=:R_i}
\end{align}
is a spectral decomposition of $\rho^{-1/2}\sigma\rho^{-1/2}$. 
Recall from \eqref{normalized alpha-mean} that
\begin{align*}
\what{\sigma\#_{\alpha}\rho}:=\frac{1}{Q_{\alpha}^{\max}}\sigma\#_{\alpha}\rho
=\frac{1}{Q_{\alpha}^{\max}}\rho\#_{1-\alpha}\sigma,
\end{align*}
where $Q_{\alpha}^{\max}:=Q_{\alpha}^{\max}(\rho\|\sigma)$, 
and note the following identities:
\begin{align*}
\sigma^{-1/2}\what{\sigma\#_{\alpha}\rho}\,\sigma^{-1/2}
&=
\frac{1}{Q_{\alpha}^{\max}}(\sigma^{-1/2}\rho\sigma^{-1/2})^{\alpha}=
\sum_i\frac{\lambda_i^{\alpha}}{Q_{\alpha}^{\max}}P_i,\\
\rho^{-1/2}\what{\sigma\#_{\alpha}\rho}\,\rho^{-1/2}
&=
\frac{1}{Q_{\alpha}^{\max}}(\rho^{-1/2}\sigma\rho^{-1/2})^{1-\alpha}=
\sum_i\frac{\lambda_i^{\alpha-1}}{Q_{\alpha}^{\max}}R_i,
\end{align*}
where in the last line we used \eqref{sigma sharp rho}.

Recall that $\pi_{\hil}=I/d$ denotes the maximally mixed state on $\hil$. We have
\begin{align*}
&\frac{d}{dt}D^{\max}\bz (1-t)\what{\sigma\#_{\alpha}\rho}+t\pi_{\hil}\big\|\sigma\jz\Big\vert_{t=0}\\
&\ds=
\Tr
\underbrace{\sigma^{1/2}\bz \pi_{\hil}-\what{\sigma\#_{\alpha}\rho}\jz\sigma^{-1/2}}_{
=I/d-\frac{1}{Q_{\alpha}^{\max}}\sigma(\sigma^{-1/2}\rho\sigma^{-1/2})^{\alpha}}
\log\underbrace{\sigma^{-1/2}\what{\sigma\#_{\alpha}\rho}\,\sigma^{-1/2}}_{=
\frac{1}{Q_{\alpha}^{\max}}(\sigma^{-1/2}\rho\sigma^{-1/2})^{\alpha}}\\
&\ds +\Tr\underbrace{\sigma^{1/2}\what{\sigma\#_{\alpha}\rho}\,\sigma^{-1/2}}_{=
\frac{1}{Q_{\alpha}^{\max}}\sigma(\sigma^{-1/2}\rho\sigma^{-1/2})^{\alpha}}
\sum_{i,j}\log^{[1]}\bz\frac{\lambda_i^{\alpha}}{Q_{\alpha}^{\max}},\frac{\lambda_j^{\alpha}}{Q_{\alpha}^{\max}}\jz
\underbrace{P_i\sigma^{-1/2}\bz\pi_{\hil}-\what{\sigma\#_{\alpha}\rho}\jz \sigma^{-1/2}P_j}_{=d\inv P_i\sigma\inv P_j-\frac{1}{Q_{\alpha}^{\max}}\delta_{i,j}\lambda_i^{\alpha}P_i}\\
&\ds=
-\log Q_{\alpha}^{\max}+\frac{\alpha}{d}\Tr\log(\sigma^{-1/2}\rho\sigma^{-1/2})\\
&\ds\ds+\frac{1}{Q_{\alpha}^{\max}}(\log Q_{\alpha}^{\max})
\underbrace{\Tr\sigma(\sigma^{-1/2}\rho\sigma^{-1/2})^{\alpha}}_{=Q_{\alpha}^{\max}}
-\frac{\alpha}{Q_{\alpha}^{\max}}\Tr\sigma(\sigma^{-1/2}\rho\sigma^{-1/2})^{\alpha}
\log(\sigma^{-1/2}\rho\sigma^{-1/2})\\
&\ds\ds+
\frac{1}{dQ_{\alpha}^{\max}}
\sum_{i,j}\log^{[1]}\bz\frac{\lambda_i^{\alpha}}{Q_{\alpha}^{\max}},\frac{\lambda_j^{\alpha}}{Q_{\alpha}^{\max}}\jz
\Tr\sigma\underbrace{(\sigma^{-1/2}\rho\sigma^{-1/2})^{\alpha}P_i\sigma\inv P_j}_{=
\lambda_i^{\alpha}P_i\sigma\inv P_j}\\
&\ds\ds-
\frac{1}{Q_{\alpha}^{\max}}\Tr\sigma(\sigma^{-1/2}\rho\sigma^{-1/2})^{\alpha}
\sum_i\underbrace{\log^{[1]}\bz\frac{\lambda_i^{\alpha}}{Q_{\alpha}^{\max}},\frac{\lambda_i^{\alpha}}{Q_{\alpha}^{\max}}\jz\frac{\lambda_i^{\alpha}}{Q_{\alpha}^{\max}}}_{=1}P_i\\
&\ds=
\frac{\alpha}{d}\Tr\log(\sigma^{-1/2}\rho\sigma^{-1/2})
-\frac{\alpha}{Q_{\alpha}^{\max}}\Tr\sigma(\sigma^{-1/2}\rho\sigma^{-1/2})^{\alpha}
\log(\sigma^{-1/2}\rho\sigma^{-1/2})\\
&\ds\ds+\frac{1}{dQ_{\alpha}^{\max}}
\sum_{i,j}\log^{[1]}\bz\frac{\lambda_i^{\alpha}}{Q_{\alpha}^{\max}},\frac{\lambda_j^{\alpha}}{Q_{\alpha}^{\max}}\jz\lambda_i^{\alpha}
\Tr\sigma P_i\sigma\inv P_j-
\frac{1}{Q_{\alpha}^{\max}}\underbrace{\Tr\sigma(\sigma^{-1/2}\rho\sigma^{-1/2})^{\alpha}}_{
=Q_{\alpha}^{\max}}\,.
\end{align*}
An exactly analogous calculation yields 
(by exchanging $\rho$ and $\sigma$ and replacing $\alpha$ with $1-\alpha$, and noting 
\eqref{sigma per rho spectral})
\begin{align*}
&\frac{d}{dt}D^{\max}\bz (1-t)\what{\sigma\#_{\alpha}\rho}+t\pi_{\hil}\big\|\rho\jz\Big\vert_{t=0}\\
&\ds=
\frac{1-\alpha}{d}\Tr\log(\rho^{-1/2}\sigma\rho^{-1/2})
-\frac{1-\alpha}{Q_{\alpha}^{\max}}\Tr\rho(\rho^{-1/2}\sigma\rho^{-1/2})^{1-\alpha}
\log(\rho^{-1/2}\sigma\rho^{-1/2})\\
&\ds\ds+\frac{1}{dQ_{\alpha}^{\max}}
\sum_{i,j}\log^{[1]}\bz\frac{\lambda_i^{\alpha-1}}{Q_{\alpha}^{\max}},\frac{\lambda_j^{\alpha-1}}{Q_{\alpha}^{\max}}\jz\lambda_i^{\alpha-1}
\Tr\rho R_i\rho\inv R_j-1.
\end{align*}
Thus,
\begin{align}
\partial_{\pi_{\hil}}&=\frac{d}{dt}\left[\alpha D^{\max}\bz (1-t)\what{\sigma\#_{\alpha}\rho}+t\pi_{\hil}\big\|\rho\jz
+
(1-\alpha)D^{\max}\bz (1-t)\what{\sigma\#_{\alpha}\rho}+t\pi_{\hil}\big\|\sigma\jz\right]
\Bigg\vert_{t=0}\nn\\
&\ds=
\frac{\alpha(1-\alpha)}{d}
\underbrace{\Big[\underbrace{\Tr\log(\rho^{-1/2}\sigma\rho^{-1/2})}_{=\sum_i\log\lambda_i\inv}+
\underbrace{\Tr\log(\sigma^{-1/2}\rho\sigma^{-1/2})}_{=\sum_i\log\lambda_i}
\Big]}_{=0}\nn\\
&\ds\ds
-\frac{\alpha(1-\alpha)}{Q_{\alpha}^{\max}}
\underbrace{\Big[\Tr\rho(\rho^{-1/2}\sigma\rho^{-1/2})^{1-\alpha}\log(\rho^{-1/2}\sigma\rho^{-1/2})
+\Tr\sigma(\sigma^{-1/2}\rho\sigma^{-1/2})^{\alpha}
\log(\sigma^{-1/2}\rho\sigma^{-1/2})
\Big]}_{=0}\nn\\
&\ds\ds+
\frac{\alpha}{d}
\sum_{i,j}\log^{[1]}\bz\lambda_i^{\alpha-1},\lambda_j^{\alpha-1}\jz\lambda_i^{\alpha-1}
\Tr\rho R_i\rho\inv R_j+\frac{1-\alpha}{d}
\sum_{i,j}\log^{[1]}\bz\lambda_i^{\alpha},\lambda_j^{\alpha}\jz\lambda_i^{\alpha}
\Tr\sigma P_i\sigma\inv P_j\s-1\nn\\
&\ds=
\frac{\alpha}{d}
\sum_{i,j}\log^{[1]}\bz\lambda_i^{\alpha-1},\lambda_j^{\alpha-1}\jz\lambda_i^{\alpha-1}
\Tr\rho R_i\rho\inv R_j
+\frac{1-\alpha}{d}
\sum_{i,j}\log^{[1]}\bz\lambda_i^{\alpha},\lambda_j^{\alpha}\jz\lambda_i^{\alpha}
\Tr\sigma P_i\sigma\inv P_j\s-1,\label{derivative proof1}
\end{align}
where the first expression above is equal to $0$ due to \eqref{sigma per rho spectral}, and the second expression is equal to $0$ according to
\eqref{max Renyi var proof 3}.

Note that by \eqref{rho per sigma polar}, 
\begin{align*}
U=X|X|\inv=\rho^{1/2}\sigma^{-1/2}(\rho/\sigma)^{-1/2},
\end{align*}
whence
\begin{align*}
\begin{array}{ccc}
U^* &=& (\rho/\sigma)^{-1/2}\sigma^{-1/2}\rho^{1/2}\\
\verteq & & \\
U\inv &=& (\rho/\sigma)^{1/2}\sigma^{1/2}\rho^{-1/2},
\end{array}
\end{align*}
which in turn yields
\begin{align*}
U=(U\inv)^*=\rho^{-1/2}\sigma^{1/2}(\rho/\sigma)^{1/2}.
\end{align*}
Thus,
\begin{align*}
\Tr\rho R_i\rho\inv R_j&=\Tr\rho (UP_iU^*)\rho\inv (UP_jU^*)\\
&=
\Tr\rho\underbrace{\rho^{-1/2}\sigma^{1/2}(\rho/\sigma)^{1/2}}_{=U}P_i
\underbrace{(\rho/\sigma)^{-1/2}\sigma^{-1/2}\rho^{1/2}}_{=U^*}\rho\inv
\underbrace{\rho^{1/2}\sigma^{-1/2}(\rho/\sigma)^{-1/2}}_{=U}P_j
\underbrace{(\rho/\sigma)^{1/2}\sigma^{1/2}\rho^{-1/2}}_{=U^*}\\
&=\Tr\sigma^{1/2}\underbrace{(\rho/\sigma)^{1/2}P_i(\rho/\sigma)^{-1/2}}_{=P_i}
\sigma\inv
\underbrace{(\rho/\sigma)^{-1/2}P_j(\rho/\sigma)^{1/2}}_{=P_j}\sigma^{1/2}\\
&=
\Tr\sigma P_i\sigma\inv P_j.
\end{align*}
Writing this back into \eqref{derivative proof1}, we get 
\begin{align}\label{derivative1}
\partial_{\pi_{\hil}}
&=
-1+\frac{1}{d}\sum_{i,j}\Tr\sigma P_i\sigma\inv P_j
\underbrace{\Big[
\alpha\log^{[1]}\bz\lambda_i^{\alpha-1},\lambda_j^{\alpha-1}\jz\lambda_i^{\alpha-1}
+(1-\alpha)\log^{[1]}\bz\lambda_i^{\alpha},\lambda_j^{\alpha}\jz\lambda_i^{\alpha}
\Big]}_{=:\Lambda_{\alpha,i,j}}.
\end{align}
It follows by a straightforward computation that $\Lambda_{\alpha,i,j}$ can be written as in 
\eqref{Lambda def}.
Note that $\Lambda_{\alpha}$ is symmetric, i.e., 
$\Lambda_{\alpha,i,j}=\Lambda_{\alpha,j,i}$. Exchanging the indices $i$ and $j$ in 
\eqref{derivative1} yields \eqref{Dmax derivative1}.
\end{proof}

\newpage

\section{Different approaches to multi-variate R\'enyi divergences}
\label{sec:multiRenyi}

Figure \ref{fig:multiRenyi} below shows some relations between $2$-variable and multi-variate
weighted geometric means, R\'enyi divergences, weighted power means, divergence radii and 
centers, and fixpoint equations, an arrow indicating that one quantity can be obtained from 
the other. These work when all operators are commuting; non-commutative extensions of any of these quantities may be obtained by extending any quantity in the diagram to non-commutative variables, and trying to follow the arrows to obtain non-commutative extensions of other quantities. 
For instance, a non-commutative multi-variate weighted geometric mean (the so-called Karcher mean \cite{Bhatia_Holbrook2006,Moakher_matrixmean}) was obtained in \cite{Lim_Palfia2012} by taking the $\alpha\to0$ limit of 
the solution of the fixed point equation $\gamma=\sum_xP(x)\gamma\#_{\alpha}W_x$.
In Sections \ref{sec:barycentric}--\ref{sec:ex} of this paper we focused on the top right corner of the diagram. One might try the same procedure starting with a different 
$2$-variable weighted geometric mean; for instance, 
Proposition A.24 and Remark A.27 in \cite{MO-cqconv-cc} show that 
the Tsallis divergence center and the solution of the fixed point equation 
corresponding to the Petz-type weighted geometric mean 
$G_{\alpha,1}(\rho\|\sigma):=\sigma^{\frac{1-\alpha}{2}}\rho^{\alpha} \sigma^{\frac{1-\alpha}{2}}$
give the same result $\bz\sum_xP(x)W_x^{\alpha}\jz^{1/\alpha}$, and it is easy to see that 
(at least when $W_x^0=I$, $x\in\supp P$) its limit at $\alpha\to 1$ is 
$G_P^{\DU}(W)=\exp(\sum_x P(x)\log W_x)$.
This shows, in particular, that the diagram is not commutative in the quantum case. 
Another example of this is that the solution of the fixpoint equation might differ from 
the Tsallis divergence center corresponding to the same 
non-commutative $2$-variable weighted geometric mean; see, e.g., \cite{PV_Hellinger} for the case $G_{\alpha}=\#_{\alpha}$.

\vspace{1cm}

\begin{figure}[h]
\caption{Different approaches to multi-variate R\'enyi divergences}
\label{fig:multiRenyi}

\vspace{1cm}

\includegraphics[trim = 0cm 0cm 0cm 0cm, clip, width=16cm]{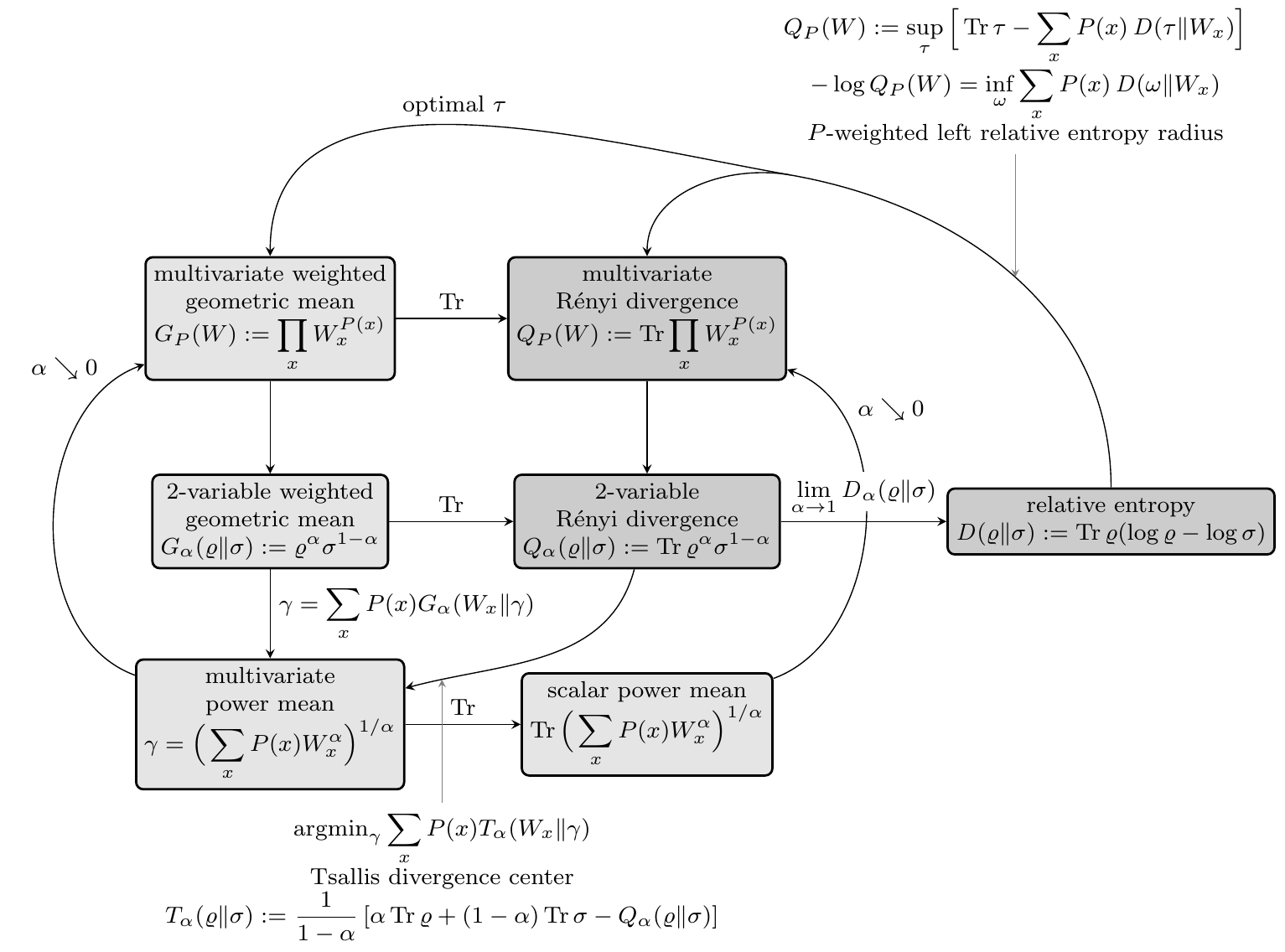}

\end{figure}

\newpage

\section{Order of relative entropies and R\'enyi divergences}
\label{sec:ordering}

Figure \ref{fig:barycentric} below illustrates the known relations between
the most relevant relative entropies and $2$-variable R\'enyi divergences studied previously 
in the literature, and the new relative entropies and R\'enyi divergences
introduced in this paper. A divergence higher in the picture gives a higher value 
when evaluated on a given pair of inputs than a divergence lower in the picture with the same 
$\alpha$ coordinate. 
It is easy to see from the Araki-Lieb-Thirring inequality \cite{Araki,LT} that the 
R\'enyi $(\alpha,z)$-divergences are monotone increasing in the $z$ parameter when 
$\alpha\in(0,1)$, and monotone decreasing when $\alpha>1$, 
(see, e.g., \cite{LinTomamichel15}), which is why 
the $z$ coordinate is transformed into $1/z$
for $\alpha>1$ in the representation of the R\'enyi $(\alpha,z)$-divergences. 
The diagonally shaded area shows the region in the $(\alpha,z)$-plane where the 
R\'enyi $(\alpha,z)$-divergences are monotone \cite{Zhang2018}.
The vertically shaded area is an illustration of the position of the barycentric R\'enyi 
divergences compared to the R\'enyi $(\alpha,z)$-divergences when the 
generating relative entropies are monotone and additive, and hence are between 
$\DU$ and $D^{\max}$.

\vspace{1cm}

\begin{figure}[h]
\caption{Order of $2$-variable quantum R\'enyi divergences}
\label{fig:barycentric}
\vspace{1cm}

\includegraphics[trim = 0cm 0cm 0cm 0cm, clip, width=13cm]{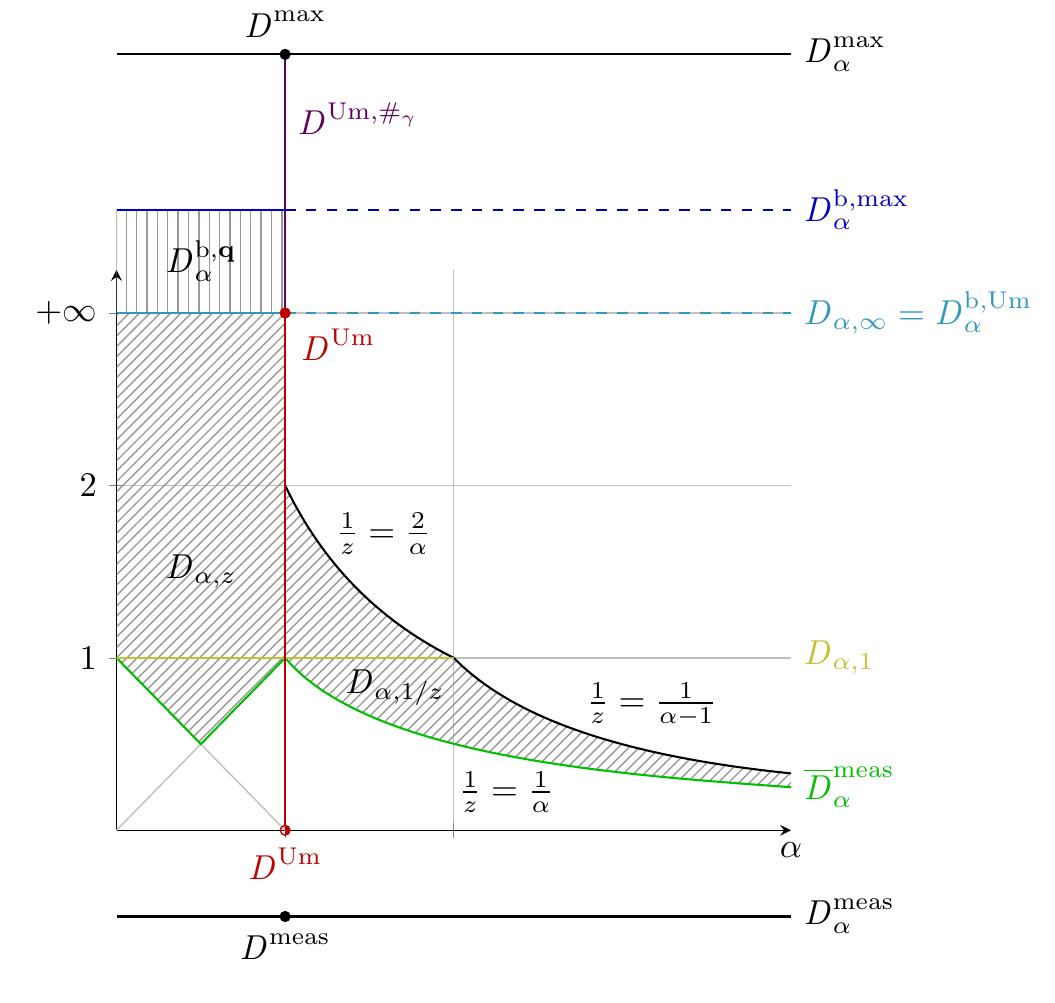}

\end{figure}

\newpage

\bibliography{bibliography240102}

\begin{thebibliography}{100}

\bibitem{AC2011}
Martial Agueh and Guillaume Carlier.
\newblock Barycenters in the {W}asserstein space.
\newblock {\em {SIAM} Journal on Mathematical Analysis}, 43(2):904--924, 2011.

\bibitem{AndersonTrapp1975}
W.~N. Anderson{, Jr.} and G.~E. Trapp.
\newblock Shorted operators. {II}.
\newblock {\em SIAM Journal on Applied Mathematics}, 28(1):60--71, 1975.

\bibitem{AndoLebesgue1976}
T.~Ando.
\newblock Lebesgue-type decomposition of positive operators.
\newblock {\em Acta Sci. Math.}, 38:253--260, 1976.

\bibitem{Ando}
T.~Ando.
\newblock Concavity of certain maps on positive definite matrices and
  applications to {Hadamard} products.
\newblock {\em Linear Algebra Appl.}, 26:203--241, 1979.

\bibitem{AndoLiMathias2004}
T.~Ando, Chi-Kwong Li, and Roy Mathias.
\newblock Geometric means.
\newblock {\em Linear Algebra and its Applications}, 385:305--334, 2004.

\bibitem{AndoHiai1994}
Tsuyoshi Ando and Fumio Hiai.
\newblock Log majorization and complementary {G}olden-{T}hompson type
  inequalities.
\newblock {\em Linear Algebra and its Applications}, 197-198:113--131, 1994.

\bibitem{Araki}
H.~Araki.
\newblock On an inequality of {Lieb} and {Thirring}.
\newblock {\em Letters in Mathematical Physics}, 19:167--170, 1990.

\bibitem{ANSzV}
K.~M.~R. Audenaert, M.~Nussbaum, A.~Szkola, and F.~Verstraete.
\newblock Asymptotic error rates in quantum hypothesis testing.
\newblock {\em Communications in Mathematical Physics}, 279:251--283, 2008.
\newblock arXiv:0708.4282.

\bibitem{AD}
Koenraad M.~R. Audenaert and Nilanjana Datta.
\newblock $\alpha$-$z$-relative {R}enyi entropies.
\newblock {\em J.~Math.~Phys.}, 56:022202, 2015.
\newblock arXiv:1310.7178.

\bibitem{Beigi}
Salman Beigi.
\newblock Sandwiched {R\'enyi} divergence satisfies data processing inequality.
\newblock {\em Journal of Mathematical Physics}, 54(12):122202, December 2013.
\newblock arXiv:1306.5920.

\bibitem{BS}
V.~P. Belavkin and P.~Staszewski.
\newblock {$C^\ast$}-algebraic generalization of relative entropy and entropy.
\newblock {\em Ann. Inst. H. Poincar\'e Phys. Th\'eor.}, 37:51--58, 1982.

\bibitem{BFT_variational}
Mario Berta, Omar Fawzi, and Marco Tomamichel.
\newblock On variational expressions for quantum relative entropies.
\newblock {\em Letters in Mathematical Physics}, 107(12):2239--2265, 2017.
\newblock arXiv:1512.02615.

\bibitem{BST}
Mario Berta, Volkher~B. Scholz, and Marco Tomamichel.
\newblock R\'enyi divergences as weighted non-commutative vector-valued
  ${L}_p$-spaces.
\newblock {\em Ann. Henri Poincar\'e}, 19:1843--1867, 2018.
\newblock arXiv:1608.05317.

\bibitem{Bhatia}
Rajendra Bhatia.
\newblock {\em Matrix Analysis}.
\newblock Number 169 in Graduate Texts in Mathematics. Springer, 1997.

\bibitem{BGJ2019}
Rajendra Bhatia, Stephane Gaubert, and Tanvi Jain.
\newblock Matrix versions of the {H}ellinger distance.
\newblock {\em Letters in Mathematical Physics}, 2019.

\bibitem{Bhatia_Holbrook2006}
Rajendra Bhatia and John Holbrook.
\newblock Riemannian geometry and matrix geometric means.
\newblock {\em Linear Algebra and its Applications}, 413(2):594--618, 2006.
\newblock Special Issue on the 11th Conference of the International Linear
  Algebra Society, Coimbra, 2004.

\bibitem{BJL2018}
Rajendra Bhatia, Tanvi Jain, and Yongdo Lim.
\newblock On the {B}ures-{W}asserstein distance between positive definite
  matrices.
\newblock {\em Expositiones Mathematicae}, 2018.

\bibitem{Brandao_etal_secondlaws}
F.~G. S.~L. Brandao, M.~Horodecki, N.~H.~Y. Ng, J.~Oppenheim, and S.~Wehner.
\newblock The second laws of quantum thermodynamics.
\newblock {\em Proceedings of the National Academy of Sciences USA},
  112(1):3275--3279, 2014.

\bibitem{bunth2021equivariant}
Gergely Bunth and P{\'e}ter Vrana.
\newblock Equivariant relative submajorization.
\newblock {\em IEEE Transactions on Information Theory}, 69(2):1057--1073,
  2023.
\newblock arXiv:2108.13217.

\bibitem{CFL}
E.~A. Carlen, R.~L. Frank, and E.~H. Lieb.
\newblock Some operator and trace function convexity theorems.
\newblock {\em Linear Algebra Appl.}, 490:174--185, 2016.

\bibitem{CFL_AD_conjecture}
Eric~A. Carlen, Rupert~L. Frank, and Elliott~H. Lieb.
\newblock Inequalities for quantum divergences and the {A}udenaert-{D}atta
  conjecture.
\newblock {\em Journal of Physics A: Mathematical and Theoretical},
  51(48):483001, nov 2018.

\bibitem{CL-MII}
Eric~A. Carlen and Elliott~H. Lieb.
\newblock A {M}inkowski type trace inequality and strong subadditivity of
  quantum entropy {II}: Convexity and concavity.
\newblock {\em Letters in Mathematical Physics}, 83:107--126, 2008.

\bibitem{Csiszar_fdiv}
I.~Csisz\'ar.
\newblock Information type measure of difference of probability distributions
  and indirect observations.
\newblock {\em Studia Sci.~Math.~Hungar.}, 2:299--318, 1967.

\bibitem{Csiszar}
Imre Csisz\'ar.
\newblock Generalized cutoff rates and {R\'enyi's} information measures.
\newblock {\em IEEE Transactions on Information Theory}, 41(1):26--34, January
  1995.

\bibitem{CsM2003}
Imre Csisz\'ar and Franti\v sek Mat\'u\v{s}.
\newblock Information projections revisited.
\newblock {\em IEEE Transactions on Information Theory}, 49(6):1474--1490,
  2003.

\bibitem{Datta}
Nilanjana Datta.
\newblock Min- and max-relative entropies and a new entanglement monotone.
\newblock {\em IEEE Transactions on Information Theory}, 55(6):2816--2826,
  2009.

\bibitem{ENG}
A.~Ebadian, I.~Nikoufar, and M.~Eshaghi Gordji.
\newblock Perspectives of matrix convex functions.
\newblock {\em Proc. Natl. Acad. Sci. USA}, 108(18):7313--7314, 2011.

\bibitem{Effros}
E.~Effros.
\newblock A matrix convexity approach to some celebrated quantum inequalities.
\newblock {\em Proc. Natl. Acad. Sci. USA}, 106(4):1006--1008, 2009.

\bibitem{EH}
E.~Effros and F.~Hansen.
\newblock Non-commutative perspectives.
\newblock {\em Ann. Funct. Anal.}, 5:74--79, 2014.
\newblock arXiv:1309.7701.

\bibitem{farooq2023asymptotic}
Muhammad~Usman Farooq, Tobias Fritz, Erkka Haapasalo, and Marco Tomamichel.
\newblock Matrix majorization in large samples.
\newblock {\em IEEE Transactions on Information Theory}, 2024.
\newblock arXiv:2301.07353.

\bibitem{FawziFawzi2021}
Hamza Fawzi and Omar Fawzi.
\newblock Defining quantum divergences via convex optimization.
\newblock {\em {Quantum}}, 5:387, January 2021.

\bibitem{FL}
Rupert~L. Frank and Elliott~H. Lieb.
\newblock Monotonicity of a relative {R\'enyi} entropy.
\newblock {\em Journal of Mathematical Physics}, 54(12):122201, December 2013.
\newblock arXiv:1306.5358.

\bibitem{Frenkel_integral}
P\'eter~E. Frenkel.
\newblock Integral formula for quantum relative entropy implies data processing
  inequality.
\newblock {\em Quantum}, 7:1102, 2023.
\newblock arXiv:2208.12194.

\bibitem{FuruyaLashkariOuseph2023}
Keiichiro Furuya, Nima Lashkari, and Shoy Ouseph.
\newblock Monotonic multi-state quantum $f$-divergences.
\newblock {\em J. Math. Phys.}, 64:042203, 2023.

\bibitem{HansenPedersen2003}
Frank Hansen and Gert~K. Pedersen.
\newblock Jensen's operator inequality.
\newblock {\em Bulletin of the London Mathematical Society}, 35(4):553--564,
  2003.

\bibitem{Hayashicq}
Masahito Hayashi.
\newblock Error exponent in asymmetric quantum hypothesis testing and its
  application to classical-quantum channel coding.
\newblock {\em Physical Review A}, 76(6):062301, December 2007.
\newblock arXiv:quant-ph/0611013.

\bibitem{hayashi2002error}
Masahito Hayashi, Masato Koashi, Keiji Matsumoto, Fumiaki Morikoshi, and
  Andreas Winter.
\newblock Error exponents for entanglement concentration.
\newblock {\em Journal of Physics A: Mathematical and General}, 36(2):527,
  2002.

\bibitem{HT14}
Masahito Hayashi and Marco Tomamichel.
\newblock Correlation detection and an operational interpretation of the
  {R}\'enyi mutual information.
\newblock {\em Journal of Mathematical Physics}, 57:102201, 2016.

\bibitem{Hiai_book}
F.~Hiai.
\newblock Matrix analysis: Matrix monotone functions, matrix means, and
  majorization.
\newblock {\em Interdisciplinary Information Sciences}, 16:139--248, 2010.

\bibitem{HiaiMosonyi2017}
F.~Hiai and M.~Mosonyi.
\newblock Different quantum $f$-divergences and the reversibility of quantum
  operations.
\newblock {\em Rev.~Math.~Phys.}, 29:1750023, 2017.

\bibitem{HMPB}
F.~Hiai, M.~Mosonyi, D.~Petz, and C.~B\'eny.
\newblock Quantum $f$-divergences and error correction.
\newblock {\em Rev.~Math.~Phys.}, 23:691--747, 2011.

\bibitem{HiaiPetzMatrix}
F.~Hiai and D.~Petz.
\newblock {\em Introduction to Matrix Analysis and Applications}.
\newblock Universitext. Springer Cham, 2014.

\bibitem{Hiai_concavity2001}
Fumio Hiai.
\newblock Concavity of certain matrix trace functions.
\newblock {\em Taiwanese Journal of Mathematics}, 5(3):535 -- 554, 2001.

\bibitem{Hiai-convexity}
Fumio Hiai.
\newblock Concavity of certain matrix trace and norm functions.
\newblock {\em Linear Algebra Appl.}, 439:1568--1589, 2013.

\bibitem{Hiai_fdiv_standard}
Fumio Hiai.
\newblock Quantum $f$-divergences in von {N}eumann algebras. {I}. {S}tandard
  $f$-divergences.
\newblock {\em J. Math. Phys.}, 59:102202, 2018.

\bibitem{Hiai_fdiv_max}
Fumio Hiai.
\newblock Quantum $f$-divergences in von {N}eumann algebras. {II}. {M}aximal
  $f$-divergences.
\newblock {\em J. Math. Phys.}, 50:012203, 2019.

\bibitem{Hiai_fdiv_Springer}
Fumio Hiai.
\newblock {\em Quantum f-{D}ivergences in von {N}eumann Algebras}.
\newblock Springer, 2021.

\bibitem{HiaiLim2020}
Fumio Hiai and Yongdo Lim.
\newblock Operator means of probability measures.
\newblock {\em Advances in Mathematics}, 365:107038, 2020.

\bibitem{sc_vN}
Fumio Hiai and Mil\'an Mosonyi.
\newblock Quantum {R\'enyi} divergences and the strong converse exponent of
  state discrimination in operator algebras.
\newblock {\em Annales Henri Poincar\'e}, 2022.
\newblock arXiv:2110.07320.

\bibitem{HP}
Fumio Hiai and D\'enes Petz.
\newblock The proper formula for relative entropy and its asymptotics in
  quantum probability.
\newblock {\em Communications in Mathematical Physics}, 143(1):99--114,
  December 1991.

\bibitem{HP_GT}
Fumio Hiai and D\'enes Petz.
\newblock The {Golden-Thompson} trace inequality is complemented.
\newblock {\em Linear Algebra Appl.}, 181:153--185, 1993.

\bibitem{HiaiUeadaWada2022}
Fumio Hiai, Yoshimichi Ueda, and Shuhei Wada.
\newblock Pusz-{W}oronowicz functional calculus and extended operator convex
  perspectives.
\newblock {\em Integr. Equ. Oper. Theory}, 94(1), 2022.

\bibitem{HircheTomamichel_integral}
Christoph Hirche and Marco Tomamichel.
\newblock Quantum {R}\'enyi and $f$-divergences from integral representations.
\newblock arXiv:2306.12343, 2023.

\bibitem{Jencova_rev16}
A.~Jen{\v c}ov{\'a}.
\newblock Preservation of a quantum {R}\'enyi relative entropy implies
  existence of a recovery map.
\newblock {\em J. Phys. A}, 50:085303, 2017.
\newblock arXiv:1604.02831.

\bibitem{Jencova_NCLp}
A.~Jen{\v c}ov{\'a}.
\newblock {R}\'enyi relative entropies and noncommutative $l_p$-spaces.
\newblock {\em Ann. Henri Poincar\'e}, 19:2513--2542, 2018.

\bibitem{Jencova_NCLpII}
A.~Jen{\v c}ov{\'a}.
\newblock {R}\'enyi relative entropies and noncommutative ${L}_p$-spaces {II}.
\newblock {\em Annales Henri Poincar\'e}, 22:3235--3254, 2021.
\newblock arXiv:1707.00047.

\bibitem{JP}
A.~Jen{\v c}ov{\'a} and D.~Petz.
\newblock Sufficiency in quantum statistical inference.
\newblock {\em Commun. Math. Phys.}, 263(1):259--276, 2006.

\bibitem{jensen2019asymptoticmajorization}
Asger~Kj{\ae}rulff Jensen.
\newblock Asymptotic majorization of finite probability distributions.
\newblock {\em IEEE Transactions on Information Theory}, 65(12):8131--8139,
  2019.
\newblock arXiv:1808.05157.

\bibitem{jensen2019asymptotic}
Asger~Kj{\ae}rulff Jensen and P{\'e}ter Vrana.
\newblock The asymptotic spectrum of {LOCC} transformations.
\newblock {\em IEEE Transactions on Information Theory}, 66(1):155--166,
  January 2019.
\newblock arXiv:1807.05130.

\bibitem{KimLee2015}
Sejong Kim and Hosoo Lee.
\newblock The power mean and the least squares mean of probability measures on
  the space of positive definite matrices.
\newblock {\em Linear Algebra and its Applications}, 465:325--346, 2015.

\bibitem{Klimesh2007}
M.~Klimesh.
\newblock Inequalities that collectively completely characterize the catalytic
  majorization relation.
\newblock arXiv:0709.3680, 2007.

\bibitem{Kosaki1982}
Hideki Kosaki.
\newblock Interpolation theory and the {Wigner-Yanase-Dyson-Lieb} concavity.
\newblock {\em Commun Math Phys}, 87:315--329, 1982.

\bibitem{Kosaki_ac}
Hideki Kosaki.
\newblock Remarks on absolute continuity of positive operators.
\newblock {\em International Journal of Mathematics}, 34(10), 2023.

\bibitem{KA}
F.~Kubo and T.~Ando.
\newblock Means of positive linear operators.
\newblock {\em Math. Ann.}, 246:205--224, 1980.

\bibitem{Lawson_Lim2011}
J.~Lawson and Y.~Lim.
\newblock Monotonic properties of the least squares mean.
\newblock {\em Math. Ann.}, 351:267--279, 2011.

\bibitem{LawsonLim2014}
J.~Lawson and Y.~Lim.
\newblock Karcher means and {K}archer equations of positive definite operators.
\newblock {\em Trans. Amer. Math. Soc. Series B}, 1:1--22, 2014.

\bibitem{LiYao_Reliability_2021}
Ke~Li and Yongsheng Yao.
\newblock Reliability function of quantum information decoupling via the
  sandwiched {R}\'enyi divergence.
\newblock arXiv:2111.06343, 2021.

\bibitem{LiYao_Channel_simulation}
Ke~Li and Yongsheng Yao.
\newblock Reliable simulation of quantum channels.
\newblock arXiv:2112.04475, 2021.

\bibitem{LiYao_EA_sc}
Ke~Li and Yongsheng Yao.
\newblock Strong converse exponent for entanglement-assisted communication.
\newblock {\em IEEE Transactions on Information Theory}, 2023.
\newblock arXiv:2209.00555.

\bibitem{LiYao2022}
Ke~Li and Yongsheng Yao.
\newblock Operational interpretation of the sandwiched {R}\'enyi divergence of
  order $1/2$ to $1$ as strong converse exponents.
\newblock {\em Commun. Math. Phys.}, 405(22), 2024.
\newblock arXiv:2209.00554.

\bibitem{LiYaoHayashi2023}
Ke~Li, Yongsheng Yao, and Masahito Hayashi.
\newblock Tight exponential analysis for smoothing the max-relative entropy and
  for quantum privacy amplification.
\newblock {\em IEEE Transactions on Information Theory}, 69(3):1680--1694,
  2023.

\bibitem{LT}
E.H. Lieb and W.~Thirring.
\newblock {\em Studies in mathematical physics}.
\newblock University Press, Princeton, 1976.

\bibitem{Lieb-Ruskai}
Elliott~H. Lieb and Mary~Beth Ruskai.
\newblock A fundamental property of quantum-mechanical entropy.
\newblock {\em Phys. Rev. Lett.}, 30:434, 1973.

\bibitem{Lim_Palfia2012}
Yongdo Lim and Mikl\'os P\'alfia.
\newblock Matrix power means and the {K}archer mean.
\newblock {\em Journal of Functional Analysis}, 262(4):1498--1514, February
  2012.

\bibitem{LinTomamichel15}
Mingyan~Simon Lin and Marco Tomamichel.
\newblock Investigating properties of a family of quantum renyi divergences.
\newblock {\em Quantum Information Processing}, 14(4):1501--1512, 2015.

\bibitem{Matsumoto_newfdiv}
K.~Matsumoto.
\newblock A new quantum version of $f$-divergence.
\newblock In {\em Nagoya Winter Workshop 2015: Reality and Measurement in
  Algebraic Quantum Theory}, pages 229--273, 2018.

\bibitem{MishraWildeNussbaum2023}
H.K. Mishra and M.M.~Wilde M.~Nussbaum.
\newblock On the optimal error exponents for classical and quantum
  antidistinguishability.
\newblock arXiv:2309.03723, 2023.

\bibitem{Moakher_matrixmean}
Maher Moakher.
\newblock A differential geometric approach to the geometric mean of symmetric
  positive-definite matrices.
\newblock {\em SIAM Journal on Matrix Analysis and Applications},
  26(3):735--747, 2005.

\bibitem{MH}
Mil\'an Mosonyi and Fumio Hiai.
\newblock On the quantum {R\'enyi} relative entropies and related capacity
  formulas.
\newblock {\em IEEE Transactions on Information Theory}, 57(4):2474--2487,
  April 2011.

\bibitem{MH-testdiv}
Mil\'an Mosonyi and Fumio Hiai.
\newblock Test-measured {R}\'enyi divergences.
\newblock {\em IEEE Transactions on Information Theory}, 69(2):1074--1092,
  2023.
\newblock arXiv:2201.05477.

\bibitem{Hiai_Mosonyi_cont_2023}
Mil\'an Mosonyi and Fumio Hiai.
\newblock Some continuity properties of quantum {R\'enyi} divergences.
\newblock {\em IEEE Transactions on Information Theory}, 70(4):2674--2700,
  2024.
\newblock arXiv:2209.00646.

\bibitem{MO}
Mil\'an Mosonyi and Tomohiro Ogawa.
\newblock Quantum hypothesis testing and the operational interpretation of the
  quantum {R\'enyi} relative entropies.
\newblock {\em Communications in Mathematical Physics}, 334(3):1617--1648,
  2015.
\newblock arXiv:1309.3228.

\bibitem{MO-cqconv}
Mil\'an Mosonyi and Tomohiro Ogawa.
\newblock Strong converse exponent for classical-quantum channel coding.
\newblock {\em Communications in Mathematical Physics}, 355(1):373--426, June
  2017.
\newblock arXiv:1409.3562.

\bibitem{MO-cqconv-cc}
Mil\'an Mosonyi and Tomohiro Ogawa.
\newblock Divergence radii and the strong converse exponent of
  classical-quantum channel coding with constant compositions.
\newblock {\em IEEE Transactions on Information Theory}, 67(3):1668--1698,
  2021.
\newblock arXiv:1811.10599.

\bibitem{Mu_Econometrica2020}
X.~Mu, L.~Pomatto, P.~Strack, and O.~Tamuz.
\newblock From {B}lackwell dominance in large samples to {R}\'enyi divergences
  and back again.
\newblock {\em Econometrica}, 89(1):475--506, 2021.

\bibitem{MHR}
Alexander M\"uller-Hermes and David Reeb.
\newblock Monotonicity of the quantum relative entropy under positive maps.
\newblock {\em Ann. Henri Poincar\'e}, 18:1777--1788, 2017.

\bibitem{Renyi_new}
Martin {M\"uller}-Lennert, Fr\'ed\'eric Dupuis, Oleg Szehr, Serge Fehr, and
  Marco Tomamichel.
\newblock On quantum {R\'enyi} entropies: A new generalization and some
  properties.
\newblock {\em Journal of Mathematical Physics}, 54(12):122203, December 2013.
\newblock arXiv:1306.3142.

\bibitem{Nagaoka}
Hiroshi Nagaoka.
\newblock The converse part of the theorem for quantum {Hoeffding} bound.
\newblock arXiv:quant-ph/0611289, November 2006.

\bibitem{NC}
Michael~A. Nielsen and Isaac~L. Chuang.
\newblock {\em Quantum Computation and Quantum Information}.
\newblock Cambridge University Press, 2000.

\bibitem{OP}
M.~Ohya and D.~Petz.
\newblock {\em Quantum Entropy and its Use}.
\newblock Springer, 1993.

\bibitem{Petz_QE_vN}
D\'enes Petz.
\newblock Quasi-entropies for states of a von {N}eumann algebra.
\newblock {\em Publ. {RIMS}, Kyoto Univ.}, 21:787--800, 1985.

\bibitem{Petz_Properties1986}
D\'enes Petz.
\newblock Properties of the relative entropy of sates of von neumann algebras.
\newblock {\em Acta Math.~Hung.}, 47(1-2):65--72, 1986.

\bibitem{P86}
D\'enes Petz.
\newblock Quasi-entropies for finite quantum systems.
\newblock {\em Reports in Mathematical Physics}, 23:57--65, 1986.

\bibitem{Petz_Temesi2005}
D\'{e}nes Petz and R\'{o}bert Temesi.
\newblock Means of positive numbers and matrices.
\newblock {\em SIAM Journal on Matrix Analysis and Applications},
  27(3):712--720, 2005.

\bibitem{PV_Hellinger}
J.~Pitrik and D.~Virosztek.
\newblock Quantum {H}ellinger distances revisited.
\newblock {\em Letters in Mathematical Physics}, 110:2039--2052, 2020.

\bibitem{Renyi}
Alfr\'ed R\'enyi.
\newblock On measures of entropy and information.
\newblock In {\em Proc.~4th Berkeley Sympos.~Math.~Statist.~and Prob.},
  volume~I, pages 547--561. Univ. California Press, Berkeley, California, 1961.

\bibitem{Strasser1985}
Helmut Strasser.
\newblock {\em Mathematical Theory of Statistics}.
\newblock Walter de Gruyter, Berlin, New York, 1985.

\bibitem{TomamichelBook}
M.~Tomamichel.
\newblock {\em Quantum Information Processing with Finite Resources}, volume~5
  of {\em Mathematical Foundations, SpringerBriefs in Math. Phys.}
\newblock Springer, 2016.

\bibitem{Tropp}
Joel~A. Tropp.
\newblock From joint convexity of quantum relative entropy to a concavity
  theorem of {L}ieb.
\newblock {\em Proceedings of the American Mathematical Society},
  140(5):1757--1760, 2011.

\bibitem{turgut2007catalytic}
Sadi Turgut.
\newblock Catalytic transformations for bipartite pure states.
\newblock {\em Journal of Physics A: Mathematical and Theoretical},
  40(40):12185, 2007.
\newblock arXiv:0707.0444.

\bibitem{Umegaki}
H.~Umegaki.
\newblock Conditional expectation in an operator algebra, {IV}: Entropy and
  information.
\newblock {\em K\=odai Math.~Sem.~Rep.}, 14:59--85, 1962.

\bibitem{WWY}
Mark~M. Wilde, Andreas Winter, and Dong Yang.
\newblock Strong converse for the classical capacity of entanglement-breaking
  and {Hadamard} channels via a sandwiched {R\'enyi} relative entropy.
\newblock {\em Communications in Mathematical Physics}, 331(2):593--622,
  October 2014.
\newblock arXiv:1306.1586.

\bibitem{Zhang2018}
Haonan Zhang.
\newblock From {W}igner-{Y}anase-{D}yson conjecture to {C}arlen-{F}rank-{L}ieb
  conjecture.
\newblock {\em Advances in Mathematics}, 365:107053, 2020.
\newblock arXiv:1811.01205.

\end{thebibliography}

\end{document}